%% file: OTJointPO.tex
\documentclass[11 pt]{article}

\usepackage[margin=1 in]{geometry}
\usepackage[inline]{enumitem}
\usepackage{comment, graphicx, cancel, dsfont, bbm, float, verbatim, multirow, pdflscape, hyperref, natbib, setspace, booktabs}
\usepackage{amsmath, amssymb, mathtools, mathrsfs, amsthm, thmtools, thm-restate, fancyvrb}
\VerbatimFootnotes

\graphicspath{{/Users/daniel/Documents/UCLA/Research/MissingDataBreakdownPoint/Paper/Figs/}}

\declaretheorem[numberwithin=section, style=remark]{remark}

\newcommand*{\thisdraft}{This draft: 15 November, 2023}
\date{\thisdraft} 

\hypersetup{
	colorlinks=true, 
	linkcolor=black, 
	urlcolor=blue, 
	citecolor=blue, 
	linktoc=all, 
}

\usepackage[framemethod=TikZ]{mdframed}
\mdfdefinestyle{MyFrame}{%
	linecolor=black,
	outerlinewidth=2pt,
	roundcorner=20pt,
	innertopmargin=\baselineskip,
	innerbottommargin=\baselineskip,
	innerrightmargin=20pt,
	innerleftmargin=20pt,
	backgroundcolor=gray!30!white}

\DeclareMathOperator*{\argmax}{arg\,max}
\DeclareMathOperator*{\argmin}{arg\,min}

\title{Estimating Functionals of the Joint Distribution of Potential Outcomes with Optimal Transport}

\author{Daniel Ober-Reynolds\thanks{\protect\linespread{1}\protect\selectfont Email: doberreynolds [at] gmail.com, website: \href{https://danieloberreynolds.com}{https://danieloberreynolds.com}. I want to thank my advisor Andres Santos for his continuous guidance and support. I also want to thank Denis Chetverikov and Jinyong Hahn for their helpful comments and suggestions.}}

\begin{document}
	\maketitle
	

	\input{./sections/OTJointPO_section_abstract.tex}
	
	\pagenumbering{gobble} 
	
	\clearpage
	
	\pagenumbering{arabic} 
	
	\input{./sections/OTJointPO_section_introduction.tex}

	\input{./sections/OTJointPO_section_setting.tex}

	\input{./sections/OTJointPO_section_optimaltransport.tex}

	\input{./sections/OTJointPO_section_identification.tex}

	\input{./sections/OTJointPO_section_estimators.tex}
	
	\input{./sections/OTJointPO_section_application.tex}

	\input{./sections/OTJointPO_section_extensions.tex}

	\input{./sections/OTJointPO_section_conclusion.tex}
	
	\newpage
	
	\bibliography{OTJointPO_bibliography}
	
	\appendix
	
	\newpage
	
	\input{./appendix/OTJointPO_appendix_identification.tex}
	
	\newpage
	
	\input{./appendix/OTJointPO_appendix_OTProperties.tex}

	\newpage
	
	\input{./appendix/OTJointPO_appendix_weak_convergence.tex}

	\newpage
	
	\input{./appendix/OTJointPO_appendix_inference.tex}

	\newpage
	
	\input{./appendix/OTJointPO_appendix_duality.tex}

	\newpage
	
	\input{./appendix/OTJointPO_appendix_differentiation.tex}

\end{document}

%% file: sections/OTJointPO_section_abstract.tex
\begin{abstract}
	Many causal parameters depend on a moment of the joint distribution of potential outcomes. 
	Such parameters are especially relevant in policy evaluation settings, where noncompliance is common and accommodated through the model of \cite{imbens1994late}.
	This paper shows that the sharp identified set for these parameters is an interval with endpoints characterized by the value of optimal transport problems. 
	Sample analogue estimators are proposed based on the dual problem of optimal transport.
	These estimators are $\sqrt{n}$-consistent and converge in distribution under mild assumptions. 
	Inference procedures based on the bootstrap are straightforward and computationally convenient.
	The ideas and estimators are demonstrated in an application revisiting the National Supported Work Demonstration job training program. 
	I find suggestive evidence that workers who would see below average earnings without treatment tend to see above average benefits from treatment.
\end{abstract}

\bigskip

\begin{singlespace}
	\begin{center}
		optimal transport 
	\end{center}
\end{singlespace}

%% file: sections/OTJointPO_section_introduction.tex
\section{Introduction}
\label{Section: introduction}

Researchers studying the causal effects of a binary treatment see an observation's treated or untreated outcome, but never both. 
As a result, the data identify the marginal distributions of each potential outcome, but not their joint distribution.
This ``fundamental problem of causal inference'' \citep{holland1986statistics} leaves parameters depending on the joint distribution partially identified.

In this paper I study a wide class of parameters that depend on a moment of the joint distribution of potential outcomes. 
My setting is the canonical potential outcomes framework with binary treatment, a binary instrument satisfying a monotonicity restriction, and finitely supported covariates \citep{imbens1994late, abadie2003semiparametric}.
In this setting, I show the sharp identified set for such parameters is an interval with endpoints characterized by the value of optimal transport problems. 
I propose sample analogue estimators based on the dual problem of optimal transport, which facilitates both computation and asymptotic analysis.
Through the functional delta method, I show these estimators converge in distribution allowing for straightforward inference procedures based on the bootstrap.

The proposed estimators are especially attractive due to their wide applicability and computational simplicity. 
The class of parameters under study is broad, including the correlation between potential outcomes, the probability of benefitting from treatment, and many more examples discussed in section \ref{Section: setting and parameter class}. 
As argued in \cite{heckman1997making}, such parameters are of particular interest to policymakers and economists carrying out econometric policy evaluation.
Noncompliance with the assigned treatment status is common in these settings.
Most studies accomodate noncompliance with the same framework adopted in this paper, and could make use of these estimators with no additional identifying assumptions.
Computing the estimator and constructing confidence sets entails nothing more challenging than solving linear programming problems, for which there are fast and efficient algorithms readily available.

This paper contributes to a large econometrics literature studying parameters of the joint distribution of potential outcomes. 
Many papers in this literature focus on a subset of the parameters considered here, especially the cumulative distribution function (cdf) or quantiles of treatment effects \citep{manski1997monotone, heckman1997making, firpo2007efficient, fan2010sharp, fan2012confidence, firpo2019partial, callaway2021bounds, frandsen2021partial}. 
This limited focus allows greater use of known analytical expressions when deriving sharp bounds, especially the famed Makarov bounds on the cdf and Fr\'echet-Hoeffding bounds on the joint distribution.
Several recent works develop methods applicable to broad parameters classes by employing procedures that do not require analytical expressions for the identified set. 
\cite{russell2021sharp} studies continuous functionals of the joint distribution of discrete potential outcomes, through a computationally intensive (sometimes infeasible) search over all permissible distributions of model primitives.
\cite{fan2023partial} study parameters identified through moment conditions in several incomplete data settings -- including potential outcomes -- by searching over an infinite dimensional space of smooth copulas. 
This paper occupies a middle ground: by focusing on parameters that depend on a scalar moment of the joint distribution and working with optimal transport, I obtain expressions for the bounds with tractable sample analogues. 
This approach allows consideration of a wide variety of parameters while maintaining computational tractability.

This paper also contributes to a growing literature on applications of optimal transport to econometrics; see \cite{galichon2017survey} for a recent survey. Several recent working papers utilize optimal transport for issues related to casual inference, including inverse propensity weighting \citep{dunipace2021optimal}, matching on covariates \citep{gunsilius2021matching}, and obtaining counterfactual distributions \citep{torous2021optimal}. In concurrent and highly complementary work, \cite{ji2023model} consider a very similar class of parameters to the present paper and also propose inference based on the dual problem of optimal transport. Their focus, accomodating non-discrete covariates without resorting to parametric models, leads to theory based on cross fitting and high-level assumptions on first stage estimators. The goal of the present paper is to provide simple, low-level conditions and computationally convenient estimators in the common case where covariates are discrete. This leads to theory based on Hadamard directional differentiability and the functional delta method quite distinct from that of \cite{ji2023model}.

The remainder of this paper is organized as follows. Section \ref{Section: setting and parameter class} formalizes the setting and introduces the class of parameters under study. Optimal transport is introduced in section \ref{Section: optimal transport}, and used in identification in section \ref{Section: identification}. Section \ref{Section: estimators} proposes the estimators and contains the asymptotic results.
Section \ref{Section: application} contains the application, showing suggestive evidence that the the National Supported Work Demonstration job training program was especially beneficial for workers who would otherwise see below average incomes. 
Section \ref{Section: extensions} discusses straightforward extensions, and section \ref{Section: conclusion} concludes.

%% file: sections/OTJointPO_section_setting.tex
\section{Setting and parameter class}
\label{Section: setting and parameter class}

\subsection{Setting}
\label{Section: setting and parameter class, subsection setting}

Consider a potential outcomes framework with binary treatment, a binary instrument, and finitely supported covariates (\cite{imbens1994late}, \cite{abadie2003semiparametric}). Let $Y$ denote the scalar, real-valued outcome of interest and $D \in \{0,1\}$ indicate treatment status. Further let $Y_1$ denote the potential outcome when treated and $Y_0$ the potential outcome when untreated. The observed outcome $Y$ is given by
\begin{equation}
	Y = D Y_1 + (1-D) Y_0. \label{Display: observed and potential outcomes}
\end{equation}
The difference in potential outcomes, $Y_1 - Y_0$, is called the treatment effect. 

The binary instrument is denoted $Z \in \{0,1\}$. Let $D_1$ denote the treatment status when $Z = 1$, and $D_0$ the treatment status when $Z = 0$. The observed treatment status $D$ is given by
\begin{equation}
	D = Z D_1 + (1-Z) D_0. \label{Display: observed and potential treatment status}
\end{equation}
It is assumed that the instrument itself does not affect the outcome.\footnote{\protect\linespread{1}\protect\selectfont
	One could hypothesize potential outcomes varying with the value of the instrument, i.e. $Y_{dz}$ for each $(d,z)$. The exposition here implicitly assumes \textit{instrument exclusion}, also known as the \textit{Stable Unit Treatment Value Assumption}: that $P(Y_{d1} = Y_{d0}) = 1$ for each $d$.
	
} 
Units with $1 = D_1 > D_0 = 0$ are known as \textit{compliers}. 

Assumption \ref{Assumption: setting} formalizes the setting.
\begin{restatable}[Setting]{assumption}{assumptionSetting}
	\label{Assumption: setting}
	\singlespacing
	
	$\{Y_i, D_i, Z_i, X_i\}_{i=1}^n$ is an i.i.d. sample with $(Y, D, Z, X) \sim P$,
	\begin{align}
		&Y \in \mathcal{Y} \subseteq \mathbb{R}, &&D \in \{0,1\}, &&Z \in \{0,1\}, &&X \in \mathcal{X} = \{x_1, \ldots, x_M\} \subseteq \mathbb{R}^{d_x} \label{Assumption: setting, support}
	\end{align}
	where $Y$, $D$, and $Z$ are related to $(Y_1, Y_0, D_1, D_0)$ through equations \eqref{Display: observed and potential outcomes} and \eqref{Display: observed and potential treatment status}, and the random vector $(Y_1, Y_0, D_1, D_0, Z, X)$ satisfies
	\begin{enumerate}[label=(\roman*), noitemsep]
		\item Instrument independence: $(Y_1, Y_0, D_1, D_0) \perp Z \mid X$, \label{Assumption: setting, instrument independence}
		\item Monotonicity: $P(D_1 \geq D_0) = 1$, \label{Assumption: setting, monotonicity}
		\item Existence of compliers: $P(D_1 > D_0, X = x) > 0$ for each $x$, and \label{Assumption: setting, existence of compliers}
		\item $P(X = x, Z = z) > 0$ for each $(x,z)$. \label{Assumption: setting, common support}
	\end{enumerate}
\end{restatable}

Assumption \ref{Assumption: setting} is essentially equivalent to assumption 2.1 in \cite{abadie2003semiparametric}, with the addition that covariates are finitely supported. 
Instrument independence is sometimes referred to as \textit{ignorability}, and satisfied in most randomized controlled trials, where $Z$ indicates being assigned to treatment. Monotonicity is typically a weak assumption in such settings.

It is worth emphasizing that this setting nests the case where treatment is exogenous. Specifically, when $D_1 = 1$ and $D_0 = 0$ (degenerately), every unit is a complier. In this case equation \eqref{Display: observed and potential treatment status} shows treatment status equals the instrument: $D = Z$. Instrument independence simplifies to $(Y_1, Y_0) \perp D \mid X$, and monotonicity is trivially satisfied.

\subsubsection{Distributions of compliers}
\label{Section: setting and parameter class, subsection setting, subsubsection distributions of compliers}

Interest focuses on the distribution of compliers. Such focus is especially policy relevant when ``the policy is the instrument'' i.e., the proposed change in policy is to assign $Z=1$ to all units. \cite{abadie2003semiparametric} shows that assumption \ref{Assumption: setting} suffices to identify the marginal distributions of $Y_1$ and $Y_0$ for the subpopulation of compliers.
\begin{restatable}[\cite{abadie2003semiparametric}]{lemma}{lemmaLATEIVMarginalDistributionIdentification}
	\label{Lemma: identification, LATE IV marginal distribution identification}
	\singlespacing
	
	Suppose assumption \ref{Assumption: setting} holds. Then the marginal distributions of $Y_d$ conditional on $D_1 > D_0$ and $X = x$, denoted $P_{d \mid x}$, are identified by	
	\begin{align}
		E_{P_{d \mid x}}[f(Y_d)] &\equiv E[f(Y_d) \mid D_1 > D_0, X = x] \notag \\
		&= \frac{E[f(Y) \mathbbm{1}\{D = d\} \mid Z = d, X = x] - E[f(Y) \mathbbm{1}\{D = d\} \mid Z = 1-d, X = x]}{P(D = d \mid Z = d, X = x) - P(D = d \mid Z = 1-d, X = x)} \label{Display: lemma, identification, LATE IV marginal distribution identification, conditional distributions}
	\end{align}
	for any integrable function $f$. Furthermore, the distribution of $X$ conditional on $D_1 > D_0$ is identified by
	\begin{align}
		s_x &\equiv P(X = x \mid D_1 > D_0) \notag \\
		&= \frac{\left[P(D = 1 \mid Z = 1, X = x) - P(D = 1 \mid Z = 0, X = x) \right]P(X = x) }{\sum_{x'} \left[P(D = 1 \mid Z = 1, X = x') - P(D = 1 \mid Z = 0, X = x') \right]P(X = x')} \label{Display: lemma, identification, LATE IV marginal distribution identification, conditional probability X = x}
	\end{align}
\end{restatable}

The joint distribution of potential outcomes is not identified. This is a result of the fundamental problem of causal inference: there is no unit where both $Y_1$ and $Y_0$ are observed, and as a result the joint distribution of $(Y_1, Y_0)$ is not identified for any subpopulation. Let $P_{1,0}$ denote the joint distribution of $(Y_1,Y_0)$ conditional on compliance, and $P_{1, 0 \mid x}$ denote the joint distribution conditional on compliance and $X = x$. These are related through the law of iterated expectations; for any function $c(y_1,y_0)$ with values in $\mathbb{R}$,
\begin{equation*}
	E_{P_{1,0}}[c(Y_1,Y_0)] = E[E[c(Y_1,Y_0) \mid D_1 > D_0, X] \mid D_1 > D_0] = \sum_x s_x E_{P_{1,0 \mid x}}[c(Y_1,Y_0)].
\end{equation*}
This relation can also be expressed as $P_{1,0} = \sum_x s_x P_{1,0 \mid x}$.

A joint distribution with marginals $P_{1 \mid x}$ and $P_{0 \mid x}$ is called a \textit{coupling} of $P_{1 \mid x}$ and $P_{0 \mid x}$. $P_{1,0 \mid x}$ is such a coupling, and is otherwise unrestricted by assumption \ref{Assumption: setting}.
Thus the identified set for $P_{1,0 \mid x}$ is the set of distributions $\pi_{1,0 \mid x}$ for $(Y_1,Y_0)$ with marginals $\pi_{1 \mid x} = P_{1 \mid x}$ and $\pi_{0 \mid x} = P_{0 \mid x}$, denoted
\begin{equation}
	\Pi(P_{1 \mid x}, P_{0 \mid x}) = \left\{\pi_{1,0 \mid x} \; : \; \pi_{1 \mid x} = P_{1 \mid x}, \; \pi_{0 \mid x} = P_{0 \mid x}\right\}. \label{Display: identified set for P10x}
\end{equation}
Moreover, the identified set for $P_{1,0}$ is $\left\{\pi_{1,0} = \sum_x s_x \pi_{1,0 \mid x} \; : \; \pi_{1,0 \mid x} \in \Pi(P_{1 \mid x}, P_{0 \mid x})\right\}$.

\subsection{Parameter class}
\label{Section: setting and parameter class, subsection parameter class}

The idea at the core of this paper is to bound a moment of the joint distribution of potential outcomes by optimization. Accordingly, the focus is on scalar parameters of the form
\begin{equation}
	\gamma = g(\theta, \eta) \label{Display: gamma = g(theta, eta)}
\end{equation}
where $g$ is a known function and $\theta = E_{P_{1,0}}[c(Y_1, Y_0)] \in \mathbb{R}$ is a scalar moment of the joint distribution of $(Y_1,Y_0)$ conditional on compliance. The function $c$ is known, and referred to as a ``cost function'' in connection with the optimal transport literature. This class of parameters is broad, as illustrated by the examples given below. In each of these examples $\eta$ is a finite collection of moments of the marginal distributions conditional on compliers: $\eta = (E_{P_1}[\eta_1(Y_1)], E_{P_0}[\eta_0(Y_0)]) \in \mathbb{R}^{K_1 + K_0}$. The formal results focus on this case, but could be generalized to allow $\eta$ to be other point identified nuisance parameters. 

The following conditions are stronger than necessary for identification of the sharp identified set of $\gamma$, but will be used when constructing and studying estimators. Assumption \ref{Assumption: cost function} places restrictions on the cost function to ensure optimal transport can be used characterize and estimate the sharp identified set for $\theta$.
\begin{restatable}[Cost function]{assumption}{assumptionCostFunction}
	\label{Assumption: cost function}
	\singlespacing
	
	Either
	\begin{enumerate}[label=(\roman*)]
		\item $c(y_1,y_0)$ is Lipschitz continuous and $\mathcal{Y}$ is compact, or  \label{Assumption: cost function, smooth costs}
		
		\item $c(y_1,y_0) = \mathbbm{1}\{y_1 - y_0 \leq \delta\}$ for a known $\delta \in \mathbb{R}$ and the cumulative distribution functions $F_{d \mid x}(y) = P(Y_d \leq y \mid D_1 > D_0, X = x)$ are continuous. \label{Assumption: cost function, CDF}
	\end{enumerate}
\end{restatable}

Assumption \ref{Assumption: cost function} covers every example listed below. Continuous cost functions $c$ are given a unified analysis, but for reasons discussed in section \ref{Section: optimal transport} discontinuous cost functions must be handled on a case-by-case basis. I focus on the leading case of interest in applications, $c(y_1,y_0) = \mathbbm{1}\{y_1 - y_0 \leq \delta\}$, corresponding to the cumulative distribution of treatment effects. The approach developed in this paper could likely be generalized to cover other discontinuous cost functions; for example, results in the appendix allow estimation of the sharp lower bound of $P((Y_1, Y_0) \in C)$ for any open, convex set $C \subseteq \mathbb{R}^2$. 

Assumption \ref{Assumption: cost function} \ref{Assumption: cost function, CDF} requires the cdfs $F_{d \mid x}$ be continuous. As discussed in section \ref{Section: identification}, this ensures the set being estimated is the sharp identified set for the parameter of interest. However, the estimation and inference results of section \ref{Section: estimators} hold \textit{regardless} of whether the cdfs are continuous or not; when the cdfs are not continuous, the estimand is a valid outer identified set.

Under assumptions \ref{Assumption: setting} and \ref{Assumption: cost function}, the sharp identified set for $\theta$ is an interval $[\theta^L, \theta^H]$. Assumption \ref{Assumption: parameter, function of moments} contains conditions on $g$ and $\eta$. 

\begin{restatable}[Function of moments]{assumption}{assumptionParameterFunctionOfMoments}
	\label{Assumption: parameter, function of moments} 
	\singlespacing
	
	The parameter is $\gamma = g(\theta, \eta) \in \mathbb{R}$, where
	\begin{align*}
		&\theta = E[c(Y_1, Y_0) \mid D_1 > D_0] \in \mathbb{R}, &&\eta = E
		\begin{bmatrix}
			\eta_1(Y_1), \eta_0(Y_0) \mid D_1 > D_0
		\end{bmatrix}
		\in \mathbb{R}^{K_1 + K_0}
	\end{align*}
	for known functions $g$, $c$, $\eta_1$ and $\eta_0$ such that 
	\begin{enumerate}[label=(\roman*)]
		\item $E[\lVert \eta_d(Y)\rVert^2] < \infty$ for $d = 1,0$,  \label{Assumption: parameter, function of moments, nuisance moments have finite variance}
		\item $g(\cdot, \eta)$ is continuous, and \label{Assumption: parameter, function of moments, g is continuous}
		\item the functions 
		\begin{align*}
			&g^L(t^L, t^H, e) = \min_{t \in [t^L, t^H]} g(t, e), &&g^H(t^L, t^H, e) = \max_{t \in [t^L, t^H]} g(t, e)
		\end{align*}
		are continuously differentiable at $(t^L, t^H, e) = (\theta^L, \theta^H, \eta)$. \label{Assumption: parameter, function of moments, sup and inf of g are differentiable}
	\end{enumerate}
\end{restatable}
Note that when $\theta$ itself is of interest, assumption \ref{Assumption: parameter, function of moments} is satisfied with $g(\theta, \eta) = \theta$. Assumption \ref{Assumption: parameter, function of moments} \ref{Assumption: parameter, function of moments, g is continuous} ensures the identified set for $\gamma$ is the interval $[\gamma^L, \gamma^H]$, and assumption \ref{Assumption: parameter, function of moments} \ref{Assumption: parameter, function of moments, sup and inf of g are differentiable} is used to apply the delta method. It is straightforward to show assumption \ref{Assumption: parameter, function of moments} \ref{Assumption: parameter, function of moments, sup and inf of g are differentiable} holds when $g$ is continuously differentiable in both arguments and $g(\cdot, \eta)$ is strictly increasing, as the latter condition implies $g^L(\theta^L, \theta^H, \eta) = g(\theta^L, \eta)$ and $g^H(\theta^L, \theta^H, \eta) = g(\theta^H, \eta)$ and the former condition implies they are continuously differentiable. This argument applies to every parameter listed below. When $g$ is differentiable but $g(\cdot, \eta)$ is not monotonic, it is often possible to use the implicit function theorem applied to first order conditions to derive sufficient conditions for the corresponding $\argmin$ and $\argmax$ to be differentiable, and thus for assumption \ref{Assumption: parameter, function of moments} \ref{Assumption: parameter, function of moments, sup and inf of g are differentiable} to hold. 

\subsubsection{Examples}
\label{Section: setting and parameter class, subsection parameter class, subsubsection examples}

The following examples are intended both to fix ideas and illustrate the broad scope of the parameter class described above.

\begin{restatable}[Summary statistics]{example}{exampleSummaryStatistics}
	\singlespacing
	Many summary statistics can be rewritten in the form $\gamma = g(\theta,\eta)$. For example, suppose interest is in the variance of treatment effects for compliers: $\gamma = \text{Var}(Y_1 - Y_0 \mid D_1 > D_0)$. This parameter can be rewritten as
	\begin{equation*}
		\gamma = \text{Var}(Y_1 - Y_0 \mid D_1 > D_0) = E_{P_{1,0}}[(Y_1 - Y_0)^2] - (E_{P_1}[Y_1] - E_{P_0}[Y_0])^2,
	\end{equation*}
	This parameter fits the form $\gamma = g(\theta, \eta)$ required of display \eqref{Display: gamma = g(theta, eta)}, with $\theta = E_{P_{1,0}}[(Y_1 - Y_0)^2]$, $\eta = (\eta^{(1)}, \eta^{(2)}) = (E_{P_1}[Y_1], E_{P_0}[Y_0])$, and $g(\theta, \eta) = \theta - (\eta^{(1)} - \eta^{(2)})^2$. The cost function $c(y_1,y_0) = (y_1 - y_0)^2$ satisfies assumption \ref{Assumption: cost function} \ref{Assumption: cost function, smooth costs} when $\mathcal{Y}$, the support of the outcome $Y$, is bounded.
	
	Similarly, suppose the researcher is interested in the correlation between $Y_1$ and $Y_0$ for compliers. Set $\gamma = \text{Corr}(Y_1, Y_0 \mid D_1 > D_0)$, which can be rewritten as
	\begin{equation*}
		\gamma = \text{Corr}(Y_1, Y_0 \mid D_1 > D_0) = \frac{E_{P_{1,0}}[Y_1 Y_0] - E_{P_1}[Y_1]E_{P_0}[Y_0]}{\sqrt{E_{P_1}[Y_1^2] - (E_{P_1}[Y_1])^2}\sqrt{E_{P_0}[Y_0^2] - (E_{P_0}[Y_0])^2}}
	\end{equation*}
	This parameter also fits the form $\gamma = g(\theta, \eta)$ in display \eqref{Display: gamma = g(theta, eta)}, with $\theta = E_{P_{1,0}}[Y_1 Y_0]$, $\eta = (\eta^{(1)}, \eta^{(2)}, \eta^{(3)}, \eta^{(4)}) = (E_{P_1}[Y_1], E_{P_1}[Y_1^2], E_{P_0}[Y_0], E_{P_0}[Y_0^2])$, and $g(\theta,\eta) = \frac{\theta - \eta^{(1)} \times \eta^{(3)}}{\sqrt{\eta^{(2)} - (\eta^{(1)})^2}\sqrt{\eta^{(4)} - (\eta^{(3)})^2}}$. The cost function $c(y_1,y_0) = y_1 y_0$ satisfies assumption \ref{Assumption: cost function} \ref{Assumption: cost function, smooth costs} when $\mathcal{Y}$ is bounded.
\end{restatable}

\begin{restatable}[Expected percent change]{example}{exampleExpectedPercentChange}
	\singlespacing
	The expected percent change in the outcome can be written as $100 \times E\left[\frac{Y_1 - Y_0}{Y_0} \mid D_1 > D_0\right] \%$. This is a unit-invariant causal parameter that is a natural summary measure when $Y_0$ exhibits considerably variation. For example, a treatment effect of $Y_1 - Y_0 = 5$ is typically of greater economic significance when the untreated outcome is small, say $ Y_0 = 10$, than when $Y_0 = 100$.
	
	The expected percent change is proportional to 
	\begin{equation*}
		\gamma = E\left[\frac{Y_1 - Y_0}{Y_0} \mid D_1 > D_0\right] = E_{P_{1,0}}\left[\frac{Y_1 - Y_0}{Y_0}\right],
	\end{equation*}
	which fits the form of display \eqref{Display: gamma = g(theta, eta)}, with $\gamma = \theta =  E_{P_{1,0}}\left[\frac{Y_1 - Y_0}{Y_0}\right]$. The cost function $c(y_1,y_0) = \frac{y_1 - y_0}{y_0}$ satisfies assumption \ref{Assumption: cost function} \ref{Assumption: cost function, smooth costs} when $\mathcal{Y}$ is bounded and bounded away from zero. 
\end{restatable}

\begin{restatable}[Equitable policies]{example}{exampleCovarianceTreatmentEffectUntreatedOutcome}
	\label{Example: equitable policies}
	\singlespacing
	
	Policy makers are often interested in whether a policy is equitable -- that is, whether the benefits are concentrated among those who would have undesirable outcomes without treatment. 
	
	One parameter that speaks to these concerns is the covariance between treatment effects and untreated outcomes among compliers: $\gamma = \text{Cov}(Y_1 - Y_0, Y_0\mid D_1 > D_0)$. Notice that $\gamma < 0$ implies those with below average $Y_0$ tend to see above average treatment effects. This parameter can be rewritten as
	\begin{equation*}
		\gamma = \text{Cov}(Y_1 - Y_0, Y_0\mid D_1 > D_0) = E_{P_{1,0}}[(Y_1 - Y_0) Y_0] - (E_{P_1}[Y_1] - E_{P_0}[Y_0])E_{P_0}[Y_0]
	\end{equation*}
	and fits the form $g(\theta, \eta)$ with $\theta = E_{P_{1,0}}[(Y_1 - Y_0)Y_0]$, $\eta = (E_{P_1}[Y_1], E_{P_0}[Y_0])$, and $g(\theta,\eta) = \theta - (\eta^{(1)} - \eta^{(2)})\eta^{(2)}$. The cost function $c(y_1,y_0) = (y_1 - y_0)y_0$ satisfies assumpion \ref{Assumption: cost function} \ref{Assumption: cost function, smooth costs} when $\mathcal{Y}$ is bounded. 
	
	Many related parameters share a sign with $\text{Cov}(Y_1 - Y_0, Y_0 \mid D_1 > D_0) $ and are also suitable for such an analysis. One such example is the OLS slope when regressing $Y_1 - Y_0$ on $Y_0$ and a constant: $\gamma = \frac{\text{Cov}(Y_1 - Y_0, Y_0 \mid D_1 > D_0)}{\text{Var}(Y_0 \mid D_1 > D_0)}$. This parameter can be rewritten as
	\begin{equation*}
		\gamma = \frac{\text{Cov}(Y_1 - Y_0, Y_0 \mid D_1 > D_0)}{\text{Var}(Y_0 \mid D_1 > D_0)} = \frac{E_{P_{1,0}}[(Y_1 - Y_0) Y_0] - (E_{P_1}[Y_1] - E_{P_0}[Y_0])E_{P_0}[Y_0]}{E_{P_0}[Y_0^2] - (E_{P_0}[Y_0])^2}
	\end{equation*}
	where $\theta = E_{P_{1,0}}[(Y_1 - Y_0)Y_0]$, $\eta = (E_{P_1}[Y_1], E_{P_0}[Y_0], E_{P_0}[Y_0^2])$, and $g(\theta,\eta) = \frac{\theta - (\eta^{(1)} - \eta^{(2)})\eta^{(2)}}{\eta^{(3)} - (\eta^{(2)})^2}$.
\end{restatable}

\begin{restatable}[Proportion that benefit]{example}{exampleShareBenefiting}
	\singlespacing
	The share of compliers benefiting from treatment, written 
	\begin{equation*}
		\gamma = P(Y_1 > Y_0 \mid D_1 > D_0),
	\end{equation*}
	is naturally of interest in applications where theory gives little indication whether the treatment will have a positive or negative effect. For example, \cite{allcott2020welfare} study the effect of deactivating facebook on subjective well-being. The authors find significant positive average effects of deactivation, but find substantial heterogeneity in follow-up interviews. 
	
	This parameter fits the form of display \eqref{Display: gamma = g(theta, eta)}, with $\gamma = \theta = E_{P_{1,0}}[\mathbbm{1}\{Y_1 - Y_0 \leq 0\}]$. The cost function $c(y_1,y_0) = \mathbbm{1}\{y_1 - y_0 \leq 0\}$ satisfies assumption \ref{Assumption: cost function} \ref{Assumption: cost function, CDF} if the cdfs $F_{d \mid x}(y)$ are continuous. 
	
	The share benefiting from treatment is also of particular interest when the intervention comes at a financial cost and the outcome of interest is a pecuniary return. Common examples include job training programs intended to increase a worker's income (e.g. the National Supported Work Demonstration studied in \cite{couch1992long}) or management practices intended to raise a firm's accounting profit (e.g. the employee referral program studied in \cite{friebel2023employee}). To illustrate, suppose the researcher observes $\{R_i, C_i, D_i, Z_i\}_{i=1}^n$, where $R$ is observed revenue and $C$ is the observed cost. These are related to treatment status $D \in \{0,1\}$, potential revenues $(R_1, R_0)$, and potential costs $(C_1, C_0)$ by
	\begin{align*}
		&R = D R_1 + (1-D)R_0, &&C = D C_1 + (1-D) C_0
	\end{align*}
	The observed profit, $Y = R - C$, is related to treatment status by 
	\begin{align*}
		Y = D \underbrace{(R_1 - C_1)}_{\coloneqq Y_1} + (1-D)\underbrace{(R_0 - C_0)}_{\coloneqq Y_0}
	\end{align*}
	The probability the change in revenue exceeds the change in cost is
	\begin{align*}
		P(R_1 - R_0 > C_1 - C_0 \mid D_1 > D_0) = P(Y_1 > Y_0 \mid D_1 > D_0)
	\end{align*}
\end{restatable}

\begin{restatable}[Quantiles]{example}{exampleQuantilesTreatmentEffects}
\singlespacing
\label{Example: pseudo quantiles}

Suppose the parameter of interest is any $q_\tau$ solving 
\begin{equation}
	P(Y_1 - Y_0 \leq q_\tau) = \tau \label{Display: pseudo quantile definition}
\end{equation}
This parameter has a similar interpretation to the $\tau$-th quantile.\footnote{The $\tau$-th quantile is usually defined as the unique  value $\tilde{q}_\tau = \inf\{y \; ; \; P(Y_1 - Y_0 \leq y) \geq \tau\}$. When the $\tau$ level set of the cumulative distribution function $P(Y_1 - Y_0 \leq \cdot)$ is nonempty, the $\tau$-th quantile has the interpretation that $100 \times \tau \%$ of the population has treatment effect less than or equal to $\tilde{q}_\tau$. Every $q_\tau$ solving \eqref{Display: pseudo quantile definition} has the same interpretation.} 
$q_\tau$ cannot be viewed as $\gamma = g(\theta,\eta)$. However, by viewing $\theta(\delta) = P(Y_1 - Y_0 \leq \delta \mid D_1 > D_0) = E_{P_{1,0}}[\mathbbm{1}\{Y_1 - Y_0 \leq \delta\}]$ as a function of $\delta$, the results below can be adapted to construct a confidence set for the identified set of this parameter as described in section \ref{Section: extensions, subsection quantiles}.
\end{restatable}

%% file: sections/OTJointPO_section_optimaltransport.tex
\section{Optimal Transport}
\label{Section: optimal transport}

This section defines and discusses optimal transport, which is used to characterize the identified set and construct estimators. 

Given any marginal distributions $P_1$ and $P_0$ and a ``cost function'' $c(y_1, y_0)$, the Monge-Kantorovich formulation of \textbf{optimal transport} is the problem of choosing a coupling $\pi \in \Pi(P_1,P_0)$ to minimize $E_\pi[c(Y_1,Y_0)]$:
\begin{equation}
	OT_c(P_1, P_0) = \inf_{\pi \in \Pi(P_1, P_0)} E_\pi[c(Y_1,Y_0)]. \label{Defn: optimal transport primal problem, main text}
\end{equation}
This minimization problem in \eqref{Defn: optimal transport primal problem, main text} is referred to as the \textbf{primal problem}, and will be used to characterize the identified set of $\theta$. 

The dual problem of optimal transport will be used to construct and analyze estimators. Let $\Phi_c$ denote the set of functions $\varphi(y_1)$ and $\psi(y_0)$ whose pointwise sum is less than $c(y_1,y_0)$:
\begin{equation}
	\Phi_c = \left\{(\varphi,\psi) \; ; \; \varphi(y_1) + \psi(y_0) \leq c(y_1, y_0)\right\}. \label{Defn: Phi_c, main text}
\end{equation}
The \textbf{dual problem} chooses a pair of functions in $\Phi_c$ to maximize the sum of the corresponding expectations:
\begin{equation}
	\sup_{(\varphi,\psi) \in \Phi_c} E_{P_1}[\varphi(Y_1)] + E_{P_0}[\psi(Y_0)]. \label{Display: optimal transport dual problem, main text}
\end{equation}
When the cost function is lower semicontinuous and bounded from below, the primal problem is attained and \textbf{strong duality} holds:
\begin{equation}
	OT_c(P_1, P_0) = \min_{\pi \in \Pi(P_1, P_0)} E_\pi[c(Y_1, Y_0)] = \sup_{(\varphi,\psi) \in \Phi_c} E_{P_1}[\varphi(Y_1)] + E_{P_0}[\psi(Y_0)]. \label{Display: strong duality, main text}
\end{equation}
The dual problem will be used to construct and analyze estimators. Indeed, the identification of $P_{d \mid x}$ in lemma \ref{Lemma: identification, LATE IV marginal distribution identification} suggests straightforward sample analogues estimating $E_{P_{d \mid x}}[f(Y_d)]$ for a given $f$, which makes it possible to form a sample analogue of the dual problem.  

Although it is clear how to form a sample analogue of the dual problem, it is not immediately clear how to analyze the resulting estimator. Fortunately, the dual problem can often be simplified by restricting the maximization problem to a smaller set of functions. Estimators based on this restricted dual problem can then be studied with empirical process techniques. 

The dual feasible set is restricted with the concept of $c$-concavity. Notice the dual problem's objective is monotonic, in the sense that $\varphi(y_1) \leq \tilde{\varphi}(y_1)$ for all $y_1$ implies
\begin{align*}
	E_{P_1}[\varphi(Y_1)] + E_{P_0}[\psi(Y_0)] \leq E_{P_1}[\tilde{\varphi}(Y_1)] + E_{P_0}[\psi(Y_0)].
\end{align*}
Increasing $\psi$ pointwise will also increase the dual objective. Speaking loosely, any function pair $(\varphi,\psi) \in \Phi_c$ for which the constraint $\varphi(y_1) + \psi(y_0) \leq c(y_1,y_0)$ is ``slack'' cannot be a solution to the dual problem and can therefore be ignored. This motivates the definition of the \textbf{$c$-transforms} of a function $\varphi$:
\begin{align*}
	&\varphi^c(y_0) = \inf_{y_1} \{c(y_1, y_0) - \varphi(y_1)\}, &&\varphi^{cc}(y_1) = \inf_{y_0} \{c(y_1,y_0) - \varphi^c(y_0)\}.
\end{align*}
For any pair of functions $(\varphi, \psi) \in \Phi_c$, these definitions imply $\psi(y_0) \leq \varphi^c(y_0)$, $\varphi(y_1) \leq \varphi^{cc}(y_1)$, and $\varphi^{cc}(y_1) + \varphi^c(y_0) \leq c(y_1,y_0)$. Further $c$-transformations are irrelevant because $(\varphi^{cc})^c = \varphi^c$, so a function $\varphi$ is called \textbf{$c$-concave} if $\varphi^{cc} = \varphi$. If the $c$-transforms are integrable, the dual problem can be restricted to $c$-concave conjugate pairs, $(\varphi^{cc}, \varphi^c)$. Furthermore, $c$-concave functions often ``inherit'' properties of the cost function $c$; for example, if $c$ is Lipschitz continuous then $\varphi^c$ and $\varphi^{cc}$ are Lipschitz continuous as well. These properties can be used to define sets of functions $\mathcal{F}_c$ and $\mathcal{F}_c^c$ (depending on the cost function $c$ but not on the distributions $P_1$, $P_0$) such that
\begin{equation}
	\sup_{(\varphi,\psi) \in \Phi_c} E_{P_1}[\varphi(Y_1)] + E_{P_0}[\psi(Y_0)] =  \sup_{(\varphi,\psi) \in \Phi_c \cap (\mathcal{F}_c \times \mathcal{F}_c^c)} E_{P_1}[\varphi(Y_1)] + E_{P_0}[\psi(Y_0)]. \label{Display: strong duality, with smaller feasible set}
\end{equation}

Two cases suffice for the parameters considered in this paper. When the cost function $c(y_1,y_0)$ is Lipschitz continuous and $\mathcal{Y}$ is compact, define
\begin{align}
	\mathcal{F}_c &= \left\{\varphi : \mathcal{Y} \rightarrow \mathbb{R} \; ; \;  -\lVert c \rVert_\infty \leq \varphi(y_1) \leq \lVert c \rVert_\infty, \; \lvert \varphi(y_1) - \varphi(y_1') \rvert \leq L \lvert y_1 - y_1'\rvert \right\} \label{Defn: F_c for smooth costs} \\
	\mathcal{F}_c^c &= \left\{\psi : \mathcal{Y} \rightarrow \mathbb{R} \; ; \; -2\lVert c \rVert_\infty \leq \psi(y_0) \leq 0, \; \lvert \psi(y_0) - \psi(y_0') \rvert \leq L \lvert y_0 - y_0'\rvert \right\} \label{Defn: F_c^c for smooth costs}
\end{align}
where $\lVert c \rVert_\infty = \sup_{(y_1,y_0)} \lvert c(y_1, y_0) \rvert$ and $L$ is the Lipschitz constant of $c$. When $c(y_1,y_0) = \mathbbm{1}\{(y_1,y_0) \in C\}$ for an open, convex set $C$, let
\begin{align}
	\mathcal{F}_c &= \left\{\varphi : \mathcal{Y} \rightarrow \mathbb{R} \; ; \; \varphi(y_1) = \mathbbm{1}\{y_1 \in I\} \text{ for some interval } I\right\} \label{Defn: F_c for indicator costs of convex C} \\
	\mathcal{F}_c^c &= \left\{\psi : \mathcal{Y} \rightarrow \mathbb{R} \; ; \; \psi(y_0) = -\mathbbm{1}\{y_0 \in I^c\} \text{ for some  interval } I\right\} \label{Defn: F_c^c for indicator costs of convex C}
\end{align}

Equation \eqref{Display: strong duality, with smaller feasible set} shows the optimal transport functional $OT_c(P_1, P_0)$ depends only on the values of $E_{P_1}[\varphi(Y_1)]$ and $ E_{P_0}[\psi(Y_0)]$ for $(\varphi,\psi) \in \mathcal{F}_c \times \mathcal{F}_c^c$. For any set $A$, let $\ell^\infty(A)$ denote the space of real-valued bounded functions defined on $A$, equipped with the supremum norm: $\ell^\infty(A) = \left\{f : A \rightarrow \mathbb{R} \; ; \; \lVert f \rVert_\infty = \sup_{a \in A} \lvert f(a) \rvert < \infty \right\}$.
Optimal transport can be viewed as the map $OT_c : \ell^\infty(\mathcal{F}_c) \times \ell^\infty(\mathcal{F}_c^c) \rightarrow \mathbb{R}$ given by
\begin{equation}
	OT_c(P_1, P_0) = \sup_{(\varphi,\psi) \in \Phi_c \cap (\mathcal{F}_c \times \mathcal{F}_c^c)} E_{P_1}[\varphi(Y_1)] + E_{P_0}[\psi(Y_0)]. \label{Defn: optimal transport dual formulation, main text optimal transport}
\end{equation}
This problem will be referred to as the \textbf{restricted dual problem}. Estimators formed with this map can be studied with empirical process techniques.

In summary, $OT_c(P_1, P_0)$ will be viewed as the functional in \eqref{Defn: optimal transport primal problem, main text} when considering identification, and as the functional given in \eqref{Defn: optimal transport dual formulation, main text optimal transport} when considering estimation. By ensuring $c$ is either Lipschitz continuous or the indicator of an open convex set, strong duality and $c$-concavity ensures these functionals agree on the space of probability distributions.

%% file: sections/OTJointPO_section_identification.tex
\section{Identification}
\label{Section: identification}

Recall the parameter of interest is $\gamma = g(\theta, \eta)$, where $\eta$ is a point identified parameter, $\theta = E_{P_{1,0}}[c(Y_1,Y_0)] \in \mathbb{R}$, and $g$ and $c$ are known functions. 

Begin by rewriting $\theta = E_{P_{1, 0}}[c(Y_1,Y_0)] = E[c(Y_1,Y_0) \mid D_1 > D_0]$ with the law of iterated expectations: 
\begin{equation*}
	\theta = E[E[c(Y_1,Y_0) \mid D_1 > D_0, X] \mid D_1 > D_0] = E[\theta_X \mid D_1 > D_0] = \sum_x s_x \theta_x
\end{equation*}
where $s_x = P(X = x \mid D_1 > D_0)$ and $\theta_x = E[c(Y_1,Y_0) \mid D_1 > D_0, X = x] = E_{P_{1,0 \mid x}}[c(Y_1,Y_0)]$. As noted in section \ref{Section: setting and parameter class, subsection setting, subsubsection distributions of compliers}, the identified set for $P_{1,0 \mid x}$ is the set of couplings of $P_{1\mid x}$ and $P_{0 \mid x}$, denoted $\Pi(P_{1 \mid x},P_{0 \mid x})$. Thus the identified set for $\theta_x$ is $\Theta_{I,x} = \left\{t \in \mathbb{R} \; : \; t = E_\pi[c(Y_1,Y_0)] \text{ for some } \pi \in \Pi(P_{1 \mid x},P_{0 \mid x})\right\}$. $\Pi(P_{1 \mid x},P_{0 \mid x})$ is convex, implying that $\Theta_{I,x}$ is an interval. Let $\theta_x^L$ and $\theta_x^H$ denote its lower and upper endpoint respectively. 

To ensure the restricted dual problem can be used for estimation, $\theta_x^L$ and $\theta_x^H$ are characterized through an optimal transport problem with a suitable cost function $c$. When assumption \ref{Assumption: cost function} \ref{Assumption: cost function, smooth costs} holds ($c(y_1,y_0)$ is Lipschitz continuous and $\mathcal{Y}$ is compact), define
\begin{align}
	&c_L(y_1, y_0) = c(y_1,y_0), &&c_H(y_1, y_0) = -c(y_1,y_0) \notag \\
	&\theta^L(P_{1 \mid x}, P_{0 \mid x}) = OT_{c_L}(P_{1 \mid x}, P_{0 \mid x}), &&\theta^H(P_{1 \mid x}, P_{0 \mid x}) = -OT_{c_H}(P_{1 \mid x}, P_{0 \mid x}). \label{Display: thetaL, thetaH when c is continuous}
\end{align}
Note that $\theta_x^L = \theta^L(P_{1 \mid x}, P_{0 \mid x})$ and $\theta_x^H = \theta^H(P_{1 \mid x}, P_{0 \mid x})$.

The cumulative distribution function of $Y_1 - Y_0$ corresponds to the cost function $c(y_1, y_0) = \mathbbm{1}\{y_1 - y_0 \leq \delta\}$, which is not lower semicontinuous. This challenge is circumvented by a small change in the cost function. When assumption \ref{Assumption: cost function} \ref{Assumption: cost function, CDF} holds (the cost function is $c(y_1,y_0) = \mathbbm{1}\{y_1 - y_0 \leq \delta\}$) define 
\begin{align}
	&c_L(y_1,y_0) = \mathbbm{1}\{y_1 - y_0 < \delta\}, &&c_H = \mathbbm{1}\{y_1 - y_0 > \delta\} \notag \\
	&\theta^L(P_{1 \mid x}, P_{0 \mid x}) = OT_{c_L}(P_{1 \mid x}, P_{0 \mid x}), &&\theta^H(P_{1 \mid x}, P_{0 \mid x}) = 1 - OT_{c_H}(P_{1 \mid x}, P_{0 \mid x}) \label{Display: thetaL, thetaH when c is for CDF}
\end{align}
It follows from definitions that $\theta_x^H = \theta^H(P_{1 \mid x}, P_{0 \mid x})$. Moreover, $c_L(y_1,y_0) \leq c(y_1,y_0)$ implies $\theta^L(P_{1 \mid x}, P_{0 \mid x})$ is a valid lower bound for $\theta_x$. It is sharp if $P_{1 \mid x}$, $P_{0\mid x}$ have continuous cumulative distribution functions, in which case $\theta_x^L =\theta^L(P_{1 \mid x}, P_{0 \mid x})$. It is worth emphasizing again that the estimation and inference results of section \ref{Section: estimators} hold \textit{regardless} of whether the cdfs are continuous or not; when the cdfs are not continuous, the estimand is a valid outer identified set.

Under assumptions \ref{Assumption: setting} and \ref{Assumption: cost function}, the identified set for $\theta = E_{P_{1,0}}[c(Y_1,Y_0)] = E[c(Y_1,Y_0) \mid D_1 > D_0]$ is the compact interval $[\theta^L, \theta^H]$ with endpoints
\begin{align*}
	&\theta^L = E[\theta_X^L \mid D_1 > D_0] = \sum_x s_x \theta_x^L, &&\theta^H = E[\theta_X^H \mid D_1 > D_0] = \sum_x s_x \theta_x^H
\end{align*}
Under assumptions \ref{Assumption: setting}, \ref{Assumption: cost function}, and \ref{Assumption: parameter, function of moments}, the identified set for $\gamma$ is $[\gamma^L, \gamma^H]$, with endpoints
\begin{align}
	&\gamma^L = g^L(\theta^L, \theta^H, \eta) = \inf_{t \in [\theta^L, \theta^H]} g(t, \eta), &&\gamma^H = g^H(\theta^L, \theta^H, \eta) = \sup_{t \in [\theta^L, \theta^H]} g(t, \eta) \label{Display: main idea, gamma bounds}
\end{align}
The following theorem summarizes the discussion above. Let $\theta^L(\cdot, \cdot)$ and $\theta^H(\cdot, \cdot)$ be given by \eqref{Display: thetaL, thetaH when c is continuous} or \eqref{Display: thetaL, thetaH when c is for CDF} depending on the cost function, and set
\begin{align}
	&\theta_x^L = \theta^L(P_{1 \mid x}, P_{0 \mid x}), &&\theta_x^H = \theta^H(P_{1 \mid x}, P_{0 \mid x}), \label{Defn: theta_x^L, theta_x^H formal definition} \\
	&\theta^L = \sum_x s_x \theta_x^L, &&\theta^H = \sum_x s_x \theta_x^H, \label{Defn: theta^L, theta^H formal definition} \\
	&\gamma^L = g^L(\theta^L, \theta^H, \eta), &&\gamma^H = g^H(\theta^L, \theta^H, \eta) \label{Defn: gamma^L, gamma^H formal definition}
\end{align}

\begin{restatable}[Identification of functions of moments]{theorem}{theoremIdentificationFunctionOfMoments}
	\label{Theorem: identification, function of moments}
	\singlespacing
	
	Suppose assumptions \ref{Assumption: setting}, \ref{Assumption: cost function}, and \ref{Assumption: parameter, function of moments} are satisfied. Then the sharp identified set for $\gamma$ is $[\gamma^L, \gamma^H]$.
	
\end{restatable}
All results are proven in the appendix.

It is worth pausing to consider the role of covariates. When covariates are available, ignoring them leads to wider bounds that are not sharp. Specifically, the marginal distributions $P_1$ and $P_0$ could be used to form a lower bound on $\theta$ with $\theta^L(P_1, P_0) = \inf_{\pi \in \Pi(P_1,P_0)} E_\pi[c_L(Y_1,Y_0)]$. This bound minimizes over the whole set $\Pi(P_1, P_0) = \left\{\pi_{1,0} \; ; \; \pi_1 = P_1, \pi_0 = P_0\right\}$, but the identified set for $P_{1,0}$ is the subset of $\Pi(P_1, P_0)$ given by $\left\{\pi_{1,0} = \sum_x s_x \pi_{1, 0 \mid x} \; ; \; \pi_{1,0 \mid x} \in \Pi(P_{1 \mid x}, P_{0 \mid x})\right\}$. The bound defined through equations \eqref{Defn: theta_x^L, theta_x^H formal definition} and \eqref{Defn: theta^L, theta^H formal definition} is found while enforcing the additional constraints that $\pi_{1,0 \mid x} \in \Pi(P_{1 \mid x}, P_{0 \mid x})$ for each $x$. These additional constraints imply $\theta^L(P_1, P_0) \leq \theta^L$, 
and similarly $\theta^H \leq \theta^H(P_1, P_0)$. 

Extreme cases illustrate when covariates are informative. If $X$ is independent of $(Y_1,Y_0)$ conditional on $D_1 > D_0$, then $P_{d \mid x} = P_d$ for each $x$, $\Pi(P_{1 \mid x}, P_{0 \mid x}) = \Pi(P_1, P_0)$, and the inequalities above hold as equalities. On the other hand, if $P_{d \mid x}$ is degenerate for either $d = 1$ or $d = 0$, then there is only one possible coupling of $P_{1 \mid x}$ and $P_{0 \mid x}$. Since $\Pi(P_{1 \mid x}, P_{0 \mid x})$ is a singleton, $\theta_x^L = \theta_x^H$ and $\theta_x = E[c(Y_1,Y_0) \mid D_1 > D_0, X = x]$ is point identified. If this occurs for all $x \in \mathcal{X}$, $\theta$ and $\gamma$ are point identified.

\begin{remark}[Makarov bounds]
	\singlespacing
	\label{Remark: Makarov bounds}
	
	The proof of theorem \ref{Theorem: identification, function of moments} given in the appendix uses properties of optimal transport to argue that under assumptions \ref{Assumption: setting} and \ref{Assumption: cost function} \ref{Assumption: cost function, CDF}, $[\theta^L, \theta^H]$ is the sharp identified set for $P(Y_1 - Y_0 \leq \delta \mid D_1 > D_0)$. Nonetheless, it is interesting to note that the proof shows
	\begin{align*}
		\theta_x^L &= OT_{c_L}(P_{1 \mid x}, P_{0 \mid x}) = \sup_y \{F_{1 \mid x}(y) - F_{0 \mid x}(y - \delta)\} \\
		\theta_x^H &= 1 - OT_{c_H}(P_{1 \mid x}, P_{0 \mid x}) = 1 - \sup_y\{F_{0 \mid x}(y - \delta) - F_{1 \mid x}(y)\} = 1 + \inf_y\{F_{1\mid x}(y) - F_{0 \mid x}(y - \delta)\}
	\end{align*}
	which are the Makarov bounds on $P(Y_1 - Y_0 \leq \delta \mid D_1 > D_0, X = x)$ studied in \cite{fan2010sharp}.
\end{remark}

\begin{remark}[Pointwise vs. uniformly sharp CDF bounds]
	\label{Remark: Pointwise vs. uniformly sharp CDF bounds}
	\singlespacing
	
	Under assumptions \ref{Assumption: setting} and \ref{Assumption: cost function} \ref{Assumption: cost function, CDF}, $[\theta^L, \theta^H]$ is the sharp identified set for $P(Y_1 - Y_0 \leq \delta \mid D_1 > D_0)$ at the \textit{point} $\delta$. Viewing these bounds as functions of $\delta$, $\theta^L(\delta)$ and $\theta^H(\delta)$ are not \textit{uniformly} sharp bounds for the cumulative distribution function $P(Y_1 - Y_0 \leq \delta \mid D_1 > D_0)$, in the sense that not every CDF $F(\cdot)$ satisfying $\theta^L(\delta) \leq F(\delta) \leq \theta^H(\delta)$ for all $\delta$ could be the CDF of $Y_1 - Y_0$. See \cite{firpo2019partial} for a detailed discussion of this point.
\end{remark}

%% file: sections/OTJointPO_section_estimators.tex
\section{Estimators}
\label{Section: estimators}

Sample analogues of the expressions identifying $P_{1 \mid x}$, $P_{0 \mid x}$, and $s_x$ in lemma \ref{Lemma: identification, LATE IV marginal distribution identification} provide convenient plug-in estimators of $\gamma^L$ and $\gamma^H$.

The following notation simplifies expressions for the sample analogues. Let $P$ denote the distribution of an observation $(Y, D, Z, X)$, and $f$ be a real-valued function. Use $P(f)$ to mean $E_P[f(Y, D, Z, X)]$. Similarly, let $P_{d \mid x}(f) = E_{P_{d \mid x}}[f(Y_d)] = E[f(Y_d) \mid D_1 > D_0, X = x]$. Let $\mathbb{P}_n$ denote the empirical distribution formed from the sample $\{Y_i, D_i, Z_i, X_i\}_{i=1}^n$, and $\mathbb{P}_n(f) = \frac{1}{n}\sum_{i=1}^n f(Y_i, D_i, Z_i, X_i)$. The following indicator function notation also simplifies expressions:
\begin{gather*}
	\mathbbm{1}_{d,x,z}(D,X,Z) = \mathbbm{1}\{D = d, X = x, Z = z\}, \\
	\mathbbm{1}_{x,z}(X,Z) = \mathbbm{1}\{X = x, Z = z\}, \quad \quad \quad \mathbbm{1}_x(X) = \mathbbm{1}\{X = x\}
\end{gather*}
For example, $P(D = d, X = x, Z = z)$ shortens to $P(\mathbbm{1}_{d,x,z})$, and $\frac{1}{n}\sum_{i=1}^n \mathbbm{1}\{D_i = 1, X_i = x, Z_i = 0\}$ to $\mathbb{P}_n(\mathbbm{1}_{1,x,0})$.

The probabilities $p_{d,x,z} = P(\mathbbm{1}_{d,x,z})$, $p_{x,z} = P(\mathbbm{1}_{x,z})$, and $p_x = P(\mathbbm{1}_x)$ are estimated with empirical analogues: 
\begin{align*}
	&\hat{p}_{d,x,z} = \mathbb{P}_n(\mathbbm{1}_{d,x,z}), &&\hat{p}_{x,z} = \mathbb{P}_n(\mathbbm{1}_{x,z}), &&\hat{p}_x = \mathbb{P}_n(\mathbbm{1}_x)
\end{align*}
In this notation, $s_x = P(X = x \mid D_1 > D_0)$ and its empirical analogue $\hat{s}_x$ are
\begin{align}
	&s_x = \frac{(p_{1,x,1}/p_{x,1} - p_{1,x,0}/p_{x,0})p_x}{\sum_{x'} (p_{1,x',1}/p_{x',1} - p_{1,x',0}/p_{x',0})p_x'}, &&\hat{s}_x = \frac{(\hat{p}_{1,x,1}/\hat{p}_{x,1} - \hat{p}_{1,x,0}/\hat{p}_{x,0})\hat{p}_x}{\sum_{x'} (\hat{p}_{1,x',1}/\hat{p}_{x',1} - \hat{p}_{1,x',0}/\hat{p}_{x',0})\hat{p}_{x'}} \label{Display: T_1, map to conditional distributions s_x, main text}
\end{align}
The maps $P_{d \mid x}$ and their empirical analogues are
\begin{align}
	P_{d \mid x}(f) &= \frac{P(\mathbbm{1}_{d,x,d} \times f)/p_{x,d} - P(\mathbbm{1}_{d,x,1-d} \times f)/p_{x,1-d}}{p_{d,x,d}/p_{x,d} - p_{d,x,1-d}/p_{x,1-d}} \notag \\
	\hat{P}_{d \mid x}(f) &= \frac{\mathbb{P}_n(\mathbbm{1}_{d,x,d} \times f)/\hat{p}_{x,d} - \mathbb{P}_n(\mathbbm{1}_{d,x,1-d} \times f)/\hat{p}_{x,1-d}}{\hat{p}_{d,x,d}/\hat{p}_{x,d} - \hat{p}_{d,x,1-d}/\hat{p}_{x,1-d}} \label{Display: T_1, map to conditional distributions, main text}
\end{align}
Under assumption \ref{Assumption: parameter, function of moments}, $\eta = (\eta_1, \eta_0) = (E_{P_1}[\eta_1(Y_1)], E_{P_0}[\eta_0(Y_0)])$. Each vector $\eta_d \in \mathbb{R}^{K_d}$ has coordinates $\eta_d^{(k)} = \sum_x s_x P_{d \mid x}(\eta_d^{(k)})$. Empirical analogues $\hat{\eta} = (\hat{\eta}_1, \hat{\eta}_0)$ are formed by $\hat{\eta}_d^{(k)} = \sum_x \hat{s}_x \hat{P}_{d \mid x}(\eta_d^{(k)})$. 

Computing $\hat{P}_{d \mid x}(f)$ for a known $f$ is straightforward:
\begin{align*}
	\hat{P}_{d \mid x}(f) &= \frac{\frac{1}{\hat{p}_{x,d}}\frac{1}{n}\sum_{i=1}^n \mathbbm{1}_{d,x,d}(D_i, X_i, Z_i) f(Y_i) - \frac{1}{\hat{p}_{x,1-d}}\frac{1}{n}\sum_{i=1}^n \mathbbm{1}_{d,x,1-d}(D_i,X_i,Z_i)f(Y_i)}{\hat{p}_{d,x,d}/\hat{p}_{x,d} - \hat{p}_{d,x,1-d}/\hat{p}_{x,1-d}} \\
	&= \sum_{i=1}^n \omega_{d,x,i} \times f_i
\end{align*}
where $f_i = f(Y_i)$ and the weights $\omega_{d,x,i}$ can be computed directly from data:
\begin{equation}
	\omega_{d,x,i} = \frac{1}{n} \times \frac{\mathbbm{1}_{d,x,d}(D_i, X_i, Z_i)/\hat{p}_{x,d} - \mathbbm{1}_{d,x,1-d}(D_i, X_i, Z_i)/\hat{p}_{x,1-d}}{\hat{p}_{d,x,d}/\hat{p}_{x,d} - \hat{p}_{d,x,1-d}/\hat{p}_{x,1-d}} \label{Display: weights to compute Pdx(f)}
\end{equation}

Sample analogue estimators of $\gamma^L$ and $\gamma^H$ are based on equations \eqref{Display: thetaL, thetaH when c is continuous}, \eqref{Display: thetaL, thetaH when c is for CDF}, \eqref{Defn: theta_x^L, theta_x^H formal definition}, \eqref{Defn: theta^L, theta^H formal definition}, and \eqref{Defn: gamma^L, gamma^H formal definition}. These expressions involve the optimal transport functional $OT_c(P_{1 \mid x}, P_{0 \mid x})$. The sample analogue of the simplified dual problem discussed in section \ref{Section: optimal transport} is written
\begin{equation}
	OT_c(\hat{P}_{1 \mid x}, \hat{P}_{0 \mid x}) = \sup_{(\varphi,\psi) \in \Phi_c \cap (\mathcal{F}_c \times \mathcal{F}_c^c)} \hat{P}_{1 \mid x}(\varphi) + \hat{P}_{0 \mid x}(\psi) \label{Defn: optimal transport dual formulation, main text estimators}
\end{equation}
Here $\mathcal{F}_c$, $\mathcal{F}_c^c$, and the functions $\theta^L(\cdot)$, $\theta^H(\cdot)$ are defined according to the cost function:
\begin{enumerate}[label=(\roman*)]
	\item When assumption \ref{Assumption: cost function} \ref{Assumption: cost function, smooth costs} holds (the cost function $c(y_1,y_0)$ is Lipschitz continuous and $\mathcal{Y}$ is compact), $\mathcal{F}_c$ and $\mathcal{F}_c^c$ are given by: 
	\begin{align*}
		\mathcal{F}_c &= \left\{\varphi : \mathcal{Y} \rightarrow \mathbb{R} \; ; \;  -\lVert c \rVert_\infty \leq \varphi(y_1) \leq \lVert c \rVert_\infty, \; \lvert \varphi(y_1) - \varphi(y_1') \rvert \leq L \lvert y_1 - y_1'\rvert \right\} \\
		\mathcal{F}_c^c &= \left\{\psi : \mathcal{Y} \rightarrow \mathbb{R} \; ; \; -2\lVert c \rVert_\infty \leq \psi(y_0) \leq 0, \; \lvert \psi(y_0) - \psi(y_0') \rvert \leq L \lvert y_0 - y_0'\rvert \right\}
	\end{align*}
	and $\theta^L(\hat{P}_{1 \mid x}, \hat{P}_{0 \mid x})$, $\theta^H(\hat{P}_{1 \mid x}, \hat{P}_{0 \mid x})$ are analogues of equation \eqref{Display: thetaL, thetaH when c is continuous}:
	\begin{align*}
		&c_L(y_1, y_0) = c(y_1,y_0), &&c_H(y_1, y_0) = -c(y_1,y_0) \\
		&\theta^L(\hat{P}_{1 \mid x}, \hat{P}_{0 \mid x}) = OT_{c_L}(\hat{P}_{1 \mid x}, \hat{P}_{0 \mid x}), &&\theta^H(\hat{P}_{1 \mid x}, \hat{P}_{0 \mid x}) = -OT_{c_H}(\hat{P}_{1 \mid x}, \hat{P}_{0 \mid x}). 
	\end{align*}
	
	\item When assumption \ref{Assumption: cost function} \ref{Assumption: cost function, CDF} holds (the cost function is $c(y_1, y_0) = \mathbbm{1}\{y_1 - y_0 \leq \delta\}$), $\mathcal{F}_c$ and $\mathcal{F}_c^c$  are given by: 
	\begin{align*}
		\mathcal{F}_c &= \left\{\varphi : \mathcal{Y} \rightarrow \mathbb{R} \; ; \; \varphi(y_1) = \mathbbm{1}\{y_1 \in I\} \text{ for some interval } I\right\} \\
		\mathcal{F}_c^c &= \left\{\psi : \mathcal{Y} \rightarrow \mathbb{R} \; ; \; \psi(y_0) = -\mathbbm{1}\{y_0 \in I^c\} \text{ for some  interval } I\right\}
	\end{align*}
	and $\theta^L(\hat{P}_{1 \mid x}, \hat{P}_{0 \mid x})$, $\theta^H(\hat{P}_{1 \mid x}, \hat{P}_{0 \mid x})$ are analogues of equation \eqref{Display: thetaL, thetaH when c is for CDF}:
	\begin{align*}
		&c_L(y_1,y_0) = \mathbbm{1}\{y_1 - y_0 < \delta\}, &&c_H = \mathbbm{1}\{y_1 - y_0 > \delta\} \\
		&\theta^L(\hat{P}_{1 \mid x}, \hat{P}_{0 \mid x}) = OT_{c_L}(\hat{P}_{1 \mid x}, \hat{P}_{0 \mid x}), &&\theta^H(\hat{P}_{1 \mid x}, \hat{P}_{0 \mid x}) = 1 - OT_{c_H}(\hat{P}_{1 \mid x}, \hat{P}_{0 \mid x}) 
	\end{align*}
\end{enumerate}
The sample analogue estimators are given by
\begin{align}
	&\hat{\theta}_x^L = \theta^L(\hat{P}_{1 \mid x}, \hat{P}_{0 \mid x}), &&\hat{\theta}_x^H = \theta^H(\hat{P}_{1 \mid x}, \hat{P}_{0 \mid x}), \label{Defn: theta_x^L, theta_x^H estimators} \\
	&\hat{\theta}^L = \sum_x \hat{s}_x \hat{\theta}_x^L, &&\hat{\theta}^H = \sum_x \hat{s}_x \hat{\theta}_x^H, \label{Defn: theta^L, theta^H estimators} \\
	&\hat{\gamma}^L = g^L(\hat{\theta}^L, \hat{\theta}^H, \hat{\eta}), &&\hat{\gamma}^H = g^H(\hat{\theta}^L, \hat{\theta}^H, \hat{\eta}) \label{Defn: gamma^L, gamma^H estimators}
\end{align}

The optimization problems in $\theta^L(\hat{P}_{1 \mid x}, \hat{P}_{0 \mid x})$ and $\theta^H(\hat{P}_{1 \mid x}, \hat{P}_{0 \mid x})$ are especially straightforward when treatment is exogenous. Recall the claim of equation \eqref{Display: strong duality, with smaller feasible set}: the supremum of $P_{1 \mid x}(\varphi) + P_{0 \mid x}(\psi)$ over the larger set $\Phi_c$ is the same value when restricted to $\Phi_c \cap (\mathcal{F}_c \times \mathcal{F}_c^c)$. The argument behind this claim uses monotonicity of the maps $P_{d \mid x}$. When treatment is exogenous, $\hat{P}_{d \mid x}$ corresponds to a probability distribution and is therefore also monotonic. Thus the claim holds replacing $P_{d \mid x}$ with $\hat{P}_{d \mid x}$, implying the function classes $\mathcal{F}_c$ and $\mathcal{F}_c^c$ can be ignored in computation:
\begin{align}
	OT_c(\hat{P}_{1 \mid x}, \hat{P}_{0 \mid x}) &= \sup_{(\varphi,\psi) \in \Phi_c \cap (\mathcal{F}_c \times \mathcal{F}_c^c)} \hat{P}_{1 \mid x}(\varphi) + \hat{P}_{0 \mid x}(\psi) = \sup_{(\varphi,\psi) \in \Phi_c} \hat{P}_{1 \mid x}(\varphi) + \hat{P}_{0 \mid x}(\psi) \notag \\
	&= \sup_{\{\varphi_i, \psi_j\}_{i,j}} \sum_{i=1}^n \omega_{1, x,i} \varphi_i + \sum_{j=1}^n \omega_{0, x, j} \psi_j \label{Display: computation with exogenous treatment} \\
	&\hspace{1 cm} \text{s.t. } \varphi_i + \psi_j \leq c(Y_i, Y_j) \text{ for all } 1 \leq i, j \leq n \notag
\end{align}
the final problem in this display is a linear programming problem with $2n$ choice variables and $n^2$ constraints, and can be further simplified by removing choice variables (and the corresponding constraints) whose weights $\omega_{d,x,i}$ equal zero. Many weights do equal zero, as only observations with $X_i = x$ correspond to nonzero weights.

When there is noncompliance in the sample, $\hat{P}_{d \mid x}$ does not correspond to a probability distribution. This is easily seen by noting that for observations $i$ where $Z_i$ differs from $D_i$, the weight $\omega_{d,x,i}$ defined in \eqref{Display: weights to compute Pdx(f)} is negative. Nonetheless, it remains computationally tractable to search over $\Phi_c \cap (\mathcal{F}_c \times \mathcal{F}_c^c)$. For example, when the cost function is continuous $OT_c(\hat{P}_{1 \mid x}, \hat{P}_{0 \mid x})$ remains a linear programming problem, with additional linear constraints enforcing $\lvert \varphi_i + \psi_j \rvert \leq L\lvert Y_i - Y_j\rvert$, $-\lVert c \rVert_\infty \leq \varphi_i \leq \lVert c \rVert_\infty$, and $-2\lVert c \rVert_\infty \leq \psi_j \leq 0$.

\subsection{Asymptotic analysis}

The estimators proposed above are especially attractive because they are a (Hadamard directionally) differentiable map of the empirical distribution. Specifically, there exists a collection of functions $\mathcal{F}$ and a map $T : \ell^\infty(\mathcal{F}) \rightarrow \mathbb{R}^2$ described by equations \eqref{Display: T_1, map to conditional distributions s_x, main text}, \eqref{Display: T_1, map to conditional distributions, main text}, \eqref{Defn: theta_x^L, theta_x^H estimators}, \eqref{Defn: theta^L, theta^H estimators}, and \eqref{Defn: gamma^L, gamma^H estimators} such that 
\begin{align*}
	&(\hat{\gamma}^L, \hat{\gamma}^H) = T(\mathbb{P}_n), &&(\gamma^L, \gamma^H) = T(P)
\end{align*}
The set $\mathcal{F}$ consists of the functions in $\mathcal{F}_c$, $\mathcal{F}_c^c$, and the coordinate functions defining $\eta$, multiplied by various indicator functions. It is formally defined in appendix \ref{Appendix: weak convergence}. Under assumption \ref{Assumption: setting}, \ref{Assumption: cost function}, and \ref{Assumption: parameter, function of moments}, $\mathcal{F}$ is a Donsker set and $T(\cdot)$ is continuous at $P$, which implies the esimators are consistent:
\begin{equation}
	(\hat{\gamma}^L, \hat{\gamma}^H) = T(\mathbb{P}_n) \overset{p}{\rightarrow} T(P) = (\gamma^L, \gamma^H) \label{Display: consistency, main text}
\end{equation}

\subsubsection{Weak convergence}
\label{Section: estimators, subsection weak convergence} 

The map $T(\cdot)$ is not only continuous under assumptions \ref{Assumption: setting}, \ref{Assumption: cost function}, and \ref{Assumption: parameter, function of moments}, but Hadamard directionally differentiable. An application of the functional delta method gives the conclusion $\sqrt{n}((\hat{\gamma}^L, \hat{\gamma}^H) - (\gamma^L, \gamma^H))$ converges in distribution, a result stated formally in theorem \ref{Theorem: weak convergence, weak convergence of estimators} below. 

In order to build hypothesis tests or construct confidence intervals based on the asymptotic distribution of $\sqrt{n}((\hat{\gamma}^L, \hat{\gamma}^H) - (\gamma^L, \gamma^H))$, one must be able to estimate the asymptotic distribution. This is possible under assumptions \ref{Assumption: setting}, \ref{Assumption: cost function}, and \ref{Assumption: parameter, function of moments}, but involves a more complex procedure described in section \ref{Section: estimators, subsection inference, subsubsection consistent alternative}. Under an additional assumption, a straightforward bootstrap will do.

For each instance of the restricted dual problem used in defining $T(\cdot)$, the set of maximizers 
\begin{equation}
	\Psi_c(P_{1 \mid x}, P_{0 \mid x}) = \argmax_{(\varphi,\psi) \in \Phi_c \cap (\mathcal{F}_c \times \mathcal{F}_c^c)} P_{1 \mid x}(\varphi) + P_{0 \mid x}(\psi)
\end{equation}
is nonempty. If the solutions are suitably unique for each instance, the map $T(\cdot)$ is fully Hadamard differentiable at $P$ and a straightforward bootstrap will consistently estimate the asymptotic distribution. 

Assumption \ref{Assumption: full differentiability} states this high-level uniqueness condition, while the following lemma \ref{Lemma: weak convergence, simple sufficient conditions for full differentiability} gives low-level sufficient conditions for it to hold. Let $\mathcal{Y}_{d,x}$ be the support of $Y$ conditional on $D = d$ and $X = x$, and $\mathbbm{1}_{\mathcal{Y}_{d,x}}(y) = \mathbbm{1}\{y \in \mathcal{Y}_{d,x}\}$ be the indicator function for this set.
\begin{restatable}[]{assumption}{assumptionFullDifferentiability}
	\label{Assumption: full differentiability}
	\singlespacing
	
	For each $x \in \mathcal{X}$, each $c \in \{c_L, c_H\}$, and any $(\varphi_1,\psi_1), (\varphi_2,\psi_2) \in \Psi_c(P_{1 \mid x},P_{0 \mid x})$, there exists $s \in \mathbb{R}$ such that 
	\begin{align*}
		&\mathbbm{1}_{\mathcal{Y}_{1,x}} \times \varphi_1 = \mathbbm{1}_{\mathcal{Y}_{1,x}} \times (\varphi_2 + s), \; P\text{-a.s.} &&\text{ and } &&\mathbbm{1}_{\mathcal{Y}_{0,x}} \times \psi_1 = \mathbbm{1}_{\mathcal{Y}_{0,x}} \times (\psi_2 - s), \; P\text{-a.s.}
	\end{align*}
	
\end{restatable}

\begin{restatable}[]{lemma}{lemmaSimpleConditionsForFullDifferentiability}
	\label{Lemma: weak convergence, simple sufficient conditions for full differentiability}
	\singlespacing
	
	Suppose that
	\begin{enumerate}[label=(\roman*)]
		\item assumption \ref{Assumption: cost function} \ref{Assumption: cost function, smooth costs} holds, with cost function $c(y_1,y_0)$ that is continuously differentiable, and \label{Assumption: full differentiability lemma, cost is continuously differentiable}
		\item for each $(d,x)$, the support of $P_{d \mid x}$ is $\mathcal{Y}_{d, x}$, which is a bounded interval. \label{Assumption: full differentiability lemma, support of compliers}
	\end{enumerate}
	Then assumption \ref{Assumption: full differentiability} holds.
\end{restatable}

When treatment is exogenous, condition \ref{Assumption: full differentiability lemma, support of compliers} of lemma \ref{Lemma: weak convergence, simple sufficient conditions for full differentiability} simplifies to the assumption that the distribution of $Y_d \mid X = x$ has bounded support $[y_{d,x}^\ell, y_{d,x}^u]$. In general, this condition requires the support of $Y_d$ for the subpopulation of compliers with covariate value $x$ is a bounded interval that contains the support of the relevant subpopulation of non-compliers. Specifically, the support of $Y_1$ for compliers is a bounded interval containing the support of $Y_1$ for always-takers, and the support of $Y_0$ for compliers is a bounded interval containing the support of $Y_0$ for never-takers.

Assumption \ref{Assumption: full differentiability} can hold even when the conditions of lemma \ref{Lemma: weak convergence, simple sufficient conditions for full differentiability} do not. For example, when interest is in the cumulative distribution function and assumption \ref{Assumption: cost function} \ref{Assumption: cost function, CDF} is satisfied, the dual problem is essentially optimizing over the difference of CDFs (see remark \ref{Remark: Makarov bounds}). Although the cost functions are not continuously differentiable, it is still plausible for this optimization problem to have a unique solution in well-behaved cases. For further discussion of uniqueness of the dual solutions of optimal transport, see \cite{staudt2022uniqueness}. 

The following theorem gives the main weak convergence result.
\begin{restatable}[]{theorem}{theoremWeakConvergenceOfEstimators}
	\label{Theorem: weak convergence, weak convergence of estimators}
	\singlespacing
	
	Suppose assumptions \ref{Assumption: setting}, \ref{Assumption: cost function}, and \ref{Assumption: parameter, function of moments} hold, and let $\mathbb{G}$ be the weak limit of $\sqrt{n}(\mathbb{P}_n - P)$ in $\ell^\infty(\mathcal{F})$. Then $T$ is Hadamard directionally differentiable at $P$ tangentially to the support of $\mathbb{G}$, and
	\begin{align*}
		\sqrt{n}((\hat{\gamma}^L, \hat{\gamma}^H) - (\gamma^L, \gamma^H)) = \sqrt{n}(T(\mathbb{P}_n) - T(P)) \overset{L}{\rightarrow} T_P'(\mathbb{G})
	\end{align*}
	If assumption \ref{Assumption: full differentiability} also holds, then $T_P'$ is linear on the support of $\mathbb{G}$ and $T_P'(\mathbb{G})$ is bivariate normal.
\end{restatable}

\subsection{Inference}
\label{Section: estimators, subsection inference}

To make use of the weak convergence result of theorem \ref{Theorem: weak convergence, weak convergence of estimators} for inference, this section develops methods of estimating the law of $T_P'(\mathbb{G})$ by utilizing the bootstrap. The ``exchangeable bootstrap'' procedures discussed in \cite{vaart1997weak} are computationally convenient for reasons discussed below. These procedures define a new map $\mathbb{P}_n^* \in \ell^\infty(\mathcal{F})$ pointwise with 
\begin{equation}
	\mathbb{P}_n^*(f) = \frac{1}{n}\sum_{i=1}^n W_i f(Y_i, D_i, Z_i, X_i) \label{Defn: exchangeable bootstrap}
\end{equation}
for nonnegative random variables $\{W_i\}_{i=1}^n$ independent of the data $\{Y_i, D_i, Z_i, X_i\}_{i=1}^n$, and satisfying technical conditions omitted here. I focus on two notable examples, the nonparametric bootstrap of \cite{efron1979bootstrap} and the ``Bayesian'' bootstrap of \cite{rubin1981bayesian}. Either bootstrap can be used to estimate the asymptotic distribution. The Bayesian bootstrap may be preferable in small samples for reasons discussed below.
\begin{restatable}[Nonparametric bootstrap]{definition}{definitionNonparametricBootstrap}
	\label{Definition: exchangeable bootstrap, nonparametric bootstrap}
	\singlespacing
	
	Let $(W_1, \ldots, W_n) \sim \text{Multinomial}(n, (1/n, \ldots, 1/n))$ be independent of the data $\{Y_i, D_i, Z_i, X_i\}_{i=1}^n$. Define $\mathbb{P}_n^* \in \ell^\infty(\mathcal{F})$ pointwise with \eqref{Defn: exchangeable bootstrap}.
\end{restatable}

\begin{restatable}[Bayesian bootstrap]{definition}{definitionBayesianBootstrap}
	\label{Definition: exchangeable bootstrap, bayesian bootstrap}
	\singlespacing
	
	Let $\{\xi_i\}_{i=1}^n$ be i.i.d. exponentially distributed random variables with mean $1$, independent of the data $\{Y_i, D_i, Z_i, X_i\}_{i=1}^n$. Set $W_i = \xi_i/(n^{-1}\sum_{i=1}^n \xi_i)$, and define $\mathbb{P}_n^* \in \ell^\infty(\mathcal{F})$ pointwise with \eqref{Defn: exchangeable bootstrap}.
\end{restatable}

The map $\mathbb{P}_n^*$ in \eqref{Defn: exchangeable bootstrap} can be used to compute $(\hat{\gamma}^{L*}, \hat{\gamma}^{H*}) = T(\mathbb{P}_n^*)$ in much the same way that $T(\mathbb{P}_n)$ is computed. Specifically, bootstrap analogues of $\hat{p}_{d,x,z}$, $\hat{p}_{x,z}$, and $\hat{p}_x$ are given by
\begin{align*}
	&\hat{p}_{d,x,z}^* = \frac{1}{n}\sum_{i=1}^n W_i \mathbbm{1}_{d,x,z}(D_i, X_i, Z_i), &&\hat{p}_{x,z}^* = \frac{1}{n}\sum_{i=1}^n W_i \mathbbm{1}_{x,z}(X_i, Z_i), &&\hat{p}_x^* = \frac{1}{n}\sum_{i=1}^n W_i \mathbbm{1}_x(X_i),
\end{align*}
and the bootstrap analogue of $\hat{s}_x$ is
\begin{equation*}
	\hat{s}_x^* = \frac{(\hat{p}_{1,x,1}^*/\hat{p}_{x,1}^* - \hat{p}_{1,x,0}^*/\hat{p}_{x,0}^*)\hat{p}_x^*}{\sum_{x'} (\hat{p}_{1,x',1}^*/\hat{p}_{x',1}^* - \hat{p}_{1,x',0}^*/\hat{p}_{x',0}^*)\hat{p}_{x'}^*}
\end{equation*}
The maps $\hat{P}_{d \mid x}$ have bootstrap analogues 
\begin{equation*}
	\hat{P}_{d \mid x}^*(f) = \frac{\mathbb{P}_n^*(\mathbbm{1}_{d,x,d} \times f)/\hat{p}_{x,d}^* - \mathbb{P}_n^*(\mathbbm{1}_{d,x,1-d} \times f)/\hat{p}_{x,1-d}^*}{\hat{p}_{d,x,d}^*/\hat{p}_{x,d}^* - \hat{p}_{d,x,1-d}^*/\hat{p}_{x,1-d}^*} = \sum_{i=1}^n \omega_{d,x,i}^* f_i
\end{equation*}
where $f_i = f(Y_i)$ and $\omega_{d,x,i}^*$ are bootstrap versions of the weights in \eqref{Display: weights to compute Pdx(f)}:
\begin{equation}
	\omega_{d,x,i}^* = \frac{W_i}{n} \times \frac{\mathbbm{1}_{d,x,d}(D_i, X_i, Z_i)/\hat{p}_{x,d}^* - \mathbbm{1}_{d,x,1-d}(D_i, X_i, Z_i)/\hat{p}_{x,1-d}^*}{\hat{p}_{d,x,d}^*/\hat{p}_{x,d}^* - \hat{p}_{d,x,1-d}^*/\hat{p}_{x,1-d}^*} \label{Display: weights to compute Pdx*(f)}
\end{equation}
Finally, $(\hat{\gamma}^{L*}, \hat{\gamma}^{H*})$ can be computed with
\begin{align}
	&\hat{\theta}_x^{L*} = \theta^L(\hat{P}_{1 \mid x}^*, \hat{P}_{0 \mid x}^*), &&\hat{\theta}_x^{H*} = \theta^H(\hat{P}_{1 \mid x}^*, \hat{P}_{0 \mid x}^*), \label{Defn: theta_x^L, theta_x^H boostrap} \\
	&\hat{\theta}^{L*} = \sum_x \hat{s}_x^* \hat{\theta}_x^{L*}, &&\hat{\theta}^{H*} = \sum_x \hat{s}_x^* \hat{\theta}_x^{H*}, \label{Defn: theta^L, theta^H boostrap} \\
	&\hat{\gamma}^{L*} = g^L(\hat{\theta}^{L*}, \hat{\theta}^{H*}, \hat{\eta}^*), &&\hat{\gamma}^{H*} = g^H(\hat{\theta}^{L*}, \hat{\theta}^{H*}, \hat{\eta}^*) \label{Defn: gamma^L, gamma^H boostrap}
\end{align}

\subsubsection{Simple bootstrap with full differentiability}
\label{Section: estimators, subsection inference, subsubsection simple bootstrap}

Under assumption \ref{Assumption: full differentiability}, estimating the distribution of $T_P'(\mathbb{G})$ is straightforward.

\begin{restatable}[]{theorem}{theoremBootstrapWorksWithFullDifferentiability}
	\label{Theorem: inference, bootstrap, bootstrap works with full differentiability}
	\singlespacing
	
	Suppose assumptions \ref{Assumption: setting}, \ref{Assumption: cost function}, \ref{Assumption: parameter, function of moments}, and \ref{Assumption: full differentiability} hold, and let $\mathbb{P}_n^*$ be given by definition \ref{Definition: exchangeable bootstrap, nonparametric bootstrap} or \ref{Definition: exchangeable bootstrap, bayesian bootstrap}. Then conditional on $\{Y_i, D_i, Z_i, X_i\}_{i=1}^n$, 
	\begin{equation*}
		\sqrt{n}(T(\mathbb{P}_n^*) - T(\mathbb{P}_n)) \overset{L}{\rightarrow} T_P'(\mathbb{G})
	\end{equation*}
	in outer probability. 
\end{restatable}

It is worth emphasizing the computationally convenience of the bootstrap $\mathbb{P}_n^*$ given in \eqref{Defn: exchangeable bootstrap} when treatment is exogenous. The weights given in display \eqref{Display: weights to compute Pdx*(f)} simplify to
\begin{equation}
	\omega_{d,x,i}^* = \frac{W_i}{n} \times \frac{\mathbbm{1}\{D_i = d, X_i = x\}}{\hat{p}_{x,d}^*} \label{Display: Pdx for bootstrap written as weighted sum}
\end{equation}
As these weights are nonnegative and sum to one, $\hat{P}_{d \mid x}^*$ is a probability distribution. Accordingly, $\theta^L(\hat{P}_{1 \mid x}^*, \hat{P}_{0 \mid x}^*)$ and $\theta^H(\hat{P}_{1 \mid x}^*, \hat{P}_{0 \mid x}^*)$ can be computed ignoring the function classes $\mathcal{F}_c$ and $\mathcal{F}_c^c$ for the same reasons discussed around display \eqref{Display: computation with exogenous treatment}:
\begin{align*}
	OT_c(\hat{P}_{1 \mid x}^*, \hat{P}_{0 \mid x}^*) &= \sup_{(\varphi,\psi) \in \Phi_c \cap (\mathcal{F}_c \times \mathcal{F}_c^c)} \hat{P}_{1 \mid x}^*(\varphi) + \hat{P}_{0 \mid x}^*(\psi) = \sup_{(\varphi,\psi) \in \Phi_c} \hat{P}_{1 \mid x}^*(\varphi) + \hat{P}_{0 \mid x}^*(\psi) \notag \\
	&= \sup_{\{\varphi_i, \psi_j\}_{i,j}} \sum_{i=1}^n \omega_{1, x,i}^* \varphi_i + \sum_{j=1}^n \omega_{0, x, j}^* \psi_j \\ 
	&\hspace{1 cm} \text{s.t. } \varphi_i + \psi_j \leq c(Y_i, Y_j) \text{ for all } 1 \leq i, j \leq n \notag
\end{align*}
 
A researcher utilizing the nonparametric bootstrap runs the risk of a boostrap draw including no observations with $\mathbbm{1}\{D_i = d, X_i = x\}$. As $\hat{p}_{x,d}^* = \frac{1}{n}\sum_{i=1}^n W_i \mathbbm{1}\{D_i = d, X_i = x\}$, this would result in the formula in \eqref{Display: Pdx for bootstrap written as weighted sum} attempting to divide by zero. This problem cannot arise when using the Bayesian bootstrap suggested in \ref{Definition: exchangeable bootstrap, bayesian bootstrap}; in this procedure $W_i > 0$ for each $i$, and thus $\hat{p}_{x,d}^* = \frac{1}{n}\sum_{i=1}^n W_i \mathbbm{1}\{D_i = d, X_i = x\} > 0$ as long as $\hat{p}_{d,x} > 0$.

\subsubsection{Alternative for directional differentiability}
\label{Section: estimators, subsection inference, subsubsection consistent alternative}

The solutions to optimal transport may not be unique as assumption \ref{Assumption: full differentiability} requires. As emphasized in the statement of theorem \ref{Theorem: weak convergence, weak convergence of estimators}, assumption \ref{Assumption: full differentiability} is not needed to obtain the asymptotic distribution of the estimators. However, without assumption \ref{Assumption: full differentiability} the procedure suggested by lemma \ref{Theorem: inference, bootstrap, bootstrap works with full differentiability} may not consistently estimate that limiting distribution. When in doubt, researchers can make use of an alternative procedure based on the results of \cite{fang2019inference} and described below. 

Additional notation is needed to describe this alternative. Let $\eta_{d, x}^{(k)} = P_{d \mid x}(\eta_d^{(k)})$, and let $T_1(\cdot)$ denote the  ``first stage'' function computing $P_{1 \mid x}$, $P_{0 \mid x}$, $\eta_{1,x}$, $\eta_{0, x}$, and $s_x$ for each $x$:
\begin{align*}
	T_1(P) = \left(\left\{P_{1 \mid x}, P_{0 \mid x}, \eta_{1,x}, \eta_{0, x}, s_x\right\}_{x \in \mathcal{X}}\right)
\end{align*} 
Here $\{a_x\}_{x \in \mathcal{X}} = (a_{x_1}, \ldots, a_{x_M})$. Let $\{\kappa_n\}_{n=1}^\infty$ be a sequence in $\mathbb{R}$ satisfying $\kappa_n \uparrow \infty$ and $\kappa_n / \sqrt{n} \rightarrow 0$. Define the set of empirical approximate maximizers:
\begin{equation*}
	\widehat{\Psi}_{c, x} = \left\{(\varphi,\psi) \in \Phi_c \cap (\mathcal{F}_c \times \mathcal{F}_c^c) \; ; \; OT_c(\hat{P}_{1 \mid x}, \hat{P}_{0 \mid x}) \leq \hat{P}_{1 \mid x}(\varphi) + \hat{P}_{0 \mid x}(\psi) + \frac{\kappa_n}{\sqrt{n}} \right\}
\end{equation*}
and the maps
\begin{equation*}
	\widehat{OT}_{c, x}'(H_1, H_0) = \sup_{(\varphi,\psi) \in \widehat{\Psi}_{c,x}} H_1(\varphi) + H_0(\psi),
\end{equation*}
and
\begin{align*}
	&\widehat{T}_{2,T_1(P)}'\left(\{H_{1,x}, H_{0,x}, h_{\eta_1, x}, h_{\eta_0, x}, h_{s, x}\}_{x \in \mathcal{X}}\right) \\
	&\hspace{1 cm} = \left(\left\{\widehat{OT}_{c_L, x}'(H_{1,x}, H_{0,x}), -\widehat{OT}_{c_H, x}'(H_{1,x}, H_{0,x}), h_{\eta_1, x}, h_{\eta_0, x}, h_{s, x}\right\}_{x \in \mathcal{X}}\right)
\end{align*}

The alternative procedure uses the conditional law of 
\begin{equation*}
	\hat{D}_4 \hat{D}_3 \widehat{T}_{2, T_1(P)}'\left(\sqrt{n}(T_1(\mathbb{P}_n^*) - T_1(\mathbb{P}_n))\right)
\end{equation*}
given the data, where $\hat{D}_4$ and $\hat{D}_3$ are matrices given by
\begin{align*}
	&\hat{D}_3 = \underset{(2+d_\eta) \times M(3 + d_\eta)}{
		\begin{bmatrix}
			\hat{D}_{3, x_1} & \hat{D}_{s, x_2} & \ldots & \hat{D}_{s, x_M}
	\end{bmatrix}},
	&&\hat{D}_{3,x} = \underset{(2+d_\eta) \times (3 + d_\eta)}{
		\begin{bmatrix}
			\hat{s}_x & 0 & 0 & 0 & \hat{\theta}_x^L \\
			0 & \hat{s}_x & 0 & 0 & \hat{\theta}_x^H \\
			0 & 0 & \hat{s}_x I_{K_1} & 0 & \hat{\eta}_{1,x} \\
			0 & 0 & 0  & \hat{s}_x I_{K_0}& \hat{\eta}_{0,x} \\
	\end{bmatrix}}, \\
	&D_4 = 
	\underset{2 \times (2+d_\eta)}{
		\begin{bmatrix}
			\nabla g^L(\hat{\theta}^L, \hat{\theta}^H, \hat{\eta})^\intercal \\
			\nabla g^H(\hat{\theta}^L, \hat{\theta}^H, \hat{\eta})^\intercal \\
	\end{bmatrix}}, 
\end{align*}

\begin{restatable}[]{theorem}{theoremFangAndSantosAlternativeWorks}
	\label{Theorem: inference, bootstrap, Fang and Santos alternative works}
	\singlespacing
	
	Suppose assumptions \ref{Assumption: setting}, \ref{Assumption: cost function}, and \ref{Assumption: parameter, function of moments} hold, let $\mathbb{P}_n^*$ be given by definition \ref{Definition: exchangeable bootstrap, nonparametric bootstrap} or \ref{Definition: exchangeable bootstrap, bayesian bootstrap}, and $\{\kappa_n\}_{n=1}^\infty \subseteq \mathbb{R}$ satisfy $\kappa_n \rightarrow \infty$ and $\kappa_n / \sqrt{n} \rightarrow 0$. Then conditional on $\{Y_i, D_i, Z_i, X_i\}_{i=1}^n$, 
	\begin{equation*}
		\hat{D}_4 \hat{D}_3 \widehat{T}_{2,T_1(P)}(\sqrt{n}(T_1(\mathbb{P}_n^*) - T_1(\mathbb{P}_n))) \overset{L}{\rightarrow} T_P'(\mathbb{G})
	\end{equation*}
	in outer probability.
\end{restatable}

\subsubsection{Confidence sets}
\label{Section: estimators, subsection confidence sets}

Theorems \ref{Theorem: inference, bootstrap, bootstrap works with full differentiability} and \ref{Theorem: inference, bootstrap, Fang and Santos alternative works} make it straightforward to conduct inference. For example, a simple confidence set for the identified set $[\gamma^L, \gamma^H]$ is given by
\begin{equation*}
	\left[\hat{\gamma}^L - \hat{c}_{1-\alpha}/\sqrt{n}, \hat{\gamma}^H + \hat{c}_{1-\alpha}/\sqrt{n}\right]
\end{equation*}
where $\hat{c}_{1-\alpha}$ is a consistent estimator of the $1-\alpha$ quantile of $\max\{T_P'(\mathbb{G})^{(1)}, -T_P'(\mathbb{G})^{(2)}\}$. When assumptions \ref{Assumption: setting} through \ref{Assumption: full differentiability} hold, let $(\hat{\gamma}^{L*}, \hat{\gamma}^{H*}) = T(\mathbb{P}_n^*)$. When assumptions \ref{Assumption: setting} through \ref{Assumption: parameter, function of moments} hold but assumption \ref{Assumption: full differentiability} is doubtful, let $(\hat{\gamma}^{L*}, \hat{\gamma}^{H*}) = (\hat{\gamma}^L, \hat{\gamma}^H) + \frac{1}{\sqrt{n}}\hat{D}_4 \hat{D}_3 \widehat{T}_{2,T_1(P)}(\sqrt{n}(T_1(\mathbb{P}_n^*) - T_1(\mathbb{P}_n)))$. In either case, compute
\begin{equation*}
	\hat{c}_{1-\alpha} = \inf\left\{c \; ; \; P\left(\max\left\{\sqrt{n}(\hat{\gamma}^{L*} - \hat{\gamma}^L), -\sqrt{n}(\hat{\gamma}^{H*} - \hat{\gamma}^H)\right\} \leq c \mid \{Y_i, D_i, Z_i, X_i\}_{i=1}^n\right) \geq 1-\alpha\right\}
\end{equation*}
through simulation:
\begin{enumerate}
	\item Compute $(\hat{\gamma}^L, \hat{\gamma}^H) = T(\mathbb{P}_n)$ and, if necessary, $\hat{D}_4$, and $\hat{D}_3$. 
	\item Generate $N$ boostrap samples, $\{W_{i,b}\}_{i=1}^n$ for each $b = 1, \ldots, N$ according to definition \ref{Definition: exchangeable bootstrap, nonparametric bootstrap} or \ref{Definition: exchangeable bootstrap, bayesian bootstrap}. For each bootstrap sample $b$, compute $(\hat{\gamma}_b^{L*}, \hat{\gamma}_b^{H*})$ as described above.
	\item Let $\hat{c}_{1-\alpha}$ be the $1-\alpha$ quantile of $\{\max\{\sqrt{n}(\hat{\gamma}_b^{L*} - \hat{\gamma}^L), -\sqrt{n}(\hat{\gamma}_b^{H*} - \hat{\gamma}^H)\}_{b=1}^N$.
\end{enumerate}
Under the further assumption that the cumulative distribution function of $\max\{T_P'(\mathbb{G})^{(1)}, -T_P'(\mathbb{G})^{(2)}\}$ is continuous and strictly increasing at its $1-\alpha$ quantile, 
\begin{equation*}
	\lim_{n \rightarrow \infty} P\left([\gamma^L, \gamma^H] \subseteq \left[\hat{\gamma}^L - \hat{c}_{1-\alpha}/\sqrt{n}, \hat{\gamma}^H + \hat{c}_{1-\alpha}/\sqrt{n}\right] \right) = 1-\alpha
\end{equation*}
Confidence sets for the parameter could be constructed following \cite{imbens2004confidence}.

%% file: sections/OTJointPO_section_application.tex
\section{Application: job training experiment}
\label{Section: application}

In this section I demonstrate the estimators in revisiting the famous National Supported Work Demonstration program (\cite{lalonde1986evaluating}). This program was implemented in the 1970s with the aim of helping socially and economically disadvantaged workers obtain job skills. Those randomly selected into the program were guaranteed a job lasting six to eighteen months, and frequently met with a counselor to discuss performance. 

I make use of the ``LaLonde'' sample studied in \cite{diamond2013genetic}. This sample consists of male participants and includes 297 treated and 425 control observations. The outcome of interest is real earnings in 1978. Observed covariates include age, years of education, real earnings in months 13 to 24 prior to randomization, and indicators for whether a participant is a high school dropout, black, hispanic, or married. Averages and standard deviations of these covariates by treatment status are reported in table \ref{Table: balance table for NSW}:

\begin{table}[H]
	\begin{center}
		\renewcommand{\arraystretch}{1.2}
		\caption{Balance table} 
		\label{Table: balance table for NSW}
		\begin{tabular}{l||cccccccc}
			\toprule
			&       base inc. &   age &  yrs. educ &  HS dropout &  black &  hispanic &  married & $N$ \\ \hline \hline 
			\multirow{2}{*}{control} & 3672.49 & 24.45 & 10.19 &      0.81 &   0.80 &  0.11 &     0.16 & \multirow{2}{*}{425} \\
			& (6521.53) &  (6.59) & (1.62) &      (0.39) &   (0.40) &  (0.32) &     (0.36) & \\ \hline
			\multirow{2}{*}{treated} & 3571.00 & 24.63 & 10.38 &      0.73 &   0.80 &  0.09 &     0.17 & \multirow{2}{*}{297} \\
			& (5773.13) &  (6.69) &  (1.82) &      (0.44) &   (0.40) &  (0.29) &     (0.37) & \\
			\bottomrule
		\end{tabular}
	\end{center}
	\vspace{-0.3 cm} 
	\hspace{0.5 cm} {\scriptsize \textit{Note:} Standard deviations in parentheses.}
\end{table}

There is no reported noncompliance, so I interpret the setting as one of exogenous treatment. The parameter of interest is the OLS slope coefficient of regressing treatment effects on a constant and $Y_0$:
\begin{align*}
	\gamma = \frac{\text{Cov}(Y_1 - Y_0, Y_0)}{\text{Var}(Y_0)} = \frac{E_{P_{1,0}}[(Y_1 - Y_0) Y_0] - (E_{P_1}[Y_1] - E_{P_0}[Y_0])E_{P_0}[Y_0]}{E_{P_0}[Y_0^2] - (E_{P_0}[Y_0])^2}
\end{align*}
as described in example \ref{Example: equitable policies}, the sign of this parameter describes who receives larger benefits from treatment: $\gamma < 0$ implies those with below average untreated outcomes tend to see above average treatment effects.

Discretized versions of baseline income and age are found to be informative covariates. Baseline income is binned as: $[0,0]$ or $(0, \infty)$, while age is binned as $(16,20]$, $(20, 26]$, or $(26, \infty)$. $X$ is the cartesian product of bins. The resulting $(d,x)$ bins have a minimum of 31 observations per bin, and an average of 60.2 observations per bin.

The point estimates are $(\hat{\gamma}^L, \hat{\gamma}^H) = (-1.73, -0.004)$. The negative upper bound point estimates suggests that the treatment was especially beneficial for participants who would otherwise have incomes below average (for the eligible population). Covariates are found to be informative, especially for the upper bound. Ignoring covariates, the lower bound point estimate is $-1.78$ and the upper bound point estimate is $0.189$. The $95\%$ confidence set for the identified based on 500 bootstrap draws is $[-1.94, 0.20]$, suggesting $\gamma$ may still be zero or slightly positive once accounted for sample uncertainty.

%% file: sections/OTJointPO_section_extensions.tex
\section{Extensions}
\label{Section: extensions}

This section briefly describes simple extensions.

\subsection{Conditioning on $X \in A$}

In many applications parameters conditional on a covariate taking a particular value are of interest. For example, the share of compliers of a particular demographic benefiting from treatment is $P(Y_1 > Y_0 \mid D_1 > D_0, \text{demographic})$. 

Such parameters can be written in the form
\begin{align*}
	\gamma_A = g(\theta_A, \eta_A)
\end{align*}
where for a known set $A \subseteq \mathcal{X}$,
\begin{align*}
	&\theta_A \equiv E[c(Y_1,Y_0) \mid D_1 > D_0, X \in A], &&\eta_A \equiv E[\eta_1(Y_1), \eta_0(Y_0) \mid D_1 > D_0, X \in A]
\end{align*} 
The identified set for $\gamma_A$ is straightforward to characterize and estimate. First note that 
\begin{align*}
	\theta_A = E[\theta_X \mid D_1 > D_0, X \in A] = \frac{1}{s_A} \sum_{x \in A} s_x \theta_x
\end{align*}
where $s_A = \sum_{x \in A} s_x$. The proof of theorem \ref{Theorem: identification, function of moments} shows that the sharp identified set for $(\theta_{x_1}, \ldots, \theta_{x_M})$ is in fact $[\theta_{x_1}^L, \theta_{x_1}^H] \times \ldots \times [\theta_{x_M}^L, \theta_{x_M}^H]$. It follows that the sharp identified set for $\theta_A$ is $[\theta_A^L, \theta_A^H]$, where
\begin{align*}
	&\theta_A^L = \frac{1}{s_A} \sum_{x \in A} s_x \theta_x^L, &&\theta_A^H = \frac{1}{s_A} \sum_{x \in A} s_x \theta_x^H
\end{align*}
and the sharp identified set for $\gamma_A$ is $[\gamma_A^L, \gamma_A^H]$ where
\begin{align*}
	&\gamma_A^L = \min_{t \in [\theta_A^L, \theta_A^H]} g(t, \eta_A), &&\gamma_A^H = \max_{t \in [\theta_A^L, \theta_A^H]} g(t, \eta_A),
\end{align*}

Let $\hat{s}_x$, $\hat{\theta}_x^L$, and $\hat{\theta}_x^H$ be as defined in section \ref{Section: estimators}. Let $\hat{s}_A = \sum_{x \in A} \hat{s}_x$ and 
\begin{align*}
	&\hat{\theta}_A^L = \frac{1}{\hat{s}_A} \sum_{x \in A} \hat{s}_x \hat{\theta}_x^L, &&\hat{\theta}^H(A) = \frac{1}{\hat{s}_A} \sum_{x \in A} \hat{s}_x \hat{\theta}_x^H \\
	&\hat{\gamma}_A^L = \min_{t \in [\hat{\theta}_A^L, \hat{\theta}_A^H]} g(t, \hat{\eta}_A), &&\hat{\gamma}_A^H = \max_{t \in [\hat{\theta}_A^L, \hat{\theta}_A^H]} g(t, \hat{\eta}_A),
\end{align*}
Under assumptions \ref{Assumption: setting}, \ref{Assumption: cost function}, and \ref{Assumption: parameter, function of moments}, $\sqrt{n}((\hat{\gamma}_A^L, \hat{\gamma}_A^H) - (\gamma_A^L, \gamma_A^H)$ will converge weakly. With assumption \ref{Assumption: full differentiability} the straightforward bootstrap will consistently estimate its asymptotic distribution.

\subsection{Quantiles}
\label{Section: extensions, subsection quantiles}

Example \ref{Example: pseudo quantiles} considers the parameter $q_\tau$ solving 
\begin{equation*}
	P(Y_1 - Y_0 \leq q_\tau \mid D_1 > D_0) = \tau 
\end{equation*}
As noted in that example, the sharp identification results for $P(Y_1 - Y_0 \leq \delta \mid D_1 > D_0)$ can be adapted to characterize the sharp identified set for $q_\tau$. First view the bounds on the cumulative distribution function as functions of $\delta$:
\begin{align*}
	&c_{L,\delta}(y_1,y_0) = \mathbbm{1}\{y_1 - y_0 < \delta\}, &&c_{H,\delta} (y_1,y_0) = \mathbbm{1}\{y_1 - y_0 > \delta\}, \\
	&\theta_x^L(\delta) = OT_{c_{L,\delta}}(P_{1 \mid x}, P_{0 \mid x}), &&\theta_x^H(\delta) = 1 - OT_{c_{H, \delta}}(P_{1 \mid x}, P_{0 \mid x}) \\
	&\theta^L(\delta) = \sum_x s_x \theta_x^L(\delta) &&\theta^H(\delta) = \sum_x s_x \theta_x^H(\delta) 
\end{align*}
Let $Q_{I, \tau}$ denote the sharp identified set for $q_\tau$.

\begin{restatable}[Identification of $q_\tau$]{lemma}{lemmaIdentificationPseudoQuantile}
	\label{Lemma: identification, pseudo quantile}
	\singlespacing
	
	Suppose assumptions \ref{Assumption: setting} and \ref{Assumption: cost function} \ref{Assumption: cost function, CDF} hold. Then $q \in Q_{I, \tau}$ if and only if $\theta^L(q) \leq \tau \leq \theta^H(q)$.
\end{restatable}

Lemma \ref{Lemma: identification, pseudo quantile} implies that inverting a test of $H_0 : \theta^L(q) \leq \tau \leq \theta^H(q)$ against the alternative $H_1 : \tau < \theta^L(q) \text{ or } \theta^H(q) < \tau$ will lead to valid confidence sets for $q_\tau$.

\begin{remark}
	\label{Remark: pseudo quantile interpretation and definition}
	\singlespacing
	
	Consider instead defining $q_\tau$ to be the closed subset of $\mathbb{R}$ given by
	\begin{equation*}
		q_\tau = [\inf\{y \; ; \; P(Y_1 - Y_0 \leq y) \geq \tau\}, \inf\{y \; ; \; P(Y_1 - Y_0 \leq y) > \tau\}] 
	\end{equation*}
	Note that this $q_\tau$ is the singleton $\inf\{y \; ; \; P(Y_1 - Y_0 \leq y) \geq \tau\}$, unless $P(Y_1 - Y_0 \leq \cdot)$ is flat when equal to $\tau$, in which case it equals the $\tau$-level set $\{y \; ; \; P(Y_1 - Y_0 \leq y) = \tau\}$. (Compare \cite{ehm2016quantiles}, who define the $\tau$-th quantile equivalently as $q_\tau = [\sup\{y \; ; \; P(Y_1 - Y_0 \leq y) < \tau\}, \sup\{y \; ; \; P(Y_1 - Y_0 \leq y) \leq \tau\}]$.) Let $Q_{I, \tau}$ denote the identified set of $q_\tau$ as defined in this remark. Lemma \ref{Lemma: identification, pseudo quantile, alternative definition} in appendix \ref{Appendix: identification} shows that under assumptions \ref{Assumption: setting} and \ref{Assumption: cost function} \ref{Assumption: cost function, CDF}, $q \in Q_{I, \tau}$ if and only if $\theta^L(q) \leq \tau \leq \theta^H(q)$.
\end{remark}

\subsection{Multiple treatment arms with exogenous treatment}

The identification results and estimators proposed above are easily extended to a setting with multiple treatment arms and exogenous treatment. Let the mutually exclusive treatment arms indexed by $d \in \{0, 1, \ldots, J\}$, with $d = 0$ indicating control. Let $Y_d$ be the potential outcome with treatment $d$, $D_d$ equal one if the unit has treatment $d$ and zero otherwise. The observed outcome is
\begin{align*}
	Y = \sum_{d=0}^J D_d Y_d
\end{align*}
Let $D = (D_0, D_1, \ldots, D_J)$ and assume 
\begin{align*}
	(Y_0, Y_1, \ldots, Y_J) \perp D \mid X
\end{align*}
Note that the marginal distributions of $Y_d \mid X = x$, denoted $P_{d \mid x}$, are identified with the relation
\begin{align*}
	E_{P_{d \mid x}}[f(Y_d)] = E[f(Y_d) \mid X = x] = \frac{E[f(Y) D_d \mid X = x]}{P(D_d = 1 \mid X = x)}
\end{align*}

Let $\gamma_d = g(\theta_d, \eta_d)$ where $\theta_d = E[c(Y_d, Y_0)]$. Consider estimating the sharp identified set for $(\gamma_1, \ldots \gamma_J)$. For example, an RCT with two treatment arms may have similar average treatment effects. The treatment arms may be further distinguished by comparing $P(Y_1 - Y_0 > 0)$ with $P(Y_2 - Y_0 > 0)$, or $\text{Cov}(Y_1 - Y_0, Y_0)$ with $\text{Cov}(Y_2 - Y_0, Y_0)$.

Let $\theta_{d,x} = E[c(Y_1,Y_0) \mid X = x]$. The sharp identified set for $(\theta_{1,x}, \ldots, \theta_{J,x})$ is given by
\begin{align*}
	[\theta_{1,x}^L, \theta_{1,x}^H] \times \ldots \times [\theta_{J,x}^L, \theta_{J,x}^H]
\end{align*}
where $\theta_{d,x}^L = \theta^L(P_{d \mid x}, P_{0 \mid x})$ and $\theta_{d,x}^H = \theta^H(P_{d \mid x}, P_{0 \mid x})$ as in section \ref{Section: identification}.\footnote{This follows from existing results and the \textit{gluing lemma}, found in \cite{villani2009optimal} (pp. 11-12).} 
The sharp identified set for $\theta_d$ is $[\theta_d^L, \theta_d^H]$ where $\theta_d^L = \sum_x s_x \theta_{d,x}^L$ and $\theta_d^H = \sum_x s_x \theta_{d,x}^H$, and the sharp identified set for $(\gamma_1, \ldots \gamma_J)$ is 
\begin{align*}
	[\gamma_1^L, \gamma_1^H] \times \ldots \times [\gamma_J^L, \gamma_J^H]
\end{align*}
Sample analogues $(\hat{\gamma}_1^L, \hat{\gamma}_1^H, \ldots, \hat{\gamma}_J^L, \hat{\gamma}_J^H)$ can be formed just as in section \ref{Section: estimators}. Under natural adjustments to assumptions \ref{Assumption: cost function}, \ref{Assumption: parameter, function of moments}, and \ref{Assumption: full differentiability}, the same arguments work to show
\begin{align*}
	\sqrt{n}((\hat{\gamma}_1^L, \hat{\gamma}_1^H, \ldots, \hat{\gamma}_J^L, \hat{\gamma}_J^H) - (\gamma_1^L, \gamma_1^H, \ldots, \gamma_J^L, \gamma_J^H))
\end{align*}
is asymptotically Gaussian and the bootstrap consistently estimates its asymptotic distribution.

%% file: sections/OTJointPO_section_conclusion.tex
\section{Conclusion}
\label{Section: conclusion}

This paper studies a large class of causal parameters that depend on a moment of the joint distribution of potential outcomes. The sharp identified set of such parameters is characterized with optimal transport. Estimators based on this identification are $\sqrt{n}$-consistent and converge in distribution under mild assumptions, and inference procedures based on the bootstrap are straightforward and computationally convenient. 

%% file: appendix/OTJointPO_appendix_identification.tex
\section{Appendix: identification}
\label{Appendix: identification}

Following \cite{kitagawa2015test}, let $T$ denote the ``type'' of a unit:
\begin{equation}
	T = 
	\begin{cases}
		a, \text{ always-taker,} &\text{ if } (D_1, D_0) = (1,1) \\
		c, \text{ complier,} &\text{ if } (D_1, D_0) = (1,0) \\
		n, \text{ never-taker,} &\text{ if } (D_1, D_0) = (0,0) \\
		df, \text{ defier,} &\text{ if } (D_1, D_0) = (0,1) 
	\end{cases}
	\label{Defn: T, type of a unit}
\end{equation}
Note that the primitives $(Y_1, Y_0, D_1, D_0, Z, X)$ are equivalent to $(Y_1, Y_0, T, Z, X)$.

\begin{restatable}[Identification of moments]{lemma}{lemmaSharpIdentificationOfMoments}
	\label{Lemma: identification, moments}
	\singlespacing
	
	Suppose assumptions \ref{Assumption: setting} and \ref{Assumption: cost function} hold. Then the sharp identified set for $\theta$ is $[\theta^L, \theta^H]$.
\end{restatable}

\begin{proof}
	\singlespacing
	
	Let $T$ be as defined in \eqref{Defn: T, type of a unit}, and note that the primitives of the model $(Y_1, Y_0, D_1, D_0, Z, X)$ are equivalent to $(Y_1, Y_0, T, Z, X)$. Moreover, the event $D_1 > D_0$ is the event $T = c$; thus $P_{d \mid x}$ is the distribution of $Y_d \mid T = c, X = x$. \\
	
	In steps:
	\begin{enumerate}
		\item The identified set for $(P_{1, 0 \mid x_1}, \ldots, P_{1,0 \mid x_M})$, the conditional distributions of $(Y_1, Y_0) \mid T = c, X = x$ for each $x \in \mathcal{X} = \{x_1, \ldots, x_M\}$, is $\Pi(P_{1 \mid x_1}, P_{0 \mid x_1}) \times \ldots \times \Pi(P_{1 \mid x_M}, P_{0 \mid x_M})$. \\
		
		That $(P_{1, 0 \mid x_1}, \ldots, P_{1,0 \mid x_M}) \in \Pi(P_{1 \mid x_1}, P_{0 \mid x_1}) \times \ldots \times \Pi(P_{1 \mid x_M}, P_{0 \mid x_M})$ is immediate. To see that any element of $\Pi(P_{1 \mid x_1}, P_{0 \mid x_1}) \times \ldots \times \Pi(P_{1 \mid x_M}, P_{0 \mid x_M})$ is possible given the assumptions and distribution of the observables $(Y, D, Z, X)$, fix a distribution of the observables generated by a distribution of the primitives consistent with the assumptions. Note that the distribution of observables is summarized by $P(D = d, Z = z, X = x)$ for each $(d,z,x)$ and the conditional distributions
		\begin{equation*}
			Y \mid D = d, Z = z, X = x
		\end{equation*}
		Use this observation and the claims of lemma \ref{Lemma: identification, LATE IV model empirical content} to see that any two distributions of the primitives $(Y_1, Y_0, T, Z, X)$ (consistent with the assumptions), sharing the same distribution of $(T, Z, X)$, and the same marginal, conditional distributions for
		\begin{align*}
			&Y_1 \mid T = a, X = x &&Y_0 \mid T=n, X = x \\
			&Y_1 \mid T = c, X = x, &&Y_0 \mid T=c, X = x
		\end{align*}
		will produce this distribution of observables. Thus, replacing $(P_{1, 0 \mid x_1}, \ldots, P_{1,0 \mid x_M})$ from the distribution of primitives with any
		\begin{equation*}
			(\pi_{x_1}, \ldots, \pi_{x_M}) \in \Pi(P_{1 \mid x_1}, P_{0 \mid x_1}) \times \ldots \times \Pi(P_{1 \mid x_M}, P_{0 \mid x_M})
		\end{equation*}
		will generate the same observed distribution of $(Y, D, Z, X)$, without violating assumption \ref{Assumption: setting} or \ref{Assumption: cost function}. The claim follows.
		
		\item The identified set for $(\theta_{x_1}, \ldots, \theta_{x_M}) \in \mathbb{R}^M$ is $[\theta_{x_1}^L, \theta_{x_1}^H] \times \ldots \times [\theta_{x_M}^L, \theta_{x_M}^H]$.
		
		Recall that $\theta_x = E[c(Y_1, Y_0) \mid X = x]$, and let $\Theta_{I,x}$ denote its identified set. Note that the previous step implies
		\begin{equation*}
			\Theta_{I,x} = \left\{t \in \mathbb{R} \; ; \; t = E_{\pi_x}[c(Y_1,Y_0)] \text{ for some } \pi_x \in \Pi(P_{1 \mid x}, P_{0 \mid x})\right\}
		\end{equation*}
		$\Pi(P_{1 \mid x}, P_{0 \mid x})$ is convex. Notice that for any $\lambda \in (0,1)$ and $\pi_x^1, \pi_x^0 \in \Pi(P_{1 \mid x}, P_{0 \mid x})$, $E_{\lambda \pi_x^1 + (1-\lambda)\pi_x^0}[c(Y_1, Y_0)] = \lambda E_{\pi_x^1}[c(Y_1,Y_0)] + (1-\lambda) E_{\pi_x^0}[c(Y_1,Y_0)]$.	Together these imply $\Theta_{I,x}$ is convex. 
		
		It suffices to show that for any $x$, $\Theta_{I,x} = [\theta_x^L, \theta_x^H]$ There are two cases:
		\begin{enumerate}[label=(\roman*)]
			\item If assumption  \ref{Assumption: cost function} \ref{Assumption: cost function, smooth costs} holds, then for each $x$,
			\begin{align*}
				&\theta_x^L = OT_c(P_{1 \mid x}, P_{0 \mid x}) = \inf_{\pi_x \in \Pi(P_{1 \mid x}, P_{0 \mid x})} E_{\pi_x}[c(Y_1, Y_0)] \\
				&\theta_x^H = -OT_{-c}(P_{1 \mid x}, P_{0 \mid x}) = \sup_{\pi_x \in \Pi(P_{1 \mid x}, P_{0 \mid x})} E_{\pi_x}[c(Y_1, Y_0)] 
			\end{align*}
			Since $c$ is continuous, lemma \ref{Lemma: optimal transport is attained} implies the optimal transport problems are attained, say by $\pi_x^L$ and $\pi_x^H$ respectively. It follows that $\theta_x^L, \theta_x^H \in \Theta_{I,x}$, and it is clear from their definitions that they bound $\Theta_{I,x}$. Since $\Theta_{I,x}$ is convex, it follows that $\Theta_{I,x} = [\theta_x^L, \theta_x^H]$. \\
			
			\item If Assumption \ref{Assumption: cost function} \ref{Assumption: cost function, CDF} holds, then 
			\begin{align*}
				&c_L(y_1,y_0) = \mathbbm{1}\{y_1 - y_0 < \delta\}, &&c_H(y_1,y_0) = \mathbbm{1}\{y_1 - y_0 > \delta\}, \\
				&\theta_x^L = OT_{c_L}(P_{1 \mid x}, P_{0 \mid x}), &&\theta_x^H = 1 - OT_{c_H}(P_{1 \mid x}, P_{0 \mid x})
			\end{align*}
			
			Let $\pi_x^L, \pi_x^H \in \Pi(P_{1 \mid x}, P_{0 \mid x})$ be such that $\theta_x^L = E_{\pi_x^L}[\mathbbm{1}\{Y_1 - Y_0 < \delta\}] = P_{\pi_x^L}(Y_1 - Y_0 < \delta)$ and $\theta_x^H = 1 - E_{\pi_x^H}[\mathbbm{1}\{Y_1 - Y_0 > \delta\}] = P_{\pi_x^H}(Y_1 - Y_0 \leq \delta)$. Notice that $\theta_x^H \in \Theta_{I,x}$. Furthermore, $\mathbbm{1}\{y_1 - y_0 < \delta\} \leq \mathbbm{1}\{y_1 - y_0 \leq \delta\}$ implies 
			\begin{align*}
				\theta_x^L = \inf_{\pi_x \in \Pi(P_{1 \mid x}, P_{0 \mid x})} E_{\pi_x}[\mathbbm{1}\{Y_1 - Y_0 < \delta\}] \leq \inf_{\pi_x \in \Pi(P_{1 \mid x}, P_{0 \mid x})} E_{\pi_x}[\mathbbm{1}\{Y_1 - Y_0 \leq \delta\}]
			\end{align*}
			and thus $\theta_x^L$ is a lower bound for $\Theta_{I,x}$. Since $\Theta_{I,x}$ is convex, it suffices to show that $\theta_x^L \in \Theta_{I,x}$. \\
			
			Corollary \ref{Lemma: duality special case CDF, corollary} implies that $\theta_x^L = P_{\pi_x^L}(Y_1 - Y_0 < \delta) = \sup_y \left\{F_{1 \mid x}(y) - F_{0 \mid x}(y - \delta)\right\}$. Moreover, \cite{villani2009optimal} theorem 5.10 part (iii) implies the dual problem $\sup_y \left\{F_{1 \mid x}(y) - F_{0 \mid x}(y - \delta)\right\}$ is attained as well, say by $y^*$. Thus
			\begin{equation}
				\int \mathbbm{1}\{y_1 - y_0 \leq \delta\} d\pi_x^L(y_1,y_0) = \int \mathbbm{1}\{y_1 \leq y^*\}dP_{1 \mid x}(y_1) - \int \mathbbm{1}\{y_0 \leq y^* - \delta\} dP_{0 \mid x}(y_0) \label{Display: theorem proof, identification, moments, lower bound on CDF is sharp, strong duality}
			\end{equation}
			Next, notice that
			\begin{equation}
				\mathbbm{1}\{y_1 \leq y^*\} - \mathbbm{1}\{y_0 \leq y^* - \delta\} \leq \mathbbm{1}\{y_1 - y_0 < \delta\} \label{Display: theorem proof, identification, moments, lower bound on CDF is sharp, inequality}
			\end{equation}
			which holds for all $(y_1, y_0)$, must hold with equality $\pi_x^L$-almost surely. Indeed, let $N$ be the set where the inequality in \eqref{Display: theorem proof, identification, moments, lower bound on CDF is sharp, inequality} is strict and suppose $N$ is $\pi_x^L$-non-negligible. Since $\pi_x^L \in \Pi(P_{1 \mid x}, P_{0 \mid x})$, 
			\begin{align*}
				&\int \mathbbm{1}\{y_1 \leq y^*\}dP_{1 \mid x}(y_1) - \int \mathbbm{1}\{y_0 \leq y^* - \delta\} dP_{0 \mid x}(y_0) =  \int \mathbbm{1}\{y_1 \leq y^*\} - \mathbbm{1}\{y_0 \leq y^* - \delta\} d\pi_x^L(y_1, y_0) \\
				&= \int_N \mathbbm{1}\{y_1 \leq y^*\} - \mathbbm{1}\{y_0 \leq y^* - \delta\} d\pi_x^L(y_1, y_0) + \int_{N^c} \mathbbm{1}\{y_1 \leq y^*\} - \mathbbm{1}\{y_0 \leq y^* - \delta\} d\pi_x^L(y_1, y_0) \\
				&< \int_N \mathbbm{1}\{y_1 - y_0 < \delta\} d\pi_x^L(y_1, y_0) + \int_{N^c}  \mathbbm{1}\{y_1 - y_0 < \delta\} d\pi_x^L(y_1, y_0) \\
				&= \int \mathbbm{1}\{y_1 - y_0 \leq \delta\} d\pi_x^L(y_1,y_0)
			\end{align*}
			contradicts \eqref{Display: theorem proof, identification, moments, lower bound on CDF is sharp, strong duality}. This implies that $\pi_x^L$ concentrates on
			\begin{gather*}
				\underbrace{\left\{(y_1, y_0) \; ; \; y_1 \leq y^*, y_0 > y^* - \delta, y_1 - y_0 < \delta\right\}}_{\text{both sides of \eqref{Display: theorem proof, identification, moments, lower bound on CDF is sharp, inequality} equal } 1} \cup \underbrace{\left\{(y_1, y_0) \; ; \; y_1 > y^*, y_0 > y^* - \delta, y_1 - y_0 \geq \delta\right\}}_{\text{both sides of \eqref{Display: theorem proof, identification, moments, lower bound on CDF is sharp, inequality} equal } 0} \\
				\cup \underbrace{\left\{(y_1, y_0) \; ; \; y_1 \leq y^*, y_0 \leq y^* - \delta, y_1 - y_0 \geq \delta\right\}}_{\text{both sides of \eqref{Display: theorem proof, identification, moments, lower bound on CDF is sharp, inequality} equal } 0}
			\end{gather*}
			Notice the only point in the set $\{(y_1, y_0) \; ; \; y_1 - y_0 =\delta\}$ where $\pi_x^L$ could put positive mass is the point $(y_1, y_0) = (y^*, y^* - \delta)$. But since $P_{1 \mid x}$ has a continuous CDF,
			\begin{align*}
				0 \leq \pi_x^L(\{(y^*, y^* - \delta)\}) \leq \pi_x^L(\{y^*\} \times \mathcal{Y}_0) = P_{1 \mid x}(\{y^*\}) = 0
			\end{align*}
			Thus $P_{\pi_x^L}(Y_1 - Y_0 = \delta) = 0$, and so $P_{\pi_x^L}(Y_1 - Y_0 \leq \delta) = P_{\pi_x^L}(Y_1 - Y_0 < \delta) = \theta^L(x)$. Thus $\theta_x^L \in \Theta_{I,x}$, and hence $\Theta_{I,x} = [\theta^L(x), \theta^H(x)]$.

		\end{enumerate}
		
		Therefore the identified set for $\theta_x$ is $[\theta_x^L, \theta_x^H]$. It follows from this and step one above that the identified set $(\theta_{x_1}, \ldots, \theta_{x_M})$ is $[\theta_{x_1}^L, \theta_{x_1}^H] \times \ldots \times [\theta_{x_M}^L, \theta_{x_M}^H]$. 
		
		\item Recall that $\theta = E[c(Y_1,Y_0)] = E[E[c(Y_1,Y_0) \mid X]] = \sum_x s_x \theta_x$. Since $s_x = P(X = x \mid T = c)$ is point identified for each $x$, it follows from step two above that the identified set for $\theta$ is $[\theta^L, \theta^H]$ where
		\begin{align*}
			&\theta^L = \sum_x s_x \theta_x^L, &&\theta^H = \sum_x s_x \theta_x^H
		\end{align*}
	\end{enumerate}
	This concludes the proof.
\end{proof}

\theoremIdentificationFunctionOfMoments*
\begin{proof}
	\singlespacing
	
	Lemma \ref{Lemma: identification, moments} shows that under assumptions \ref{Assumption: setting} and \ref{Assumption: cost function}, the sharp identified set for $\theta$ is $[\theta^L, \theta^H]$. Let $\Gamma_I$ be the identified set for $\gamma$, and note that 
	\begin{equation*}
		\Gamma_I = \{\gamma \in \mathbb{R} \; ; \; \gamma = g(t, \eta) \text{ for some } t \in [\theta^L, \theta^H]\}
	\end{equation*}
	
	Assumption \ref{Assumption: cost function} implies $c$ is bounded; under assumption \ref{Assumption: cost function} \ref{Assumption: cost function, smooth costs} the continuous $c : \mathcal{Y} \times \mathcal{Y}\rightarrow \mathbb{R}$ takes a maximum and minimum on the compact set $\mathcal{Y} \times \mathcal{Y}$, while under assumption \ref{Assumption: cost function} \ref{Assumption: cost function, CDF} the cost function only takes values $0$ or $1$. It follows that $\theta^L$ and $\theta^H$ are finite and thus $[\theta^L, \theta^H]$ is compact.
	
	Assumption \ref{Assumption: parameter, function of moments} \ref{Assumption: parameter, function of moments, g is continuous} is that $g(\cdot,\eta)$ is continuous, and thus the extreme value theorem implies $\gamma^L = \inf_{t \in [\theta^L, \theta^H]} g(t, \eta)$ and $\gamma^H = \sup_{t \in [\theta^L, \theta^H]} g(t, \eta)$ are both elements of $\Gamma_I$. The intermediate value theorem then implies $\Gamma_I = [\gamma^L, \gamma^H]$.
\end{proof}

\lemmaIdentificationPseudoQuantile*
\begin{proof}
	\singlespacing
	
	By definition, $q \in \Gamma_{I,\tau}$ if and only if there exists a distribution of the primitives, $\pi$, consistent with the observed distribution, such that $P_\pi(Y_1 - Y_0 \leq q) = \tau$. Lemma \ref{Lemma: identification, moments} shows that $\theta^L(q) \leq \tau \leq \theta^H(q)$ if and only if there exists a distribution of the primitives, $\pi$, such that $\theta^L(q) \leq \tau \leq \theta^H(q)$. This concludees the proof.
\end{proof}

\begin{restatable}[Identification: $\tau$-th quantile]{lemma}{lemmaQuantileIdentification}
	\label{Lemma: identification, pseudo quantile, alternative definition}
	\singlespacing
	
	Let $q_\tau$ be defined as
	\begin{equation*}
		q_\tau = [\inf\{y \; ; \; P(Y_1 - Y_0 \leq y) \geq \tau\}, \inf\{y \; ; \; P(Y_1 - Y_0 \leq y) > \tau\}]
	\end{equation*}
	
	Suppose assumption \ref{Assumption: setting} and \ref{Assumption: cost function} \ref{Assumption: cost function, CDF} hold, and let $Q_{I,\tau}$ denote the identified set of $q_\tau$ defined above. Then $q \in Q_{I, \tau}$ if and only if $\theta^L(q) \leq \tau \leq \theta^H(q)$. 
\end{restatable}
\begin{proof}
	\singlespacing
	
	Suppose $\theta^L(q) \leq \tau \leq \theta^H(q)$. Lemma \ref{Lemma: identification, moments} implies there exists a distribution $\pi$ of the primitives consistent with assumption \ref{Assumption: cost function} \ref{Assumption: cost function, CDF} such that $P_\pi(Y_1 - Y_0 \leq q) = \tau$. Thus $q \in [\inf\{y \; ; \; P_\pi(Y_1 - Y_0 \leq y) \geq \tau\}, \inf\{y \; ; \; P_\pi(Y_1 - Y_0 \leq y) > \tau\}]$ and hence $q \in Q_{I, \tau}$.
	
	Before showing the other direction, we next show that assumption \ref{Assumption: cost function} \ref{Assumption: cost function, CDF} implies $\theta^L(\delta)$ is continuous. Specifically, apply corollary \ref{Lemma: duality special case CDF, corollary} to find $\theta_x^L(\delta) = \sup_y \{F_{1 \mid x}(y) - F_{0 \mid x}(y - \delta)\}$. So for any $\delta, \delta'$,  
	\begin{align*}
		\theta_x^L(\delta) - \theta_x^L(\delta') &= \sup_y \{F_{1 \mid x}(y) - F_{0 \mid x}(y - \delta)\} - \sup_y \{F_{1 \mid x}(y) - F_{0 \mid x}(y - \delta')\}  \\
		&\leq \sup_y \left\{F_{0 \mid x}(y - \delta') - F_{0 \mid x}(y - \delta)\right\} \\
		&\leq \sup_y \left\lvert F_{0 \mid x}(y - \delta') - F_{0 \mid x}(y - \delta)\right\rvert
	\end{align*}
	and thus $\lvert \theta_x^L(\delta) - \theta_x^L(\delta') \rvert \leq \sup_y \left\lvert F_{0 \mid x}(y - \delta') - F_{0 \mid x}(y - \delta)\right\rvert$. Recall that any continuous CDF is in fact uniformly continuous, and so $F_{0 \mid x}$ is in fact uniformly continuous. Let $\varepsilon > 0$, choose $\eta > 0$ such that for any $y, y' \in \mathbb{R}$ with $\lvert y - y' \rvert < \eta$, one has $\lvert F_{0 \mid x}(y) - F_{0 \mid x}(y') \rvert < \varepsilon/2$, and notice that 
	\begin{align*}
		\lvert \delta - \delta' \rvert < \eta \implies \sup_y \left\lvert F_{0 \mid x}(y - \delta') - F_{0 \mid x}(y - \delta) \right\rvert \leq \varepsilon/2 < \varepsilon
	\end{align*}
	This shows $\theta_x^L(\delta)$ is continuous, and so $\theta^L(\delta) =\sum_x s_x \theta_x^L$ is continuous.
	
	Return to showing the other direction, through the contrapositive. Suppose it is not the case that $\theta^L(q) \leq \tau \leq \theta^H(q)$. There are two possibilities:
	\begin{enumerate}
		\item Suppose $\theta^H(q) < \tau$. Then there is no distribution $\pi$ of the primitives such that $P_\pi(Y_1 - Y_0 \leq q) \geq \tau$, hence there is no distribution where $q \in [\inf\{y \; ; \; P(Y_1 - Y_0 \leq y) \geq \tau\}, \inf\{y \; ; \; P(Y_1 - Y_0 \leq y) > \tau\}]$ and thus $q \not \in Q_{I, \tau}$.

		\item Suppose $\tau < \theta^L(q)$. If one further supposes that $q \in Q_{I, \tau}$, then $\theta^L(\cdot)$ would have a jump discontinuity at $q$, contradicting the continuity shown above. 
		
		Specifically, if $\tau < \theta^L(q)$ and $q \in Q_{I, \tau}$, then there exists a distribution $\pi$ of the primitives such that $P_\pi(Y_1 - Y_0 \leq q) > \tau$ and $q \in [\inf\{y \; ; \; P_\pi(Y_1 - Y_0 \leq y) \geq \tau\}, \inf\{y \; ; \; P_\pi(Y_1 - Y_0 \leq y) > \tau\}]$, implying that  $P_\pi(Y_1 - Y_0 \leq \cdot)$ jumps at $q$ from below $\tau$ to above $\theta^L(q)$:
		\begin{equation*}
			\lim_{\epsilon \rightarrow 0} P_\pi(Y_1 - Y_0 \leq q - \epsilon) < \tau < \theta^L(q) \leq P_\pi(Y_1 - Y_0 \leq q)
		\end{equation*}
		This jump discontinuity at $q$ is at least of size $\varepsilon = \theta^L(q) - \tau > 0$. But then $\theta^L(\cdot)$ would have a jump discontinuity of at least size $\varepsilon$ at $q$ as well, a contradiction of the continuity of $\theta^L(\cdot)$ shown above.
		
		Thus if $\tau < \theta^L(q)$, then $q \not \in Q_{I,\tau}$.
	\end{enumerate}
	In either case, $q \not \in Q_{I, \tau}$. This completes the proof.
\end{proof}

\subsection{Additional identification lemmas}
\label{Appendix: identification, subsection lemmas}

The lemmas below contain results well known in the literature. They are included here with proofs for completeness. 

\begin{restatable}[]{lemma}{lemmaDegenerateImpliesSingletonCoupling}
	\label{Lemma: identification, degenerate distribution has only one coupling}
	\singlespacing
	
	Let $P_1$ be any distribution and $P_0$ be degenerate at $\tilde{y}_0 \in \mathbb{R}$. Then the only possible coupling of $P_1$ and $P_0$ is characterized by the cumulative distribution function
	\begin{align*}
		P(Y_1 \leq y_1, Y_0 \leq y_0) = \begin{cases}
			P(Y_1 \leq y_1) &\text{ if } y_0 \geq \tilde{y}_0 \\
			0 &\text{ if } y_0 < \tilde{y}_0
		\end{cases}
	\end{align*}
\end{restatable}
\begin{proof}
	\singlespacing
	
	First suppose $y_0 < \tilde{y}_0$. Then $0 \leq P(Y_1 \leq y_1, Y_0 \leq y_0) \leq P(Y_0 \leq y_0) = 0$.
	
	Next suppose $y_0 \geq \tilde{y}_0$. Then $1 \geq P(\{Y_1 \leq y_1\} \cup \{Y_0 \leq y_0\}) \geq P(Y_0 \leq y_0) = 1$ implies that 
	\begin{align*}
		P(Y_1 \leq y_1, Y_0 \leq y_0) &= P(Y_1 \leq y_1) + \underbrace{P(Y_0 \leq y_0)}_{=1} - \underbrace{P(\{Y_1 \leq y_1\} \cup \{Y_0 \leq y_0\})}_{=1} \\
		& = P(Y_1 \leq y_1)
	\end{align*}
	which completes the proof.
\end{proof}

Lemma \ref{Lemma: identification, LATE IV model empirical content} below summarizes the empirical content of the model described in assumption \ref{Assumption: setting}. In particular, it implies that any two distributions of the primitives consistent with assumption \ref{Assumption: setting} that share the same marginal distribution of $(T, Z, X)$ and marginal, conditional distributions of 
\begin{align*}
	&Y_1 \mid T = a, X = x && Y_0 \mid T=n, X = x \\
	&Y_1 \mid T = c, X = x, &&Y_0 \mid T = c, X = x
\end{align*}
will produce the same distribution of observables.

\begin{restatable}[]{lemma}{lemmaLATEModelEmpiricalContent}
	\label{Lemma: identification, LATE IV model empirical content}
	\singlespacing
	
	Suppose assumpion \ref{Assumption: setting} holds. Then 
	\begin{align*}
		P(D = 1 \mid Z = 0, X = x) &= P(T = a \mid X = x) \\
		P(D = 0 \mid Z = 1, X = x) &= P(T = n \mid X = x) \\
		P(D = 1 \mid Z = 1, X = x) &= P(T \in \{a,c\} \mid X = x) \\
		P(D = 0 \mid Z = 0, X = x) &= P(T \in \{c,n\} \mid X = x) 
	\end{align*}
	and for any integrable function $f$, 
	\begin{align*}
		E[f(Y) \mid D = 1, Z = 1, X = x] &= E[f(Y_1) \mid T \in \{a,c\}, X = x] \\
		E[f(Y) \mid D = 0, Z = 0, X = x] &= E[f(Y_0) \mid T \in \{c,n\}, X = x] 
	\end{align*}
	Furthermore,
	\begin{align*}
		\text{ if } P(D = 1 \mid Z = 0, X = x) > 0, \text{ then } E[f(Y) \mid D = 1, Z = 0, X = x] = E[f(Y_1) \mid T = a, X = x] \\
		\text{ if } P(D = 0 \mid Z = 1, X = x) > 0, \text{ then } E[f(Y) \mid D = 0, Z = 1, X = x] = E[f(Y_0) \mid T = n, X = x]
	\end{align*}
\end{restatable}
\begin{proof}
	\singlespacing
	
	Assumption \ref{Assumption: setting} \ref{Assumption: setting, monotonicity} implies $\mathbbm{1}\{D_1 = 0, D_0 = 1\} = 0$. The definition of $T$ in \eqref{Defn: T, type of a unit} then implies
	\begin{align*}
		\mathbbm{1}\{D_0 = 1\} &= \mathbbm{1}\{D_1 = 1, D_0 = 1\} + \cancel{\mathbbm{1}\{D_1 = 0, D_0 = 1\}} = \mathbbm{1}\{T = a\} \\
		\mathbbm{1}\{D_1 = 0\} &= \mathbbm{1}\{D_1 = 0, D_0 = 0\} + \cancel{\mathbbm{1}\{D_1 = 0, D_0 = 1\}} = \mathbbm{1}\{T = n\}\\
		\mathbbm{1}\{D_1 = 1\} &= \mathbbm{1}\{D_1 = 1, D_0 = 1\} + \mathbbm{1}\{D_1 = 1, D_0 = 0\} = \mathbbm{1}\{T \in \{a, c\}\} \\
		\mathbbm{1}\{D_0 = 0\} &= \mathbbm{1}\{D_1 = 1, D_0 = 0\} + \mathbbm{1}\{D_1 = 0, D_0 = 0\} = \mathbbm{1}\{T \in \{c, n\} \}
	\end{align*}
	These observations, equation \eqref{Display: observed and potential treatment status}, and assumption \ref{Assumption: setting} \ref{Assumption: setting, instrument independence} imply 
	\begin{align*}
		P(D = 1 \mid Z = 0, X = x) &= P(D_0 = 1 \mid X = x) = P(T = a \mid X = x), \\
		P(D = 0 \mid Z = 1, X = x) &= P(D_1 = 0 \mid X = x) = P(T = n \mid X = x), \\
		P(D = 1 \mid Z = 1, X = x) &= P(D_1 = 1 \mid X = x) = P(T \in \{a,c\} \mid X = x), \text{ and } \\
		P(D = 0 \mid Z = 0, X = x) &= P(D_0 = 0 \mid X = x) = P(T \in \{c,n\} \mid X = x)
	\end{align*}
	Note the first two equalities can be summarized as $P(D = d \mid Z = z, X = x) = P(D_z = d \mid X = x)$.
	
	Next, let $f : \mathbb{R} \rightarrow \mathbb{R}$ be integrable. Assumption \ref{Assumption: setting} \ref{Assumption: setting, instrument independence} and equations \eqref{Display: observed and potential outcomes} and \eqref{Display: observed and potential treatment status} imply that for any $(d,z,x)$,
	\begin{align*}
		&P(D = d \mid Z = z, X = x) E[f(Y) \mid D = d, Z = z, X = x] \\
		&\hspace{1 cm} = P(D_z = d \mid X = x)E[f(Y_d) \mid D_z = d, X = x] 
	\end{align*}
	and since $P(D = d \mid Z = z, X = x) = P(D_z = d \mid X = x)$, this implies
	\begin{equation}
		0 = P(D = d \mid Z = z, X = x)\Big(E[f(Y) \mid D = d, Z = z, X = x] - E[f(Y_d) \mid D_z = d, X = x] \Big) \label{Display: lemma proof, identification, LATE IV model empirical content equating distributions}
	\end{equation}
	Assumption \ref{Assumption: setting} \ref{Assumption: setting, existence of compliers} implies 
	\begin{align*}
		P(D = 1 \mid Z = 1, X = x) &= P(T \in \{a,c\} \mid X = x) \geq P(T = c \mid X = x) > 0 \\
		P(D = 0 \mid Z = 0, X = x) &= P(T \in \{c,n\} \mid X = x) \geq P(T = c \mid X = x) > 0
	\end{align*}
	Use strict positivity of $P(D = 1 \mid Z = 1, X = x)$ and $P(D = 0 \mid Z = 0, X = x)$ to see that 
	\begin{align*}
		E[f(Y) \mid D = 1, Z = 1, X = x] = E[f(Y_1) \mid D_1 = 1, X = x] = E[f(Y_1) \mid T \in \{a,c\}, X = x] \\
		E[f(Y) \mid D = 0, Z = 0, X = x] = E[f(Y_0) \mid D_0 = 0, X = x] = E[f(Y_0) \mid T \in \{c,n\}, X = x]
	\end{align*}
	Similarly, \eqref{Display: lemma proof, identification, LATE IV model empirical content equating distributions} implies 
	\begin{align*}
		\text{ if } P(D = 1 \mid Z = 0, X = x) > 0, \text{ then } E[f(Y) \mid D = 1, Z = 0, X = x] = E[f(Y_1) \mid T = a, X = x] \\
		\text{ if } P(D = 0 \mid Z = 1, X = x) > 0, \text{ then } E[f(Y) \mid D = 0, Z = 1, X = x] = E[f(Y_0) \mid T = n, X = x]
	\end{align*}
	this concludes the proof.
\end{proof}

\lemmaLATEIVMarginalDistributionIdentification*
\begin{proof}
	\singlespacing
	
	First notice that using $T$ as defined in \eqref{Defn: T, type of a unit},
	\begin{equation}
		E[f(Y_d) \mid D_1 > D_0, X = x] = E[f(Y_d) \mid T = c, X = x] = \frac{E[f(Y_d) \mathbbm{1}\{T = c\} \mid X = x]}{P(T = c \mid X = x)} \label{Display: lemma proof, identification, LATE IV marginal distribution identification, rewrite with T}
	\end{equation}
	Now notice that 
	\begin{align*}
		D_1 - D_0 = (1 - D_0) - (1 - D_1) = \mathbbm{1}\{D_d = d\} - \mathbbm{1}\{D_{1-d} = d\}
	\end{align*}
	for either $d \in \{1,0\}$. Monotonicity (assumption \ref{Assumption: setting} \ref{Assumption: setting, monotonicity}) implies that this is an indicator for $T = c$:
	\begin{align*}
		D_1 - D_0 = \mathbbm{1}\{D_1 = 1, D_0 = 0\} = \mathbbm{1}\{T = c\}
	\end{align*}
	So, 
	\begin{align}
		&E[f(Y) \mathbbm{1}\{D = d\} \mid Z = d, X = x] - E[f(Y) \mathbbm{1}\{D = d\} \mid Z = 1-d, X = x] \notag \\
		&\hspace{1 cm} = E[f(Y_d) \mathbbm{1}\{D_d = d\} \mid X = x] - E[f(Y_d) \mathbbm{1}\{D_{1-d} = d\} \mid X = x] \notag \\
		&\hspace{1 cm} = E[f(Y_d) (\mathbbm{1}\{D_d = d\} - \mathbbm{1}\{D_{1-d} = d\}) \mid X = x] \notag \\
		&\hspace{1 cm} = E[f(Y_d) \mathbbm{1}\{T= c\} \mid X = x] \label{Display: lemma proof, identification, LATE IV marginal distribution identification, numerator}
	\end{align}
	
	Lemma \ref{Lemma: identification, LATE IV model empirical content} shows that
	\begin{align*}
		&P(D = 1 \mid Z = 1, X = x) - P(D = 1 \mid Z = 0, X = x) \\
		&\hspace{1 cm} = P(T \in \{a,c\} \mid X = x) - P(T = a \mid X = x) = P(T = c \mid X = x)
	\end{align*}
	and similarly, 
	\begin{align*}
		&P(D = 0 \mid Z = 0, X = x) - P(D = 0 \mid Z = 1, X = x) \\
		&\hspace{1 cm} = P(T \in \{c, n\} \mid X = x) - P(T = n \mid X = x) = P(T = c \mid X = x)
	\end{align*}
	Thus for either $d \in \{1,0\}$,
	\begin{equation}
		P(D = d \mid Z = d, X = x) - P(D = d \mid Z = 1 - d, X = x) = P(T = c \mid X = x). \label{Display: lemma proof, identification, LATE IV marginal distribution identification, denomenator}
	\end{equation}
	
	It follows from \eqref{Display: lemma proof, identification, LATE IV marginal distribution identification, rewrite with T}, \eqref{Display: lemma proof, identification, LATE IV marginal distribution identification, numerator}, and \eqref{Display: lemma proof, identification, LATE IV marginal distribution identification, denomenator} that
	\begin{align*}
		E_{P_{d \mid x}}[f(Y_d)] &= E[f(Y_d) \mid D_1 > D_0, X = x] \\
		&= \frac{E[f(Y) \mathbbm{1}\{D = d\} \mid X = x, Z = d] - E[f(Y) \mathbbm{1}\{D = d\} \mid X = x, Z = 1-d]}{P(D = d \mid X = x, Z = d) - P(D = d \mid X = x, Z = 1-d)},
	\end{align*}
	and from \eqref{Display: lemma proof, identification, LATE IV marginal distribution identification, denomenator} that
	\begin{align*}
		s_x &= P(X = x \mid D_1 > D_0) = P(X = x \mid T = c) = \frac{P(T = c \mid X = x)P(X = x)}{\sum_{x'} P(T = c \mid X = x')P(X = x')} \\
		&= \frac{\left[P(D = 1 \mid X = x, Z = 1) - P(D = 1 \mid X = x, Z = 0) \right]P(X = x) }{\sum_{x'} \left[P(D = 1 \mid X = x', Z = 1) - P(D = 1 \mid X = x', Z = 0) \right]P(X = x')}.
	\end{align*}
	This concludes the proof.
\end{proof}

%% file: appendix/OTJointPO_appendix_OTProperties.tex
\section{Appendix: properties of optimal transport}
\label{Appendix: properties of optimal transport}

Suppose that strong duality holds:
\begin{equation}
	\inf_{\pi \in \Pi(P_1,P_0)} \int c(y_1, y_0) d\pi(y_1,y_0) = \sup_{(\varphi, \psi) \in \Phi_c \cap (\mathcal{F}_c \times \mathcal{F}_c^c)} \int \varphi(y_1) dP_1(y_1) + \int \psi(y_0) dP_0(y_0) \label{Display: strong duality, differentiability of optimal transport}
\end{equation}
for sets of universally bounded functions $\mathcal{F}_c \subseteq L^1(P_1)$ and $\mathcal{F}_c^c \subseteq L^1(P_0)$. See lemmas \ref{Lemma: c-concave functions, smooth costs, strong duality} and \ref{Lemma: c-concave functions, indicator of convex set costs, strong duality} for examples.\footnote{
	\label{Footnote: how to find F_c and F_c^c}
	$\mathcal{F}_c$ and $\mathcal{F}_c^c$ are typically found with the following steps:
	\begin{enumerate}[label=(\roman*)]
		\item Start with a known strong duality result; for some $\Phi_{cs} \subseteq \Phi_c$, 
		\begin{equation*}
			\inf_{\pi \in \Pi(P_1,P_0)} \int c(y_1, y_0) d\pi(y_1,y_0) = \sup_{(\varphi, \psi) \in \Phi_{cs}} \int \varphi(y_1) dP_1(y_1) + \int \psi(y_0) dP_0(y_0)
		\end{equation*}
		\item Compute $\mathcal{F}_c(\Phi_{cs})$ and $\mathcal{F}_c^c(\Phi_{cs})$ defined by \eqref{Defn: F_c^c for arbitrary set of functions}. 
		\item Notice that $\mathcal{F}_c(\Phi_{cs}) \subseteq \mathcal{F}_c$ and $\mathcal{F}_c^c(\Phi_{cs}) \subseteq \mathcal{F}_c^c$ for known and easy to study sets $\mathcal{F}_c$, $\mathcal{F}_c^c$ 
	\end{enumerate}
	Lemma \ref{Lemma: c-concave functions, universal bound for c-concave functions} and remark \ref{Remark: discussion of universal bounds on c-concave functions} are useful to ensure $\mathcal{F}_c$ and $\mathcal{F}_c^c$ are universally bounded. 
}
 Then for suitable sets $\mathcal{F}_1$ and $\mathcal{F}_0$ with $\mathcal{F}_c \subseteq \mathcal{F}_1$ and $\mathcal{F}_c^c \subseteq \mathcal{F}_0$, the map $OT_c(P_1, P_0) = \inf_{\pi \in \Pi(P_1,P_0)} \int c(y_1, y_0) d\pi(y_1,y_0)$ can be viewed as 
\begin{align}
	&OT_c : \ell^\infty(\mathcal{F}_1) \times \ell^\infty(\mathcal{F}_0) \rightarrow \mathbb{R}, &&OT_c(P_1, P_0) = \sup_{(\varphi, \psi) \in \Phi_c \cap (\mathcal{F}_c \times \mathcal{F}_c^c)} P_1(\varphi) + P_0(\psi) \label{Display: optimal transport for continuity, differentiability}
\end{align}
where $P_d(f) = \int f(y_d) dP_d(y_d) = E_{P_d}[f(Y_d)]$. 

The functional in \eqref{Display: optimal transport for continuity, differentiability} is defined over the familiar Banach space $\ell^\infty(\mathcal{F}_1) \times \ell^\infty(\mathcal{F}_0)$. This makes it straightforward to show that optimal transport, as a functional from this space to $\mathbb{R}$, has certain desirable properties. 

\subsection{Continuity}
\label{Appendix: properties of optimal transport, continuity}

\begin{restatable}[Optimal transport is uniformly continuous]{lemma}{lemmaContinuityOfOptimalTransport}
	\label{Lemma: optimal transport is continuous}
	\singlespacing
	
	Suppose that for some universally bounded $\mathcal{F}_c \subseteq L^1(P_1)$ and $\mathcal{F}_c^c \subseteq L^1(P_0)$, \eqref{Display: strong duality, differentiability of optimal transport} holds. Then the optimal transport functional, given by \eqref{Display: optimal transport for continuity, differentiability}, is uniformly continuous.
\end{restatable}
\begin{proof}
	\singlespacing
	
	Define 
	\begin{align*}
		&\mathcal{S} : \ell^\infty(\mathcal{F}_1) \times \ell^\infty(\mathcal{F}_0) \rightarrow \ell^\infty(\mathcal{F}_1 \times \mathcal{F}_0), &&\mathcal{S}(H_1, H_0)(\varphi,\psi) = H_1(\varphi) + H_0(\psi) \\
		&\Xi_c : \ell^\infty(\mathcal{F}_1 \times \mathcal{F}_0) \rightarrow \mathbb{R}, &&\Xi_c[G] = \sup_{(\varphi,\psi) \in \Phi_c \cap (\mathcal{F}_c \times \mathcal{F}_c^c)} G(\varphi,\psi)
	\end{align*}
	and notice that $OT_c(H_1,H_0) = \Xi_c(\mathcal{S}(H_1, H_0))$. Since $s : \mathbb{R}^2 \rightarrow \mathbb{R}$ given by $s(h_1, h_2) = h_1 + h_2$ is uniformly continuous, we have that $\mathcal{S}$ is uniformly continuous (see lemma \ref{Lemma: bounded function spaces, continuity with f: R^K -> R^M}). Lemma \ref{Lemma: bounded function spaces, uniform continuity of restricted sup} shows that $\Xi_c$ is uniformly continuous. The composition of uniformly continuous functions is uniformly continuous, implying $OT_c$ is uniformly continuous. This completes the proof.
\end{proof}

\subsection{Directional Differentiability}
\label{Appendix: properties of optimal transport, subsection differentiability}

The optimal transport functional given by \eqref{Display: optimal transport for continuity, differentiability} is Hadamard directionally differentiable.\footnote{
	Recall the definition, found in \cite{fang2019inference}: let $\mathbb{D}$, $\mathbb{E}$ be Banach spaces (complete, normed, vector spaces), and $\phi : \mathbb{D}_\phi \subseteq \mathbb{D} \rightarrow \mathbb{E}$. $\phi$ is \textbf{Hadamard directionally differentiable} at $x_0 \in \mathbb{D}_\phi$ tangentially to $\mathbb{D}_T \subseteq \mathbb{D}$ if there exists a continuous map $\phi_{x_0}' : \mathbb{D}_T \rightarrow\mathbb{E}$ such that
	\begin{equation*}
		\lim_{n \rightarrow \infty} \left\lVert \frac{\phi(x_0 + t_n h_n) - \phi(x_0)}{t_n} - \phi_{x_0}'(h) \right\rVert_{\mathbb{E}} = 0
	\end{equation*}
	for all sequences $\{h_n\}_{n=1}^\infty \subseteq \mathbb{D}$ and $\{t_n\}_{n=1}^\infty\subseteq \mathbb{R}_+$ such that $h_n \rightarrow h \in \mathbb{D}_T$ and $t_n \downarrow 0$ as $n \rightarrow \infty$, and $x_0 + t_n h_n \in \mathbb{D}_\phi$ for all $n$.
	}
The formal result, stated below, requires that $\mathcal{F}_c$ and $\mathcal{F}_c^c$ each be equipped with a semimetric. The semimetrics chosen must be such that $P_1 \in \ell^\infty(\mathcal{F}_c)$ and $P_0 \in \ell^\infty(\mathcal{F}_c^c)$ are continuous and the product space $\mathcal{F}_c \times \mathcal{F}_c^c$ and its subset $\Phi_c \cap (\mathcal{F}_c \times \mathcal{F}_c^c)$ are compact. 

The setting suggests a very convenient semimetric. Let $P$ be the distribution of an observation, i.e. $(Y, D, Z, X) \sim P$. Note that under assumption \ref{Assumption: setting}, the distributions $P_{d \mid x}$ are dominated by $P$ with bounded densities $\frac{dP_{d \mid x}}{dP}$. Specifically, recall that 
\begin{align*}
	E_{P_{d \mid x}}[f(Y_d)] &= E[f(Y_d) \mid D_1 > D_0, X = x] \\
	&= \frac{E[f(Y) \mathbbm{1}\{D = d\} \mid Z = d, X = x] - E[f(Y) \mathbbm{1}\{D = d\} \mid Z = 1-d, X = x]}{P(D = d \mid Z = d, X = x) - P(D = d \mid Z = 1-d, X = x)}
\end{align*}
Let $\mathbbm{1}_{d,x,z}(D,X,Z) = \mathbbm{1}\{D = d, X = x, Z = z\}$, $p_{d,x,z} = P(D = d, X = x, Z = z)$, and $p_{x,z} = P(X = x, Z = z)$. Observe that
\begin{align*}
	&E[f(Y_d) \mid D_1 > D_0, X = x] = E\left[f(Y) \frac{\mathbbm{1}_{d,x,d}(D,X,Z)/p_{x,d} - \mathbbm{1}_{d,x,1-d}(D,X,Z)/p_{x,1-d}}{p_{d,x,d}/p_{x,d} - p_{d,x,1-d}/p_{x,1-d}}\right] \\
	&\hspace{1 cm} = E\left[f(Y) E\left[\frac{\mathbbm{1}_{d,x,d}(D,X,Z)/p_{x,d} - \mathbbm{1}_{d,x,1-d}(D,X,Z)/p_{x,1-d}}{p_{d,x,d}/p_{x,d} - p_{d,x,1-d}/p_{x,1-d}} \mid Y\right]\right]
\end{align*}
reveals the densities to be $\frac{dP_{d \mid x}}{dP}(Y) = E\left[\frac{\mathbbm{1}_{d,x,d}(D,X,Z)/p_{x,d} - \mathbbm{1}_{d,x,1-d}(D,X,Z)/p_{x,1-d}}{p_{d,x,d}/p_{x,d} - p_{d,x,1-d}/p_{x,1-d}} \mid Y\right]$. 

We now drop the subscript $x$ for the remainder of this appendix. Because $P$ dominates both $P_1$ and $P_0$ with bounded densities, the $L_{2,P}$ semimetric works very well:
\begin{equation}
	L_{2,P}(f_1,f_2) = \sqrt{P((f_1 - f_2)^2)} = \sqrt{E_P[(f_1(Y) - f_2(Y))^2]} \label{Defn: L2 semimetric, P}
\end{equation}
Equip the product space $\mathcal{F}_1 \times \mathcal{F}_0$ with the product semimetric:
\begin{equation}
	L_2((f_1,g_1), (f_2,g_2)) = \sqrt{L_{2,P}(f_1, f_2)^2 + L_{2,P}(g_1, g_2)^2} \label{Defn: L2 semimetric, product space}
\end{equation}
To apply the $L_{2,P}$ semimetric, each $f \in \mathcal{F}_1$ and $f \in \mathcal{F}_0$ are defined on whole domain $\mathcal{Y}$.

\begin{restatable}[Hadamard directional differentiability of optimal transport]{lemma}{lemmaHadamardDirectionalDifferentiabilityOfOptimalTransport}
	\label{Lemma: Hadamard differentiability, optimal transport}
	\singlespacing
	
	Let $c : \mathcal{Y} \times \mathcal{Y} \rightarrow \mathbb{R}$ be lower semicontinuous, $\mathcal{F}_1, \mathcal{F}_0$ be sets of measurable functions mapping $\mathcal{Y}$ to $\mathbb{R}$, and $\mathcal{F}_c \subseteq \mathcal{F}_1$ and $\mathcal{F}_c^c \subseteq \mathcal{F}_0$ be universally bounded subsets. Suppose that
	\begin{enumerate}
		\item Strong duality holds:
		\begin{equation*}
			\inf_{\pi \in \Pi(P_1,P_0)} \int c(y_1, y_0) d\pi(y_1,y_0) = \sup_{(\varphi, \psi) \in \Phi_c \cap (\mathcal{F}_c \times \mathcal{F}_c^c)} \int \varphi(y_1) dP_1(y_1) + \int \psi(y_0) dP_0(y_0),
		\end{equation*}
		\item $P$ dominates $P_1$ and $P_0$ with bounded densities, 
		\item $\mathcal{F}_d$ is $P$-Donsker and $\sup_{f \in \mathcal{F}_d} \lvert P(f) \rvert < \infty$ for each $d = 1, 0$, and
		\item $(\mathcal{F}_1 \times \mathcal{F}_0, L_2)$ and the subset
		\begin{equation*}
			\Phi_c \cap (\mathcal{F}_c \times \mathcal{F}_c^c) = \left\{(\varphi,\psi) \in \mathcal{F}_c \times \mathcal{F}_c^c \; ; \; \varphi(y_1) + \psi(y_0) \leq c(y_1, y_0) \right\}
		\end{equation*}
		are complete. 
	\end{enumerate}
	Then $OT_c : \ell^\infty(\mathcal{F}_1) \times \ell^\infty(\mathcal{F}_0) \rightarrow \mathbb{R}$ defined by
	\begin{equation*}
		OT_c(P_1, P_0) = \sup_{(\varphi, \psi) \in \Phi_c \cap (\mathcal{F}_c \times \mathcal{F}_c^c)} P_1(\varphi) + P_0(\psi) 
	\end{equation*}
	is Hadamard directionally differentiable at $(P_1, P_0)$ tangentially to 
	\begin{equation}
		\mathbb{D}_{Tan} = \mathcal{C}(\mathcal{F}_1, L_{2, P}) \times \mathcal{C}(\mathcal{F}_0, L_{2, P}). \label{Defn: tangent set for optimal transport derivative}
	\end{equation}
	The set of maximizers $\Psi_c(P_1, P_0) = \argmax_{(\varphi,\psi) \in \Phi_c \cap (\mathcal{F}_c \times \mathcal{F}_c^c)} P_1(\varphi) + P_0(\psi)$ is nonempty, and the derivative $OT_{c, (P_1,P_0)}' : \mathbb{D}_{Tan} \rightarrow \mathbb{R}$ is given by
	\begin{equation*}
		OT_{c,(P_1,P_0)}'(H_1, H_0) = \sup_{(\varphi,\psi) \in \Psi_c(P_1, P_0)} H_1(\varphi) + H_0(\psi)
	\end{equation*}
\end{restatable}
\begin{proof}
	\singlespacing
	
	For legibility, the proof is broken down into four steps:
	\begin{enumerate}
		\item Define 
		\begin{align*}
			&\mathcal{S} : \ell^\infty(\mathcal{F}_1) \times \ell^\infty(\mathcal{F}_0) \rightarrow \ell^\infty(\mathcal{F}_1 \times \mathcal{F}_0), &&\mathcal{S}(H_1, H_0)(\varphi,\psi) = H_1(\varphi) + H_0(\psi) \\
			&\Xi_c : \ell^\infty(\mathcal{F}_1 \times \mathcal{F}_0) \rightarrow \mathbb{R}, &&\Xi_c[G] = \sup_{(\varphi,\psi) \in \Phi_c \cap (\mathcal{F}_c \times \mathcal{F}_c^c)} G(\varphi,\psi)
		\end{align*}
		and notice that $OT_c(H_1,H_0) = \Xi_c(\mathcal{S}(H_1, H_0))$. This suggests application of the chain rule.
		
		\item $\mathcal{S}$ is linear and continuous at every point of $\ell^\infty(\mathcal{F}_1) \times \ell^\infty(\mathcal{F}_0)$, which implies it is (fully) Hadamard differentiable at any $(H_1,H_0) \in \ell^\infty(\mathcal{F}_1) \times \ell^\infty(\mathcal{F}_0)$ tangentially to $\ell^\infty(\mathcal{F}_1) \times \ell^\infty(\mathcal{F}_0)$, and is its own derivative. Indeed, for any $(H_{1n}, H_{0n}) \rightarrow (H_1, H_0) \in \ell^\infty(\mathcal{F}_1) \times \ell^\infty(\mathcal{F}_0)$ and any $t_n \downarrow 0$,
		\begin{align*}
			&\lim_{n \rightarrow \infty} \left\lVert \frac{\mathcal{S}((H_1, H_0) + t_n (H_{1n}, H_{0n})) - \mathcal{S}(H_1, H_0)}{t_n} - \mathcal{S}(H_1, H_0) \right\rVert_{\mathcal{F}_c \times \mathcal{F}_c^c} \\
			&\hspace{3 cm} = \lim_{n \rightarrow \infty} \left\lVert \mathcal{S}(H_{1n}, H_{0n}) - \mathcal{S}(H_1, H_0) \right\rVert_{\mathcal{F}_c \times \mathcal{F}_c^c} = 0
		\end{align*}
		
		\item Consider $\Xi_c$. Verify the conditions of lemma \ref{Lemma: Hadamard differentiability, supremum of bounded function}:
		\begin{enumerate}
			\item $(\mathcal{F}_1 \times \mathcal{F}_0, L_2)$ and the subset $\Phi_c \cap (\mathcal{F}_c \times \mathcal{F}_c^c)$ are compact. 
			
			First recall that a subset of semimetric space is compact if and only if it is totally bounded and complete.\footnote{See \cite{vaart1997weak}, footnote on p. 17.} Completeness of both sets is assumed, so it suffices to show they are totally bounded. Since $\Phi_c \cap (\mathcal{F}_c \times \mathcal{F}_c^c)$ is a subset of $\mathcal{F}_1 \times \mathcal{F}_0$, it suffices to show the latter set is totally bounded. 
			
			Using the assumption that $\mathcal{F}_d$ is $P$-Donsker and $\sup_{f \in \mathcal{F}_d} \lvert P(f) \rvert < \infty$, we have that $\sup_{\varphi \in \mathcal{F}_c} \lvert P(\varphi) \rvert < \infty$ and $(\mathcal{F}_d, L_{2, P})$ is totally bounded (see \cite{vaart1997weak} problem 2.1.2.). It follows that the product space $(\mathcal{F}_1 \times \mathcal{F}_0, L_2)$ is totally bounded.\footnote{
				For $\varepsilon > 0$, let $(f_1,\ldots, f_K)$ be the centers of $L_{2,P}$-balls of radius $\varepsilon / \sqrt{2}$ that cover $\mathcal{F}_1$, and $(g_1, \ldots, g_M)$ be the center of $L_{2,P}$-balls of radius $\varepsilon / \sqrt{2}$ that cover $\mathcal{F}_0$. Then for any $(f, g) \in \mathcal{F}_1 \times \mathcal{F}_0$, there exists $f_k$ and $g_m$ such that $L_{2,P}(f,f_k) < \varepsilon / \sqrt{2}$ and $L_{2,P}(g, g_m) < \varepsilon/  \sqrt{2}$, and so
				\begin{equation*}
					L_2((f,g), (f_k, g_m) = \sqrt{L_{2,P}(f, f_k)^2 + L_{2,P}(g,g_m)^2} < \sqrt{(\varepsilon/\sqrt{2})^2 + (\varepsilon/\sqrt{2})^2} = \varepsilon
				\end{equation*}
				and thus the $KM$ balls in $(\mathcal{F}_1 \times \mathcal{F}_0)$ of radius $\varepsilon$ centered at $(f_k, g_m)$ for some $k,m$ cover $\mathcal{F}_1 \times \mathcal{F}_0$.
			}
			
			\item $\mathcal{S}(P_1, P_0) \in \mathcal{C}(\mathcal{F}_1 \times \mathcal{F}_0, L_2)$. 
			
			Notice that 
			\begin{equation*}
				\lvert P_1(f_1) - P_1(f_2) \rvert \leq P_1(\lvert f_1 - f_2 \rvert) \leq \sqrt{P_1((f_1 - f_2)^2)} = L_{2, P_1}(f_1, f_2)
			\end{equation*}
			where the second inequality is an applications of Jensen's inequality. This implies $P_1 \in \mathcal{C}(\mathcal{F}_1, L_{2, P_1})$. Moreover, since $P_1 \ll P$ and $\frac{dP_1}{dP} \leq K_1 < \infty$ for some $K_1 \in \mathbb{R}$,
			\begin{align*}
				L_{2,P_1}(f_1, f_2) = \left(\int (f_1 - f_2)^2 \frac{dP_1}{dP} dP \right)^{1/2} \leq K_1^{1/2} \left(\int (f_1 - f_2)^2 dP\right)^{1/2} = K_1^{1/2} L_{2,P}(f_1, f_2) 
			\end{align*}
			shows that $\mathcal{C}(\mathcal{F}_1, L_{2,P_1}) \subseteq \mathcal{C}(\mathcal{F}_1, L_{2,P})$ and so $P_1 \in \mathcal{C}(\mathcal{F}_1, L_{2,P})$. A similar argument shows $P_0 \in \mathcal{C}(\mathcal{F}_0, L_{2, P})$. 
			
			Use the inequalities above to see that 
			\begin{align*}
				&\lvert \mathcal{S}(P_1, P_0)(f_1,g_1) - \mathcal{S}(P_1, P_0)(f_2,g_2) \rvert = \lvert P_1(f_1) - P_1(f_2) + P_0(g_1) - P_0(g_2) \rvert \notag \\
				&\hspace{1 cm} \leq L_{2, P_1}(f_1, f_2) + L_{2, P_0}(g_1, g_2) \leq K_1^{1/2}L_{2,P}(f_1, f_2) + K_0^{1/2} L_{2,P}(\psi_1, \psi_2) \\
				&\hspace{1 cm} \leq 2\max\{K_1^{1/2}, K_0^{1/2}\}\max\{L_{2, P}(f_1, f_2), L_{2, P}(g_1, g_2)\} \notag \\
				&\hspace{1 cm} = 2\max\{K_1^{1/2}, K_0^{1/2}\} \sqrt{\max\{L_{2, P}(f_1, f_2)^2, L_{2, P}(g_1, g_2)^2\}} \\
				&\hspace{1 cm} \leq 2\max\{K_1^{1/2}, K_0^{1/2}\}\sqrt{L_{2, P}(f_1, f_2)^2 + L_{2, P}(g_1, g_2)^2} \notag \\
				&\hspace{1 cm} = 2\max\{K_1^{1/2}, K_0^{1/2}\}L_2((f_1, g_1), (f_2, g_2)) 
			\end{align*}
			hence $L_2((f_1, g_1), (f_2, g_2)) < \varepsilon /(2\max\{K_1^{1/2}, K_0^{1/2}\})$ implies 
			\begin{equation*}
				\lvert \mathcal{S}(P_1, P_0)(f_1, g_1) - \mathcal{S}(P_1, P_0)(f_2, g_2) \rvert < \varepsilon
			\end{equation*}
			and therefore $\mathcal{S}(P_1, P_0) \in \mathcal{C}(\mathcal{F}_1 \times \mathcal{F}_0, L_2)$.
		\end{enumerate}
		Lemma \ref{Lemma: Hadamard differentiability, supremum of bounded function} shows that $\Xi_c$ is Hadamard directionally differentiable at $\mathcal{S}(P_1, P_0)$ tangentially to $\mathcal{C}(\mathcal{F}_1 \times \mathcal{F}_0, L_2)$, with derivative
		\begin{align*}
			&\Xi_{c, \mathcal{S}(P_1, P_0)}' : \mathcal{C}(\mathcal{F}_1 \times \mathcal{F}_0, L_2) \rightarrow \mathbb{R}, &&\Xi_{c, \mathcal{S}(P_1, P_0)}'(H) = \sup_{(\varphi, \psi) \in \Psi_c(P_1,P_0)} H(\varphi,\psi)
		\end{align*}
		where $\Psi_c(P_1,P_0) = \argmax_{(\varphi,\psi) \in \Phi_c \cap (\mathcal{F}_c \times \mathcal{F}_c^c)} P_1(\varphi) + P_0(\psi)$ is nonempty, because $P_1 + P_0 = \mathcal{S}(P_1,P_0)$ is continuous and $\Phi_c \cap (\mathcal{F}_c \times \mathcal{F}_c^c)$ is compact. 
		
		\item Now consider the tangent spaces to ensure the composition of the derivatives is well defined. Observe that if $(H_1, H_0) \in \mathcal{C}(\mathcal{F}_1, L_{2,P}) \times \mathcal{C}(\mathcal{F}_0, L_{2, P})$ then $\mathcal{S}(H_1, H_0) = H_1 + H_0 \in \mathcal{C}(\mathcal{F}_1 \times \mathcal{F}_0, L_2)$.\footnote{
			Fix $(f,g) \in \mathcal{F}_1 \times \mathcal{F}_0$ and let $\delta_1 > 0$ and $\delta_0 > 0$ be such that $L_{2,P_1}(f, \tilde{f}) < \delta_1$ implies $H_1(f, \tilde{f}) < \varepsilon/2$ and $L_{2,P_0}(g,\tilde{g}) < \delta_0$ implies $H_0(g,\tilde{g}) < \varepsilon/2$. The inequality 
			\begin{align*}
				L_{2, P}(f, \tilde{f}) + L_{2, P}(g,\tilde{g}) &\leq 2\max\{L_{2, P}(f, \tilde{f}), L_{2, P}(g,\tilde{g})\} \\
				&= 2\sqrt{\max\{L_{2, P}(f, \tilde{f})^2, L_{2, P}(g,\tilde{g})^2\}} = 2L_2((f, g), (\tilde{f}, \tilde{g}))
			\end{align*}
			implies that if $L_2((f, g), (\tilde{f}, \tilde{g})) < \min\{\delta_1,\delta_2\}/2$ then $\lvert \mathcal{S}(H_1, H_0)(f,g) - \mathcal{S}(H_1, H_0)(\tilde{f}, \tilde{g}) \rvert \leq \lvert H_1(f) - H_1(\tilde{f}) \rvert + \lvert H_0(g) - H_0(\tilde{g}) \rvert < \varepsilon$.
			} 
			It follows from the chain rule (lemma \ref{Lemma: Hadamard differentiability, chain rule}) that $OT_c$ is Hadamard directionally differentiable at $(P_1, P_0)$ tangentially to $\mathcal{C}(\mathcal{F}_1, L_{2,P}) \times \mathcal{C}(\mathcal{F}_0, L_{2, P})$ with derivative $OT_c : \mathcal{C}(\mathcal{F}_1, L_{2,P}) \times \mathcal{C}(\mathcal{F}_0, L_{2, P}) \rightarrow \mathbb{R}$ given by
		\begin{equation*}
			OT_{c, (P_1,P_0)}'(H_1, H_0) = \Xi_{c, \mathcal{S}(P_1,P_0)}'(\mathcal{S}_{(P_1,P_0)}'(H_1,H_0)) = \sup_{(\varphi,\psi) \in \Psi_c(P_1,P_0)} H_1(\varphi) + H_0(\psi)
		\end{equation*}
	\end{enumerate}	
\end{proof}

\subsection{Full differentiability}
\label{Appendix: properties of optimal transport, subsection differentiability, subsubsection full differentiability}

The property distinguishing directional from full differentiability on a subspace is linearity of the derivative (\cite{fang2019inference}, proposition 2.1). In the case of optimal transport, the derivative found in lemma \ref{Lemma: Hadamard differentiability, optimal transport} is linear on a large subspace of the tangent space when the solution to the dual problem is suitably unique. When it holds, this is sufficient for simpler bootstrap procedures to work for inference. 

The dual solutions
\begin{equation*}
	(\varphi,\psi) \in \Psi_c(P_1, P_0) = \argmax_{(\varphi,\psi) \in \Phi_c \cap (\mathcal{F}_c \times \mathcal{F}_c^c)} P_1(\varphi) + P_0(\psi)
\end{equation*}
are referred to as \textbf{Kantorovich potentials}. Notice that for any $s \in \mathbb{R}$, 
\begin{equation*}
	P_1(\varphi + s) + P_0(\psi - s) = P_1(\varphi) + P_0(\psi)
\end{equation*}
shows the most one can hope for is uniqueness up to a constant; if $(\varphi,\psi) \in \Psi_c(P_1,P_0)$, then $(\varphi + s, \psi - s) \in \Psi_c(P_1, P_0)$ as well.\footnote{
	See \cite{staudt2022uniqueness} for extended discussion on uniqueness of Kantorovich potentials.
	}
It is well known in the optimal transport literature that when the distributions $P_1$, $P_0$ have full support on a convex, compact subset of $\mathbb{R}$ and $c$ is differentiable, the Kantorovich potential is indeed unique in this way on the supports of $P_1$ and $P_0$.

\begin{restatable}[]{lemma}{lemmaKantorovichPotentialSufficientConditionsForUniqueness}
	\label{Lemma: Kantorovich potential, sufficient conditions for uniqueness}
	\singlespacing
	
	Suppose that 
	\begin{enumerate}
		\item $c(y_1, y_0)$ is continuously differentiable. 
		\item $P_d$ has compact support $\mathcal{Y}_d = [y_d^\ell, y_d^u] \subseteq \mathbb{R}$, and
	\end{enumerate}
	Let $\mathcal{F}_c$ and $\mathcal{F}_c^c$ be defined by \eqref{Defn: F_c for smooth costs} and \eqref{Defn: F_c^c for smooth costs} respectively, and 
	\begin{equation*}
		\Psi_c(P_1, P_0) = \argmax_{(\varphi,\psi) \in \Phi_c \cap (\mathcal{F}_c \times \mathcal{F}_c^c)} P_1(\varphi) + P_0(\psi)
	\end{equation*}
	Then for any $(\varphi_1,\psi_1), (\varphi_2, \psi_2) \in \Psi_c(P_1, P_0)$, there exists $s \in \mathbb{R}$ such that for all $(y_1, y_0) \in \mathcal{Y}_1 \times \mathcal{Y}_0$
	\begin{align*}
		&\varphi_1(y_1) - \varphi_2(y_1) = s, &&\psi_1(y_0) - \psi_2(y_0) = -s
	\end{align*}
\end{restatable}
\begin{proof}
	\singlespacing
	
	The proof is quite similar to that of \cite{santambrogio2015optimal} proposition 7.18. 
	
	Let $(\varphi_1, \psi_1), (\varphi_2,\psi_2) \in \Psi_c(P_1,P_0)$. For $k = 1,2$, $\varphi_k$ and $\psi_k$ (being elements of $\mathcal{F}_c$ and $\mathcal{F}_c^c$ respectively) are $L$-Lipschitz and hence absolutely continuous. This implies all four functions are differentiable Lebesgue-almost everywhere, and that for any $(y_1,y_0) \in \mathcal{Y}_1 \times \mathcal{Y}_0$,
	\begin{align*}
		&\varphi_k(y_1) = \varphi_k(y_1^\ell) + \int_{y_1^\ell}^{y_1} \varphi_k'(y) dy &&\psi_k(y_0) = \psi_k(y_0^\ell) + \int_{y_0^\ell}^{y_0} \psi_k'(y) dy
	\end{align*}
	Notice that the subset of $\mathcal{Y}_1$ where both $\varphi_1$ and $\varphi_2$ are differentiable also has full Lebesgue measure. It suffices to show that $\varphi_1'(y_1) = \varphi_2'(y_1)$ on this set (and $\psi_1'(y_0) = \psi_2'(y_0)$ on the subset of $\mathcal{Y}_0$ where both $\psi_1$ and $\psi_2$ are differentiable, which also has full Lebesgue measure), from which it will follow that for any $(y_1,y_0) \in \mathcal{Y}_1 \times \mathcal{Y}_0$,
	\begin{align*}
		\varphi_1(y_1) - \varphi_2(y_1) = \varphi_1(y_1^\ell) - \varphi_2(y_1^\ell) + \int_{y_1^\ell}^{y_1} (\varphi_1'(y) - \varphi_2'(y)) dy = \underbrace{\varphi_1(y_1^\ell) - \varphi_2(y_1^\ell)}_{\coloneqq s_\varphi} \\
		\psi_1(y_0) - \psi_2(y_0) = \psi_1(y^\ell) - \psi_2(y^\ell) + \int_{y_0^\ell}^{y_0} (\psi_1'(y) - \psi_2'(y)) dy = \underbrace{\psi_1(y_0^\ell) - \psi_2(y_0^\ell)}_{\coloneqq s_\psi}
	\end{align*}
	Finally, observe that $P_1(\varphi_2) + P_0(\varphi_2) = P_1(\varphi_1) + P_0(\psi_1) = P_1(\varphi_2 + s_\varphi) + P_0(\psi_2 + s_\psi) = P_1(\varphi_2) + P_0(\psi_2) + s_\varphi + s_\psi$ implies $s_\varphi = -s_\psi$. \\

	The remainder of the proof shows that for any $\bar{y}_1$ in the set where both $\varphi_1$ and $\varphi_2$ are differentiable, $\varphi_1'(\bar{y}_1) = \varphi_2'(\bar{y}_1)$. The same arguments work to show the corresponding claim regarding $\psi_1$ and $\psi_2$. \\
	
	There exists $\pi \in \Pi(P_1,P_0)$ that solves the primal problem (see lemma \ref{Lemma: optimal transport is attained}). For any such $\pi$,
	\begin{enumerate}
		\item $\text{Supp}(P_1) = \left\{y_1 \in \mathcal{Y}_1 \; ; \; \exists y_0 \in \mathcal{Y}_0 \text{ s.t. } (y_1, y_0) \in \text{Supp}(\pi)\right\}$ \label{Claim: Kantorovich potential, sufficient conditions for uniqueness, projection of support equals support}
		
		This follows because $\text{Pr}_1(\text{Supp}(\pi)) \coloneqq \left\{y_1 \in \mathcal{Y}_1 \; ; \; \exists y_0 \in \mathcal{Y}_0 \text{ s.t. } (y_1, y_0) \in \text{Supp}(\pi)\right\}$ is dense in $\text{Supp}(P_1)$, and $\text{Pr}_1(\text{Supp}(\pi))$ is closed because $\mathcal{Y}_0$ is compact.\footnote{
			Specifically, for any $A \subseteq \mathcal{Y}_1 \times \mathcal{Y}_0 \subseteq \mathbb{R}^2$, let $\text{Pr}_1(A) = \left\{y_1 \in \mathcal{Y}_1 \; ; \; \exists y_0 \in \mathcal{Y}_0 \text{ s.t. } (y_1,y_0) \in A\right\}$ be the cartesian projection of the set $A$ onto the first coordinate. Let $P_1 \in \mathcal{P}(\mathcal{Y}_1)$, $P_0 \in \mathcal{P}(\mathcal{Y}_0)$, and $\pi \in \Pi(P_1,P_0)$. As noted in \cite{staudt2022uniqueness} (Remark 1), $\text{Pr}_1(\text{Supp}(\pi)) \subseteq \text{Supp}(P_1)$ with the possibility that inclusion is strict.
			
			However, $\text{Pr}_1(\text{Supp}(\pi))$ is always dense in $\text{Supp}(P_1)$: let $y_1 \in \text{Supp}(P_1)$ and $\delta > 0$ be arbitrary, and suppose for contradiction that $B_\delta(y_1) \cap \text{Pr}_1(\text{Supp}(\pi)) = \varnothing$. Then $\big(B_\delta(y_1) \times \mathcal{Y}_0\big) \cap \text{Supp}(\pi) = \varnothing$ follows from the definition of $\text{Pr}_1(\text{Supp}(\pi))$, and thus
			\begin{align*}
				0 &= \pi\left(\big(B_\delta(y_1) \times \mathcal{Y}_0\big) \cap \text{Supp}(\pi)\right) = \pi\left(\big(B_\delta(y_1) \times \mathcal{Y}_0\big)\right) + \pi\left(\text{Supp}(\pi)\right) - \pi\left(\big(B_\delta(y_1) \times \mathcal{Y}_0\big) \cup \text{Supp}(\pi)\right) \\
				&= \pi\left(\big(B_\delta(y_1) \times \mathcal{Y}_0\big)\right) = P_1(B_\delta(y_1)) > 0
			\end{align*}
			a contradiction showing $B_\delta(y_1) \cap \text{Pr}_1(\text{Supp}(\pi)) \neq \varnothing$. Thus $\text{Pr}_1(\text{Supp}(\pi))$ is dense in $\text{Supp}(P_1)$.
			
			Moreover, if $\mathcal{Y}_0$ is compact then the map $\text{Pr}_1$ is closed: suppose $A \subseteq \mathcal{Y}_1 \times \mathcal{Y}_0 \subseteq \mathbb{R}^2$ is closed, and $\{y_{1n}\}_{n=1}^\infty \subseteq \text{Pr}_1(A)$ converges to $y_1$. Then there exists $\{y_{0n}\}_{n=1}^\infty \subseteq \mathcal{Y}_0$ such that $(y_{1n}, y_{0n}) \in A$ for each $n$. Since $\mathcal{Y}_0$ is compact, there exists a subsequence $\{y_{0n_k}\}_{k=1}^\infty$ and $y_0$ such that $\lim_{k\rightarrow \infty} y_{0n_k} = y_0$. Then notice that $\lim_{k \rightarrow \infty} (y_{1n_k}, y_{0n_k}) = (y_1,y_0)$. Since $A$ is closed, $(y_1,y_0) \in A$. 
			
			$\text{Supp}(\pi)$ is closed by definition, hence $\text{Pr}_1(\text{Supp}(\pi))$ is closed and dense in $\text{Supp}(P_1)$, from which it follows that $\text{Supp}(\pi) = \text{Supp}(P_1)$.
		}
		
		\item For all $(y_1, y_0) \in \text{Supp}(\pi)$, $\varphi_k(y_1) + \psi_k(y_0) = c(y_1, y_0)$.
		
		It is easy to see that the equality holds $\pi$-almost surely. To see it holds specifically on the support, notice that optimality of $\pi$ and $(\varphi_k, \psi_k)$ implies that 
		\begin{equation*}
			\int c(y_1, y_0) d\pi(y_1, y_0) = \int \varphi_k(y_1) dP(y_1) + \int \psi_k(y_0) dP_0(y_0) 
		\end{equation*}
		and recall that $\varphi_k(y_1) + \psi_k(y_0) \leq c(y_1, y_0)$ holds for all $(y_1, y_0) \in \mathcal{Y} \times \mathcal{Y}$. If the inequality were strict for some $(y_1', y_0') \in \text{Supp}(\pi)$, then continuity of $\varphi_k$, $\psi_k$, and $c$ would imply the inequality is sharp on a ball centered at $(y_1, y_0)$ of some positive radius, denoted $B$, leading to the contradiction
		\begin{align*}
			\int c(y_1, y_0) d\pi(y_1, y_0) &= \int_B c(y_1, y_0) d\pi(y_1,y_0) + \int_{B^c} c(y_1, y_0) d\pi(y_1,y_0) \\
			&> \int_B \varphi_k(y_1) + \psi_k(y_0) d\pi(y_1,y_0) + \int_{B^c} \varphi_k(y_1) + \psi_k(y_0) d\pi(y_1,y_0) \\
			&= \int \varphi_k(y_1) + \psi_k(y_0) d\pi(y_1,y_0) = \int \varphi_k(y_1) dP_1(y_1) + \int \psi_k(y_0) dP_0(y_0)
		\end{align*}
		
		\item For any $\bar{y}_1 \in \text{Supp}(P_1)$, the above implies there there exists $\bar{y}_0 \in \mathcal{Y}_0$ such that $(\bar{y}_1, \bar{y}_0) \in \text{Supp}(\pi)$, and hence $\varphi_k(\bar{y}_1) + \psi_k(\bar{y}_0) = c(\bar{y}_1, \bar{y}_0)$. For any such $\bar{y}_0$, 
		\begin{equation}
			y_1 \mapsto \varphi_k(y_1) - c(y_1, \bar{y}_0) \text{ is maximized at } \bar{y}_1 \label{Display: lemma proof, Kantorovich potential, sufficient conditions for uniqueness, map must be maximized}
		\end{equation}
		
		Indeed, if there were $y_1' \in \mathcal{Y}_1$ such that $\varphi_k(y_1') - c(y_1', \bar{y}_0) > \varphi_k(\bar{y}_1) - c(\bar{y}_1, \bar{y}_0)$, then by adding $\psi_k(\bar{y}_0)$ to both sides we find
		\begin{align*}
			\varphi_k(y_1') + \psi_k(\bar{y}_0) - c(y_1', \bar{y}_0) > \varphi_k(\bar{y}_1) + \psi_k(\bar{y}_0) - c(\bar{y}_1, \bar{y}_0) = 0
		\end{align*}
		This implies $\varphi_k(y_1') + \psi_k(\bar{y}_0) > c(y_1', \bar{y}_0)$, which contradicts $\varphi_k(y_1') + \psi_k(\bar{y}_0) \leq c(y_1', \bar{y}_0)$ for all $(y_1, y_0) \in \mathcal{Y}_1 \times \mathcal{Y}_0$.
		
		\item Now observe that if $\bar{y}_1 \in (y_1^\ell, y_1^u)$ is a point at which $\varphi_k$ is differentiable, then \eqref{Display: lemma proof, Kantorovich potential, sufficient conditions for uniqueness, map must be maximized} implies $\varphi_k'(\bar{y}_1) = \frac{\partial c}{\partial y_1}(\bar{y}_1, \bar{y}_0)$.\footnote{Notice that the ``choice'' of $\pi$ or $\bar{y}_0$ doesn't matter, because $\varphi_k'(\bar{y}_1)$ can take only one value.} Thus if $\bar{y}_1 \in (y_1^\ell, y_1^u)$ is a point at which both $\varphi_1$ and $\varphi_2$ are differentiable, then 
		\begin{equation*}
			\varphi_1(\bar{y}_1) = \frac{\partial c}{\partial y_1}(\bar{y}_1, \bar{y}_0) = \varphi_2(\bar{y}_1)
		\end{equation*}
	\end{enumerate}
	This completes the proof.
\end{proof}

To specify the subset of the tangent space on which $OT_{c,(P_1,P_0)}'$ is linear, let $\mathcal{Y}_d \subseteq \mathcal{Y}$ and $\mathbbm{1}_{\mathcal{Y}_d}(y) = \mathbbm{1}\{y \in \mathcal{Y}_d\}$. Let $\mathcal{G}$ denote a set of real-valued functions $g : \mathcal{Y} \rightarrow \mathbb{R}$ with the following property: if $g \in \mathcal{G}$, then $\mathbbm{1}_{\mathcal{Y}_d} \times g \in \mathcal{G}$.\footnote{If we have a set $\tilde{\mathcal{G}}$ that does not satisfy this property, the set $\mathcal{G} = \tilde{\mathcal{G}} \cup \left\{\mathbbm{1}_{\mathcal{Y}_d} \times g \; ; \; g \in \tilde{\mathcal{G}}\right\}$ will satisfy it.} Let $\ell_{\mathcal{Y}_d}^\infty(\mathcal{G})$ be the set of bounded, linear functions $H : \mathcal{G} \rightarrow \mathbb{R}$ that evaluate constant functions to zero and ``ignore'' the value of functions outside of $\mathcal{Y}_d$. Specifically, define
\begin{align}
	\ell_{\mathcal{Y}_d}^\infty(\mathcal{G})  &= \Big\{H \in \ell^\infty(\mathcal{G}) \; ; \; \text{ for all } a,b \in \mathbb{R} \text{ and } f, g \in \mathcal{G}, \notag \\
	&\hspace{3.5 cm} (i) \; H(f) = H(\mathbbm{1}_{\mathcal{Y}_d} \times f), \; \; (ii) \; \text{ if } a \in \mathcal{G} \text{ then } H(a) = 0, \text{ and } \notag \\
	&\hspace{3.5 cm} (iii) \; \text{ if } a f + b g \in \mathcal{G} \text{ then } H(a f + b g) = aH(f) + bH(g) \Big\}  \label{Defn: ell_Yd^infty(G)}
\end{align}
Here we slightly abuse notation; $a \in \mathcal{G}$ refers to the function mapping each point in $\mathcal{Y}$ to the constant $a \in \mathbb{R}$. Equip $\ell_{\mathcal{Y}_d}^\infty(\mathcal{G})$ with the supremum norm, $\lVert H \rVert_\mathcal{G} = \lVert H \rVert_\infty = \sup_{g \in \mathcal{G}} \lvert H(g) \rvert$. As shown in appendix \ref{Appendix: weak convergence}, first stage estimators of $(P_1, P_0)$ based on the empirical distribution have weak limits concentrated on $\ell_{\mathcal{Y}_1}^\infty(\mathcal{F}_c) \times \ell_{\mathcal{Y}_0}^\infty(\mathcal{F}_c^c)$ where $\mathcal{Y}_d$ is the support of $P_d$. 

\begin{restatable}[]{lemma}{BoundedFunctionSpacesSubspaceOfLinearFunctionsAssigningZeroToConstantsIsClosed}
	\label{Lemma: bounded function spaces, subspace of linear functions assigning zero to constants and ignoring values outside Y_d is closed}
	\singlespacing
	
	$\ell_{\mathcal{Y}_d}^\infty(\mathcal{G})$ defined by \eqref{Defn: ell_Yd^infty(G)} is closed.
	
\end{restatable}
\begin{proof}
	\singlespacing
	
	Let $\{H_n\}_{n=1}^\infty \subseteq \ell_{\mathcal{Y}_d}^\infty(\mathcal{G})$ be Cauchy, and let $H$ be its limit in the Banach space $\ell^\infty(\mathcal{G})$. It suffices to show $H \in \ell_{\mathcal{Y}_d}^\infty(\mathcal{G})$. 
	
	Toward this end, first notice that $\lVert H_n - H \rVert_\mathcal{G} \rightarrow 0$ implies that for any $f \in \mathcal{G}$, $\lvert H_n(f) - H(f) \rvert \rightarrow 0$. Next observe that if the constant function $a \in \mathcal{G}$, then $0 = \lim_{n \rightarrow \infty} \lvert H_n(a) - H(a) \rvert = \lim_{n \rightarrow \infty} \lvert H(a) \rvert = \lvert H(a) \rvert$. For any function $f \in \mathcal{G}$, since $H_n(f) = H_n(\mathbbm{1}_{\mathcal{Y}_d} \times f)$,
	\begin{align*}
		0 \leq \lvert H(f) - H(\mathbbm{1}_{\mathcal{Y}_d} \times f) \rvert \leq \lvert H(f) - H_n(f) \rvert + \lvert H(\mathbbm{1}_{\mathcal{Y}_d} \times f) - H_n(\mathbbm{1}_{\mathcal{Y}_d} \times f) \rvert \rightarrow 0 
	\end{align*}
	and thus $H(\mathbbm{1}_{\mathcal{Y}_d} \times f) = H(f)$. Finally, suppose $a, b \in \mathbb{R}$ and $f, g \in \mathcal{G}$ are such that $a f + b g \in \mathcal{G}$. Similar to the argument above, since $H_n(a f + bg) = aH_n(f) + b H_n(g)$,
	\begin{align*}
		0 &\leq \lvert H(a f + bg) - aH(f) - bH(g) \rvert \\
		&\leq \lvert H(af + bg) - H_n(af + bg) \rvert + \lvert a H_n(f) + bH_n(f) - a H(f) - b H_n(g) \rvert \\
		&\leq \lvert H(af + bg) - H_n(af + bg) \rvert + \lvert a\rvert \lvert H_n(f) - H(f) \rvert + \lvert b\rvert \lvert H_n(g) - H_n(g) \rvert \rightarrow 0
	\end{align*}
	and thus $H(a f + bg) = aH(f) + bH(g)$. 
	
	This shows $H \in \ell_{\mathcal{Y}_d}^\infty(\mathcal{G})$, and completes the proof.
\end{proof}

\begin{restatable}[Full differentiability of optimal transport]{lemma}{lemmaHadamardDifferentiabilityOptimalTransportFullDifferentiability}
	
	\label{Lemma: Hadamard differentiability, optimal transport, full differentiability}
	\singlespacing
	
	Let $c : \mathcal{Y} \times \mathcal{Y} \rightarrow \mathbb{R}$ be lower semicontinuous, $\mathcal{F}_1, \mathcal{F}_0$ be sets of measurable functions mapping $\mathcal{Y}$ to $\mathbb{R}$, and $\mathcal{F}_c \subseteq \mathcal{F}_1$ and $\mathcal{F}_c^c \subseteq \mathcal{F}_0$ be universally bounded subsets. Suppose that
	\begin{enumerate}
		\item Strong duality holds:
		\begin{equation*}
			\inf_{\pi \in \Pi(P_1,P_0)} \int c(y_1, y_0) d\pi(y_1,y_0) = \sup_{(\varphi, \psi) \in \Phi_c \cap (\mathcal{F}_c \times \mathcal{F}_c^c)} \int \varphi(y_1) dP_1(y_1) + \int \psi(y_0) dP_0(y_0),
		\end{equation*}
		\item $P$ dominates $P_1$ and $P_0$ with bounded densities, 
		\item $\mathcal{F}_d$ is $P$-Donsker and $\sup_{f \in \mathcal{F}_d} \lvert P(f) \rvert < \infty$ for each $d = 1, 0$, and
		\item $(\mathcal{F}_1 \times \mathcal{F}_0, L_2)$ and the subset
		\begin{equation*}
			\Phi_c \cap (\mathcal{F}_c \times \mathcal{F}_c^c) = \left\{(\varphi,\psi) \in \mathcal{F}_c \times \mathcal{F}_c^c \; ; \; \varphi(y_1) + \psi(y_0) \leq c(y_1, y_0) \right\}
		\end{equation*}
		are complete. 
	\end{enumerate}
	Let $\mathcal{Y}_1, \mathcal{Y}_0 \subseteq \mathcal{Y}$ and $\Psi_c(P_1,P_0) = \argmax_{(\varphi,\psi) \in \Phi_c \cap (\mathcal{F}_c \times \mathcal{F}_c^c)} P_1(\varphi) + P_0(\psi)$, and further assume 
	\begin{enumerate}
		\setcounter{enumi}{3}
		\item For any $(\varphi_1,\psi_1), (\varphi_2,\psi_2) \in \Psi_c(P_1,P_0)$, there exists $s \in \mathbb{R}$ such that 
		\begin{align*}
			&\mathbbm{1}_{\mathcal{Y}_1} \times \varphi_1 = \mathbbm{1}_{\mathcal{Y}_1} \times (\varphi_2 + s), \; P\text{-a.s.} &&\text{ and } &&\mathbbm{1}_{\mathcal{Y}_0} \times \psi_1 = \mathbbm{1}_{\mathcal{Y}_0} \times (\psi_2 - s), \; P\text{-a.s.}
		\end{align*}
	\end{enumerate}
	
	Then $OT_c : \ell^\infty(\mathcal{F}_1) \times \ell^\infty(\mathcal{F}_0) \rightarrow \mathbb{R}$ defined by
	\begin{equation*}
		OT_c(P_1, P_0) = \sup_{(\varphi, \psi) \in \Phi_c \cap (\mathcal{F}_c \times \mathcal{F}_c^c)} P_1(\varphi) + P_0(\psi) 
	\end{equation*}
	is fully Hadamard differentiable at $(P_1, P_0)$ tangentially to
	\begin{equation}
		\mathbb{D}_{Tan,Full} = \Big(\ell_{\mathcal{Y}_1}^\infty(\mathcal{F}_c) \times \ell_{\mathcal{Y}_0}^\infty(\mathcal{F}_c^c)\Big) \cap \Big(\mathcal{C}(\mathcal{F}_1, L_{2,P}) \times \mathcal{C}(\mathcal{F}_0, L_{2,P})\Big)\label{Defn: tangent set for optimal transport derivative, when fully differentiable}
	\end{equation}
	with derivative $OT_{c, (P_1,P_0)}' : \mathbb{D}_{Tan,Full} \rightarrow \mathbb{R}$ given by
	\begin{equation*}
		OT_{c,(P_1,P_0)}'(H_1, H_0) = \sup_{(\varphi,\psi) \in \Psi_c(P_1, P_0)} H_1(\varphi) + H_0(\psi)
	\end{equation*}
\end{restatable}
\begin{proof}
	\singlespacing
	
	The first four assumptions allow application of lemma \ref{Lemma: Hadamard differentiability, optimal transport} to find that $OT_c : \ell^\infty(\mathcal{F}_1) \times \ell^\infty(\mathcal{F}_0) \rightarrow \mathbb{R}$ given by
	\begin{equation*}
		OT_c(P_1, P_0) = \sup_{(\varphi, \psi) \in \Phi_c \cap (\mathcal{F}_c \times \mathcal{F}_c^c)} P_1(\varphi) + P_0(\psi) 
	\end{equation*}
	is Hadamard directionally differentiable at $(P_1, P_0)$ tangentially to $\mathbb{D}_{Tan} = \mathcal{C}(\mathcal{F}_1, L_{2,P}) \times \mathcal{C}(\mathcal{F}_0, L_{2,P})$. The set of maximizers $\Psi_c(P_1, P_0) = \argmax_{(\varphi,\psi) \in \Phi_c \cap (\mathcal{F}_c \times \mathcal{F}_c^c)} P_1(\varphi) + P_0(\psi)$ is nonempty, and the derivative $OT_{c, (P_1,P_0)}' : \mathbb{D}_{Tan} \rightarrow \mathbb{R}$ is given by
	\begin{equation*}
		OT_{c,(P_1,P_0)}'(H_1, H_0) = \sup_{(\varphi,\psi) \in \Psi_c(P_1, P_0)} H_1(\varphi) + H_0(\psi)
	\end{equation*}
	
	Next observe that for any $(H_1,H_0) \in \mathbb{D}_{Tan,Full}$, $H_1 + H_0$ is flat on $\Psi_c(P_1,P_0)$. Specifically, for any $(\varphi_1,\psi_1), (\varphi_2,\psi_2) \in \Psi_c(P_1,P_0)$, let $s$ be such that 
	\begin{align*}
		\mathbbm{1}_{\mathcal{Y}_1} \times \varphi_1 = \mathbbm{1}_{\mathcal{Y}_1} \times (\varphi_2 + s), \; P\text{-a.s.} &&\text{ and } &&\mathbbm{1}_{\mathcal{Y}_0} \times \psi_1 = \mathbbm{1}_{\mathcal{Y}_0} \times (\psi_2 - s), \; P\text{-a.s.}
	\end{align*}
	Then
	\begin{align*}
		H_1(\varphi_1) + H_0(\psi_1) &= H_1(\mathbbm{1}_{\mathcal{Y}_1} \times \varphi_1) + H_0(\mathbbm{1}_{\mathcal{Y}_0} \times \psi_1) \\
		&= H_1(\mathbbm{1}_{\mathcal{Y}_1} \times (\varphi_2 + s)) + H_0(\mathbbm{1}_{\mathcal{Y}_0}\times (\psi_2 - s)) \\
		&=H_1(\varphi_2 + s) + H_0(\psi_2 - s) \\
		&= H_1(\varphi_2) + H_1(s) + H_0(\psi_2) - H_0(s) \\
		&= H_1(\varphi_2) + H_0(\psi_2)
	\end{align*}
	where the first, third, fourth, and fifth equalities hold because $(H_1, H_0) \in \ell_{\mathcal{Y}_1}^\infty(\mathcal{F}_c) \times \ell_{\mathcal{Y}_0}^\infty(\mathcal{F}_c^c)$, and the second because $(H_1, H_0) \in \mathcal{C}(\mathcal{F}_1, L_{2,P}) \times \mathcal{C}(\mathcal{F}_0, L_{2,P})$. 
	
	Now use this ``flatness'' to observe the derivative is linear. Let $(H_1,H_0), (G_1,G_0) \in \mathbb{D}_{Tan,Full}$, $a,b \in \mathbb{R}$, and $(\tilde{\varphi},\tilde{\psi}) \in \Psi_c(P_1,P_0)$, and notice that 
	\begin{align*}
		&OT_{c,(P_1,P_0)}'(a(H_1, H_0) + b(G_1,G_0)) = \sup_{(\varphi, \psi) \in \Psi(P_1,P_0)} (aH_1 + bG_1)(\varphi) +  (aH_0 + bG_0)(\psi) \\
		&\hspace{1 cm} = a H_1(\tilde{\varphi}) + b G_1(\tilde{\varphi}) + a H_0(\tilde{\psi}) + bG_0(\tilde{\psi}) = a(H_1(\tilde{\varphi}) + H_0(\tilde{\psi})) + b(G_1(\tilde{\varphi}) + G_0(\tilde{\psi}) \\
		&\hspace{1 cm} = a\times \sup_{(\varphi, \psi) \in \Psi(P_1,P_0)}\left\{ H_1(\varphi) + H_0(\psi)\right\} + b \times \sup_{(\varphi, \psi) \in \Psi(P_1,P_0)}\left\{G_1(\varphi) + G_0(\psi)\right\} \\
		&\hspace{1 cm} = a OT_{c, (P_1,P_0)}'(H_1, H_0) + b OT_{c, (P_1,P_0)}'(G_1, G_0)
	\end{align*}
	
	Since $OT_{c,(P_1,P_0)}'$ is linear on the subspace $\mathbb{D}_{Tan,Full}$, \cite{fang2019inference} proposition 2.1 implies $OT_c$ is fully Hadamard differentiable at $(P_1,P_0)$ tangentially to $\mathbb{D}_{Tan,Full}$. 
\end{proof}

%% file: appendix/OTJointPO_appendix_weak_convergence.tex
\section{Appendix: weak convergence}
\label{Appendix: weak convergence}

Recall that 
\begin{align*}
	&\theta_x^L = \theta^L(P_{1 \mid x}, P_{0 \mid x}), &&\theta_x^H = \theta^H(P_{1 \mid x}, P_{0 \mid x}) \\
	&\theta^L = \sum_x s_x \theta_x^L, &&\theta^L = \sum_x s_x \theta_x^H \\
	&\gamma^L = \inf_{t \in [\theta^L, \theta^H]} g(t, \eta) &&\gamma^H = \sup_{t \in [\theta^L, \theta^H]} g(t, \eta)
\end{align*}
where $\eta = (\eta_1, \eta_0)$, with $\eta_d \in \mathbb{R}^{K_d}$ having coordinates
\begin{align*}
	\eta_d^{(k)} &= \sum_x P(X = x \mid D_1 > D_0) E[\eta_d^{(k)}(Y_d)\mid D_1 > D_0, X = x] = \sum_x s_x \eta_{d,x}^{(k)}
\end{align*}
Here $\eta_{d,x}^{(k)} = P_{d \mid x}(\eta_d^{(k)})$, which are collected as $\eta_{d,x} = (\eta_{d,x}^{(1)}, \ldots, \eta_{d,x}^{(K_d)})$.

Define the following sets of functions:
\begin{align}
	\tilde{\mathcal{F}}_{1} &= \left\{f : \mathcal{Y}\rightarrow \mathbb{R} \; ; \; f = \varphi \text{ for some } \varphi \in \mathcal{F}_c, \text{ or } f = \eta_1^{(k)} \text{ for some } k = 1, \ldots, K_1\right\} \label{Defn: F_dx, the set on which P_{d|x} is defined} \\
	\tilde{\mathcal{F}}_{0} &= \left\{f : \mathcal{Y} \rightarrow \mathbb{R} \; ; \; f = \psi \text{ for some } \psi \in \mathcal{F}_c^c, \text{ or } f = \eta_0^{(k)} \text{ for some } k = 1, \ldots, K_0\right\} \notag \\
	\mathcal{F}_{d,x} &= \left\{f : \mathcal{Y}\rightarrow \mathbb{R} \; ; \; f = g \text{ or } \mathbbm{1}_{\mathcal{Y}_{d, x}} \times g  \text{ for some } g \in \tilde{\mathcal{F}}_d\right\} \notag 
\end{align}
where $\mathcal{Y}_{d,x}$ is the support of $Y \mid D = d, X = x$, and $\mathbbm{1}_{\mathcal{Y}_{d,x}}(y) = \mathbbm{1}\{y \in \mathcal{Y}_{d, x}\}$. The additional functions of the form $f(y) = \mathbbm{1}_{\mathcal{Y}_{d,x}}(y) g(y)$ are used to characterize the support of the weak limit of $\sqrt{n}(\hat{P}_{d\mid x} - P_{d \mid x})$ in $\ell^\infty(\mathcal{F}_{d,x})$. The maps $P_{d \mid x}$ can be written as
\begin{align}
	&P_{d \mid x} : \mathcal{F}_{d, x} \rightarrow \mathbb{R}, &&P_{d \mid x}(f) = \frac{P(\mathbbm{1}_{d,x,d} \times f)/P(\mathbbm{1}_{x,d}) - P(\mathbbm{1}_{d,x,1-d} \times f)/P(\mathbbm{1}_{x,1-d})}{P(\mathbbm{1}_{d,x,d})/P(\mathbbm{1}_{x,d}) - P(\mathbbm{1}_{d,x,1-d})/P(\mathbbm{1}_{x,1-d})} \label{Display: T_1, map to conditional distributions, weak convergence appendix} 
\end{align}
and finally, define the set
\begin{equation}
	\mathcal{F} = \bigcup_{d,x,z} \left\{\mathbbm{1}_{d,x,z} \times f \; ; \; f \in \mathcal{F}_{d,x}\right\} \cup \{\mathbbm{1}_{d,x,z}, \mathbbm{1}_{x,z}, \mathbbm{1}_x\}. \label{Defn: F large donsker set}
\end{equation}

This appendix defines and studies the map $T : \mathbb{D}_C \subseteq \ell^\infty(\mathcal{F}) \rightarrow \mathbb{R}^2$ given by $(\gamma^L, \gamma^H) = T(P)$. The coming results show that $\mathcal{F}$ is $P$-Donsker, and the map $T$ is Hadamard directionally differentiable at $P$. Together these imply, through the functional delta method, the weak convergence of $\sqrt{n}(T(\mathbb{P}_n) - T(P))$ (\cite{fang2019inference}).

Several operations in the definition of the map $T$ are repeated for each $x \in \mathcal{X} = \{x_1, \ldots, x_M\}$, leading to large expressions. These are shortened with the notation $\{a_x\}_{x\in\mathcal{X}}$, which refers to $(a_{x_1}, \ldots, a_{x_M})$. For example, 
\begin{equation*}
	 \left(\left\{P_{1 \mid x}, P_{0 \mid x}, \eta_{1,x}, \eta_{0,x}, s_{x}\right\}_{x \in \mathcal{X}}\right) = (P_{1 \mid x_1}, P_{0 \mid x_1}, \eta_{1,x_1}, \eta_{0,x_1}, s_{x_1}, \ldots, P_{1 \mid x_M}, P_{0 \mid x_M}, \eta_{1,x_M}, \eta_{0,x_M}, s_{x_M})
\end{equation*}
is an element of $\prod_{m=1}^M \ell^\infty(\mathcal{F}_{1,x_m}) \times \ell^\infty(\mathcal{F}_{0,x_m}) \times \mathbb{R}^{K_1}\times \mathbb{R}^{K_0} \times \mathbb{R}$.

The function $T$ is viewed as the composition of four functions: $T(P) = T_4(T_3(T_2(T_1(P))))$. 
\begin{enumerate}
	\item $T_1$ is the map to the conditional distributions and $\eta_{d,x}$: $T_1(P) = (\{P_{1 \mid x}, P_{0 \mid x}, \eta_{1,x}, \eta_{0,x}, s_x\}_{x \in \mathcal{X}})$, 
	\item $T_2$ involves optimal transport: $T_2(\{(P_{1 \mid x}, P_{0 \mid x}, \eta_{1,x}, \eta_{0,x}, s_x)\}_{x \in \mathcal{X}}) = (\{\theta_x^L, \theta_x^H, \eta_{1,x}, \eta_{0,x}, s_x\}_{x\in \mathcal{X}})$,
	\item $T_3$ takes expectations over covariates: $T_3(\{(\theta_x^L, \theta_x^H, \eta_{1,x}, \eta_{0,x}, s_x)\}_{x \in \mathcal{X}}) \mapsto (\theta^L, \theta^H, \eta)$, 
	\item $T_4$ optimizes over $t \in [\theta^L, \theta^H]$: $T^4(\theta^L, \theta^H, \eta) = (\gamma^L, \gamma^H)$.
\end{enumerate}

\subsection{Verifying Donsker conditions}

Before studying this map, this subsection shows the relevant sets are Donsker. The function classes $\mathcal{F}_c$ and $\mathcal{F}_c^c$ given by \eqref{Defn: F_c for smooth costs} and \eqref{Defn: F_c^c for smooth costs}, or by \eqref{Defn: F_c for indicator costs of convex C} and \eqref{Defn: F_c^c for indicator costs of convex C}, are well known Donsker classes as noted below. The results of \cite{vaart1997weak} chapter 2.10 allow these to be extended to show $\mathcal{F}_{1,x}$ and $\mathcal{F}_{0,x}$ are Donsker. It follows quickly that $\mathcal{F}$ is Donsker.

\begin{restatable}[]{lemma}{lemmaCConcaveFunctionsLipschitzCostCompactSupportsDonsker}
	\singlespacing
	\label{Lemma: weak convergence, Donsker results, c-concave functions for smooth costs}
	
	Suppose that $\mathcal{Y} \subset \mathbb{R}$ is compact and $c : \mathcal{Y} \times \mathcal{Y} \rightarrow \mathbb{R}$ is $L$-Lipschitz. Let $\mathcal{F}_c$, $\mathcal{F}_c^c$ be given by \eqref{Defn: F_c for smooth costs} and \eqref{Defn: F_c^c for smooth costs} respectively. Then $\mathcal{F}_c$ and $\mathcal{F}_c^c$ are universally Donsker.
\end{restatable}
\begin{proof}
	\singlespacing
	Note that any distribution defined on the compact $\mathcal{Y}$ has a finite $2 + \delta$ moment. The result follows from the bracketing number bound given by \cite{vaart1997weak} corollary 2.7.4.
\end{proof}

\begin{restatable}[]{lemma}{lemmaCConcaveFunctionsForIndicatorCostsAreDonsker}
	\label{Lemma: weak convergence, Donsker results, c-concave functions for indicator of convex set costs}
	\singlespacing
	
	$\mathcal{F}_c$ and $\mathcal{F}_c^c$ given by \eqref{Defn: F_c for indicator costs of convex C} and \eqref{Defn: F_c^c for indicator costs of convex C} are universally Donsker.
\end{restatable}
\begin{proof}
	\singlespacing
	
	The intervals (convex subsets of $\mathbb{R}$) form a well-known VC class with VC-dimension at most 3. Consider an arbitrary set of three real numbers $\{y_1, y_2, y_3\}$ with $y_1 < y_2 < y_3$, and notice that no interval can pick out the set $\{y_1, y_3\}$; that is, there does not exist an interval $I$ with $\{y_1, y_3\} = \{y_1, y_2, y_3\} \cap I$. Since the intervals cannot shatter finite sets of size $3$, the VC-dimension of the intervals is at most $3$. 
	
	Similarly, the complements of intervals form a VC class of VC-dimension at most 4. Consider $\{y_1, y_2, y_3, y_4\}$ with $y_1 < y_2 < y_3 < y_4$ and notice that no complement of an interval can pick out $\{y_1, y_3\}$. Since the complements of intervals cannot shatter finite sets of size $4$, the VC-dimension of the complements of intervals is at most $4$. 
	
	The claim follows, because any (suitably measurable) VC class is Donsker for any probability measure (\cite{vaart1997weak} section 2.6.1).
\end{proof}

\begin{restatable}[]{lemma}{lemmaDonskerTimesIndicatorIsDonsker}
	\label{Lemma: weak convergence, Donsker results, Donsker times indicator is Donsker}
	\singlespacing
	
	Let $\mathcal{G}$ be $P$-Donsker and $\mathbbm{1}_A$ be the indicator function for the set $A$. Then the set $\left\{\mathbbm{1}_A \times g \; ; \; g \in \mathcal{G}\right\}$ is $P$-Donsker. 
\end{restatable}
\begin{proof}
	\singlespacing
	
	The proof is an application of \cite{vaart1997weak} theorem 2.10.6. Specifically, let $\phi : \mathcal{G} \times \{\mathbbm{1}_A\} \rightarrow \mathbb{R}$ be the map $\phi(g, \mathbbm{1}_a) = \mathbbm{1}_A \times g$. Notice that for any $f,g  \in \mathcal{G}_1 \times \{\mathbbm{1}_A\}$, 
	\begin{align*}
		\lvert \phi \circ f(w) - \phi \circ g(w) \rvert^2 &= \left\lvert \mathbbm{1}_A(w) \times f_1(w) - \mathbbm{1}_A(w) \times g_1(w) \right\rvert^2 \\
		&= \mathbbm{1}_A(w) \times \lvert f_1(w) - g_2(w) \rvert^2 \\
		&\leq \lvert f_1(w) - g_1(w) \rvert^2 = \sum_{\ell=1}^k (f_\ell(w) - g_\ell(w))^2
	\end{align*}
	and thus \cite{vaart1997weak} condition (2.10.5) holds. Moreover, notice that for any $g \in \mathcal{G}$, $(\mathbbm{1}_A \times g)^2 \leq g^2$ and $P$-square integrability of $g \in \mathcal{G}$ implies $\mathbbm{1}_A \times g$ is $P$-square integrable. Thus \cite{vaart1997weak} theorem 2.10.6 implies $\left\{\mathbbm{1}_A \times g \; ; \; g \in \mathcal{G}\right\}$ is $P$-Donsker. 
\end{proof}

\begin{restatable}[$\mathcal{F}_{d,x}$ are $P$-Donsker]{lemma}{lemmaDonskerResultForFd}
	\label{Lemma: weak convergence, Donsker results, Donsker result for F_dx}
	\singlespacing
	
	Suppose assumptions \ref{Assumption: setting}, \ref{Assumption: cost function}, and \ref{Assumption: parameter, function of moments} hold. Let $\mathcal{F}_c$ and $\mathcal{F}_c^c$ be given by \eqref{Defn: F_c for smooth costs} and \eqref{Defn: F_c^c for smooth costs}, or by \eqref{Defn: F_c for indicator costs of convex C} and \eqref{Defn: F_c^c for indicator costs of convex C}. Let $\mathcal{F}_{d,x}$ be as defined in \eqref{Defn: F_dx, the set on which P_{d|x} is defined}. Then $\mathcal{F}_{d,x}$ is $P$-Donsker and $\sup_{f \in \mathcal{F}_{d,x}} \lvert P(f) \rvert < \infty$. 
	
\end{restatable}
\begin{proof}
	\singlespacing
	
	\begin{enumerate}
		\item We first show $\tilde{\mathcal{F}}_d$ is $P$-Donsker and $\sup_{g \in \tilde{F}_d} \lvert P(f) \rvert < \infty$. The argument shows the argument for $\tilde{\mathcal{F}}_1$, as the same argument works when applied to $\tilde{\mathcal{F}}_0$.
		
		Begin by noticing that 
		\begin{align*}
			\tilde{\mathcal{F}}_1 &= \left\{f : \mathcal{Y}\rightarrow \mathbb{R} \; ; \; f = \varphi \text{ for some } \varphi \in \mathcal{F}_c, \text{ or } f = \eta_1^{(k)} \text{ for some } k = 1, \ldots, K_1\right\} \\
			&= \mathcal{F}_c \cup \left\{\eta_1^{(1)}, \ldots, \eta_1^{(K_1)}\right\} 
		\end{align*}
		Since $\left\{\eta_1^{(1)}, \ldots, \eta_1^{(K_1)}\right\}$ is a finite number of functions which, by assumption \ref{Assumption: parameter, function of moments} \ref{Assumption: parameter, function of moments, nuisance moments have finite variance}, have finite second $P$-moment: $P((\eta_1^{(k)})^2) < \infty$. Thus $\left\{\eta_1^{(1)}, \ldots, \eta_1^{(K_1)}\right\}$ is Donsker. $\mathcal{F}_c$ is Donsker by lemma \ref{Lemma: weak convergence, Donsker results, c-concave functions for smooth costs} or \ref{Lemma: weak convergence, Donsker results, c-concave functions for indicator of convex set costs}, and so $\tilde{\mathcal{F}}_1 = \mathcal{F}_c \cup \left\{\eta_1^{(1)}, \ldots, \eta_1^{(K_1)}\right\}$ is the union of two $P$-Donsker sets. Since 
		\begin{equation*}
			\lVert P \rVert_{\tilde{\mathcal{F}}_1} = \max\{\sup_{\varphi \in \mathcal{F}_c} \lvert P(\varphi) \rvert, \lvert P(\eta_1^{(1)}) \rvert, \ldots, \lvert P(\eta_1^{(K_1)}) \rvert\} < \infty
		\end{equation*}
		\cite{vaart1997weak} example 2.10.7 shows $\tilde{\mathcal{F}}_1$ is $P$-Donsker. Note we have also shown that $\sup_{g \in \tilde{F}_1} \lvert P(f) \rvert < \infty$. 
		
		\item Now notice that 
		\begin{align*}
			\mathcal{F}_{d,x} &= \left\{f : \mathcal{Y}\rightarrow \mathbb{R} \; ; \; f = g \text{ or } \mathbbm{1}_{\mathcal{Y}_{d, x}} \times g  \text{ for some } g \in \tilde{\mathcal{F}}_d\right\} \\
			&= \tilde{\mathcal{F}}_d \cup \left\{\mathbbm{1}_{\mathcal{Y}_{d, x}} \times g  \; ; \; g \in \tilde{\mathcal{F}}_d \right\}
		\end{align*}
		Lemma \ref{Lemma: weak convergence, Donsker results, Donsker times indicator is Donsker} shows $\left\{\mathbbm{1}_{\mathcal{Y}_{d, x}} \times g  \; ; \; g \in \tilde{\mathcal{F}}_d \right\}$ is $P$-Donsker. Moreover, since $\mathcal{F}_c$ is universally bounded, 
		\begin{align*}
			\lVert P \rVert_{\left\{\mathbbm{1}_{\mathcal{Y}_{d, x}} \times g  \; ; \; g \in \tilde{\mathcal{F}}_d \right\}} = \max\left\{\sup_{\varphi \in \mathcal{F}_c} \lvert P(\mathbbm{1}_{\mathcal{Y}_{d,x}} \times \varphi) \rvert, \lvert P(\mathbbm{1}_{\mathcal{Y}_{d,x}} \times \eta_1^{(1)}) \rvert, \ldots, \lvert P(\mathbbm{1}_{\mathcal{Y}_{d,x}} \times \eta_1^{(K_1)}) \rvert\right\} < \infty
		\end{align*}
		It follows that
		\begin{align*}
			\lVert P \rVert_{\mathcal{F}_{d,x}} = \sup_{f \in \mathcal{F}_{d,x}} \lvert P(f) \rvert = \max\left\{\sup_{f \in \tilde{\mathcal{F}}_d} \lvert P(f) \rvert, \sup_{f \in \{\mathbbm{1}_{\mathcal{Y}_{d,x}} \times g \; ; \; g \in \tilde{\mathcal{F}}_d\}} \lvert P(f) \rvert\right\} < \infty
		\end{align*}
		Thus \cite{vaart1997weak} example 2.10.7 implies $\mathcal{F}_1$ is $P$-Donsker.
	\end{enumerate}
\end{proof}

\begin{restatable}[$\mathcal{F}$ is $P$-Donsker]{lemma}{lemmaLargeDonskerSetIsDonsker}
	\label{Lemma: weak convergence, Donsker results, large Donsker set F is Donsker}
	\singlespacing
	
	Suppose assumptions \ref{Assumption: setting}, \ref{Assumption: cost function} and \ref{Assumption: parameter, function of moments} hold. Then $\mathcal{F}$ is $P$-Donsker, implying
	\begin{align*}
		&\sqrt{n}(\mathbb{P}_n - P) \overset{L}{\rightarrow} \mathbb{G} &&\text{ in } \ell^\infty(\mathcal{F}),
	\end{align*}
	where $\mathbb{G}$ is a tight, mean-zero Gaussian process with $P(\mathbb{G} \in \mathcal{C}(\mathcal{F}, L_{2,P}) = 1$. 
\end{restatable}
\begin{proof}
	\singlespacing
	
	Lemma \ref{Lemma: weak convergence, Donsker results, Donsker times indicator is Donsker} shows $\left\{\mathbbm{1}_{d,x,z} \times f \; ; \; f \in \mathcal{F}_{d,x}\right\}$ is $P$-Donsker. Moreover, $\mathcal{F}_{d,x}$ is the union of a subset of universally bounded functions (in either $\mathcal{F}_c$ or $\mathcal{F}_c^c$) and a finite subset of square integrable functions. It follows that 
	\begin{align*}
		\lVert P \rVert_{\left\{\mathbbm{1}_{d,x,z} \times g \; ; \; g \in \mathcal{F}_{d,x}\right\}} = \sup_{f \in \left\{\mathbbm{1}_{d,x,z} \times g \; ; \; g \in \mathcal{F}_{d,x}\right\}} \lvert P(f) \rvert < \infty
	\end{align*}
	
	Next notice that 
	\begin{equation*}
		\mathcal{F} = \bigcup_{d,x,z} \left\{\mathbbm{1}_{d,x,z} \times f \; ; \; f \in \mathcal{F}_{d,x}\right\} \cup \{\mathbbm{1}_{d,x,z}, \mathbbm{1}_{x,z}, \mathbbm{1}_x\}
	\end{equation*}
	is the union of a finite number of $P$-Donsker sets, with 
	\begin{align*}
		\lVert P \rVert_\mathcal{F} = \max_{d,x,z}\left\{\max\left\{\sup_{f \in \left\{\mathbbm{1}_{d,x,z} \times g \; ; \; g \in \mathcal{F}_{d,x}\right\}} \lvert P(f) \rvert, \lvert P(\mathbbm{1}_{d,x,z}) \rvert, \lvert P(\mathbbm{1}_{x,z}) \rvert, \lvert P(\mathbbm{1}_x) \rvert, \right\}\right\} < \infty
	\end{align*}
	It follows from \cite{vaart1997weak} example 2.10.7 that $\mathcal{F}$ is $P$-Donsker, which implies $\sqrt{n}(\mathbb{P}_n - P) \overset{L}{\rightarrow} \mathbb{G}$ in $\ell^\infty(\mathcal{F})$, where $\mathbb{G}$ is a tight, mean-zero Gaussian process. Moreover, \cite{vaart1997weak} section 2.1.2 and problem 2.1.2 imply that $P(\mathbb{G} \in \mathcal{C}(\mathcal{F}, L_{2,P}) = 1$.
\end{proof}

\subsection{Conditional Distributions, $T_1(P) = (\{P_{1 \mid x}, P_{0 \mid x}, \eta_{1,x}, \eta_{0,x}, s_x\}_{x \in \mathcal{X}})$}

\label{Appendix: weak convergence, subsection conditional distributions}

Lemma \ref{Lemma: identification, LATE IV marginal distribution identification} shows that the distributions of $Y_d \mid D_1 > D_0, X = x$, denoted $P_{d \mid x}$, are identified by
\begin{align*}
	P_{d \mid x}(f) &= E_{P_{d \mid x}}[f(Y_d)] = E[f(Y_d) \mid D_1 > D_0, X = x] \\
	&= \frac{E[f(Y) \mathbbm{1}\{D = d\} \mid Z = d, X = x] - E[f(Y) \mathbbm{1}\{D = d\} \mid Z = 1-d, X = x]}{P(D = d \mid Z = d, X = x) - P(D = d \mid Z = 1-d, X = x)}
\end{align*}
and the distribution of $X$ conditional on $D_1 > D_0$ is identified by
\begin{align*}
	s_x &= P(X = x \mid D_1 > D_0) \\
	&= \frac{\left[P(D = 1 \mid Z = 1, X = x) - P(D = 1 \mid Z = 0, X = x) \right]P(X = x) }{\sum_{x'} \left[P(D = 1 \mid Z = 1, X = x') - P(D = 1 \mid Z = 0, X = x') \right]P(X = x')}
\end{align*}

Recall the notation shortening indicators
\begin{align*}
	&\mathbbm{1}_{d,x,z}(D,X,Z) = \mathbbm{1}\{D = d, X = x, Z = z\}, &&\mathbbm{1}_{x,z}(X,Z) = \mathbbm{1}\{X = x, Z = z\}, &&\mathbbm{1}_x(X) = \mathbbm{1}\{X = x\} 
\end{align*}
and notice that $P_{d \mid x} : \ell^\infty(\mathcal{F}_d) \rightarrow \mathbb{R}$ and $s_x \in \mathbb{R}$, given by
\begin{align*}
	P_{d \mid x}(f) &= \frac{P(\mathbbm{1}_{d,x,d} \times f)/P(\mathbbm{1}_{x,d}) - P(\mathbbm{1}_{d,x,1-d} \times f)/P(\mathbbm{1}_{x,0})}{P(\mathbbm{1}_{d,x,d})/P(\mathbbm{1}_{x,d})- P(\mathbbm{1}_{d,x,1-d})/P(\mathbbm{1}_{x,1-d})}, \\
	s_x &= \frac{[P(\mathbbm{1}_{1,x,1})/P(\mathbbm{1}_{x,1}) - P(\mathbbm{1}_{1,x,0})/P(\mathbbm{1}_{x,0})]P(\mathbbm{1}_x)}{\sum_{x'} [P(\mathbbm{1}_{1,x',1})/P(\mathbbm{1}_{x',1}) - P(\mathbbm{1}_{1,x',0})/P(\mathbbm{1}_{x',0})]P(\mathbbm{1}_{x'})},
\end{align*}
are functions of $P \in \ell^\infty(\mathcal{F})$. Moreover, $\eta_{d,x}^{(k)} = E[\eta_d^{(k)}(Y_d) \mid D_1 > D_0, X = x] = P_{d \mid x}(\eta_d^{(k)})$ and $\eta_{d,x} = (\eta_{d,x}^{(1)}, \ldots, \eta_{d,x}^{(K_1)})$ is simply an evaluation of $P_{d\mid x}$ at the points $\eta_d^{(k)} \in \mathcal{F}_{d,x}$. 

This map is given by
\begin{align*}
	&T_1 : \mathbb{D}_C \subseteq \ell^\infty(\mathcal{F}) \rightarrow \prod_{m=1}^M \ell^\infty(\mathcal{F}_{1,x_m}) \times \ell^\infty(\mathcal{F}_{0,x_m}) \times \mathbb{R} \times \mathbb{R}^{(K_1)}\times \mathbb{R}^{(K_0)} \\
	&T_1(P) = \left(\left\{P_{1 \mid x}, P_{0 \mid x}, \eta_{1,x}, \eta_{0,x}, s_{x}\right\}_{x \in \mathcal{X}}\right)\\
	&\hspace{1 cm} = (P_{1 \mid x_1}, P_{0 \mid x_1}, \eta_{1,x_1}, \eta_{0,x_1}, s_{x_1}, \ldots, P_{1 \mid x_M}, P_{0 \mid x_M}, \eta_{1,x_M}, \eta_{0,x_M}, s_{x_M})
\end{align*}
where the domain, $\mathbb{D}_C  \subseteq \ell^\infty(\mathcal{F})$, ensures the map never divide by zero:
\begin{align}
	&\mathbb{D}_C = \big\{G \in \ell^\infty(\mathcal{F}) \; ; \; \text{ for all } (d,x,z), \; G(\mathbbm{1}_x) > 0, \; G(\mathbbm{1}_{x,z}) > 0, \text{ and } \notag \\
	&\hspace{4 cm} G(\mathbbm{1}_{d,x,d})/G(\mathbbm{1}_{x,d}) - G(\mathbbm{1}_{d,x,1-d})/G(\mathbbm{1}_{x,1-d}) > 0 \big\} \label{Defn: D_C, domain of the overall map}
\end{align}
Note that assumption \ref{Assumption: setting} implies $P \in \mathbb{D}_C$, a claim shown in the proof of lemma \ref{Lemma: Hadamard differentiability, T1 is fully Hadamard differentiable} below.

Lemma \ref{Lemma: Hadamard differentiability, stacking functions} shows that Hadamard differentiable functions with the same domain can be ``stacked''. Moreover, the coordinates corresponding to the $\eta$ terms are evaluations of the $P_{d \mid x}$ at specific coordinates; since evaluation is linear and continuous, the map defining these terms is fully Hadamard differentiable if the other maps are fully Hadamard differentiable. Thus it suffices to ensure the maps $C_{d, x} : \mathbb{D}_C \rightarrow \mathbb{R}$ and $C_{s,x} : \mathbb{D}_C \rightarrow \mathbb{R}$ given by $C_{d, x}(P) = P_{d \mid x}$ and $C_{s,x}(P) = s_x$ are fully Hadamard differentiable at $P$ tangentially to $\ell^\infty(\mathcal{F})$.

\begin{restatable}[Maps to conditional distributions are fully Hadamard differentiable]{lemma}{lemmaHadamardDifferentiabilityOfConditionalDistributions}
	\label{Lemma: Hadamard differentiability, conditional distributions}
	\singlespacing
	
	Let $\mathcal{F}$ be defined by \eqref{Defn: F large donsker set}, and $\mathbb{D}_C$ be defined by \eqref{Defn: D_C, domain of the overall map}. Define the functions $C_{1,x}$, $C_{0,x}$, and $C_{s,x}$ with
	\begin{align*}
		&C_{d,x} : \mathbb{D}_C \rightarrow \ell^\infty(\mathcal{F}_{d,x}), &&C_{d,x}(G)(f) = \frac{G(\mathbbm{1}_{d,x,d} \times f)/G(\mathbbm{1}_{x,d}) - G(\mathbbm{1}_{d,x,1-d}\times f)/G(\mathbbm{1}_{x,1-d})}{G(\mathbbm{1}_{d,x,d})/G(\mathbbm{1}_{x,d})- G(\mathbbm{1}_{d,x,1-d})/G(\mathbbm{1}_{x,1-d})}, \\
		&C_{s,x} : \mathbb{D}_C \rightarrow \mathbb{R}, &&C_{s,x}(G) = \frac{[G(\mathbbm{1}_{1,x,1})/G(\mathbbm{1}_{x,1}) - G(\mathbbm{1}_{1,x,0})/G(\mathbbm{1}_{x,0})]G(\mathbbm{1}_x)}{\sum_{x'} [G(\mathbbm{1}_{1,x',1})/G(\mathbbm{1}_{x',1}) - G(\mathbbm{1}_{1,x',0})/G(\mathbbm{1}_{x',0})]G(\mathbbm{1}_{x'})}
	\end{align*}
	All three functions are fully Hadamard differentiable at any $G \in \mathbb{D}_C$ tangentially to $\ell^\infty(\mathcal{F})$, with derivatives $C_{d, x, G}' : \ell^\infty(\mathcal{F}) \rightarrow \ell^\infty(\mathcal{F}_{d,x})$ and $C_{s, x, G}' : \ell^\infty(\mathcal{F}) \rightarrow \mathbb{R}$ described in the proof.
	
\end{restatable}
\begin{proof}
	\singlespacing
	
	In steps:
	
	\begin{enumerate}
		\item We first show differentiability of $C_{1,x}$. The argument applies the chain rule. An inner function ``rearranges'' elements of $\mathbb{D}_C \subseteq \ell^\infty(\mathcal{F})$, which can be viewed as a fully Hadamard differentiable mapping (see lemma \ref{Lemma: Hadamard differentiability, rearranging Donsker sets}). An outer function maps that rearrangment to $\ell^\infty(\mathcal{F}_1)$, and is shown fully Hadamard differentiable at $G \in \mathbb{D}_C$ by applying corollary \ref{Lemma: Hadamard differentiability, maps between bounded function spaces, corollary}.  
		
		In detailed steps:
		\begin{enumerate}
			\item Define $\mathbb{D}_q = \left\{(n_1, p_{11}, p_1, n_0, p_{10}, p_0) \in \mathbb{R}^6 \; ; \; p_1 > 0, \; p_0 > 0, \; p_{11}/p_1 - p_{10}/p_0 > 0\right\}$ and 
			\begin{align*}
				&q : \mathbb{D}_q \rightarrow \mathbb{R}, &&q(n_1, p_{11}, p_1, n_0, p_{10}, p_0) = \frac{n_1/p_1 - n_0/p_0}{p_{11}/p_1 - p_{10}/p_0}
			\end{align*}
			Recall the following notation from corollary \ref{Lemma: Hadamard differentiability, maps between bounded function spaces, corollary}: 
			\begin{align*}
				\ell^\infty(\mathcal{F}_1, \mathbb{D}_q) &= \left\{r : \mathcal{F}_1 \rightarrow \mathbb{R}^6 \; ; \; r(\varphi) \in \mathbb{D}_q, \; \sup_{\varphi \in \mathcal{F}_1} \lVert r(f) \rVert < \infty\right\} \subseteq \ell^\infty(\mathcal{F}_1)^6 \\
				\ell_q^\infty(\mathcal{F}_1, \mathbb{D}_q) &= \left\{r \in \ell^\infty(\mathcal{F}_1, \mathbb{D}_q) \; ; \; \sup_{f \in \mathcal{F}_1} \lvert q(r(f)) \rvert < \infty\right\}
			\end{align*}
			For elements $r \in \ell^\infty(\mathcal{F}_1, \mathbb{D}_q)$, the composition $q(r(\varphi))$ is well defined for any $\varphi \in \mathcal{F}_1$. For elements $r \in \ell_q^\infty(\mathcal{F}_1, \mathbb{D}_q)$, composition defines a bounded map; that is, $\varphi \mapsto q(r(\varphi))$ defines an element of $\ell^\infty(\mathcal{F}_1)$. Finally, define
			\begin{align*}
				&Q : \ell_q^\infty(\mathcal{F}_1, \mathbb{D}_q) \rightarrow \ell^\infty(\mathcal{F}_1), &&Q(r)(\varphi) = q(r(\varphi))
			\end{align*}
			
			\item For the rearrangement, define $\tilde{\mathcal{F}}_{1,x,1} = \left\{\mathbbm{1}_{1,x,1} \times f \; ; \; f \in \mathcal{F}_1\right\}$, $\tilde{\mathcal{F}}_{1,x,0} = \left\{\mathbbm{1}_{1,x,0} \times f \; ; \; f \in \mathcal{F}_1\right\}$, and
			\begin{align*}
				&\tilde{R}_{1,x} : \mathbb{D}_C \rightarrow \ell^\infty(\tilde{\mathcal{F}}_{1,x,1}) \times \ell^\infty(\{\mathbbm{1}_{1,x,1}\}) \times \ell^\infty(\{\mathbbm{1}_{x,1}\}) \times \ell^\infty(\tilde{\mathcal{F}}_{1,x,0}) \times \ell^\infty(\{\mathbbm{1}_{1,x,0}\}) \times \ell^\infty(\{\mathbbm{1}_{x,0}\}) \\
				&\tilde{R}_{1,x}(G)(\mathbbm{1}_{1,x,1} \times f, \mathbbm{1}_{1,x,1}, \mathbbm{1}_{x,1}, \mathbbm{1}_{1,x,0} \times f, \mathbbm{1}_{1,x,0}, \mathbbm{1}_{x,0}) \\
				&\hspace{1 cm} = (G(\mathbbm{1}_{1,x,1} \times f), G(\mathbbm{1}_{1,x,1}), G(\mathbbm{1}_{x,1}), G(\mathbbm{1}_{1,x,0} \times f), G(\mathbbm{1}_{1,x,0}), G(\mathbbm{1}_{x,0}))
			\end{align*}
			Lemma \ref{Lemma: Hadamard differentiability, rearranging Donsker sets} shows that $\tilde{R}_{1,x}$ is fully Hadamard differentiable tangentially to $\ell^\infty(\mathcal{F})$ and is its own derivative; i.e. $\tilde{R}_{1,x,g}' = \tilde{R}_{1,x}$. Now view $\tilde{R}_{1,x}$ as a map from $\mathbb{D}_C \subseteq \ell^\infty(\mathcal{F})$ to $\ell_q^\infty(\mathcal{F}_1, \mathbb{D}_q)$, i.e. define $R_{1,x} : \mathbb{D}_C \rightarrow \ell_q^\infty(\mathcal{F}_1, \mathbb{D}_q)$ pointwise with
			\begin{align*}
				R_{1,x}(G)(f) &= \tilde{R}_{1,x}(G)(\mathbbm{1}_{1,x,1} \times f, \mathbbm{1}_{1,x,1}, \mathbbm{1}_{x,1}, \mathbbm{1}_{1,x,0} \times g, \mathbbm{1}_{1,x,0}, \mathbbm{1}_{x,0}) \\
				&= (G(\mathbbm{1}_{1,x,1} \times f), G(\mathbbm{1}_{1,x,1}), G(\mathbbm{1}_{x,1}), G(\mathbbm{1}_{1,x,0} \times f), G(\mathbbm{1}_{1,x,0}), G(\mathbbm{1}_{x,0}))
			\end{align*}
			Note that $G \in \mathbb{D}_C$ implies 
			\begin{align*}
				\sup_{f \in \mathcal{F}_1} \lvert q(R_{1,x}(G)(f)) \rvert = \sup_{f \in \mathcal{F}_1} \left\lvert \frac{G(\mathbbm{1}_{1,x,1} \times f)/G(\mathbbm{1}_{x,1}) - G(\mathbbm{1}_{1,x,0}\times f)/G(\mathbbm{1}_{x,0})}{G(\mathbbm{1}_{1,x,1})/G(\mathbbm{1}_{x,1})- G(\mathbbm{1}_{1,x,0})/G(\mathbbm{1}_{x,0})} \right\rvert < \infty
			\end{align*}
			and thus $R_{1,x}(G) \in \ell_q^\infty(\mathcal{F}_1, \mathbb{D}_q)$.
			
			\item To apply corollary \ref{Lemma: Hadamard differentiability, maps between bounded function spaces, corollary}, observe that $q(n_1, p_{11}, p_1, n_0, p_{10}, p_0) = \frac{n_1/p_1 - n_0/p_0}{p_{11}/p_1 - p_{10}/p_0}$ is continuously differentiable on $\mathbb{D}_q$ with gradient $\nabla q : \mathbb{D}_q \rightarrow \mathbb{R}^6$ given by 
			\begin{align*}
				&\nabla q(n_1, p_{11}, p_1, n_0, p_{10}, p_0) = \begin{pmatrix}
					\frac{\partial q}{\partial n_1}, & \frac{\partial q}{\partial p_{11}}, & \frac{\partial q}{\partial p_1}, & \frac{\partial q}{\partial n_0}, & \frac{\partial q}{\partial p_{10}}, & \frac{\partial q}{\partial p_0}
				\end{pmatrix}^\intercal, \\
				&\frac{\partial q}{\partial n_1} = \frac{1/p_1}{p_{11}/p_1 - p_{10}/p_0} \\ 
				&\frac{\partial q}{\partial p_{11}} = -\frac{n_1/p_1 - n_0/p_0}{(p_{11}/p_1 - p_{10}/p_0)^2} \frac{1}{p_1} = \left[\frac{1/p_1}{p_{11}/p_1 - p_{10}/p_0}\right](-q) \\
				&\frac{\partial q}{\partial p_1} = \frac{(p_{11}/p_1 - p_{10}/p_0)(-n_1/p_1^2) - (n_1/p_1 - n_0/p_0)(-p_{11}/p_1^2)}{(p_{11}/p_1 - p_{10}/p_0)^2} \\
				&\hspace{1 cm} = \frac{- n_1/p_1^2}{p_{11}/p_1 - p_{10}/p_0} + \frac{q(p_{11}/p_1^2)}{p_{11}/p_1 - p_{10}/p_0} = \left[\frac{1/p_1}{p_{11}/p_1 - p_{10}/p_0}\right] \frac{q p_{11} - n_1}{p_1} \\
				&\frac{\partial q}{\partial n_0} = \frac{-1/p_0}{p_{11}/p_1 - p_{10}/p_0} \\
				&\frac{\partial q}{\partial p_{10}} = -\frac{n_1/p_1 - n_0/p_0}{(p_{11}/p_1 - p_{10}/p_0)^2} \left(-\frac{1}{p_0}\right) = \left[\frac{-1/p_0}{p_{11}/p_1 - p_{10}/p_0} \right](-q) \\
				& \frac{\partial q}{\partial p_0} = \frac{(p_{11}/p_1 - p_{10}/p_0)(n_0/p_0^2) - (n_1/p_1 - n_0/p_0)(p_{10}/p_0^2)}{(p_{11}/p_1 - p_{10}/p_0)^2} \\
				&\hspace{1 cm} = \frac{n_0/p_0^2}{p_{11}/p_1 - p_{10}/p_0} - \frac{q(p_{10}/p_0^2)}{p_{11}/p_1 - p_{10}/p_0} = \left[\frac{-1/p_0}{p_{11}/p_1 - p_{10}/p_0} \right]\frac{q p_{10} - n_0}{p_0}
			\end{align*}
			Furthermore, there exists $\delta > 0$ such that
			\begin{align*}
				R_{1,x}(G)(\mathcal{F}_1) = \left\{r \in \mathbb{R}^6 \; ; \; \inf_{f \in \mathcal{F}_1} \lVert r - R_{1,x}(G)(\varphi)\rVert \leq \delta \right\} \subseteq \mathbb{D}_q
			\end{align*}
			and so lemma \ref{Lemma: Hadamard differentiability, maps between bounded function spaces, corollary} implies $Q$ is fully Hadamard differentiable at $R_{1,x}(G)$ tangentially to $\ell^\infty(\mathcal{F}_1)^6$ with derivative $Q_{R_{1,x}(G)}' : \ell^\infty(\mathcal{F}_1)^6 \rightarrow \ell^\infty(\mathcal{F}_1)$ given pointwise by
			\begin{align*}
				&Q_{R_{1,x}(G)}'(J)(f) = \left[\nabla q(R_{1,x}(G)(\varphi))\right]^\intercal J(f)
			\end{align*}
			
			\item Finally, observe that $C_{1,x}(G) = Q(R_{1,x}(G))$ and apply the chain rule (lemma \ref{Lemma: Hadamard differentiability, chain rule}) to find that $C_{1,x}$ is fully Hadamard differentiable at $G$ tangentially to $\ell^\infty(\mathcal{F})$ with derivative
			\begin{align*}
				&C_{1, x, G}' : \ell^\infty(\mathcal{F}) \rightarrow \ell^\infty(\mathcal{F}_{1,x}), &&C_{1,x, G}'(H) = Q_{R_{1,x}(G)}'(R_{1,x}(H))
			\end{align*}
			Writing out an evaluation clarifies the notation of the derivative:
			\begin{align}
				&C_{1,x,G}'(H)(f) = Q_{R_{1,x}(G)}'(R_{1,x}(H))(f) = [\nabla q(R_{1,x}(G)(f))]^\intercal R_{1,x}(H)(f) \label{Display: lemma proof, Hadamard differentiability, conditional distributions, derivative of C_1x} \\
				&= \left[\frac{1/G(\mathbbm{1}_{x,1})}{G(\mathbbm{1}_{1,x,1})/G(\mathbbm{1}_{x,1}) - G(\mathbbm{1}_{1,x,0})/G(\mathbbm{1}_{x,0})}\right] H(\mathbbm{1}_{1,x,1} \times f)  \notag \\
				&\hspace{1 cm} + \left[\frac{1/G(\mathbbm{1}_{x,1})}{G(\mathbbm{1}_{1,x,1})/G(\mathbbm{1}_{x,1}) - G(\mathbbm{1}_{1,x,0})/G(\mathbbm{1}_{x,0})}\right] (-C_{1,x}(G)(f)) H(\mathbbm{1}_{1,x,1})  \notag \\
				&\hspace{1 cm} + \left[\frac{1/G(\mathbbm{1}_{x,1})}{G(\mathbbm{1}_{1,x,1})/G(\mathbbm{1}_{x,1}) - G(\mathbbm{1}_{1,x,0})/G(\mathbbm{1}_{x,0})}\right] \frac{C_{1,x}(G)(f) \times G(\mathbbm{1}_{1,x,1}) - G(\mathbbm{1}_{1,x,1} \times f)}{G(\mathbbm{1}_{x,1})} H(\mathbbm{1}_{x,1})  \notag \\
				&\hspace{1 cm} + \left[\frac{-1/G(\mathbbm{1}_{x,0})}{G(\mathbbm{1}_{1,x,1})/G(\mathbbm{1}_{x,1}) - G(\mathbbm{1}_{1,x,0})/G(\mathbbm{1}_{x,0})}\right] H(\mathbbm{1}_{1,x,0} \times f)  \notag \\
				&\hspace{1 cm} + \left[\frac{-1/G(\mathbbm{1}_{x,0})}{G(\mathbbm{1}_{1,x,1})/G(\mathbbm{1}_{x,1}) - G(\mathbbm{1}_{1,x,0})/G(\mathbbm{1}_{x,0})}\right] (-C_{1,x}(G)(f)) H(\mathbbm{1}_{1,x,0})  \notag \\
				&\hspace{1 cm} + \left[\frac{-1/G(\mathbbm{1}_{x,0})}{G(\mathbbm{1}_{1,x,1})/G(\mathbbm{1}_{x,1}) - G(\mathbbm{1}_{1,x,0})/G(\mathbbm{1}_{x,0})}\right] \frac{C_{1,x}(G)(f) \times G(\mathbbm{1}_{1,x,0}) - G(\mathbbm{1}_{1,x, 0} \times f)}{G(\mathbbm{1}_{x,0})} H(\mathbbm{1}_{x,0}) \notag 
			\end{align}
		\end{enumerate}

		\item The same arguments imply the claim regarding $C_{0,x}$. 
		
		Specifically, notice that $C_{0,x}$ is the same outer transformation applied to a different rearrangement: let
		\begin{align*}
			R_{1,x}(G)(\varphi) &= (G(\mathbbm{1}_{1,x,1} \times \varphi), G(\mathbbm{1}_{1,x,1}), G(\mathbbm{1}_{x,1}), G(\mathbbm{1}_{1,x,0} \times \varphi), G(\mathbbm{1}_{1,x,0}), G(\mathbbm{1}_{x,0})) \\
			R_{0,x}(G)(\varphi) &= (G(\mathbbm{1}_{0,x,0} \times \psi), G(\mathbbm{1}_{0,x,0}), G(\mathbbm{1}_{x,0}), G(\mathbbm{1}_{0,x,1}\times \psi), G(\mathbbm{1}_{0,x,1}), G(\mathbbm{1}_{x,1})) 
		\end{align*}
		observe that 
		\begin{align*}
			C_{1,x}(G)(f) &= \frac{G(\mathbbm{1}_{1,x,1} \times f)/G(\mathbbm{1}_{x,1}) - G(\mathbbm{1}_{1,x,0}\times f)/G(\mathbbm{1}_{x,0})}{G(\mathbbm{1}_{1,x,1})/G(\mathbbm{1}_{x,1})- G(\mathbbm{1}_{1,x,0})/G(\mathbbm{1}_{x,0})} = q(R_{1,x}(G)(f)) \\
			C_{0,x}(G)(f) &= \frac{G(\mathbbm{1}_{0,x,0} \times f)/G(\mathbbm{1}_{x,0}) - G(\mathbbm{1}_{0,x,1}\times f)/G(\mathbbm{1}_{x,1})}{G(\mathbbm{1}_{0,x,0})/G(\mathbbm{1}_{x,0})- G(\mathbbm{1}_{0,x,1})/G(\mathbbm{1}_{x,1})} = q(R_{0,x}(G)(f))
		\end{align*}
		Thus, the same argument shows $C_{0,x} : \mathbb{D}_C \rightarrow \ell^\infty(\mathcal{F}_{0,x})$ is fully Hadamard differentiable at any $G \in \mathbb{D}_C$ tangentially to $\ell^\infty(\mathcal{F})$, and  $C_{0, x, G}'(H)(f)$ can be found with the appropriate substitutions in \eqref{Display: lemma proof, Hadamard differentiability, conditional distributions, derivative of C_1x} above.
		
		\item Finally consider $C_{s,x}$. Notice that 
		\begin{align*}
			&\mathbb{D}_{q_{s,x}} = \Big\{\{p_{1,x,1}, p_{x,1}, p_{1,x,0},p_{x,0}, p_x\}_{x \in \mathcal{X}} \in \mathbb{R}^{5M} \; ; \; \\
			&\hspace{3 cm} p_{x,1} > 0, p_{x,0} > 0, p_{1,x,1}/p_{x,1} - p_{1,x,0}/p_{x,0} > 0, p_x > 0 \text{ for all } x \in \mathcal{X}\Big\} \\
			&q_{s,x} : \mathbb{D}_{q_{s,x}} \rightarrow \mathbb{R}, \\
			&q_{s,x}(\{p_{1,x_m,1}, p_{x_m,1}, p_{1,x_m,0},p_{x_m,0}\}_{m=1}^M) = \frac{(p_{1,x,1}/p_{x,1} - p_{1,x,0}/p_{x,0})p_x}{\sum_{m=1}^M (p_{1,x_m,1}/p_{x_m,1} - p_{1,x_m,0}/p_{x_m,0})p_{x_m}}
		\end{align*}
		is continuously differentiable at any point in $\mathbb{D}_{q_{s,x}}$ with gradient 
		\begin{equation*}
			\nabla q(\{p_{1,x_m,1}, p_{x_m,1}, p_{1,x_m,0},p_{x_m,0}, p_{x_m}\}_{m=1}^M) \in \mathbb{R}^{5M}
		\end{equation*}
		Furthermore, notice that for any $G \in \mathbb{D}_C$, $C_{s,x}(G) = q_{s,x}(R_{s,x}(G))$, where 
		\begin{align*}
			&R_{s,x} : \ell^\infty(\mathcal{F}) \rightarrow \mathbb{R}^{5M}, &&R_{s,x}(G) = (\{G(\mathbbm{1}_{1,x_m,1}), G(\mathbbm{1}_{x_m,1}), G(\mathbbm{1}_{1,x_m,0}), G(\mathbbm{1}_{x_m,0}), G(\mathbbm{1}_{x_m})\}_{m=1}^M)
		\end{align*}
		
		It follows that $C_{s,x} : \mathbb{D}_C \rightarrow \mathbb{R}$ is fully Hadamard differentiable at any $G \in \mathbb{D}_C$ tangentially to $\ell^\infty(\mathcal{F})$. The derivative is
		\begin{align*}
			C_{s,x,G}'(H) &= \sum_{m=1}^M \frac{\partial q_{s,x}}{\partial p_{1,x_m,1}}(R_{s,x}(G)) \times H(\mathbbm{1}_{1,x_m,1}) + \frac{\partial q_{s,x}}{\partial p_{x_m,1}}(R_{s,x}(G)) \times H(\mathbbm{1}_{x_m,1}) \\
			&\hspace{1.5 cm} + \frac{\partial q_{s,x}}{\partial p_{1,x_m,0}}(R_{s,x}(G)) \times H(\mathbbm{1}_{1,x_m,0}) + \frac{\partial q_{s,x}}{\partial p_{x_m,0}}(R_{s,x}(G)) \times H(\mathbbm{1}_{x_m,0}) \\
			&\hspace{1.5 cm} + \frac{\partial q_{s,x}}{\partial p_{x_m}}(R_{s,x}(G)) \times H(\mathbbm{1}_{x_m})
		\end{align*}

	\end{enumerate}
	This completes the proof. 
\end{proof}

\begin{restatable}[$T_1$ is fully Hadamard differentiable]{lemma}{lemmaTOneIsFullyHadamardDifferentiable}
	\label{Lemma: Hadamard differentiability, T1 is fully Hadamard differentiable}
	\singlespacing
	
	Let $\mathcal{F}$ be defined by \eqref{Defn: F large donsker set} and $\mathbb{D}_C$ by \eqref{Defn: D_C, domain of the overall map}. Let $C_{d,x}$ and $C_{s,x}$ be as defined in lemma \ref{Lemma: Hadamard differentiability, conditional distributions}, and
	\begin{align*}
		&\tilde{\eta}_{d,x} : \mathbb{D}_C \rightarrow \mathbb{R}^{K_d}, &&\tilde{\eta}_{d,x}(G) = \left(C_{d,x}(G)(\eta_{d,x}^{(1)}), \ldots, C_{d,x}(G)(\eta_{d,x}^{(K_d)})\right)
	\end{align*}
	Further define
	\begin{align*}
		&T_1 : \mathbb{D}_C \rightarrow \prod_{m=1}^M \ell^\infty(\mathcal{F}_{1,x_m}) \times \ell^\infty(\mathcal{F}_{0,x_m}) \times \mathbb{R}^{K_1}\times \mathbb{R}^{K_0} \times \mathbb{R} \\
		&T_1(G) = \left(\left\{C_{1,x}(G), C_{0,x}(G), \tilde{\eta}_{1,x}(G), \tilde{\eta}_{0,x}(G), C_{s,x}(G)\right\}_{x \in \mathcal{X}}\right)
	\end{align*}
	$T_1$ is fully Hadamard differentiable at any $G \in \mathbb{D}_C$ tangentially to $\ell^\infty(\mathcal{F})$.
\end{restatable}
\begin{proof}
	\singlespacing
	
	Lemma \ref{Lemma: Hadamard differentiability, conditional distributions} shows that $C_{d,x}$ and $C_{s,x}$ are fully Hadamard differentiable at any $G \in \mathbb{D}_C$ tangentially to $\ell^\infty(\mathcal{F})$.
	
	Define the evaluation maps
	\begin{align*}
		&ev_{\eta_d^{(k)}} : \ell^\infty(\mathcal{F}_{d,x}) \rightarrow \mathbb{R}, &&ev_{\eta_d^{(k)}}(H) = H(\eta_d^{(k)})
	\end{align*}
	Note that each $ev_{\eta_d^{(k)}}$ is continuous and linear, and is therefore fully Hadamard differentiable at any $H \in \ell^\infty(\mathcal{F}_{d,x})$ tangentially to $\ell^\infty(\mathcal{F}_{d,x})$ (and is its own derivative). Moreover, 
	\begin{equation*}
		\tilde{\eta}_{d,x}(G) = (ev_{\eta_d^{(1)}}(C_{d,x}(G)), \ldots, ev_{\eta_d^{(K_1)}}(C_{d,x}(G)))
	\end{equation*} 
	is the composition of an inner function that is fully Hadamard differentiable at any $G \in \mathbb{D}_C$, and an other function that is fully differentiable at any $H \in \ell^\infty(\mathcal{F}_{d,x})$. Therefore $\tilde{\eta}_{d,x}$ is fully Hadamard differentiable at any $G \in \mathbb{D}_C$ tangentially to $\ell^\infty(\mathcal{F})$.
	
	Next apply lemma \ref{Lemma: Hadamard differentiability, stacking functions} to find that 
	\begin{align*}
		&T_1 : \mathbb{D}_C \rightarrow \prod_{m=1}^M \ell^\infty(\mathcal{F}_{1,x_m}) \times \ell^\infty(\mathcal{F}_{0,x_m}) \times \mathbb{R}^{K_1}\times \mathbb{R}^{K_0} \times \mathbb{R} \\
		&T_1(G) = \left(\left\{C_{1,x}(G), C_{0,x}(G), \tilde{\eta}_{1,x}(G), \tilde{\eta}_{0,x}(G), C_{s,x}(G)\right\}_{x \in \mathcal{X}}\right)
	\end{align*}
	is fully Hadamard differentiable at any $G \in \mathbb{D}_C$ tangentially to $\ell^\infty(\mathcal{F})$. 
\end{proof}

\subsubsection{Support of the weak limit of $\sqrt{n}(T_1(\mathbb{P}_n) - T_1(P))$}

The next few lemmas study the support of the asymptotic distribution of $\sqrt{n}(T_1(\mathbb{P}_n) - T_1(P))$; in particular, it concentrates on the tangent set of the next map studied in appendix \ref{Appendix: weak convergence, subsection optimal transport}.

\begin{restatable}[Continuity of $C_{d,x,G}'(H)(\cdot)$]{lemma}{lemmaWeakConvergenceLTwoPContinuityOfConditionalExpectationDerivative}
	\label{Lemma: weak convergence, conditional distributions continuity}
	\singlespacing
	
	Let $C_{d,x}$ be as defined in lemma \ref{Lemma: Hadamard differentiability, conditional distributions}. If $G, H \in \mathcal{C}(\mathcal{F}, L_{2,P})$, then $C_{d,x,G}'(H) \in \mathcal{C}(\mathcal{F}_{d,x}, L_{2,P})$.
	
\end{restatable}
\begin{proof}
	\singlespacing
	
	Consider $C_{1,x,G}'(H)$ first. Fix $f \in \mathcal{F}_{1,x}$ and let $\varepsilon > 0$. Let 
	\begin{align*}
		&\text{Coef}_1(G) = \left[\frac{1/G(\mathbbm{1}_{x,1})}{G(\mathbbm{1}_{1,x,1})/G(\mathbbm{1}_{x,1}) - G(\mathbbm{1}_{1,x,0})/G(\mathbbm{1}_{x,0})}\right] \\
		&\text{Coef}_2(G) = \left[\frac{-1/G(\mathbbm{1}_{x,0})}{G(\mathbbm{1}_{1,x,1})/G(\mathbbm{1}_{x,1}) - G(\mathbbm{1}_{1,x,0})/G(\mathbbm{1}_{x,0})}\right]
	\end{align*}
	and use display \eqref{Display: lemma proof, Hadamard differentiability, conditional distributions, derivative of C_1x} to see that 
	\begin{align*}
		&\lvert C_{1,x,G}'(H)(f) - C_{1,x,G}'(H)(g) \rvert \\
		&= \Bigg\lvert \text{Coef}_1(G) \times \left[H(\mathbbm{1}_{1,x,1} \times f) - H(\mathbbm{1}_{1,x,1} \times g)\right] \\
		&\hspace{1 cm} + \text{Coef}_1(G) \times \left(-\left[C_{1,x}(G)(f) -C_{1,x}(G)(g)\right]\right) H(\mathbbm{1}_{1,x,1}) \\
		&\hspace{1 cm} + \text{Coef}_1(G) \times \frac{\left[C_{1,x}(G)(f) - C_{1,x}(G)(g)\right] \times G(\mathbbm{1}_{1,x,1}) - \left[G(\mathbbm{1}_{1,x,1} \times f) - G(\mathbbm{1}_{1,x,1} \times g)\right]}{G(\mathbbm{1}_{x,1})} H(\mathbbm{1}_{x,1}) \\
		&\hspace{1 cm} + \text{Coef}_2(G) \times \left[H(\mathbbm{1}_{1,x,0} \times f) - H(\mathbbm{1}_{1,x,0} \times g)\right] \\
		&\hspace{1 cm} + \text{Coef}_2(G) \times \left(-\left[(C_{1,x}(G)(f)) -C_{1,x}(G)(g)\right]\right) H(\mathbbm{1}_{1,x,0}) \\
		&\hspace{1 cm} + \text{Coef}_2(G) \times \frac{\left[C_{1,x}(G)(f) - C_{1,x}(G)(g)\right]\times G(\mathbbm{1}_{1,x,0}) - \left[G(\mathbbm{1}_{1,x,0} \times f) - G(\mathbbm{1}_{1,x,0} \times g)\right]}{G(\mathbbm{1}_{x,0})} H(\mathbbm{1}_{x,0}) \bigg\rvert 
	\end{align*}
	Recall that $C_{1,x}(G)(f) = \frac{G(\mathbbm{1}_{1,x,1} \times f)/G(\mathbbm{1}_{x,1}) - G(\mathbbm{1}_{1,x,0} \times f)/G(\mathbbm{1}_{x,0})}{G(\mathbbm{1}_{1,x,1})/G(\mathbbm{1}_{x,1}) - G(\mathbbm{1}_{1,x,0})/G(\mathbbm{1}_{x,0})}$, and thus 
	\begin{equation*}
		C_{1,x}(G)(f) -C_{1,x}(G)(g) = \frac{[G(\mathbbm{1}_{1,x,1} \times f) - G(\mathbbm{1}_{1,x,1} \times g)] / G(\mathbbm{1}_{x,1}) - [G(\mathbbm{1}_{1,x,0} \times f) - G(\mathbbm{1}_{1,x,0} \times g)]/G(\mathbbm{1}_{x,0})}{G(\mathbbm{1}_{1,x,1})/G(\mathbbm{1}_{x,1}) - G(\mathbbm{1}_{1,x,0})/G(\mathbbm{1}_{x,0})}
	\end{equation*}
	use this to see that 
	\begin{align}
		&\lvert C_{1,x,G}'(H)(f) - C_{1,x,G}'(H)(g) \rvert \notag\\
		&\hspace{1 cm} \leq A_1 \times \lvert H(\mathbbm{1}_{1,x,1} \times f) - H(\mathbbm{1}_{1,x,1} \times g) \rvert + A_2 \times \lvert G(\mathbbm{1}_{1,x,1} \times f) - G(\mathbbm{1}_{1,x,1} \times g) \rvert \notag \\
		&\hspace{2 cm} + A_3 \times \lvert H(\mathbbm{1}_{1,x,0} \times f) - H(\mathbbm{1}_{1,x,0} \times g) \rvert + A_4 \times \lvert G(\mathbbm{1}_{1,x,0} \times f) - G(\mathbbm{1}_{1,x,0} \times g) \rvert \label{Display: lemma proof, weak convergence, conditional distributions, C_1xG(H)(.) is continuous inequality 1}
	\end{align}
	for finite constants $A_1$, $A_2$, $A_3$, and $A_4$ that depend on $G$ and $H$, but not on $f$ or $g$. Now use $G,H \in \mathcal{C}(\mathcal{F}, L_{2,P})$ to choose $\delta_{z,H} > 0$ and $\delta_{z,G} > 0 $ such that 
	\begin{align}
		&L_{2,P}(\mathbbm{1}_{1,x,1} \times f, \mathbbm{1}_{1,x,1} \times g) < \delta_{1,H} &&\implies &&\lvert H(\mathbbm{1}_{1,x,1} \times f) - H(\mathbbm{1}_{1,x,1} \times g) \rvert < \varepsilon/(4A_1) \notag \\
		&L_{2,P}(\mathbbm{1}_{1,x,1} \times f, \mathbbm{1}_{1,x,1} \times g) < \delta_{1,G} &&\implies &&\lvert G(\mathbbm{1}_{1,x,1} \times f) - G(\mathbbm{1}_{1,x,1} \times g) \rvert < \varepsilon/(4A_2) \notag \\
		&L_{2,P}(\mathbbm{1}_{1,x,0} \times f, \mathbbm{1}_{1,x,0} \times g) < \delta_{0,H} &&\implies &&\lvert H(\mathbbm{1}_{1,x,0} \times f) - H(\mathbbm{1}_{1,x,0} \times g)\rvert < \varepsilon/(4A_3) \notag \\
		&L_{2,P}(\mathbbm{1}_{1,x,0} \times f, \mathbbm{1}_{1,x,0} \times g) < \delta_{0,G} &&\implies &&\lvert G(\mathbbm{1}_{1,x,0} \times f) - G(\mathbbm{1}_{1,x,0} \times g) \rvert < \varepsilon/(4A_4) \label{Display: lemma proof, weak convergence, conditional distributions, C_1xG(H)(.) is continuous inequality 2}
	\end{align}
	Finally, notice that 
	\begin{align}
		L_{2,P}(\mathbbm{1}_{1,x,z} \times f, \mathbbm{1}_{1,x,z} \times g) &= \sqrt{P((\mathbbm{1}_{1,x,z} \times f - \mathbbm{1}_{1,x,z} \times g)^2)} = \sqrt{P(\mathbbm{1}_{1,x,z} \times (f - g)^2)} \notag \\
		&\leq \sqrt{P((f - g)^2)} = L_{2,P}(f, g) \label{Display: lemma proof, weak convergence, conditional distributions, C_1xG(H)(.) is continuous inequality 3}
	\end{align}
	It follows from \eqref{Display: lemma proof, weak convergence, conditional distributions, C_1xG(H)(.) is continuous inequality 1}, \eqref{Display: lemma proof, weak convergence, conditional distributions, C_1xG(H)(.) is continuous inequality 2}, and \eqref{Display: lemma proof, weak convergence, conditional distributions, C_1xG(H)(.) is continuous inequality 3} that 
	\begin{align*}
		&L_{2,P}(f, g) < \min\{\delta_{1,H}, \delta_{1,G}, \delta_{0,H}, \delta_{0,G}\} &&\implies &&\lvert C_{1,x,G}'(H)(f) - C_{1,x,G}'(H)(g) \rvert < \varepsilon
	\end{align*}
	i.e., $C_{1,x,G}'(H)(\cdot)$ is continuous at $f$. Since $f \in \mathcal{F}_{1,x}$ and $G, H \in \mathcal{C}(\mathcal{F}, L_{2,P})$ were arbitrary, this shows that $G,H \in \mathcal{C}(\mathcal{F}, L_{2,P})$ implies $C_{1,x,G}'(H) \in \mathcal{C}(\mathcal{F}_{1,x}, L_{2,P})$. \\
	
	The same argument shows that $G,H \in \mathcal{C}(\mathcal{F}, L_{2,P})$ implies $C_{0,x,G}'(H) \in \mathcal{C}(\mathcal{F}_{0,x}, L_{2,P})$. This completes the proof.
\end{proof}

\begin{restatable}[Support of $T_{1,P}'(\mathbb{G})$]{lemma}{lemmaSupportOfFirstStage}
	\label{Lemma: weak convergence, support of T_1P(G)}
	\singlespacing
	
	Let $\mathcal{F}$ be defined by \eqref{Defn: F large donsker set} and $T_1$ be as defined in lemma \ref{Lemma: Hadamard differentiability, T1 is fully Hadamard differentiable}. 
	\begin{enumerate}
		\item If assumption \ref{Assumption: setting} holds, $P \in \mathbb{D}_C$ and hence $T_1$ is fully Hadamard differentiable at $P$ tangentially to $\ell^\infty(\mathcal{F})$. 
		\item If assumptions \ref{Assumption: setting}, \ref{Assumption: cost function}, and \ref{Assumption: parameter, function of moments} hold,
		\begin{equation*}
			\sqrt{n}(T_1(\mathbb{P}_n) - T_1(P)) \overset{L}{\rightarrow} T_{1,P}'(\mathbb{G})
		\end{equation*}
		where $\mathbb{G}$ is the Gaussian limit of $\sqrt{n}(\mathbb{P}_n - P)$ in $\ell^\infty(\mathcal{F})$ discussed in lemma \ref{Lemma: weak convergence, Donsker results, large Donsker set F is Donsker}. 
		
		\item If assumptions \ref{Assumption: setting}, \ref{Assumption: cost function}, and \ref{Assumption: parameter, function of moments}, then $P(T_{1,P}'(\mathbb{G}) \in \mathbb{D}_{Tan,Full}) = 1$ where 
		\begin{equation}
			\mathbb{D}_{Tan, Full} = \prod_{m=1}^M \Big(\ell_{\mathcal{Y}_{1, x_m}}^\infty(\mathcal{F}_{1, x_m}) \times \ell_{\mathcal{Y}_{0, x_m}}^\infty(\mathcal{F}_{0, x_m})\Big) \cap \Big(\mathcal{C}(\mathcal{F}_{1, x_m}, L_{2,P}) \times \mathcal{C}(\mathcal{F}_{0, x_m}, L_{2,P})\Big) \times \mathbb{R}^{K_1} \times \mathbb{R}^{K_0} \times \mathbb{R} \label{Defn: support of T_1(G)}
		\end{equation}
	\end{enumerate}
	
\end{restatable}
\begin{proof}
	\singlespacing
	
	In steps:
	\begin{enumerate}
		\item $P \in \mathbb{D}_C$ and differentiability of $T_1$ at $P$.
		
		Assumption \ref{Assumption: setting} implies $P \in \mathbb{D}_C$, given by \eqref{Defn: D_C, domain of the overall map}. To see this, recall that assumption \ref{Assumption: setting} \ref{Assumption: setting, common support} is that $P(\mathbbm{1}_{x,z}) = P(X = x, Z = z) > 0$ (implying $P(\mathbbm{1}_x) = P(X = x) = P(X = x, Z = 1) + P(X = x, Z = 0) > 0$). Furthermore,
		\begin{align*}
			&P(\mathbbm{1}_{d,x,d})/P(\mathbbm{1}_{x,d}) - P(\mathbbm{1}_{d,x,1-d})/P(\mathbbm{1}_{x,1-d}) \\
			&\hspace{1 cm} = P(D = d \mid X = x, Z = d) - P(D = d \mid X = x, Z = 1-d) \\
			&\hspace{1 cm} = P(D_1 > D_0 \mid X = x) > 0
		\end{align*}
		The second equality is shown in the proof of lemma \ref{Lemma: identification, LATE IV marginal distribution identification}, and the inequality is assumption \ref{Assumption: setting} \ref{Assumption: setting, existence of compliers}. Lemma \ref{Lemma: Hadamard differentiability, T1 is fully Hadamard differentiable} thus shows that $T_1$ is fully Hadamard differentiable at $P$ tangentially to $\ell^\infty(\mathcal{F})$.
		
		\item Functional delta method. 
		
		Under assumptions \ref{Assumption: setting}, \ref{Assumption: cost function}, and \ref{Assumption: full differentiability}, lemma \ref{Lemma: weak convergence, Donsker results, large Donsker set F is Donsker} shows that $\sqrt{n}(\mathbb{P}_n - P) \overset{L}{\rightarrow} \mathbb{G}$ in $\ell^\infty(\mathcal{F})$. The functional delta method (\cite{van2000asymptotic} theorem 20.8) then implies 
		\begin{align*}
			&\sqrt{n}(T_1(\mathbb{P}_n) - T_1(P)) \overset{L}{\rightarrow} T_{1,P}'(\mathbb{G}), &&\text{ in } \prod_{m=1}^M \ell^\infty(\mathcal{F}_{1,x_m}) \times \ell^\infty(\mathcal{F}_{0,x_m}) \times \mathbb{R}^{K_1}\times \mathbb{R}^{K_0} \times \mathbb{R} 
		\end{align*}
		
		\item Support of $T_{1,P}'(\mathbb{G})$.
		
		Notice that $T_P'(\mathbb{G}) = \left(\left\{C_{1,x,P}'(\mathbb{G}), C_{0,x,P}'(\mathbb{G}), \tilde{\eta}_{1,x,P}'(\mathbb{G}), \tilde{\eta}_{0,x,P}'(\mathbb{G}), C_{s,x,P}'(\mathbb{G})\right\}_{x \in \mathcal{X}}\right)$, where $\tilde{\eta}_{d,x}$ are defined in lemma \ref{Lemma: Hadamard differentiability, T1 is fully Hadamard differentiable}. Let 
		\begin{equation*}
			\mathbb{S}_x = \Big(\ell_{\mathcal{Y}_{1, x_m}}^\infty(\mathcal{F}_{1, x_m}) \times \ell_{\mathcal{Y}_{0, x_m}}^\infty(\mathcal{F}_{0, x_m})\Big) \cap \Big(\mathcal{C}(\mathcal{F}_{1, x_m}, L_{2,P}) \times \mathcal{C}(\mathcal{F}_{0, x_m}, L_{2,P})\Big) \times \mathbb{R}^{K_1} \times \mathbb{R}^{K_0} \times \mathbb{R}
		\end{equation*}
		and note that it suffices to show $P\left(C_{1,x,P}'(\mathbb{G}), C_{0,x,P}'(\mathbb{G}), \tilde{\eta}_{1,x,P}'(\mathbb{G}), \tilde{\eta}_{0,x,P}'(\mathbb{G}), C_{s,x,P}'(\mathbb{G}) \in \mathbb{S}_x\right) = 1$ for each $x$. Moreover, 
		\begin{align*}
			P\left(\left(\tilde{\eta}_{1,x,P}'(\mathbb{G}), \tilde{\eta}_{0,x,P}'(\mathbb{G}), C_{s,x,P}'(\mathbb{G}) \right)\in \mathbb{R}^{K_1}\times \mathbb{R}^{K_0} \times \mathbb{R}\right) = 1
		\end{align*}
		is immediate. To complete the proof we must show $P(C_{d,x,P}'(\mathbb{G}) \in \ell_{\mathcal{Y}_{d,x}}^\infty(\mathcal{F}_{d,x})) = P(C_{d,x,P}'(\mathbb{G}) \in \mathcal{C}(\mathcal{F}_{d,x}, L_{2,P})) = 1$. 
		
		\begin{enumerate}
			\item To see that $P(C_{d,x,P}'(\mathbb{G}) \in \mathcal{C}(\mathcal{F}_{d,x}, L_{2,P})) = 1$, first note that for any functions $f_1, f_2 \in \mathcal{F}$, 
			\begin{align*}
				\lvert P(f_1) - P(f_2) \rvert \leq P(\lvert f_1 - f_2 \rvert) = P(\sqrt{(f_1 - f_2)}) \leq \sqrt{P((f_1 - f_2)^2)} = L_{2,P}(f_1, f_2)
			\end{align*}
			where the second inequality is an application of Jensen's inequality. Thus $P \in \mathcal{C}(\mathcal{F}, L_{2,P})$.  
			
			Next apply lemma \ref{Lemma: weak convergence, conditional distributions continuity} to see that if $G \in \mathcal{C}(\mathcal{F}, L_{2,P})$ then $C_{d,x,P}'(G) \in \mathcal{C}(\mathcal{F}_{d,x}, L_{2,P})$. It follows that 
			\begin{equation*}
				1 = P\left(\mathbb{G} \in \mathcal{C}(\mathcal{F}, L_{2,P})\right) \leq P\left(C_{d,x,P}'(\mathbb{G}) \in \mathcal{C}(\mathcal{F}_{d,x}, L_{2,P})\right)
			\end{equation*}
			
			\item To see that $P(C_{d,x,P}'(\mathbb{G}) \in \ell_{\mathcal{Y}_{d,x}}^\infty(\mathcal{F}_{d,x})) = 1$, we show that $P(\sqrt{n}(C_{d,x}(\mathbb{P}_n) - C_{d,x}(P)) \in \ell_{\mathcal{Y}_{d,x}}^\infty(\mathcal{F}_{d,x})) = 1$. 
			
			First recall the definition given in \eqref{Defn: ell_Yd^infty(G)}: 
			\begin{align*}
				\ell_{\mathcal{Y}_{d,x}}^\infty(\mathcal{F}_{d,x})  &= \Big\{H \in \ell^\infty(\mathcal{F}_{d,x}) \; ; \; \text{ for all } a,b \in \mathbb{R} \text{ and } f, g \in \mathcal{F}_{d,x}, \notag \\
				&\hspace{3.5 cm} H(f) = H(\mathbbm{1}_{\mathcal{Y}_{d,x}} \times f), \text{ if } a \in \mathcal{F}_{d,x} \text{ then } H(a) = 0, \text{ and } \notag \\
				&\hspace{3.5 cm} \text{ if } a f + b g \in \mathcal{F}_{d,x} \text{ then } H(a f + b g) = aH(f) + bH(g) \Big\}  
			\end{align*}
			
			\begin{enumerate}
				\item $\sqrt{n}(C_{d,x}(\mathbb{P}_n) - C_{d,x}(P))$ is linear and evaluates constants to zero. \\
				
				This follows because $C_{d,x}(\mathbb{P}_n)$ and $C_{d,x}(P)$ are linear and ``return constants''. To see this, recall that $C_{d,x}(P) \in \ell^\infty(\mathcal{F}_{d,x})$ is given pointwise by
				\begin{align*}
					C_{d,x}(P)(f) = \frac{P(\mathbbm{1}_{d,x,d} \times f)/P(\mathbbm{1}_{x,d}) - P(\mathbbm{1}_{d,x,1-d} \times f)/P(\mathbbm{1}_{x,1-d})}{P(\mathbbm{1}_{d,x,d})/P(\mathbbm{1}_{x,d}) - P(\mathbbm{1}_{d,x,1-d})/P(\mathbbm{1}_{x,1-d})}
				\end{align*}
				Use this to see that for any $a, b \in \mathbb{R}$ and $f, g \in \mathcal{F}_{d,x}$. if $a f + b g \in \mathcal{F}_{d,x}$, then linearity of $P$ implies $C_{d,x}(P)(a f + b g) = a C_{d,x}(P)(f) + b C_{d,x}(P)(g)$ and $C_{d,x}(\mathbb{P}_n)(a f + b g) = a C_{d,x}(\mathbb{P}_n)(f) + b C_{d,x}(\mathbb{P}_n)(g)$. Similarly, if $a \in \mathcal{F}_{d,x}$ is the constant function always returning $a$, then $C_{d,x}(P)(a) = a$. The same observations apply to $C_{d,x}(\mathbb{P}) \in \ell^\infty(\mathcal{F}_{d,x})$.
				
				Therefore
				\begin{align*}
					&\sqrt{n}(C_{d,x}(\mathbb{P}_n) - C_{d,x}(P))(a f + b g) \\
					&\hspace{1 cm} = \sqrt{n}(C_{d,x}(\mathbb{P}_n)(a f + b g) =  - C_{d,x}(P)(a f + b g)) \\
					&\hspace{1 cm} = \sqrt{n}(aC_{d,x}(\mathbb{P}_n)(f) + bC_{d,x}(\mathbb{P}_n)(g) - a C_{d,x}(P)(f) - b C_{d,x}(P)(g)) \\
					&\hspace{1 cm} = a \times \sqrt{n}(C_{d,x}(\mathbb{P}_n) - C_{d,x}(P))(f) + b \times \sqrt{n}(C_{d,x}(\mathbb{P}_n) - C_{d,x}(P))(g)
				\end{align*}
				and furthermore, if $a \in \mathcal{F}_{d,x}$, then
				\begin{equation*}
					\sqrt{n}(C_{d,x}(\mathbb{P}_n) - C_{d,x}(P))(a) = \sqrt{n}(a - a) = 0
				\end{equation*}

				\item $C_{d,x}(P)$ ``ignores values outside $\mathcal{Y}_{d,x}$''; i.e. $C_{d,x}(P)(f) = C_{d,x}(P)(\mathbbm{1}_{\mathcal{Y}_{d,x}} \times f)$. 
				
				To see this, 
				\begin{align}
					&C_{d,x}(P)(f) \label{Display: lemma proof, weak convergence, conditional distributions, C_{d,x}(P) ignores values outside support} \\
					&=  \frac{E[f(Y) \mathbbm{1}\{D = d\} \mid X = x, Z = d] - E[f(Y) \mathbbm{1}\{D = d\} \mid X = x, Z = 1-d]}{P(\mathbbm{1}_{d,x,d})/P(\mathbbm{1}_{x,d}) - P(\mathbbm{1}_{d,x,1-d})/P(\mathbbm{1}_{x,1-d})} \notag \\
					&= \frac{P(D = d \mid X = x, Z = d)E[f(Y) \mid D = d, X = x, Z = d]}{P(\mathbbm{1}_{d,x,d})/P(\mathbbm{1}_{x,d}) - P(\mathbbm{1}_{d,x,1-d})/P(\mathbbm{1}_{x,1-d})} \notag \\
					&\hspace{1 cm} - \frac{P(D = d \mid X = x, Z = 1-d)E[f(Y) \mid D = d, X = x, Z = 1-d]}{P(\mathbbm{1}_{d,x,d})/P(\mathbbm{1}_{x,d}) - P(\mathbbm{1}_{d,x,1-d})/P(\mathbbm{1}_{x,1-d})} \notag
				\end{align}
				Since $\mathcal{Y}_{d,x}$ is the support of $Y \mid D = d, X = x$, 
				\begin{align*}
					&E[f(Y) \mid D = d, X = x, Z = z] \\
					&\hspace{1 cm} = E[f(Y) \mathbbm{1}\{Z = z\}\mid D = d, X = x]/P(Z = z \mid D = d, X = x) \\
					&\hspace{1 cm} = E[\mathbbm{1}\{Y \in \mathcal{Y}_{d,x}\} f(Y) \mathbbm{1}\{Z = z\}\mid D = d, X = x]/P(Z = z \mid D = d, X = x) \\
					&\hspace{1 cm} = E[\mathbbm{1}\{Y \in \mathcal{Y}_{d,x}\} f(Y) \mid D = d, X = x, Z = z]
				\end{align*}
				Along with \eqref{Display: lemma proof, weak convergence, conditional distributions, C_{d,x}(P) ignores values outside support}, this implies $C_{d,m}(P)(f) = C_{d,m}(P)(\mathbbm{1}_{\mathcal{Y}_{d,x}} \times f)$. 
				
				\item Now notice that with probability one the sample is a subset of the support, and when this is so, $C_{d,x}(\mathbb{P}_n)$ ignores values outside of $\mathcal{Y}_{d,x}$.
				
				Specifically, observe that
				\begin{align}
					&C_{d,x}(\mathbb{P}_n)(f) \label{Display: lemma proof, weak convergence, conditional distributions, C_{d,x}(P_n) ignores values outside support} \\
					&= \frac{\left[\frac{1}{n}\sum_{i=1}^n \mathbbm{1}\{D_i = d, X_i=  x\}\mathbbm{1}\{Z_i = d\}f(Y_i)\right]/\left[\frac{1}{n}\sum_{i=1}^n \mathbbm{1}\{X_i = x, Z_i = d\}\right]}{\mathbb{P}_n(\mathbbm{1}_{d,x,d})/\mathbb{P}_n(\mathbbm{1}_{x,d}) - \mathbb{P}_n(\mathbbm{1}_{d,x,1-d})/\mathbb{P}_n(\mathbbm{1}_{x,1-d})} \notag \\
					&\hspace{1 cm} - \frac{\left[\frac{1}{n}\sum_{i=1}^n \mathbbm{1}\{D_i = d, X_i=  x\}\mathbbm{1}\{Z_i = 1-d\}f(Y_i)\right]/\left[\frac{1}{n}\sum_{i=1}^n \mathbbm{1}\{X_i = x, Z_i = 1-d\}\right]}{\mathbb{P}_n(\mathbbm{1}_{d,x,d})/\mathbb{P}_n(\mathbbm{1}_{x,d}) - \mathbb{P}_n(\mathbbm{1}_{d,x,1-d})/\mathbb{P}_n(\mathbbm{1}_{x,1-d})} \notag 
				\end{align}
				
				Note that because $\mathcal{Y}_{d,x}$ is the support of $Y \mid D = d, X = x$, we have that with probability one, $\{Y_i, D_i, Z_i, X_i\}_{i=1}^n \subseteq \mathcal{S} \coloneqq \bigcup_{d,z,x} \mathcal{Y}_{d,x} \times \{d\} \times \{z\} \times \{x\}$. Indeed, since $ \mathcal{Y}_{d,x} \times \{d\} \times \{z\} \times \{x\} \subseteq \mathbb{R}^4$ are disjoint for each distinct $(d,z,x)$,
				\begin{align*}
					&P((Y_i, D_i, Z_i, X_i) \in \mathcal{S}) = P\left((Y_i, D_i, Z_i, X_i) \in \bigcup_{d,z,x} \mathcal{Y}_{d,x} \times \{d\} \times \{z\} \times \{x\}\right) \\
					&\hspace{1 cm} = \sum_{d,z,x} P(Y_i \in \mathcal{Y}_{d,x}, D_i = d, X_i = x, Z_i = z) \\
					&\hspace{1 cm} = \sum_{d,z,x} P(D_i = d, X_i = x, Z_i = z) \\
					&\hspace{3 cm} \times \underbrace{P(Y_i \in \mathcal{Y}_{d,x}, Z_i = z \mid D_i = d, X_i = x)}_{=P(Z_i = z \mid D_i = d, X_i = x)}/P(Z_i = z \mid D_i = d, X_i = x) \\
					&\hspace{1 cm} = \sum_{d,z,x} P(D_i = d, X_i = x, Z_i = z) = 1
				\end{align*}
				Since $\{Y_i, D_i, Z_i, X_i\}_{i=1}^n$ is i.i.d.,
				\begin{equation*}
					P\left(\{Y_i, D_i, Z_i, X_i\}_{i=1}^n \subseteq \mathcal{S}\right) = P\left(\bigcap_{i=1}^n \left\{(Y_i, D_i, Z_i, X_i) \in \mathcal{S}\right\}\right) = \prod_{i=1}^n P\left((Y_i, D_i, Z_i, X_i) \in \mathcal{S}\right) = 1
				\end{equation*}

				When $\{Y_i, D_i, Z_i, X_i\}_{i=1}^n \subseteq \mathcal{S}$ holds, $\mathbbm{1}\{D_i = d, X_i = x\} \leq \mathbbm{1}\{Y_i \in \mathcal{Y}_{d,x}\} = \mathbbm{1}_{\mathcal{Y}_{d,x}}(Y_i)$ and thus $\mathbbm{1}_{\mathcal{Y}_{d,x}}(Y_i) \times \mathbbm{1}\{D_i = d, X_i = x\} = \mathbbm{1}\{D_i = d, X_i = x\}$. This and \eqref{Display: lemma proof, weak convergence, conditional distributions, C_{d,x}(P_n) ignores values outside support} implies that when $\{Y_i, D_i, Z_i, X_i\}_{i=1}^n \subseteq \mathcal{S}$ holds,
				\begin{align*}
					C_{d,x}(\mathbb{P}_n)(f) = C_{d,x}(\mathbb{P}_n)(\mathbbm{1}_{\mathcal{Y}_{d,x}} \times f) 
				\end{align*}
				
				\item Use the facts established above to see that 
				\begin{align*}
					&P(\sqrt{n}(C_{d,x}(\mathbb{P}_n) - C_{d,x}(P)) \in \ell_{\mathcal{Y}_{d, x}}^\infty(\mathcal{F}_{1,x})) \\
					&\hspace{0.5 cm} = P(\sqrt{n}(C_{d,x}(\mathbb{P}_n) - C_{d,x}(P)) \in \ell_{\mathcal{Y}_{d, x}}^\infty(\mathcal{F}_{d,x}) \mid \{Y_i, D_i, Z_i, X_i\}_{i=1}^n \subseteq \mathcal{S}) \\
					&\hspace{0.5 cm} = 1
				\end{align*}
				Lemma \ref{Lemma: bounded function spaces, subspace of linear functions assigning zero to constants and ignoring values outside Y_d is closed} is that $\ell_{\mathcal{Y}_{d, x}}^\infty(\mathcal{F}_{1,x})$ is closed, so Portmanteau (\cite{vaart1997weak} theorem 1.3.4) implies
				\begin{equation*}
					1 = \limsup_{n \rightarrow \infty} P(\sqrt{n}(C_{d,x}(\mathbb{P}_n) - C_{d,x}(P)) \in \ell_{\mathcal{Y}_{d, x}}^\infty(\mathcal{F}_{1,x})) \leq P(C_{d,x,P}'(\mathbb{G}) \in \ell_{\mathcal{Y}_{d, x}}^\infty(\mathcal{F}_{1,x}))
				\end{equation*}
			\end{enumerate}
		\end{enumerate}

		In summary, we have
		\begin{align*}
			1 &= P\left(\left(\tilde{\eta}_{1,x,P}'(\mathbb{G}), \tilde{\eta}_{0,x,P}'(\mathbb{G}), C_{s,x,P}'(\mathbb{G})\right)\in \mathbb{R}^{K_1}\times \mathbb{R}^{K_0} \times \mathbb{R}\right) \\
			&= P(C_{d,x,P}'(\mathbb{G}) \in \ell_{\mathcal{Y}_{d, x}}^\infty(\mathcal{F}_{d,x})) \\
			&= P\left(C_{d,x,P}'(\mathbb{G}) \in \mathcal{C}(\mathcal{F}_{d,x}, L_{2,P})\right)
		\end{align*}
		From which it follows that 
		\begin{align*}
			1 = P\left(C_{1,x,P}'(\mathbb{G}), C_{0,x,P}'(\mathbb{G}), \tilde{\eta}_{1,x,P}'(\mathbb{G}), \tilde{\eta}_{0,x,P}'(\mathbb{G}), C_{s,x,P}'(\mathbb{G}) \in \mathbb{S}_x\right) 
		\end{align*}
		for each $x$, and therefore
		\begin{align*}
			&P(T_{1,P}'(\mathbb{G}) \in \mathbb{D}_{Tan,Full}) \\
			&\hspace{1 cm} = P\left(\bigcap_{x \in \mathcal{X}} \left\{C_{1,x,P}'(\mathbb{G}), C_{0,x,P}'(\mathbb{G}), \tilde{\eta}_{1,x,P}'(\mathbb{G}), \tilde{\eta}_{0,x,P}'(\mathbb{G}), C_{s,x,P}'(\mathbb{G}) \in \mathbb{S}_x\right\}\right) = 1
		\end{align*}
		
	\end{enumerate}
	This completes the proof.
\end{proof}

\subsection{Optimal transport, $T_2(\{P_{1 \mid x}, P_{0 \mid x}, \eta_{1,x}, \eta_{0,x}, s_x\}_{x \in \mathcal{X}}) = (\{\theta_x^L, \theta_x^H, \eta_{1,x}, \eta_{0,x}, s_x\}_{x\in \mathcal{X}})$}

\label{Appendix: weak convergence, subsection optimal transport}

The second map applies the directional differentiability of optimal transport shown in appendix \ref{Appendix: properties of optimal transport, subsection differentiability}. There are three assumptions in lemma \ref{Lemma: Hadamard differentiability, optimal transport} to verify: strong duality, Donsker conditions, and completeness. Strong duality is shown by lemmas \ref{Lemma: c-concave functions, smooth costs, strong duality} and \ref{Lemma: c-concave functions, indicator of convex set costs, strong duality}, and the Donsker conditions were shown by lemma \ref{Lemma: weak convergence, Donsker results, Donsker result for F_dx}. It remains to verify the completeness assumptions.

\subsubsection{Verifying completeness}

\begin{restatable}[Completeness of dual problem feasible set in $L_2$ for smooth cost functions]{lemma}{lemmaCompletenessOfCConcaveFunctionsWithSmoothCosts}
	\label{Lemma: weak convergence, completeness, smooth costs c-concave functions}
	\singlespacing
	
	Suppose $\mathcal{Y} \subset \mathbb{R}$ is compact and $c : \mathcal{Y} \times \mathcal{Y} \rightarrow \mathbb{R}$ is $L$-Lipschitz. Let $\mathcal{F}_c$, $\mathcal{F}_c^c$ be given by \eqref{Defn: F_c for smooth costs} and \eqref{Defn: F_c^c for smooth costs} respectively:
	\begin{align*}
		\mathcal{F}_c &= \left\{\varphi : \mathcal{Y} \rightarrow \mathbb{R} \; ; \;  -\lVert c \rVert_\infty \leq \varphi(y_1) \leq \lVert c \rVert_\infty, \; \lvert \varphi(y) - \varphi(y') \rvert \leq L \lvert y - y'\rvert \right\},  \\
		\mathcal{F}_c^c &= \left\{\psi : \mathcal{Y} \rightarrow \mathbb{R} \; ; \; -2\lVert c \rVert_\infty \leq \psi(y) \leq 0, \; \lvert \psi(y) - \psi(y') \rvert \leq L \lvert y - y'\rvert \right\},
	\end{align*}
	Further let $\Phi_c$ be defined by \eqref{Defn: Phi_c, appendix}, and $\mathcal{F}_d$ defined by \eqref{Defn: F_dx, the set on which P_{d|x} is defined}. Let $L_{2,P}$ be given by \eqref{Defn: L2 semimetric, P}, and $L_2$ be given by \eqref{Defn: L2 semimetric, product space}. Then $(\mathcal{F}_{1,x} \times \mathcal{F}_{0,x}, L_2)$ and its subset $\Phi_c \cap (\mathcal{F}_c \times \mathcal{F}_c^c)$ are complete.
\end{restatable}
\begin{proof}
	\singlespacing
	
	In steps:
	\begin{enumerate}
		\item $(\mathcal{F}_c, L_{2,P})$ and $(\mathcal{F}_c^c, L_{2,P})$ are complete.

		The proof that $(\mathcal{F}_c, L_{2,P})$ is complete is broken into steps:
		\begin{enumerate}
			\item Let $\{\varphi_n\}_{n=1}^\infty \subseteq \mathcal{F}_c$ be $L_{2,P}$-Cauchy. The $L_p$ semimetrics are complete for any probability distribution (\cite{pollard2002user} section 2.7 and chapter 2 problem [19]), thus there exists $\tilde{\varphi}$ such that $L_{2,P}(\varphi_n, \tilde{\varphi}) \rightarrow 0$. Convergence in $L_{2,P}$ implies convergence almost surely along a subsequence (\cite{pollard2002user} section 2.8). Thus there exists a subsequence $\{\varphi_{n_k}\}_{k=1}^\infty$ such that $\lim_{k \rightarrow \infty} \varphi_{n_k}(y) = \tilde{\varphi}(y)$ for $P$-almost every $y$. Let $N_1 \subseteq \mathcal{Y}$ be the $P$-negligible set where this fails.
			
			\item Observe that on $N_1^c = \mathcal{Y} \setminus N_1$, $\tilde{\varphi}$ obeys the bounds and Lipschitz continuity of $\mathcal{F}_c$. Specifically, 
			\begin{equation*}
				-\lVert c \rVert_\infty \leq \lim_{k \rightarrow \infty} -\lVert c \rVert_\infty \leq \underbrace{\lim_{k \rightarrow \infty} \varphi_{n_k}(y)}_{\tilde{\varphi}(y)} \leq \lim_{k \rightarrow \infty} \lVert c \rVert_\infty \leq \lVert c \rVert_\infty
			\end{equation*}
			Furthermore, for any $y, y' \in N_1^c$,
			\begin{align*}
				\lvert \tilde{\varphi}(y) - \tilde{\varphi}(y') \rvert &= \lvert \lim_{k \rightarrow \infty} \varphi_{n_k}(y) - \lim_{k \rightarrow \infty} \varphi_{n_k}(y') \rvert = \lim_{k \rightarrow \infty} \lvert \varphi_{n_k}(y) - \varphi_{n_k}(y') \rvert \\
				&\leq \lim_{k \rightarrow \infty} L \lvert y - y' \rvert = L \lvert y - y' \rvert
			\end{align*}
			
			\item Now define functions $\bar{\varphi}, \varphi : \mathcal{Y}\rightarrow \mathbb{R}$ with
			\begin{align*}
				&\bar{\varphi}(y_1) = \sup_{y_1' \in N_1^c} \{\tilde{\varphi}(y_1') - L \lvert y_1 - y_1' \rvert\}, &&\varphi(y_1) = \max\{\bar{\varphi}(y_1), -\lVert c \rVert_\infty\} 
			\end{align*}
			Then $L_{2,P}(\varphi_n, \varphi) \rightarrow 0$ and $\varphi \in \mathcal{F}_c$, which shows $(\mathcal{F}_c, L_{2,P})$ is complete.
			\begin{enumerate}
				\item $L_{2,P}(\varphi_n, \varphi) \rightarrow 0$ follows from $\varphi(y) = \tilde{\varphi}(y)$ for all $y \in N_1^c$. To see this, let $y \in N_1^c$. Since $\tilde{\varphi}$ is $L$-Lipschitz on $N_1^c$, it follows that for any $y' \in N_1^c$,
				\begin{equation*}
					\tilde{\varphi}(y') - L \lvert y - y'\rvert \leq \tilde{\varphi}(y)
				\end{equation*}
				and thus $\bar{\varphi}(y) = \tilde{\varphi}(y)$. This implies $\bar{\varphi}(y) = \tilde{\varphi}(y) \geq -\lVert c \rVert_\infty$, and thus $\varphi(y) = \bar{\varphi}(y) = \tilde{\varphi}(y)$. Thus $\varphi(y) = \tilde{\varphi}(y)$ for $P$-almost all $y$, implying $L_{2,P}(\tilde{\varphi}, \varphi) = 0$ and thus $L_{2,P}(\varphi_n, \varphi) \rightarrow 0$. 
				
				\item To see that $\varphi \in \mathcal{F}_c$, first notice that $\bar{\varphi}(y) = \sup_{y' \in N_1^c} \{\tilde{\varphi}(y') - L \lvert y - y'\rvert\} \leq \sup_{y' \in N_1^c} \tilde{\varphi}(y) \leq \lVert c \rVert_\infty$, and hence $\bar{\varphi}$ obeys the upper bound for $\mathcal{F}_c$. It then follows easily that $\varphi(y) = \max\{\bar{\varphi}(y), -\lVert c \rVert_\infty\}$ obeys both the upper and lower bound. Next notice that $\bar{\varphi}$ is $L$-Lipschitz on all of $\mathcal{Y}$:
				\begin{align*}
					\bar{\varphi}(y) - \bar{\varphi}(y') = &= \sup_{\tilde{y} \in N_1^c} \{\tilde{\varphi}(\tilde{y}) - L \lvert y - \tilde{y}\rvert\} - \sup_{\tilde{y}' \in N_1^c} \{\tilde{\varphi}(\tilde{y}') - L \lvert y' - \tilde{y}'\rvert\} \\
					&\leq \sup_{\tilde{y} \in N_1^c} \{\tilde{\varphi}(\tilde{y}) - L \lvert y - \tilde{y}\rvert - \left(\tilde{\varphi}(\tilde{y}) - L\lvert y' - \tilde{y}\rvert\right)\} \\
					&= \sup_{\tilde{y} \in N_1^c} L\left(\lvert y' - \tilde{y} \rvert - \lvert y - \tilde{y}\rvert\right) \leq L\lvert y - y' \rvert
				\end{align*}
				where the last inequality follows from the reverse triangle inequality. It follows that $\varphi(y_1) = \max\{\bar{\varphi}(y_1), -\lVert c \rVert_\infty\}$ is also $L$-Lipschitz, and thus $\varphi \in \mathcal{F}_c$. 
			\end{enumerate}
		\end{enumerate}
		
		\item Very similar steps show that $(\mathcal{F}_c^c, L_{2,P})$ is complete; the only substantial changes are replacing the lower bounds with $-2\lVert c \rVert$ and the upper bounds with $0$.
		
		\item Note that since $(\mathcal{F}_c \times \mathcal{F}_c^c, L_2)$ is the product space of $(\mathcal{F}_c, L_{2,P})$ and $(\mathcal{F}_c^c, L_{2,P})$, it follows that $(\mathcal{F}_c \times \mathcal{F}_c^c, L_2)$ is complete. 
		
		\item $\Phi_c \cap (\mathcal{F}_c \times \mathcal{F}_c^c)$ is complete. 
		
		To see that $\Phi_c \cap (\mathcal{F}_c \times \mathcal{F}_c^c)$ is complete, let $\{(\varphi_n, \psi_n)\}_{n=1}^\infty \subseteq \Phi_c \cap (\mathcal{F}_c \times \mathcal{F}_c^c)$ be $L_2$-Cauchy, and follow the same steps shown above to define $(\varphi, \psi) \in \mathcal{F}_c \times \mathcal{F}_c^c$ such that $L_2((\varphi_n, \psi_n), (\varphi, \psi)) \rightarrow 0$. It remains to show that $\varphi(y_1) + \psi(y_0) \leq c(y_1,y_0)$ for all $(y_1, y_0) \in \mathcal{Y} \times \mathcal{Y} \subseteq \mathbb{R}^2$. 
		
		Since $c$ is $L$-Lipschitz,
		\begin{align*}
			c(y_1, y_0) - c(y_1', y_0) \geq -L \lVert (y_1, y_0) - (y_1', y_0') \rVert \geq -L \lvert y_1 - y_1' \rvert - L\lvert y_0 - y_0' \rvert
		\end{align*}
		which implies $c(y_1', y_0') -L \lvert y_1 - y_1' \rvert - L\lvert y_0 - y_0' \rvert \leq c(y_1,y_0)$. Thus
		\begin{align*}
			\bar{\varphi}(y_1) + \bar{\varphi}(y_0) &= \sup_{y_1' \in N_1^c} \{\tilde{\varphi}(y_1') - L\lvert y_1 - y_1' \rvert\} + \sup_{y_0' \in N_0^c} \{\tilde{\psi}(y_0') - L\lvert y_0 - y_0' \rvert\} \\
			&= \sup_{(y_1', y_0') \in N_1^c \times N_0^c} \left\{\tilde{\varphi}(y_1') + \tilde{\psi}(y_0') - L\lvert y_1 - y_1' \rvert - L\lvert y_0 - y_0' \rvert\right\} \\
			&\leq \sup_{(y_1', y_0') \in N_1^c \times N_0^c} \left\{c(y_1', y_0') - L\lvert y_1 - y_1' \rvert - L\lvert y_0 - y_0' \rvert\right\} \\
			&\leq \sup_{(y_1', y_0') \in N_1^c \times N_0^c} \left\{ c(y_1, y_0) \right\} = c(y_1, y_0)
		\end{align*}
		
		Finally,
		\begin{align*}
			\varphi(y_1) + \psi(y_0) &= \max\{\bar{\varphi}(y_1), -\lVert c \rVert_\infty\} + \max\{\bar{\psi}(y_0), -2\lVert c \rVert\} \\
			&= \max\{\bar{\varphi}(y_1) + \bar{\varphi}(y_0), \bar{\varphi}(y_1) - 2\lVert c \rVert_\infty, - \lVert c \rVert_\infty + \bar{\psi}(y_0), - \lVert c \rVert_\infty - 2\lVert c \rVert \} \\
			&\leq \max\{c(y_1,y_0), -\lVert c \rVert_\infty, -\lVert c \rVert_\infty, -3\lVert c \rVert_\infty\} \\
			&\leq c(y_1,y_0)
		\end{align*}
		where the first inequality follows from $\bar{\varphi}(y_1) \leq \lVert c \rVert_\infty$ and $\bar{\psi}(y_0) \leq 0$. 
		
		\item $(\mathcal{F}_{1,x} \times \mathcal{F}_{0,x}, L_2)$ is complete.
		
		As this is the product space of $(\mathcal{F}_{1,x}, L_{2,P})$ and $(\mathcal{F}_{0,x}, L_{2,P})$, it suffices to show these individual spaces are complete. 
		
		Now recall that $\mathcal{F}_{d,x}$ is defined by \eqref{Defn: F_dx, the set on which P_{d|x} is defined}:
		\begin{align*}
			\tilde{\mathcal{F}}_{1} &= \left\{f : \mathcal{Y}\rightarrow \mathbb{R} \; ; \; f = \varphi \text{ for some } \varphi \in \mathcal{F}_c, \text{ or } f = \eta_1^{(k)} \text{ for some } k = 1, \ldots, K_1\right\} \\
			\tilde{\mathcal{F}}_{0} &= \left\{f : \mathcal{Y} \rightarrow \mathbb{R} \; ; \; f = \psi \text{ for some } \psi \in \mathcal{F}_c^c, \text{ or } f = \eta_0^{(k)} \text{ for some } k = 1, \ldots, K_0\right\} \\
			\mathcal{F}_{d,x} &= \left\{f : \mathcal{Y}\rightarrow \mathbb{R} \; ; \; f = g \text{ or } \mathbbm{1}_{\mathcal{Y}_{d, x}} \times g  \text{ for some } g \in \tilde{\mathcal{F}}_d\right\} 
		\end{align*}
		
		Recall that the union of a finite number of complete sets is complete. Since $(\mathcal{F}_c, L_{2,P})$ and $\mathcal{F}_c^c, L_{2,P})$ are complete and any finite set is complete, $\tilde{\mathcal{F}}_d$ is complete. Next recognize that $\mathcal{F}_{d,x} = \tilde{\mathcal{F}}_d \cup \left\{\mathbbm{1}_{\mathcal{Y}_{d, x}} \times g  \; ; \; g \in \tilde{\mathcal{F}}_d \right\}$ is the union of a finite number of sets, and thus it suffices to show $\left\{\mathbbm{1}_{\mathcal{Y}_{d, x}} \times g  \; ; \; g \in \tilde{\mathcal{F}}_d \right\}$ is complete. 
		
		Let $\{\mathbbm{1}_{\mathcal{Y}_{d,x}} \times g_n\}_{n=1}^\infty \subseteq \left\{\mathbbm{1}_{\mathcal{Y}_{d, x}} \times g  \; ; \; g \in \tilde{\mathcal{F}}_d \right\}$ be $L_{2,P}$-Cauchy. Lemma \ref{Lemma: weak convergence, Donsker results, Donsker result for F_dx} shows that $\mathcal{F}_{d,x}$ is Donsker and $\sup_{f \in \mathcal{F}_{d,x}} \lvert P(f) \rvert < \infty$, which implies $(\mathcal{F}_{d,x}, L_{2, P})$ is totally bounded (see \cite{vaart1997weak} problem 2.1.2.). Since $\tilde{\mathcal{F}}_d$ is a complete subset of a totally bounded set, it is compact. Thus $\{g_n\}_{n=1}^\infty \subseteq \tilde{\mathcal{F}}_d$ is a sequence in a compact semimetric space, and therefore has a convergent subsequence $\{g_{n_k}\}_{k=1}^\infty$. Let $g \in \tilde{\mathcal{F}}_d$ be its limit, and notice that 
		\begin{align*}
			0 \leq L_{2,P}(\mathbbm{1}_{\mathcal{Y}_{d, x}} \times g_{n_k}, \mathbbm{1}_{\mathcal{Y}_{d, x}} \times g) &= \sqrt{P((\mathbbm{1}_{\mathcal{Y}_{d, x}} \times g_{n_k} - \mathbbm{1}_{\mathcal{Y}_{d, x}} \times g)^2)} \\
			&\leq \sqrt{P((g_{n_k} - g)^2)} \\
			&= L_{2,P}(g_{n_k}, g) \rightarrow 0
		\end{align*}
		and thus $\mathbbm{1}_{\mathcal{Y}_{d, x}} \times \varphi_{n_k} \rightarrow \mathbbm{1}_{\mathcal{Y}_{d, x}} g$. It follows that $\mathbbm{1}_{\mathcal{Y}_{d, x}} \times \varphi_n \rightarrow \mathbbm{1}_{\mathcal{Y}_{d, x}} g$, and thus $\left\{\mathbbm{1}_{\mathcal{Y}_{d, x}} \times g  \; ; \; g \in \tilde{\mathcal{F}}_d \right\}$ is complete.
	\end{enumerate}
	This completes the proof.
\end{proof}

\begin{restatable}[Completeness of dual problem feasible set in $L_2$ for indicator cost functions]{lemma}{lemmaCompletenessOfCConcaveFunctionsWithIndicatorCosts}
	\label{Lemma: weak convergence, completeness, indicator of convex set costs c-concave functions}
	\singlespacing
	
	Let $\mathcal{Y} \subseteq \mathbb{R}$, $C \subseteq \mathcal{Y} \times \mathcal{Y}$ be nonempty, open, and convex, and let $c : \mathcal{Y} \times \mathcal{Y} \rightarrow \mathbb{R}$ be given by $c(y_1,y_0) = \mathbbm{1}_C(y_1,y_0) = \mathbbm{1}\{(y_1, y_0) \in C\}$. Let $\mathcal{F}_c$, $\mathcal{F}_c^c$ be given by \eqref{Defn: F_c for indicator costs of convex C} and \eqref{Defn: F_c^c for indicator costs of convex C} respectively:
	\begin{align*}
		\mathcal{F}_c &= \left\{\varphi : \mathcal{Y} \rightarrow \mathbb{R} \; ; \; \varphi(y_1) = \mathbbm{1}_I(y_1) \text{ for some interval } I\right\}, \\
		\mathcal{F}_c^c &= \left\{\psi : \mathcal{Y} \rightarrow \mathbb{R} \; ; \; \psi(y_0) = -\mathbbm{1}_{I^c}(y_0) \text{ for some  interval } I\right\},
	\end{align*}
	Further let $\Phi_c$ be defined by \eqref{Defn: Phi_c, appendix}, and $\mathcal{F}_{d,x}$ defined by \eqref{Defn: F_dx, the set on which P_{d|x} is defined}. Let $L_{2,P}$ be given by \eqref{Defn: L2 semimetric, P}, and $L_2$ be given by \eqref{Defn: L2 semimetric, product space}. Then $(\mathcal{F}_{1,x} \times \mathcal{F}_{0,x}, L_2)$ and its subset $\Phi_c \cap (\mathcal{F}_c \times \mathcal{F}_c^c)$ are complete.
\end{restatable}
\begin{proof}
	\singlespacing
	
	The proof is similar in structure to that of lemma \ref{Lemma: weak convergence, completeness, smooth costs c-concave functions}. 
	
	\begin{enumerate}
		\item $(\mathcal{F}_c, L_{2,P})$ is complete. 
		
		Let $\{\varphi_n\}_{n=1}^\infty \subseteq \mathcal{F}_c$ be $L_{2, P}$-Cauchy. Note that $\varphi_n(y) = \mathbbm{1}_{I_n}(y)$ for some interval $I_n$. Just as in the proof of lemma \ref{Lemma: weak convergence, completeness, smooth costs c-concave functions}, there exists $\tilde{\varphi}$ such that $L_{2, P}(\varphi_n, \tilde{\varphi}) \rightarrow 0$, and a subsequence $\{\varphi_{n_k}\}_{k=1}^\infty$ such that $\lim_{k \rightarrow \infty} \varphi_{n_k}(y) = \tilde{\varphi}(y)$ for $P$-almost every $y$. Let $N \subset \mathcal{Y}$ be the $P$-negligible set where this convergence fails. 
		
		Let $y \in N^c$, and notice that $\varphi_{n_k}(y) = \mathbbm{1}_{I_{n_k}}(y) \in \{0,1\}$ for all $k$ and $\{\varphi_{n_k}(y)\}_{k=1}^\infty$ converging in $\mathbb{R}$ implies that $\varphi_{n_k}(y)$ is eventually constant as $k$ grows. This implies $\tilde{\varphi}(y) \in \{0,1\}$, and hence for some set $A \subset \mathcal{Y}$,
		\begin{align*}
			&\tilde{\varphi}(y) = \mathbbm{1}_A(y) &&\text{ for all } y \in N^c
		\end{align*}
		
		We will show that for some interval $I$, $A \cap N^c = I\cap N^c$. Let $y_1, y_2, y_3 \in N^c$ satisfy $y_1 < y_2 < y_3$ and $y_1, y_3 \in A$, but be otherwise arbitrary. It suffices to show that $y_2 \in A$; we can then define $I$ to be the interval with endpoints $\inf A$ and $\sup A$ (including the lower endpoint if $\inf A = \min A > -\infty$, and including the upper endpoint if $\sup A = \max A < \infty$), and define the function $\varphi : \mathcal{Y}_1 \rightarrow \mathbb{R}$ with $\varphi(y_1) = \mathbbm{1}_I(y_1)$.\footnote{Explicitly, $I$ is defined as follows: 
			\begin{enumerate*}
				\item $I = (\ell, u)$ if neither $\ell = \inf A$ nor $u = \sup A$ is attained in $\mathbb{R}$
				\item $I = [\ell, u)$ if $\ell = \inf A = \min A$, but $u = \sup A$ is not attained in $\mathbb{R}$
				\item $I = (\ell, u]$ if $\ell = \inf A$ is not attained in $\mathbb{R}$, but $u = \sup A = \max A$
				\item $I = [\ell, u]$ if both $\ell = \inf A = \min A$ and $u = \sup A = \max A$. 
			\end{enumerate*}
		} 
		
		Notice that $\lim_{k \rightarrow \infty} \mathbbm{1}_{I_{n_k}}(y_3) = \mathbbm{1}_A(y_3) = 1$ and $\lim_{k \rightarrow \infty} \mathbbm{1}_{I_{n_k}}(y_3) = \mathbbm{1}_A(y_3) = 1$ implies that $\mathbbm{1}_{I_{n_k}}(y_1)$ and $\mathbbm{1}_{I_{n_k}}(y_3)$ are eventually constant and equal to $1$, i.e. there exists $K_1, K_3 \in \mathbb{N}$ such that 
		\begin{align*}
			&y_1 \in I_{n_k} \text{ for all } k \geq K_1, \text{ and } &&y_3 \in I_{n_k} \text{ for all } k \geq K_3 
		\end{align*}
		Since $I_{n_k}$ is an interval, this implies 
		\begin{equation*}
			y_2 \in I_{n_k} \text{ for all } k \geq \max\{K_1, K_3\}
		\end{equation*}
		i.e. $\mathbbm{1}_{I_{n_k}}(y_2) = 1$ for all such $k$, and therefore $\mathbbm{1}_A(y_2) = \lim_{k \rightarrow \infty} \mathbbm{1}_{A_n}(y_2) = 1$. Thus $y_2 \in A$.
		
		It follows that $\tilde{\varphi}(y) = \varphi(y) = \mathbbm{1}_I(y)$ for all $y \in N^c$. Thus $L_{2, P}(\tilde{\varphi}, \varphi) = 0$, and $L_{2, P}(\varphi_n, \varphi) \rightarrow 0$. Since $\varphi \in \mathcal{F}_c$, this completes the proof that $(\mathcal{F}_c, L_{2, P})$ is complete. 
		
		\item $(\mathcal{F}_c^c, L_{2,P})$ is complete. 
		
		The argument is similar. Let $\{\psi_n\}_{n=1}^\infty \subseteq \mathcal{F}_c$ be $L_{2, P}$-Cauchy. Note that $\psi_n(y) = \mathbbm{1}_{I_n^c}(y)$ for some interval $I_n$. There exists $\tilde{\psi}$ such that $L_{2, P}(\psi_n, \tilde{\psi}) \rightarrow 0$, and a subsequence $\{\psi_{n_k}\}_{k=1}^\infty$ such that $\lim_{k \rightarrow \infty} \psi_{n_k}(y) = \tilde{\psi}(y)$ for $P$-almost every $y$. Let $N \subset \mathcal{Y}$ be the $P$-negligible set where this convergence fails. 
		
		Since $\psi_{n_k}(y) = \mathbbm{1}_{I_{n_k}^c}(y) \in \{0,1\}$ for all $k$ and $y$, and $\lim_{k \rightarrow \infty} \psi_{n_k}(y) = \tilde{\psi}(y)$ for all $y \in N^c$, we have $\tilde{\psi}(y) \in \{0,1\}$ for all such $y$ and thus for some set $A \subseteq \mathcal{Y}$,
		\begin{align*}
			&\tilde{\psi}(y) = \mathbbm{1}_{A^c}(y) &&\text{ for all } y \in N^c
		\end{align*}
		
		Once again, it suffices to show $A \cap N^c = I \cap N^c$ for some interval $I$. Consider $y_1, y_2, y_3 \in N^c$, $y_1 < y_2 < y_3$, with $y_1, y_3 \in A$. $\lim_{k \rightarrow \infty} \psi_{n_k}(y_1) = \tilde{\psi}(y_1) = 0$ and $\lim_{k \rightarrow \infty} \psi_{n_k}(y_3) = \tilde{\psi}(y_3) = 0$ implies that $\psi_{n_k}(y_1) = \mathbbm{1}_{I_{n_k}^c}(y_1)$ and $\psi_{n_k}(y_3) = \mathbbm{1}_{I_{n_k}^c}(y_3)$ are eventually constant and equal to $0$, i.e. for some $K_1, K_3 \in \mathbb{N}$,
		\begin{align*}
			&y_1 \in I_{n_k} \text{ for all } k \geq K_1, &&y_3 \in I_{n_k} \text{ for all } k \geq K_3
		\end{align*}
		since $I_{n_k}$ is an interval for every $k$, this implies
		\begin{equation*}
			y_2 \in I_{n_k} \text{ for all } k \geq \max\{K_1, K_3\}
		\end{equation*}
		thus $\tilde{\psi}(y_2) = \lim_{k \rightarrow \infty} \psi_{n_k}(y_2) = 0$. It follows that $A \cap N^c = I \cap N^c$, where $I$ is the interval defined by endpoints $\inf A$ and $\sup A$, which are included if attained and finite. Define $\psi(y) = \mathbbm{1}_{I^c}(y)$ and notice $\psi \in \mathcal{F}_c^c$. We have $\psi(y) = \tilde{\psi}(y)$ for all $y \in N^c$ and hence $L_{2, P}(\tilde{\psi}, \psi) = 0$. Thus $L_{2, P}(\psi_n, \psi) \rightarrow 0$, showing $(\mathcal{F}_c^c, L_{2, P})$ is complete.
		
		\item Note that $(\mathcal{F}_c \times \mathcal{F}_c^c, L_2)$ is the product space of the complete spaces $(\mathcal{F}_c, L_{2,P})$ and $(\mathcal{F}_c^c, L_{2,P})$, and so is complete. 
		
		\item We next show $\Phi_c \cap (\mathcal{F}_c \times \mathcal{F}_c^c) = \left\{(\varphi, \psi) \in \mathcal{F}_c \times \mathcal{F}_c^c \; ; \; \varphi(y_1) + \psi(y_0) \leq c(y_1,y_0)\right\}$ is complete. 
		
		Let $\{(\varphi_n, \psi_n)\}_{n=1}^\infty \subseteq \Phi_c \cap (\mathcal{F}_c \times \mathcal{F}_c^c)$ be $L_2$-Cauchy, and let $(\tilde{\varphi}, \tilde{\psi})$ be a limit in $\mathcal{F}_c \times \mathcal{F}_c^c$. Since $L_{2, P}(\varphi_n, \tilde{\varphi}) \rightarrow 0$ there exists a subsequence $\{(\varphi_{n_k}, \psi_{n_k})\}_{k=1}^\infty$ such that $\lim_{k \rightarrow \infty} \varphi_{n_k}(y_1) = \tilde{\varphi}(y_1)$ for $P$-almost all $y_1$. Let $N_1$ be the negligible set where this fails. Furthermore, $L_{2, P}(\psi_{n_k}, \tilde{\psi}) \rightarrow 0$ as $k \rightarrow \infty$ and so there is a further subsequence $\{(\varphi_{n_{k_j}}, \psi_{n_{k_j}})\}_{j=1}^\infty$ such that $\lim_{j \rightarrow \infty} \psi_{n_{k_j}}(y_0) = \tilde{\psi}(y_0)$ for $P$-almost all $y_0$. Let $N_0$ be the negligible set where this fails. It is then clear that if $(y_1, y_0) \in N_1^c \times N_0^c$, then 
		\begin{equation}
			\tilde{\varphi}(y_1) + \tilde{\psi}(y_0) = \lim_{j \rightarrow \infty} \{\varphi_{n_{k_j}}(y_1) + \psi_{n_{k_j}}(y_0)\} \leq \lim_{j \rightarrow \infty} c(y_1, y_0) = \mathbbm{1}_C(y_1, y_0) \label{Display: lemma proof, completeness of c-concave functions for indicator costs, cost function inequality}
		\end{equation}
		
		Note that $\tilde{\varphi} = \mathbbm{1}_{I_{\tilde{\varphi}}}$, and $\tilde{\psi} = -\mathbbm{1}_{I_{\tilde{\psi}}^c}$ for some intervals $I_{\tilde{\varphi}}$ and $I_{\tilde{\psi}}$. Let
		\begin{align*}
			&\ell_1 = \inf I_{\tilde{\varphi}} \cap N_1^c, &&u_1 = \sup I_{\tilde{\varphi}} \cap N_1^c, &&\ell_0 = \inf I_{\tilde{\psi}} \cap N_0^c, &&u_0 = \sup I_{\tilde{\psi}} \cap N_0^c
		\end{align*} 
		and define $\varphi = \mathbbm{1}_{I_{\varphi}}$ where $I_{\varphi}$ is the interval with endpoints $\ell_1$, $u_1$ (included if the inf/sup are finite and attained), and $\psi = -\mathbbm{1}_{I_{\psi}^c}$ where $I_{\psi}^c$ is the interval with endpoints $\ell_0$, $u_0$ (included if the inf/sup are finite and attained). Notice that $I_\varphi = I_{\tilde{\varphi}}$, $P$-almost surely and $I_\psi = I_{\tilde{\psi}}$, $P$-almost surely. 
		
		Notice that for $(y_1, y_0) \in (N_1^c \times N_0^c)^c$ to satisfy $\varphi(y_1) + \psi(y_0) = \mathbbm{1}_{I_\varphi}(y_1) - \mathbbm{1}_{I_\psi^c}(y_0) >\mathbbm{1}_C(y_1, y_0)$, it would have to be the case that $(y_1, y_0) \in (I_{\tilde{\varphi}} \times  I_{\tilde{\psi}}) \cap (N_1^c \times N_0^c)^c \setminus C$. Let $(y_1, y_0) \in (I_{\varphi} \times I_{\psi}) \cap (N_1^c \times N_0^c)^c$, and note that there exists $y_1^\ell, y_1^u \in I_{\varphi} \cap N_1^c$ with $y_1^\ell \leq y_1 \leq y_1^u$ and $y_0^\ell, y_0^u \in I_{\psi} \cap N_0^c$ with $y_0^\ell \leq y_0 \leq y_0^u$. Notice that $[y_1^\ell, y_1^u] \times [y_0^\ell, y_0^u] \subseteq C$, because $C$ is convex and \eqref{Display: lemma proof, completeness of c-concave functions for indicator costs, cost function inequality} holds for the ``corners'': $(\ell_1, \ell_0), (\ell_1, u_0), (u_1, \ell_0), (u_1, u_0) \in (I_{\varphi} \times I_{\psi}) \cap (N_1^c \times N_0^c)$. Thus $(I_{\tilde{\varphi}} \times  I_{\tilde{\psi}}) \cap (N_1^c \times N_0^c)^c \setminus C = \varnothing$, showing that $\varphi(y_1) + \psi(y_0) \leq c(y_1, y_0)$ holds for all $(y_1, y_0) \in \mathcal{Y}_1 \times \mathcal{Y}_0$. This shows $\Phi_c \cap (\mathcal{F}_c \times \mathcal{F}_c^c)$ is complete.
		
		\item The argument thet $\left(\mathcal{F}_{1,x} \times \mathcal{F}_{0,x}, L_2\right)$ is complete is identical to the argument given in step 5 of the proof of lemma \ref{Lemma: weak convergence, completeness, smooth costs c-concave functions}. 
	\end{enumerate}
	This completes the proof.
\end{proof}

\subsubsection{Differentiability of $T_2$}

We first apply lemma \ref{Lemma: Hadamard differentiability, optimal transport} to show show that $\theta^L(\cdot)$ and $\theta^H(\cdot)$, given by either \eqref{Display: thetaL, thetaH when c is continuous} or \eqref{Display: thetaL, thetaH when c is for CDF} depending on the function $c$, are Hadamard differentiable.

\begin{restatable}[]{lemma}{lemmaBoundsAreHadamardDifferentiable}
	\label{Lemma: Hadamard differentiability, theta^L(.) and theta^H(.) are directionally or fully Hadamard differentiable}
	\singlespacing
	
	Suppose assumptions \ref{Assumption: setting}, \ref{Assumption: cost function}, and \ref{Assumption: parameter, function of moments} hold. Then $\theta^L$ and $\theta^H$ given by \eqref{Display: thetaL, thetaH when c is continuous} or \eqref{Display: thetaL, thetaH when c is for CDF} are Hadamard directionally differentiable at $(P_{1 \mid x}, P_{0 \mid x})$ tangentially to $\mathcal{C}(\mathcal{F}_{1,x}, L_{2,P}) \times \mathcal{C}(\mathcal{F}_{0,x}, L_{2,P})$. The argmax sets
	\begin{align*}
		\Psi_{c_L}(P_{1 \mid x}, P_{0 \mid x}) &= \argmax_{(\varphi,\psi) \in \Phi_{c_L} \cap (\mathcal{F}_c \times \mathcal{F}_c^c)} P_{1 \mid x}(\varphi) + P_{0 \mid x}(\psi) \\
		\Psi_{c_H}(P_{1 \mid x}, P_{0 \mid x}) &= \argmax_{(\varphi,\psi) \in \Phi_{c_H} \cap (\mathcal{F}_c \times \mathcal{F}_c^c)} P_{1 \mid x}(\varphi) + P_{0 \mid x}(\psi)
	\end{align*}
	are nonempty, and the derivatives $\theta_{(P_{1 \mid x}, P_{0 \mid x})}^{L\prime}, \theta_{(P_{1 \mid x}, P_{0 \mid x})}^{H\prime}: \mathcal{C}(\mathcal{F}_{1,x}, L_{2,P}) \times \mathcal{C}(\mathcal{F}_{0,x}, L_{2,P}) \rightarrow \mathbb{R}$ are given by 
	\begin{align}
		\theta_{(P_{1 \mid x}, P_{0 \mid x})}^{L\prime}(H_1, H_0) &= \sup_{(\varphi,\psi) \in \Psi_{c_L}(P_{1 \mid x}, P_{0 \mid x})} H_1(\varphi) + H_0(\psi) \label{Display: lemma, Hadamard differentiability, T2 is directionally or fully Hadamard differentiable, derivative of BL} \\
		\theta_{(P_{1 \mid x}, P_{0 \mid x})}^{H\prime}(H_1, H_0) &= -\left[\sup_{(\varphi,\psi) \in \Psi_{c_H}(P_{1 \mid x}, P_{0 \mid x})} H_1(\varphi) + H_0(\psi) \right] \label{Display: lemma, Hadamard differentiability, T2 is directionally or fully Hadamard differentiable, derivative of BH}
	\end{align}
	
	If assumption \ref{Assumption: full differentiability} also holds, then $\theta^L$ and $\theta^H$ are fully Hadamard differentiable at $(P_{1\mid x}, P_{0 \mid x})$ tangentially to
	\begin{equation*}
		\mathbb{D}_{Tan, Full, x} = \Big(\ell_{\mathcal{Y}_{1, x}}^\infty(\mathcal{F}_{1,x}) \times \ell_{\mathcal{Y}_{0, x}}^\infty(\mathcal{F}_{0,x})\Big) \cap \Big(\mathcal{C}(\mathcal{F}_{1,x}, L_{2,P}) \times \mathcal{C}(\mathcal{F}_{0,x}, L_{2,P})\Big)
	\end{equation*}
	with the derivatives $\theta_{(P_{1 \mid x}, P_{0 \mid x})}^{L\prime}, \theta_{(P_{1 \mid x}, P_{0 \mid x})}^{H\prime} : \mathbb{D}_{Tan, Full, x} \rightarrow \mathbb{R}$ also given by \eqref{Display: lemma, Hadamard differentiability, T2 is directionally or fully Hadamard differentiable, derivative of BL} and \eqref{Display: lemma, Hadamard differentiability, T2 is directionally or fully Hadamard differentiable, derivative of BH}.
	
\end{restatable}
\begin{proof}
	\singlespacing
	
	We apply lemma \ref{Lemma: Hadamard differentiability, optimal transport}. It is clear from inspection that the cost functions $c_L$ and $c_H$ are lower semicontinuous, the sets $\mathcal{F}_{d,x}$ defined by \eqref{Defn: F_dx, the set on which P_{d|x} is defined} consists of measurable functions mapping $\mathcal{Y}$ to $\mathbb{R}$, and that the subsets $\mathcal{F}_c$ and $\mathcal{F}_c^c$ given by \eqref{Defn: F_c for smooth costs} and \eqref{Defn: F_c^c for smooth costs}, or by \eqref{Defn: F_c for indicator costs of convex C} and \eqref{Defn: F_c^c for indicator costs of convex C}, are universally bounded. Moreover,
	\begin{enumerate}
		\item Strong duality holds.
		\begin{enumerate}[label=(\roman*)]
			\item If assumption \ref{Assumption: cost function} \ref{Assumption: cost function, smooth costs} holds, then lemma \ref{Lemma: c-concave functions, smooth costs, strong duality} shows that strong duality holds.
			\item If assumption \ref{Assumption: cost function} \ref{Assumption: cost function, CDF} holds, then lemma \ref{Lemma: c-concave functions, indicator of convex set costs, strong duality} shows that strong duality holds.
		\end{enumerate}
		\item Assumption \ref{Assumption: setting} implies $P$ dominates $P_{d \mid x}$ with bounded densities $\frac{dP_{d\mid x}}{dP}$. Indeed,
		\begin{align*}
			E_{P_{d \mid x}}[f(Y_d)] &= \frac{E_P[f(Y) \mathbbm{1}\{D = d\} \mid X = x, Z = d] - E_P[f(Y) \mathbbm{1}\{D = d\} \mid X = x, Z = 1-d]}{P(D = d \mid X = x, Z = d) - P(D = d \mid X = x, Z = 1-d)}\\
			&= E_P\left[f(Y) \frac{\mathbbm{1}_{d,x,d}(D,X,Z)/p_{x,d} - \mathbbm{1}_{d,x,1-d}(D,X,Z)/p_{x,1-d}}{p_{d,x,d}/p_{x,d} - p_{d,x,1-d}/p_{x,1-d}}\right] \\
			&= E_P\left[f(Y) E\left[\frac{\mathbbm{1}_{d,x,d}(D,X,Z)/p_{x,d} - \mathbbm{1}_{d,x,1-d}(D,X,Z)/p_{x,1-d}}{p_{d,x,d}/p_{x,d} - p_{d,x,1-d}/p_{x,1-d}} \mid Y\right]\right]
		\end{align*}
		Notice that $\frac{dP_{d \mid x}}{dP}(Y) = E_P\left[\frac{\mathbbm{1}_{d,x,d}(D,X,Z)/p_{x,d} - \mathbbm{1}_{d,x,1-d}(D,X,Z)/p_{x,1-d}}{p_{d,x,d}/p_{x,d} - p_{d,x,1-d}/p_{x,1-d}} \mid Y\right]$ must be nonnegative $P$-almost surely; if the set $A = \left\{y \; ; \; \frac{dP_{d \mid x}}{dP}(y) < 0\right\}$ was $P$-non-negligible, the displays above would imply the contradiction $P(Y_d \in A \mid D_1 > D_0, X = x) < 0$. Moreover, it is bounded by $K_{d,x} = \frac{1/p_{x,d}}{p_{d,x,d}/p_{x,d} - p_{d,x,1-d}/p_{x,1-d}}$
		
		\item Lemma \ref{Lemma: weak convergence, Donsker results, Donsker result for F_dx} shows that under assumptions \ref{Assumption: setting}, \ref{Assumption: cost function}, and \ref{Assumption: parameter, function of moments}, $\mathcal{F}_{d,x}$ is $P$-Donsker and $\sup_{f \in \mathcal{F}_{d,x}} \lvert P(f)\rvert < \infty$ for $d = 1,0$, and
		
		\item The set $(\mathcal{F}_1 \times \mathcal{F}_0, L_2)$ and its subset $\Phi_c \cap (\mathcal{F}_c \times \mathcal{F}_c^c)$ are complete.
		\begin{enumerate}[label=(\roman*)]
			\item If assumption \ref{Assumption: cost function} \ref{Assumption: cost function, smooth costs} holds, then lemma \ref{Lemma: weak convergence, completeness, smooth costs c-concave functions} shows these sets are complete. 
			\item If assumption \ref{Assumption: cost function} \ref{Assumption: cost function, CDF} holds, then lemma \ref{Lemma: weak convergence, completeness, indicator of convex set costs c-concave functions} shows these sets are complete.
		\end{enumerate}
	\end{enumerate}
	It follows from the chain rule that $\theta^L$ and $\theta^H$ are Hadamard directionally differentiable with the claimed directional derivatives. 
	
	Now suppose assumptions \ref{Assumption: setting}, \ref{Assumption: cost function}, \ref{Assumption: parameter, function of moments}, and \ref{Assumption: full differentiability} hold. Lemma \ref{Lemma: Hadamard differentiability, optimal transport, full differentiability} implies $\theta^L$ and $\theta^H$ are fully Hadamard differentiable at $(P_{1 \mid x}, P_{0 \mid x})$ tangentially to
	\begin{equation*}
		\mathbb{D}_{T, Full, x} = \Big(\ell_{\mathcal{Y}_{1, x}}^\infty(\mathcal{F}_{1,x}) \times \ell_{\mathcal{Y}_{0, x}}^\infty(\mathcal{F}_{0,x}\Big) \cap \Big(\mathcal{C}(\mathcal{F}_{1,x}, L_{2,P}) \times \mathcal{C}(\mathcal{F}_{0,x}, L_{2,P})\Big)
	\end{equation*}
	with derivatives given by the same expressions.
\end{proof}

We can now show the differentiability properties of $T_2$. 

\begin{restatable}[$T_2$ is Hadamard differentiable]{lemma}{lemmaTTwoIsHadamardDifferentiable}
	\label{Lemma: Hadamard differentiability, T2 is directionally or fully Hadamard differentiable}
	\singlespacing
	Let $\mathbb{D}_{Tan}$ and $\mathbb{D}_{Tan,Full}$ be given by 
	\begin{align*}
		\mathbb{D}_{Tan} &= \prod_{m=1}^M \mathcal{C}(\mathcal{F}_{1, x_m}, L_{2,P}) \times \mathcal{C}(\mathcal{F}_{0, x_m}, L_{2,P}) \times \mathbb{R}^{K_1} \times \mathbb{R}^{K_0}  \times \mathbb{R} \\
		\mathbb{D}_{Tan, Full} &= \prod_{m=1}^M \Big(\ell_{\mathcal{Y}_{1, x_m}}^\infty(\mathcal{F}_{1, x_m}) \times \ell_{\mathcal{Y}_{0, x_m}}^\infty(\mathcal{F}_{0, x_m})\Big) \cap \Big(\mathcal{C}(\mathcal{F}_{1, x_m}, L_{2,P}) \times \mathcal{C}(\mathcal{F}_{0, x_m}, L_{2,P})\Big) \times \mathbb{R}^{K_1} \times \mathbb{R}^{K_0} \times \mathbb{R}
	\end{align*}
	and define
	\begin{align*}
		&T_2 : \prod_{m=1}^M \ell^\infty(\mathcal{F}_{1,x}) \times \ell^\infty(\mathcal{F}_{0,x})\times \mathbb{R}^{K_1} \times \mathbb{R}^{K_0} \times \mathbb{R} \rightarrow \prod_{m=1}^M \mathbb{R} \times \mathbb{R} \times \mathbb{R}^{K_1} \times \mathbb{R}^{K_0} \times \mathbb{R}, \\
		&T_2(\{P_{1 \mid x}, P_{0 \mid x}, \eta_{1,x}, \eta_{0,x}, s_x\}_{x\in \mathcal{X}})  = \left(\{\theta^L(P_{1 \mid x}, P_{0 \mid x}), \theta^H(P_{1 \mid x}, P_{0 \mid x}), \eta_{1,x}, \eta_{0,x}, s_x\}_{x \in \mathcal{X}}\right)
	\end{align*}
	
	Under assumptions \ref{Assumption: setting}, \ref{Assumption: cost function}, and \ref{Assumption: parameter, function of moments}, $T_2$ is Hadamard directionally differentiable at \\ $T_1(P) = (\{P_{1 \mid x}, P_{0 \mid x}, s_{x}, \eta_{1,x}, \eta_{0,x}\}_{x\in \mathcal{X}})$ tangentially to $\mathbb{D}_{Tan}$, with derivative
	\begin{align*}
		&T_{2,T_1(P)}' : \mathbb{D}_{Tan} \rightarrow \prod_{m=1}^M \mathbb{R} \times \mathbb{R} \times \mathbb{R}^{K_1} \times \mathbb{R}^{K_0} \times \mathbb{R} \\
		&T_{2,T_1(P)}'\left(\{H_{1,x}, H_{0,x}, h_{\eta_1, x}, h_{\eta_0, x}, h_{s, x}\}_{x \in \mathcal{X}}\right) \\
		&\hspace{1 cm} = \left(\left\{\theta_{(P_{1 \mid x}, P_{0 \mid x})}^{L\prime}(H_{1,x}, H_{0,x}), \theta_{(P_{1 \mid x}, P_{0 \mid x})}^{H\prime}(H_{1,x}, H_{0,x}), h_{\eta_1, x}, h_{\eta_0, x}, h_{s, x}\right\}_{x \in \mathcal{X}}\right)
	\end{align*}
	
	If assumption \ref{Assumption: full differentiability} also holds, then $T_2$ is fully Hadamard differentiable at $T_1(P)$ tangentially to $\mathbb{D}_{Tan,Full}$, with derivative $T_{2,T_1(P)} : \mathbb{D}_{Tan, Full} \rightarrow \prod_{m=1}^M \mathbb{R} \times \mathbb{R} \times \mathbb{R}^{K_1} \times \mathbb{R}^{K_0} \times \mathbb{R}$ given by the same expression.
\end{restatable}
\begin{proof}
	\singlespacing
	
	Lemma \ref{Lemma: Hadamard differentiability, theta^L(.) and theta^H(.) are directionally or fully Hadamard differentiable} shows that under assumptions \ref{Assumption: setting}, \ref{Assumption: cost function}, and \ref{Assumption: parameter, function of moments}, $\theta^L(\cdot)$ and $\theta^H(\cdot)$ are Hadamard directionally differentiable at $(P_{1 \mid x}, P_{0 \mid x})$ tangentially to $\mathcal{C}(\mathcal{F}_{1,x}, L_{2,P}) \times \mathcal{C}(\mathcal{F}_{0,x}, L_{2,P})$ for each $x \in \mathcal{X}$. If assumption \ref{Assumption: full differentiability} also holds, lemma \ref{Lemma: Hadamard differentiability, theta^L(.) and theta^H(.) are directionally or fully Hadamard differentiable} shows these derivatives are linear on the subspace $\mathbb{D}_{Tan, Full}$, and hence $\theta^L(\cdot)$ and $\theta^H(\cdot)$ are fully Hadamard differentiable tangentially to $\mathbb{D}_{Tan, Full}$. The other coordinates are the identity mapping, which is fully Hadamard differentiable. Apply lemma \ref{Lemma: Hadamard differentiability, stacking functions} to obtain the result.
\end{proof}

\subsection{Expectations, $T_3(\{\theta_x^L, \theta_x^H, \eta_{1,x}, \eta_{0,x}, s_x\}_{x\in\mathcal{X}}) = (\theta^L, \theta^H, \eta)$}

\begin{restatable}[]{lemma}{lemmaTThreeIsFullyHadamardDifferentiable}
	\label{Lemma: Hadamard differentiability, T3 is fully Hadamard differentiable}
	\singlespacing

	Define
	\begin{align*}
		&T_3 : \prod_{m=1}^M \mathbb{R} \times \mathbb{R} \times \mathbb{R}^{K_1} \times \mathbb{R}^{K_0} \times \mathbb{R} \rightarrow \mathbb{R} \times \mathbb{R} \times \mathbb{R}^{K_1} \times \mathbb{R}^{K_0} \\
		&T_3(\{\theta_x^L, \theta_x^H, \eta_{1,x}, \eta_{0,x}, s_x\}_{x\in\mathcal{X}}) = \left(\sum_{x \in \mathcal{X}} s_x \theta_x^L, \sum_{x \in \mathcal{X}} s_x \theta_x^H, \sum_{x \in \mathcal{X}} s_x \eta_{1,x}, \sum_{x \in \mathcal{X}} s_x \eta_{0,x}\right)
	\end{align*}
	
	$T_3$ is fully (Hadamard) differentiable at any $V = (\{\theta_x^L, \theta_x^H, \eta_{1,x}, \eta_{0,x}, s_x\}_{x\in\mathcal{X}}) \in \prod_{m=1}^M \mathbb{R} \times \mathbb{R} \times \mathbb{R} \times \mathbb{R}^{K_1} \times \mathbb{R}^{K_0}$ tangentially to $\prod_{m=1}^M \mathbb{R} \times \mathbb{R} \times \mathbb{R}^{K_1} \times \mathbb{R}^{K_0} \times \mathbb{R}$ with derivative 
	\begin{align*}
		&T_{3,V}' : \prod_{m=1}^M \mathbb{R} \times \mathbb{R} \times \mathbb{R}^{K_1} \times \mathbb{R}^{K_0} \times \mathbb{R} \rightarrow \mathbb{R} \times \mathbb{R} \times \mathbb{R}^{K_1} \times \mathbb{R}^{K_0} \\
		&T_{3,V}'(\{h_x^L, h_x^H, h_{\eta_1, x}, h_{\eta_0, x}, h_{s,x}\}_{x \in \mathcal{X}}) \\
		&\hspace{1 cm} = \left(\sum_{x \in \mathcal{X}} s_x h_x^L + h_{s,x} \theta^L(x), \sum_{x \in \mathcal{X}} s_x h_x^H + h_{s,x} \theta^H(x), \sum_{x \in \mathcal{X}} s_x h_{\eta_1,x} + h_{s,x} \eta_{1,x}, \sum_{x \in \mathcal{X}} s_x h_{\eta_0,x} + h_{s,x} \eta_{0,x}\right)
	\end{align*}
\end{restatable}
\begin{proof}
	\singlespacing
	
	The inner product 
	\begin{align*}
		&IP : \mathbb{R}^M \times \mathbb{R}^M \rightarrow \mathbb{R}, &&IP(r_1, r_2) = \langle r_1, r_2 \rangle = \sum_{m=1}^M r_1^{(m)} r_2^{(m)}
	\end{align*}
	is fully Hadamard differentiable at any $(r_1,r_2) \in \mathbb{R}^M \times \mathbb{R}^M$ tangentially to $\mathbb{R}^M \times \mathbb{R}^M$ with derivative
	\begin{align*}
		&IP_{(r_1,r_2)}' : \mathbb{R}^M \times \mathbb{R}^M \rightarrow \mathbb{R}, \\
		&IP_{(r_1,r_2)}'(h_1,h_2) = \langle r_1, h_2 \rangle + \langle h_1, r_2 \rangle = \sum_{m=1}^M r_1^{(m)} h_2^{(m)} + h_1^{(m)} r_2^{(m)}
	\end{align*}
	Apply lemma \ref{Lemma: Hadamard differentiability, stacking functions} to obtain the result.
\end{proof}

\subsection{Optimization over $t \in [\theta^L,\theta^H]$: $T_4(\theta^L,\theta^H,\eta) = (\gamma^L,\gamma^H)$}

\begin{restatable}[]{lemma}{lemmaTFourIsFullyHadamardDifferentiable}
	\label{Lemma: Hadamard differentiability, T4 is fully Hadamard differentiable}
	\singlespacing
	
	Let $g^L , g^H : \mathbb{R} \times \mathbb{R} \times \mathbb{R}^{K_1} \times \mathbb{R}^{K_0} \rightarrow \mathbb{R}$ be as defined in assumption \ref{Assumption: parameter, function of moments}:
	\begin{align*}
		&g^L(\theta^L, \theta^H, \eta_1, \eta_0) = \inf_{t \in [\theta^L, \theta^H]} g(t, \eta_1, \eta_0), &&g^H(\theta^L, \theta^H, \eta_1, \eta_0) = \sup_{t \in [\theta^L, \theta^H]} g(t, \eta_1, \eta_0)
	\end{align*}
	Define 
	\begin{align*}
		&T_4 : \mathbb{R} \times \mathbb{R} \times \mathbb{R}^{K_1} \times \mathbb{R}^{K_0} \rightarrow \mathbb{R} \times \mathbb{R} \\
		&T_4(\theta^L, \theta^H, \eta_1, \eta_0) = \left(g^L(\theta^L, \theta^H, \eta_1, \eta_0), g^H(\theta^L, \theta^H, \eta_1, \eta_0) \right)
	\end{align*}
	 
	Under assumption \ref{Assumption: parameter, function of moments}, $g^L$ and $g^H$ are continuously differentiable at $(\theta^L, \theta^H, \eta_1,\eta_0) = T_3(T_2(T_1(P)))$ with gradients
	\begin{align*}
		&\nabla g^L = \nabla g^L(\theta^L,\theta^H, \eta_1,\eta_0) \in \mathbb{R}^{2+K_1+K_0}, &&\nabla g^H = \nabla g^H(\theta^L, \theta^H, \eta_1, \eta_0) \in \mathbb{R}^{2+K_1+K_0}
	\end{align*}
	Therefore $T_4$ is fully Hadamard differentiable at $(\theta^L, \theta^H, \eta_1,\eta_0)$ tangentially to $\mathbb{R} \times \mathbb{R} \times \mathbb{R}^{K_1} \times \mathbb{R}^{K_0}$, with derivative
	\begin{align*}
		&T_{4, T_3(T_2(T_1(P)))}' : \mathbb{R} \times \mathbb{R} \times \mathbb{R}^{K_1} \times \mathbb{R}^{K_0} \rightarrow \mathbb{R} \times \mathbb{R} \\
		&T_{4, T_3(T_2(T_1(P)))}'(h^L, h^H, h_{\eta_1}, h_{\eta_0}) \\
		&\hspace{1 cm} = \left(\left\langle \nabla g^L, (h^L, h^H, h_{\eta_1}, h_{\eta_0})\right\rangle, \left\langle \nabla g^H, (h^L, h^H, h_{\eta_1}, h_{\eta_0})\right\rangle \right)
	\end{align*}
	
\end{restatable}
\begin{proof}
	\singlespacing
	
	Assumption \ref{Assumption: parameter, function of moments} \ref{Assumption: parameter, function of moments, sup and inf of g are differentiable} is that $g^L$ and $g^H$ are continuously differentiable. The result follows.
\end{proof}

\begin{remark}
	\label{Remark: zero derivatives when applying the delta method for functions of moments}
	\singlespacing
	
	This remark discusses the derivatives of $g^L$ and $g^H$. In particular, note that even if $\argmin_{t \in [\theta^L, \theta^H]} g(t,\eta)$ is within $(\theta^L, \theta^H)$, the derivative of $g^L$ and $g^H$ are unlikely to be zero because the derivatives with respect to $\eta$ will not be zero.
	
	Consider $g^H(\theta^L, \theta^H, \eta) = \sup_{t \in [\theta^L, \theta^H]} g(t, \eta)$. The maximization problem has Lagrangian
	\begin{align*}
		\mathcal{L}(t, \lambda, \theta^L, \theta^H, \eta) = g(t, \eta) + \lambda^L (t - \theta^L) + \lambda^H(\theta^H - t)
	\end{align*}
	where $\lambda = (\lambda^L, \lambda^H)$ are Lagrange multipliers. Let $g_\theta(t, \eta) = \frac{\partial g}{\partial \theta}(t,\eta)$. Suppose there is unique solution $(\theta^*, \lambda^*)$. The necessary KKT conditions imply that
	\begin{align*}
		&g_\theta(\theta^*, \eta) + \lambda^{L*} - \lambda^{H*} = 0 \\
		&\theta^* - \theta^{L*} \geq 0 \text{ w.e. if } \lambda^{L*} > 0 \\
		&\theta^{H^*} - \theta^* \geq 0 \text{ w.e. if } \lambda^{H*} > 0 \\
		&\lambda^{L*}, \lambda^{H*} \geq 0
	\end{align*}
	Notice that at most one of either $\theta^* = \theta^L$ or $\theta^* = \theta^H$ is true. If $\theta^* = \theta^L$, then $\lambda^L > 0$ and $\lambda^H = 0$, and the first KKT implies $-g_\theta(\theta^L, \eta) = \lambda^L$. Similarly, if $\theta^* = \theta^H$ is true then $\lambda^L = 0$ and $g_\theta(\theta^H, \eta) = \lambda^H$. 
	
	Now use assumption \ref{Assumption: parameter, function of moments} \ref{Assumption: parameter, function of moments, sup and inf of g are differentiable} to apply the envelope theorem, finding that  
	\begin{align*}
		\nabla g^L(\theta^L, \theta^H, \eta)^\intercal &= 
		\begin{pmatrix}
			\frac{\partial \mathcal{L}}{\partial \theta^L}(\theta^*, \lambda^*, \theta^L, \theta^H, \eta) & \frac{\partial \mathcal{L}}{\partial \theta^H}(\theta^*, \lambda^*, \theta^L, \theta^H, \eta) &\frac{\partial \mathcal{L}}{\partial \eta}(\theta^*, \lambda^*, \theta^L, \theta^H, \eta) 
		\end{pmatrix} \\
		&= 
		\begin{pmatrix}
			-\lambda^{L*} & \lambda^{H*} & g_\eta(\theta^*, \eta)
		\end{pmatrix}
	\end{align*}
	The linearization of $g^L$ at $(\theta^L, \theta^H, \eta)$ is the function $g_{(\theta^L, \theta^H, \eta)}^{L'} : \mathbb{R}\times \mathbb{R} \times \mathbb{R}^{d_1 + d_2} \rightarrow \mathbb{R}$ given by 
	\begin{align*}
		g_{(\theta^L, \theta^H, \eta)}^{L'}(h_L, h_H, h_\eta) &= \nabla g^L(\theta^L, \theta^H, \eta)^\intercal h = 
		\begin{cases}
			g_\theta(\theta^*, \eta) h_L + g_\eta(\theta^*, \eta)^{\intercal} h_\eta &\text{ if } \theta^* = \theta^L \\
			g_\theta(\theta^*, \eta) h_H + g_\eta(\theta^*, \eta)^{\intercal} h_\eta &\text{ if } \theta^* = \theta^ H\\
			g_\eta(\theta^*, \eta)^{\intercal} h_\eta &\text{ if } \theta^* \in (\theta^L, \theta^H)
		\end{cases} \\
		&= 
		\begin{pmatrix}
			g_\theta(\theta^*, \eta) \mathbbm{1}\{\theta^* = \theta^L\} & g_\theta(\theta^*, \eta) \mathbbm{1}\{\theta^* = \theta^H\} & g_\eta(\theta^*,\eta)^\intercal
		\end{pmatrix}
		\begin{pmatrix}
			h_L \\
			h_H \\
			h_\eta
		\end{pmatrix}
	\end{align*}
	where $g_\eta(t,\eta) = \frac{\partial g}{\partial \eta}(t,\eta)$. In particular, notice that the first order condition $g_\theta(\theta^*, \eta) = 0$, which holds true when $\theta^* \in (0,1)$, does \textit{not} imply this linearization is the zero map, as long as $g_\eta(\theta^*, \eta)$ is not zero. 
\end{remark}

\subsection{The map $T(P) = (\gamma^L,\gamma^H)$, consistency, and weak convergence}

\begin{restatable}[]{lemma}{lemmaConsistency}
	\label{Lemma: consistency}
	\singlespacing
	
	Let $T_1$, $T_2$, $T_3$, and $T_4$ be as defined in lemmas \ref{Lemma: Hadamard differentiability, T1 is fully Hadamard differentiable}, \ref{Lemma: Hadamard differentiability, T2 is directionally or fully Hadamard differentiable}, \ref{Lemma: Hadamard differentiability, T3 is fully Hadamard differentiable}, and \ref{Lemma: Hadamard differentiability, T4 is fully Hadamard differentiable} respectively. Let 
	\begin{align*}
		\left(\left\{\hat{P}_{1 \mid x}, \hat{P}_{0 \mid x}, \hat{\eta}_{1,x}, \hat{\eta}_{0,x}, \hat{s}_x\right\}_{x \in \mathcal{X}}\right) &= T_1(\mathbb{P}_n) \\
		\left(\{\hat{\theta}_x^L, \hat{\theta}_x^H, \hat{\eta}_{1,x}, \hat{\eta}_{0,x}, \hat{s}_x\}_{x \in \mathcal{X}}\right) &= T_2(T_1(\mathbb{P}_n)) \\
		(\hat{\theta}^L, \hat{\theta}^H, \hat{\eta}) &= T_3(T_2(T_1(\mathbb{P}_n))), \\
		(\hat{\gamma}^L, \hat{\gamma}^H) &= T_4(T_3(T_2(T_1(\mathbb{P}_n))))
	\end{align*}
	be the empirical analogue estimators. If assumptions \ref{Assumption: setting}, \ref{Assumption: cost function}, and \ref{Assumption: parameter, function of moments} hold, then each of these estimators are consistent.
	
\end{restatable}
\begin{proof}
	\singlespacing
	
	Lemmas \ref{Lemma: Hadamard differentiability, T1 is fully Hadamard differentiable}, \ref{Lemma: Hadamard differentiability, T2 is directionally or fully Hadamard differentiable}, \ref{Lemma: Hadamard differentiability, T3 is fully Hadamard differentiable}, and \ref{Lemma: Hadamard differentiability, T4 is fully Hadamard differentiable} show that $T_1$, $T_2$, $T_3$, and $T_4$ are Hadamard (directionally) differentiable at $P$, $T_1(P)$, $T_2(T_1(P))$, and $T_3(T_2(T_1(P)))$ respectively, tangentially to sets that include zero. It follows that these functions are continuous at $P$, $T_1(P)$, $T_2(T_1(P))$, and $T_3(T_2(T_1(P)))$ respectively.\footnote{
		For normed spaces $\mathbb{D}$, $\mathbb{E}$, $\phi : \mathbb{D}_\phi \subseteq \mathbb{D} \rightarrow \mathbb{E}$ is continuous at $\theta \in \mathbb{D}_\phi$ if and only if for every sequence $\{\theta_n\}_{n=1}^\infty \subseteq \mathbb{D}_\phi \setminus\{\theta\}$ with $\theta_n \rightarrow \theta$, $\lVert \phi(\theta_n) - \phi(\theta) \rVert_\mathbb{E} \rightarrow 0$. For such a sequence $\{\theta_n\}_{n=1}^\infty$, let $t_n = \lVert \theta_n - \theta \rVert_\mathbb{D}^{1/2}$ and notice that $t_n \downarrow 0$, $h_n \coloneqq \frac{\theta_n - \theta}{t_n} \rightarrow 0 \in \mathbb{D}_0$, and $\theta + t_n h_n = \theta_n \in \mathbb{D}_\phi$ for all $n$. The definition of Hadamard directional differentiability then implies $\lVert \phi(\theta + t_n h_n) - \phi(\theta) - t_n \phi_\theta'(h) \rVert_\mathbb{E}\rightarrow 0$, while the reverse traingle inequality implies
		\begin{align*}
			&\lVert \phi(\theta + t_n h_n) - \phi(\theta) - t_n \phi_\theta'(h) \rVert_\mathbb{E} \geq \left\lvert \lVert \phi(\theta + t_n h_n) - \phi(\theta)\rVert_\mathbb{E} - t_n \lVert \phi_\theta'(h)\rVert_\mathbb{E}\right\rvert \geq \lVert \phi(\theta + t_n h_n) - \phi(\theta)\rVert_\mathbb{E} - t_n \lVert \phi_\theta'(h)\rVert_\mathbb{E} \\
			\implies &0\leq \lVert \phi(\theta + t_n h_n) - \phi(\theta)\rVert_\mathbb{E} \leq \lVert \phi(\theta + t_n h_n) - \phi(\theta) - t_n \phi_\theta'(h) \rVert_\mathbb{E} + t_n \lVert \phi_\theta'(h)\rVert_\mathbb{E} \rightarrow 0
		\end{align*}
		showing continuity at $\theta$.
	}
	Lemma \ref{Lemma: weak convergence, Donsker results, large Donsker set F is Donsker} implies that $\mathbb{P}_n \overset{p}{\rightarrow} P$ in $\ell^\infty(\mathcal{F})$, so it follows from the continuous mapping theorem that 
	\begin{align*}
		T_1(\mathbb{P}_n) &\overset{p}{\rightarrow} T_1(P) \\
		T_2(T_1(\mathbb{P}_n)) &\overset{p}{\rightarrow} T_2(T_1(P)) \\
		T_3(T_2(T_1(\mathbb{P}_n))) &\overset{p}{\rightarrow} T_3(T_2(T_1(P))) \\
		T_4(T_3(T_2(T_1(\mathbb{P}_n)))) &\overset{p}{\rightarrow} T_4(T_3(T_2(T_1(P))))
	\end{align*}
	In other words, the estimates are all consistent in their respective spaces.
\end{proof}

\begin{restatable}[$T$ is Hadamard directionally differentiable]{lemma}{lemmaTIsHadamardDifferentiable}
	\label{Lemma: Hadamard differentiability, T is Hadamard directionally differentiable}
	\singlespacing
	
	Let $\mathbb{D}_C$ be defined by \eqref{Defn: D_C, domain of the overall map}, and 
	\begin{align*}
		&T : \mathbb{D}_C \rightarrow \mathbb{R}^2, &&T(G) = T_4(T_3(T_2(T_1(G))))
	\end{align*}
	If assumptions \ref{Assumption: setting}, \ref{Assumption: cost function}, \ref{Assumption: parameter, function of moments} holds, then $T$ is Hadamard directionally differentiable at $P$ tangentially to $\mathcal{C}(\mathcal{F}, L_{2,P})$ with derivative given by
	\begin{align*}
		&T_P' : \mathcal{C}(\mathcal{F}, L_{2,P}) \rightarrow \mathbb{R}^2, &&T_P'(G) = T_{4,T_3(T_2(T_1(P)))}'(T_{3, T_2(T_1(P))}'(T_{2,T_1(P)}'(T_{1,P}'(G))))
	\end{align*}
	If assumption \ref{Assumption: full differentiability} also holds, then $T$ is fully Hadamard differentiable at $P$ tangentially to the support of $\mathbb{G}$ as defined in lemma \ref{Lemma: weak convergence, Donsker results, large Donsker set F is Donsker}. 
\end{restatable}
\begin{proof}
	\singlespacing

	Lemma \ref{Lemma: Hadamard differentiability, T1 is fully Hadamard differentiable} shows that $T_1$ is fully Hadamard differentiable at any point in $\mathbb{D}_C$ tangentially to $\ell^\infty(\mathcal{F})$. Lemma \ref{Lemma: Hadamard differentiability, T2 is directionally or fully Hadamard differentiable} shows that under assumptions \ref{Assumption: setting}, \ref{Assumption: cost function}, and \ref{Assumption: parameter, function of moments}, $T_2$ is Hadamard directionally differentiable at $T_1(P)$ tangentially to 
	\begin{equation*}
		\mathbb{D}_{Tan} = \prod_{m=1}^M \mathcal{C}(\mathcal{F}_{1, x_m}, L_{2,P}) \times \mathcal{C}(\mathcal{F}_{0, x_m}, L_{2,P}) \times \mathbb{R}^{K_1} \times \mathbb{R}^{K_0}  \times \mathbb{R}
	\end{equation*}
	Lemma \ref{Lemma: weak convergence, conditional distributions continuity} implies that if $H \in \mathcal{C}(\mathcal{F}, L_{2,P})$, then $T_{1,P}'(H) \in \mathbb{D}_{Tan}$. It follows from the chain rule (lemma \ref{Lemma: Hadamard differentiability, chain rule}) that $T_2 \circ T_1$ is Hadamard directionally differentiable at $P$ tangentially to $\mathcal{C}(\mathcal{F}, L_{2,P})$. Lemma \ref{Lemma: Hadamard differentiability, T3 is fully Hadamard differentiable} shows $T_3$ is fully differentiable at any point in its domain tangentially to the entire relevant space, and lemma \ref{Lemma: Hadamard differentiability, T4 is fully Hadamard differentiable} shows $T_4$ is fully differentiable at $T_3(T_2(T_1(P)))$  tangentially to the entire relevant space. The chain rule thus implies the first claim: under assumptions \ref{Assumption: setting}, \ref{Assumption: cost function}, and \ref{Assumption: parameter, function of moments}, $T = T_4 \circ T_3 \circ T_2 \circ T_1$ is Hadamard directionally differentiable at $P$ tangentially to $\mathcal{C}(\mathcal{F}, L_{2,P})$ with the claimed derivative.
	
	If assumption \ref{Assumption: full differentiability} also holds, lemma \ref{Lemma: Hadamard differentiability, T2 is directionally or fully Hadamard differentiable} implies that $T_2$ is fully differentiable at $T_1(P)$ tangentially to $\mathbb{D}_{Tan, Full}$. Lemma \ref{Lemma: weak convergence, support of T_1P(G)} shows the support of $T_{1, P}'(\mathbb{G})$ is contained within $\mathbb{D}_{Tan, Full}$. It follows that $T_P'(\cdot) = T_{4,T_3(T_2(T_1(P)))}'(T_{3, T_2(T_1(P))}'(T_{2,T_1(P)}'(T_{1,P}'(\cdot))))$ is linear on the support of $\mathbb{G}$, and hence \cite{fang2019inference} proposition 2.1 implies $T$ is fully Hadamard differentiable at $P$ tangentially to the support of $\mathbb{G}$.
\end{proof}

\lemmaSimpleConditionsForFullDifferentiability*
\begin{proof}
	\singlespacing
	
	Note that both $c_L(y_1,y_0) = c(y_1,y_0)$ and $c_H(y_1,y_0) = -c(y_1, y_0)$ are continuously differentiable. Moreover, since the support of $P_{d \mid x}$ is $\mathcal{Y}_{d,x}$ which is a bounded interval, the support can be written as $[y_{d,x}^\ell, y_{d,x}^u]$. So for any $x \in \mathcal{X}$ and either $c \in \{c_L, c_H\}$, lemma \ref{Lemma: Kantorovich potential, sufficient conditions for uniqueness} shows that for any $(\varphi_1,\psi_1), (\varphi_2, \psi_2) \in \Psi_c(P_{1 \mid x}, P_{0 \mid x})$, there exists $s \in \mathbb{R}$ such that for all $(y_1, y_0) \in \mathcal{Y}_{1,x} \times \mathcal{Y}_{0,x}$
	\begin{align*}
		&\varphi_1(y_1) - \varphi_2(y_1) = s, &&\psi_1(y_0) - \psi_2(y_0) = -s
	\end{align*}
	and thus 
	\begin{align*}
		&\mathbbm{1}_{\mathcal{Y}_{1,x}} \times \varphi_1 = \mathbbm{1}_{\mathcal{Y}_{1,x}} \times (\varphi_2 + s), \; P\text{-a.s.} &&\text{ and } &&\mathbbm{1}_{\mathcal{Y}_{0,x}} \times \psi_1 = \mathbbm{1}_{\mathcal{Y}_{0,x}} \times (\psi_2 - s), \; P\text{-a.s.}.
	\end{align*}
	Therefore assumption \ref{Assumption: full differentiability} holds.
\end{proof}

\theoremWeakConvergenceOfEstimators*
\begin{proof}
	\singlespacing
	
	The result is an application of the functional delta method (see \cite{fang2019inference} theorem 2.1) and lemma \ref{Lemma: Hadamard differentiability, T is Hadamard directionally differentiable}. 
	
	Indeed, $\ell^\infty(\mathcal{F})$ and $\mathbb{R}^2$ are Banach spaces, and under assumptions \ref{Assumption: setting}, \ref{Assumption: cost function}, and \ref{Assumption: parameter, function of moments} lemma \ref{Lemma: Hadamard differentiability, T is Hadamard directionally differentiable} shows $T$ is Hadamard directionally differentiable at $P$ tangentially to $\mathcal{C}(\mathcal{F}, L_{2,P})$. Lemma \ref{Lemma: weak convergence, Donsker results, large Donsker set F is Donsker} shows that $\sqrt{n}(\mathbb{P}_n - P) \overset{L}{\rightarrow} \mathbb{G}$ in $\ell^\infty(\mathcal{F})$, where $\mathbb{G}$ is tight and supported in $\mathcal{C}(\mathcal{F}, L_{2,P})$. \cite{fang2019inference} theorem 2.1 gives the result that $\sqrt{n}(T(\mathbb{P}_n) - T(P)) \overset{L}{\rightarrow} T_P'(\mathbb{G})$.

	If assumption \ref{Assumption: full differentiability} holds as well as assumptions \ref{Assumption: setting}, \ref{Assumption: cost function}, and \ref{Assumption: parameter, function of moments}, then lemma \ref{Lemma: Hadamard differentiability, T is Hadamard directionally differentiable} shows that $T$ is fully differentiable on the support of $\mathbb{G}$. Since $\mathbb{G}$ is Gaussian and $T_P'$ is continuous and linear on the support of $\mathbb{G}$, $T_P'(\mathbb{G}) \in \mathbb{R}^2$ is Gaussian.
\end{proof}

%% file: appendix/OTJointPO_appendix_inference.tex
\section{Appendix: inference}

\subsection{Bootstrap}

\begin{restatable}[]{lemma}{lemmaExchangeableBootstrapSatisfiesFangSantosAssumption3}
	\label{Lemma: inference, bootstrap, exchangeable bootstrap satisfies Fang and Santos assumption 3}
	\singlespacing
	
	Suppose assumptions \ref{Assumption: setting}, \ref{Assumption: cost function}, and \ref{Assumption: parameter, function of moments} are satisfied. Let $\mathbb{P}_n^*$ be given by definition \ref{Definition: exchangeable bootstrap, nonparametric bootstrap} or \ref{Definition: exchangeable bootstrap, bayesian bootstrap}. Then \cite{fang2019inference} assumption 3 is satisfied:
	\begin{enumerate}[label=(\roman*)]
		
		\item $\mathbb{P}_n^*$ is a function of $\{Y_i, D_i, Z_i, X_i, W_i\}_{i=1}^n$, with $\{W_i\}_{i=1}^n$ independent of $\{Y_i, D_i, Z_i, X_i\}_{i=1}^n$.
		
		\item $\mathbb{P}_n^*$ satisfies $\sup_{f \in \text{BL}_1} \left\lvert E\left[ f(\sqrt{n}(\mathbb{P}_n^* - \mathbb{P}_n))\mid \{Y_i, D_i, Z_i, X_i\}_{i=1}^n \right] - E[f(\mathbb{G})] \right\rvert = o_p(1)$.
		
		\item $\sqrt{n}(\mathbb{P}_n^* - \mathbb{P}_n)$ is asymptotically measurable (jointly in $\{Y_i, D_i, Z_i, X_i, W_i\}_{i=1}^n$).
		
		\item $f(\sqrt{n}(\mathbb{P}_n^* - \mathbb{P}_n)$ is a measurable function of $\{W_i\}_{i=1}^n$ outer almost surely in $\{\{Y_i, D_i, Z_i, X_i\}_{i=1}^n$ for any continuous and bounded real-valued $f$.
	\end{enumerate}
	
\end{restatable}
\begin{proof}
	\singlespacing
	
	Note that assumption 3(i) is satisfied by construction. \cite{vaart1997weak} example 3.6.9, 3.6.10, and theorem 3.6.13 implies assumpion 3(ii) holds:
	\begin{equation*}
		\sup_{f \in \text{BL}_1} \left\lvert E\left[ f(\sqrt{n}(\mathbb{P}_n^* - \mathbb{P}_n))\mid \{Y_i, D_i, Z_i, X_i\}_{i=1}^n \right] - E[f(\mathbb{G})] \right\rvert \overset{P^*}{\rightarrow} 0
	\end{equation*}
	and further that 
	\begin{equation*}
		E\left[f(\sqrt{n}(\mathbb{P}_n^* - \mathbb{P}_n))^*\right] - E\left[f(\sqrt{n}(\mathbb{P}_n^* - \mathbb{P}_n))_*\right] = o_p(1)
	\end{equation*}
	for any $f \in \text{BL}_1$, where $f(\sqrt{n}(\mathbb{P}_n^* - \mathbb{P}_n))^*$ and $f(\sqrt{n}(\mathbb{P}_n^* - \mathbb{P}_n))_*$ denote the minimal measurable majorant and maximal measurable minorant of $f(\sqrt{n}(\mathbb{P}_n^* - \mathbb{P}_n))$, respectively. Note that for any continuous and bounded $f$, $f(\sqrt{n}(\mathbb{P}_n^* - \mathbb{P}_n))$ is continuous in $\{W_i\}_{i=1}^n$, and is hence measurable satisfying \cite{fang2019inference} assumption 3(iv). \cite{fang2019inference} lemma S.3.9 then implies assumption 3(iii) is satisfied as well.
\end{proof}

\theoremBootstrapWorksWithFullDifferentiability*
\begin{proof}
	\singlespacing
	
	By application of \cite{fang2019inference} theorem 3.1. There are three numbered assumptions:
	\begin{enumerate}
		\item \cite{fang2019inference} assumption 1 is satisfied; $\ell^\infty(\mathcal{F})$ and $\mathbb{R}^2$ are indeed Banach spaces, and lemma \ref{Lemma: Hadamard differentiability, T is Hadamard directionally differentiable} shows that under this paper's assumptions \ref{Assumption: setting}, \ref{Assumption: cost function}, and \ref{Assumption: parameter, function of moments}, the map $T$ is Hadamard directionally differentiable at $P$ tangentially to $\mathcal{C}(\mathcal{F}, L_{2,P})$. 
		\item \cite{fang2019inference} assumption 2 is satisfied; lemma \ref{Lemma: weak convergence, Donsker results, large Donsker set F is Donsker} shows that $\sqrt{n}(\mathbb{P}_n - P) \overset{L}{\rightarrow} \mathbb{G}$ in $\ell^\infty(\mathcal{F})$, where $\mathbb{G}$ is tight and supported in $\mathcal{C}(\mathcal{F}, L_{2,P})$. 
		\item Lemma \ref{Lemma: inference, bootstrap, exchangeable bootstrap satisfies Fang and Santos assumption 3} shows that \cite{fang2019inference} assumption 3 is satisfied. 
	\end{enumerate}
	
	Finally, note that $\mathbb{G}$ is Gaussian and mean zero; it follows that its support is a vector subspace of $\ell^\infty(\mathcal{F})$. Thus \cite{fang2019inference} theorem 3.1 implies $T$ is (fully) Hadamard differentiable tangentially to the support of $\mathbb{G}$ if and only if 
	\begin{equation*}
		\sup_{f \in \text{BL}_1} \left\lvert E\left[f\left(\sqrt{n}(T(\mathbb{P}_n^*) - T(\mathbb{P}_n))\right) \mid \{Y_i, D_i, Z_i, X_i\}_{i=1}^n\right] - E\left[f(T_P'(\mathbb{G}))\right]\right\rvert = o_p(1)
	\end{equation*}
	Since lemma \ref{Lemma: Hadamard differentiability, T is Hadamard directionally differentiable} shows that under assumptions \ref{Assumption: setting}, \ref{Assumption: cost function}, \ref{Assumption: parameter, function of moments}, and \ref{Assumption: full differentiability}, $T$ is fully Hadamard differentiable tangentially to the support of $\mathbb{G}$, this completes the proof.
\end{proof}

\subsection{Alternative procedure}

\begin{restatable}[]{lemma}{lemmaFangAndOptimalTransportEstimatorWorks}
	\label{Lemma: inference, bootstrap, Fang and Santos alternative works for OT}
	\singlespacing
	
	Let assumptions \ref{Assumption: setting}, \ref{Assumption: cost function}, and \ref{Assumption: parameter, function of moments} hold, and $\{\kappa_n\}_{n=1}^\infty \subseteq \mathbb{R}$ satisfy $\kappa_n \rightarrow \infty$ and $\kappa_n / \sqrt{n} \rightarrow 0$. For $c \in \{c_L, c_H\}$, let 
	\begin{align*}
		\Psi_c(P_{1 \mid x}, P_{0 \mid x}) &= \argmax_{(\varphi, \psi) \in \Phi_c \cap (\mathcal{F}_c \times \mathcal{F}_c^c)} P_{1 \mid x}(\varphi) + P_{0 \mid x}(\psi) \\
		\widehat{\Psi}_{c, x} &= \left\{(\varphi,\psi) \in \Phi_c \cap (\mathcal{F}_c \times \mathcal{F}_c^c) \; ; \; OT_c(\hat{P}_{1 \mid x}, \hat{P}_{0 \mid x}) \leq \hat{P}_{1 \mid x}(\varphi) + \hat{P}_{0 \mid x}(\psi) + \frac{\kappa_n}{\sqrt{n}} \right\}
	\end{align*}
	and $OT_{c, (P_{1 \mid x}, P_{0 \mid x})}', \widehat{OT}_{c,x}': \mathcal{C}(\mathcal{F}_{1,x}, L_{2,P}) \times \mathcal{C}(\mathcal{F}_{0,x}, L_{2,P}) \rightarrow \mathbb{R}$, be given by
	\begin{align*}
		OT_{c, (P_{1, \mid x}, P_{0 \mid x})}'(H_1, H_0) &= \sup_{(\varphi,\psi) \in \Psi_c(P_{1 \mid x}, P_{0 \mid x})} H_1(\varphi) + H_0(\psi) \\
		\widehat{OT}_{c, x}'(H_1, H_0) &= \sup_{(\varphi,\psi) \in \widehat{\Psi}_{c,x}} H_1(\varphi) + H_0(\psi)
	\end{align*}
	Then for any $(H_1,H_0) \in \mathcal{C}(\mathcal{F}_{1,x}, L_{2,P}) \times \mathcal{C}(\mathcal{F}_{0,x}, L_{2,P})$, 
	\begin{equation*}
		\left\lvert \widehat{OT}_{c, x}'(H_1, H_0) - OT_{c, (P_{1, \mid x}, P_{0 \mid x})}'(H_1, H_0) \right\rvert \overset{p}{\rightarrow} 0
	\end{equation*}
\end{restatable}
\begin{proof}
	\singlespacing
	
	The proof is similar that of \cite{fang2019inference} lemma S.4.8. As the subscript $x$ plays no role, we drop it from the notation.
	
	In steps:
	\begin{enumerate}
		\item We first esteablish an inequality used several times below. Note that for any $(\tilde{\varphi}, \tilde{\psi}), (\varphi,\psi) \in \Phi_c \cap (\mathcal{F}_c \times \mathcal{F}_c^c)$,
		\begin{align*}
			\lVert \hat{P}_1 - P_1 \rVert_{\mathcal{F}_1} + \lVert \hat{P}_0 - P_0 \rVert_{\mathcal{F}_0} &\geq \hat{P}_1(\varphi) - P_1(\varphi) + \hat{P}_0(\psi) - P_0(\psi) \\
			\lVert \hat{P}_1 - P_1 \rVert_{\mathcal{F}_1} + \lVert \hat{P}_0 - P_0 \rVert_{\mathcal{F}_0} &\geq P_1(\tilde{\varphi}) - \hat{P}_1(\tilde{\varphi}) + P_0(\tilde{\psi}) - \hat{P}_0(\tilde{\psi}) 
		\end{align*}
		Add these to obtain
		\begin{align}
			&2\left(\lVert \hat{P}_1 - P_1 \rVert_{\mathcal{F}_1} + \lVert \hat{P}_0 - P_0 \rVert_{\mathcal{F}_0}\right) \notag \\
			&\hspace{1 cm} \geq \hat{P}_1(\varphi) - P_1(\varphi) + \hat{P}_0(\psi) - P_0(\psi) + P_1(\tilde{\varphi}) - \hat{P}_1(\tilde{\varphi}) + P_0(\tilde{\psi}) - \hat{P}_0(\tilde{\psi}), \label{Display: lemma proof, inference, bootstrap, Fang and Santos alternative works for OT, inequality 1}
		\end{align}
		
		\item We next show 
		\begin{equation}
			\lim_{n \rightarrow \infty} P\left(\Psi(P_1, P_0) \subseteq \widehat{\Psi}_c \right) = 1 \label{Display: lemma proof, inference, bootstrap, Fang and Santos alternative works for OT, true argmax set is eventually subset of estimated argmax set}
		\end{equation}
		
		Let $(\tilde{\varphi}, \tilde{\psi}) \in \Psi(P_1, P_0)$, and rearrange \eqref{Display: lemma proof, inference, bootstrap, Fang and Santos alternative works for OT, inequality 1} to find 
		\begin{align*}
			&2\left(\lVert \hat{P}_1 - P_1 \rVert_{\mathcal{F}_1} + \lVert \hat{P}_0 - P_0 \rVert_{\mathcal{F}_0} \right) \\
			&\hspace{1 cm} \geq \hat{P}_1(\varphi) + \hat{P}_0(\psi) - \hat{P}_1(\tilde{\varphi}) - \hat{P}(\tilde{\psi}) + \underbrace{P_1(\tilde{\varphi}) + P_0(\tilde{\psi}) - P_1(\varphi) - P_0(\psi)}_{\geq 0} \\
			&\hspace{1 cm} \geq \hat{P}_1(\varphi) + \hat{P}_0(\psi) - \hat{P}_1(\tilde{\varphi}) - \hat{P}(\tilde{\psi})
		\end{align*}
		and therefore
		\begin{align*}
			\sup_{(\varphi, \psi) \in \Phi_c \cap (\mathcal{F}_c \times \mathcal{F}_c^c)} \hat{P}_1(\varphi) + \hat{P}_0(\psi) \leq \hat{P}_1(\tilde{\varphi}) + \hat{P}(\tilde{\psi}) + 2\left(\lVert \hat{P}_1 - P_1 \rVert_{\mathcal{F}_1} + \lVert \hat{P}_0 - P_0 \rVert_{\mathcal{F}_0} \right) 
		\end{align*}
		holds for any $(\tilde{\varphi}, \tilde{\psi}) \in \Psi_c(P_1, P_0)$. It follows that $2\left(\lVert \hat{P}_1 - P_1 \rVert_{\mathcal{F}_1} + \lVert \hat{P}_0 - P_0 \rVert_{\mathcal{F}_0} \right) < \frac{\kappa_n}{\sqrt{n}}$ implies $(\tilde{\varphi}, \tilde{\psi}) \in \widehat{\Psi}_c$, and hence
		\begin{equation*}
			P\left(2\frac{\sqrt{n}}{\kappa_n}\left(\lVert \hat{P}_1 - P_1 \rVert_{\mathcal{F}_1} + \lVert \hat{P}_0 - P_0 \rVert_{\mathcal{F}_0} \right) < 1\right) \leq P\left(\Psi(P_1, P_0) \subseteq \widehat{\Psi}_c\right)
		\end{equation*}
		Lemma \ref{Lemma: consistency} implies $\lVert \hat{P}_1 - P_1 \rVert_{\mathcal{F}_1} + \lVert \hat{P}_0 - P_0 \rVert_{\mathcal{F}_0} \overset{p}{\rightarrow} 0$. Since $\frac{\sqrt{n}}{\kappa_n} \rightarrow 0$, this implies that $2\frac{\sqrt{n}}{\kappa_n}\left(\lVert \hat{P}_1 - P_1 \rVert_{\mathcal{F}_1} + \lVert \hat{P}_0 - P_0 \rVert_{\mathcal{F}_0} \right) = o_p(1)$ and therefore
		\begin{align*}
			\lim_{n \rightarrow \infty} P\left(\Psi(P_1, P_0) \subseteq \widehat{\Psi}_c \right) \geq \lim_{n \rightarrow \infty} P\left(2\frac{\sqrt{n}}{\kappa_n}\left(\lVert \hat{P}_1 - P_1 \rVert_{\mathcal{F}_1} + \lVert \hat{P}_0 - P_0 \rVert_{\mathcal{F}_0} \right) < 1\right) = 1
		\end{align*}
		as was to be shown.
		
		\item We next show that for any $\delta > 0$,
		\begin{equation}
			\lim_{n \rightarrow \infty} P\left(\widehat{\Psi}_c \subseteq \left(\Psi(P_1, P_0)\right)^\delta\right) = 1 \label{Display: lemma proof, inference, bootstrap, Fang and Santos alternative works for OT, estimated argmax set is eventually subset of expanded true argmax set}
		\end{equation}
		where $\left(\Psi(P_1, P_0)\right)^\delta$ is an open $\delta$-enlargement of $\Psi(P_1, P_0)$ under $L_2$; i.e.
		\begin{align*}
			\left(\Psi(P_1, P_0)\right)^\delta = \left\{(f,g) \; ; \; \inf_{(\varphi,\psi) \in \Psi(P_1,P_0)} L_2((\varphi,\psi), (f,g)) < \delta\right\}
		\end{align*}
		Toward this end, note that
		\begin{align*}
			\eta \equiv \left[\sup_{(\varphi,\psi) \in \Phi_c \cap (\mathcal{F}_c \times \mathcal{F}_c^c)} \left\{P_1(\varphi) + P_0(\psi)\right\} - \sup_{(\varphi,\psi) \in \Phi_c \cap (\mathcal{F}_c \times \mathcal{F}_c^c) \setminus \left(\Psi(P_1, P_0)\right)^\delta}\left\{P_1(\varphi) + P_0(\psi)\right\}\right] > 0
		\end{align*}
		$\eta > 0$ follows from compactness of $\Phi_c \cap (\mathcal{F}_c \times \mathcal{F}_c^c)$ and continuity of $P_1 + P_0$ with respect to $L_2$ (see the proof of lemma \ref{Lemma: Hadamard differentiability, optimal transport}).
		
		Rearrange \eqref{Display: lemma proof, inference, bootstrap, Fang and Santos alternative works for OT, inequality 1} to find 
		\begin{align*}
			&P_1(\tilde{\varphi}) + P_0(\tilde{\psi}) - P_1(\varphi) - P_0(\psi) \\
			&\hspace{1 cm} \leq 2\left(\lVert \hat{P}_1 - P_1 \rVert_{\mathcal{F}_1} + \lVert \hat{P}_0 - P_0 \rVert_{\mathcal{F}_0}\right) + \hat{P}_1(\tilde{\varphi}) + \hat{P}_0(\tilde{\psi}) - \hat{P}_1(\varphi) - \hat{P}_0(\psi)
		\end{align*}
		Take suprema over $(\tilde{\varphi}, \tilde{\psi}) \in \Phi_c \cap (\mathcal{F}_c \times \mathcal{F}_c^c)$ to find
		\begin{align}
			&\sup_{(\tilde{\varphi}, \tilde{\psi}) \in \Phi_c \cap (\mathcal{F}_c \times \mathcal{F}_c^c)} P_1(\tilde{\varphi}) + P_0(\tilde{\psi}) - P_1(\varphi) - P_0(\psi) \notag \\
			&\hspace{1 cm} \leq 2\left(\lVert \hat{P}_1 - P_1 \rVert_{\mathcal{F}_1} + \lVert \hat{P}_0 - P_0 \rVert_{\mathcal{F}_0}\right) + \sup_{(\tilde{\varphi}, \tilde{\psi}) \in \Phi_c \cap (\mathcal{F}_c \times \mathcal{F}_c^c)}\hat{P}_1(\tilde{\varphi}) + \hat{P}_0(\tilde{\psi}) - \hat{P}_1(\varphi) - \hat{P}_0(\psi)  \label{Display: lemma proof, inference, bootstrap, Fang and Santos alternative works for OT, inequality 2} 
		\end{align}
		
		Suppose there exists $(\varphi,\psi) \in \Phi_c \cap (\mathcal{F}_c \times \mathcal{F}_c^c) \setminus \left(\Psi(P_1, P_0)\right)^\delta$ such that $ \sup_{(\tilde{\varphi}, \tilde{\psi}) \in \Phi_c \cap (\mathcal{F}_c \times \mathcal{F}_c^c)}\hat{P}_1(\tilde{\varphi}) + \hat{P}_0(\tilde{\psi}) \leq \hat{P}_1(\varphi) + \hat{P}_0(\psi) + \frac{\kappa}{\sqrt{n}}$. For any such $(\varphi, \psi)$, \eqref{Display: lemma proof, inference, bootstrap, Fang and Santos alternative works for OT, inequality 2} implies
		\begin{align*}
			&\sup_{(\tilde{\varphi}, \tilde{\psi}) \in \Phi_c \cap (\mathcal{F}_c \times \mathcal{F}_c^c)} P_1(\tilde{\varphi}) + P_0(\tilde{\psi}) - P_1(\varphi) - P_0(\psi) \leq 2\left(\lVert \hat{P}_1 - P_1 \rVert_{\mathcal{F}_1} + \lVert \hat{P}_0 - P_0 \rVert_{\mathcal{F}_0}\right) + \frac{\kappa_n}{\sqrt{n}}
		\end{align*}
		from which it follows that
		\begin{align*}
			&2\left(\lVert \hat{P}_1 - P_1 \rVert_{\mathcal{F}_1} + \lVert \hat{P}_0 - P_0 \rVert_{\mathcal{F}_0}\right) + \frac{\kappa_n}{\sqrt{n}} \\
			&\hspace{2 cm}\geq \sup_{(\tilde{\varphi}, \tilde{\psi}) \in \Phi_c \cap (\mathcal{F}_c \times \mathcal{F}_c^c)} P_1(\tilde{\varphi}) + P_0(\tilde{\psi}) - \sup_{(\varphi,\psi) \in \Phi_c \cap (\mathcal{F}_c \times \mathcal{F}_c^c) \setminus \left(\Psi(P_1, P_0)\right)^\delta} \left\{P_1(\varphi) + P_0(\psi)\right\} \\
			&\hspace{2 cm} = \eta
		\end{align*}
		
		To summarize: if there exists $(\varphi,\psi) \in \Phi_c \cap (\mathcal{F}_c \times \mathcal{F}_c^c) \setminus \left(\Psi(P_1, P_0)\right)^\delta$ such that $\sup_{(\tilde{\varphi}, \tilde{\psi}) \in \Phi_c \cap (\mathcal{F}_c \times \mathcal{F}_c^c)}\hat{P}_1(\tilde{\varphi}) + \hat{P}_0(\tilde{\psi}) \leq \hat{P}_1(\varphi) + \hat{P}_0(\psi) + \frac{\kappa}{\sqrt{n}}$, then $2\left(\lVert \hat{P}_1 - P_1 \rVert_{\mathcal{F}_1} + \lVert \hat{P}_0 - P_0 \rVert_{\mathcal{F}_0}\right) + \frac{\kappa_n}{\sqrt{n}} \geq \eta$, from which it follows that 
		\begin{align*}
			&P\left(\widehat{\Psi}_c \not \subseteq \left(\Psi(P_1, P_0)\right)^\delta\right) \\
			&\hspace{1 cm} = P\Bigg(\sup_{(\tilde{\varphi}, \tilde{\psi}) \in \Phi_c \cap (\mathcal{F}_c \times \mathcal{F}_c^c)}\hat{P}_1(\tilde{\varphi}) + \hat{P}_0(\tilde{\psi}) \leq \hat{P}_1(\varphi) + \hat{P}_0(\psi) + \frac{\kappa}{\sqrt{n}} \\
			&\hspace{5 cm} \text{ for some } (\varphi,\psi) \in \Phi_c \cap (\mathcal{F}_c \times \mathcal{F}_c^c) \setminus \left(\Psi(P_1, P_0)\right)^\delta\Bigg) \\
			&\hspace{1 cm} \leq P\left(2\left(\lVert \hat{P}_1 - P_1 \rVert_{\mathcal{F}_1} + \lVert \hat{P}_0 - P_0 \rVert_{\mathcal{F}_0}\right) + \frac{\kappa_n}{\sqrt{n}} \geq \eta\right) \rightarrow 0
		\end{align*}
		where the final limit claim follows from $\eta > 0$, $\kappa_n/\sqrt{n} \rightarrow 0$, and $\lVert \hat{P}_1 - P_1 \rVert_{\mathcal{F}_1} + \lVert \hat{P}_0 - P_0 \rVert_{\mathcal{F}_0} = o_p(1)$.
		
		\item \eqref{Display: lemma proof, inference, bootstrap, Fang and Santos alternative works for OT, true argmax set is eventually subset of estimated argmax set} and \eqref{Display: lemma proof, inference, bootstrap, Fang and Santos alternative works for OT, estimated argmax set is eventually subset of expanded true argmax set} imply that for any $\delta > 0$, $P\left(\Psi_c(P_1, P_0) \subseteq \widehat{\Psi}_c \subseteq \Psi_c(P_1, P_0)^{\delta}\right) \rightarrow 1$. It follows that there exists a sequence $\{\delta_n\}_{n=1}^\infty \subseteq \mathbb{R}_+$ with $\delta_n \downarrow 0$ such that $P\left(\Psi(P_1, P_0) \subseteq \widehat{\Psi}_c \subseteq \Psi(P_1, P_0)^{\delta_n}\right) \rightarrow 1$. 
		Notice that when $\Psi(P_1, P_0) \subseteq \widehat{\Psi}_c \subseteq \Psi(P_1, P_0)^{\delta_n}$ holds,
		\begin{align*}
			&\left\lvert \widehat{OT}_{c, x}'(H_1, H_0) - OT_{c, (P_1, P_0)}'(H_1, H_0) \right\rvert \\
			&\leq \sup_{(\varphi, \psi) \in \Psi_c(P_1, P_0)^{\delta_n} \cap \Phi_c \cap (\mathcal{F}_c \times \mathcal{F}_c^c)} \left\{H_1(\varphi) + H_0(\psi)\right\} - \sup_{(\varphi, \psi) \in \Psi_c(P_1, P_0)}\left\{H_1(\varphi) + H_0(\psi)\right\} \\
			&\leq \sup_{(\varphi_1, \psi_1), (\varphi_2, \psi_2) \in \Phi_c \cap (\mathcal{F}_c \times \mathcal{F}_c^c); \; L_2((\varphi_1, \psi_1), (\varphi_2, \psi_2)) < \delta_n} \left\{H_1(\varphi_1) + H_0(\psi_1) - H_1(\varphi_2) - H_0(\psi_0)\right\} \\
			&= o_p(1)
		\end{align*}
		where the $o_p(1)$ claim follows from $H_1 + H_0$ being continuous and $\Phi_c \cap (\mathcal{F}_c \times \mathcal{F}_c^c)$ being compact, implying $H_1 + H_0$ is in fact uniformly continuous. 
	\end{enumerate}
	This concludes the proof.
\end{proof}

\theoremFangAndSantosAlternativeWorks*
\begin{proof}
	\singlespacing
	
	The overall strategy is to apply \cite{fang2019inference} theorem 3.2, viewing $T_1(\mathbb{P}_n)$ as the estimator for $T_1(P)$, $T_1(\mathbb{P}_n^*)$ as the bootstrap, and $T_{-1} = T_4 \circ T_3 \circ T_2$ as the directionally differentiable function. There are four assumption to verify.
	
	\begin{enumerate}
		\item To see that \cite{fang2019inference} assumption 1 holds, 
		\begin{enumerate}[label=(\roman*)]
			\item the map
			\begin{align*}
				T_4 \circ T_3 \circ T_2 : \prod_{m=1}^M \ell^\infty(\mathcal{F}_{1,x}) \times \ell^\infty(\mathcal{F}_{0,x})\times \mathbb{R}^{K_1} \times \mathbb{R}^{K_0} \times \mathbb{R} \rightarrow \mathbb{R}^2
			\end{align*}
			is a map between Banach spaces
			\item by lemmas \ref{Lemma: Hadamard differentiability, T2 is directionally or fully Hadamard differentiable}, \ref{Lemma: Hadamard differentiability, T3 is fully Hadamard differentiable}, \ref{Lemma: Hadamard differentiability, T4 is fully Hadamard differentiable} and the chain rule (lemma \ref{Lemma: Hadamard differentiability, chain rule}), $T_{-1} = T_4 \circ T_3 \circ T_2$ is Hadamard directionally differentiable at $T_1(P)$ tangentially to 
			\begin{equation*}
				\mathbb{D}_{Tan} = \prod_{m=1}^M \mathcal{C}(\mathcal{F}_{1, x_m}, L_{2,P}) \times \mathcal{C}(\mathcal{F}_{0, x_m}, L_{2,P}) \times \mathbb{R}^{K_1} \times \mathbb{R}^{K_0}  \times \mathbb{R}
			\end{equation*}
		\end{enumerate}
		
		\item To see that the estimator $T_1(\mathbb{P}_n)$ satisfies \cite{fang2019inference} assumption 2, note that 
		\begin{enumerate}[label=(\roman*)]
			\item $T_1(P) \in \prod_{m=1}^M \ell^\infty(\mathcal{F}_{1,x}) \times \ell^\infty(\mathcal{F}_{0,x})\times \mathbb{R}^{K_1} \times \mathbb{R}^{K_0} \times \mathbb{R}$ and lemma \ref{Lemma: weak convergence, support of T_1P(G)} shows
			\begin{equation*}
				T_1(\mathbb{P}_n) : \{Y_i, D_i, Z_i, X_i\}_{i=1}^n \rightarrow \prod_{m=1}^M \ell^\infty(\mathcal{F}_{1,x}) \times \ell^\infty(\mathcal{F}_{0,x})\times \mathbb{R}^{K_1} \times \mathbb{R}^{K_0} \times \mathbb{R}
			\end{equation*}
			satisfies $\sqrt{n}(T_1(\mathbb{P}_n) - T_1(P)) \overset{L}{\rightarrow} T_{1,P}'(\mathbb{G})$.
			\item $T_{1,P}'(\mathbb{G})$ is tight because $\mathbb{G}$ is tight and $T_{1,P}'$ is continuous. Lemma \ref{Lemma: weak convergence, support of T_1P(G)} also shows the support of $T_{1,P}'(\mathbb{G})$ is included in $\mathbb{D}_{Tan}$.
		\end{enumerate}
		
		\item The bootstrap $T_1(\mathbb{P}_n^*)$ satisfies \cite{fang2019inference} assumption 3:
		\begin{enumerate}[label=(\roman*)]
			\item $T_1(\mathbb{P}_n^*)$ is a function of $\{Y_i, D_i, Z_i, X_i, W_i\}_{i=1}^n$ with $\{W_i\}_{i=1}^n$ independent of $\{Y_i, D_i, Z_i, X_i\}_{i=1}^n$.
			\item $T_1$ is fully Hadamard differentiable at $P$ tangentially to $\ell^\infty(\mathcal{F})$, and hence the functional delta method implies $\sqrt{n}(T_1(\mathbb{P}_n) - T_1(P)) \overset{L}{\rightarrow} T_{1,P}'(\mathbb{G})$. Lemma \ref{Lemma: inference, bootstrap, exchangeable bootstrap satisfies Fang and Santos assumption 3} shows that $\mathbb{P}_n^*$ satisfies \cite{fang2019inference} assumption 3, and thus \cite{fang2019inference} theorem 3.1 implies
			\begin{equation*}
				\sup_{f \in \text{BL}_1} \left\lvert E\left[ f(\sqrt{n}(T_1(\mathbb{P}_n^*) - T_1(\mathbb{P}_n)))\mid \{Y_i, D_i, Z_i, X_i\}_{i=1}^n \right] - E[f(T_{1,P}'(\mathbb{G}))] \right\rvert = o_p(1)
			\end{equation*}
			
			\item Condition (iv) below holds, and hence \cite{fang2019inference} lemma S.3.9 implies $\sqrt{n}(T_1(\mathbb{P}_n^*) - T_1(\mathbb{P}_n))$ is asymptotically measurable.
			\item Note that for any continuous and bounded function $f$, $f(\sqrt{n}(T_1(\mathbb{P}_n^*) - T_1(\mathbb{P}_n)))$ is continuous in $\{W_i\}_{i=1}^n$ and hence is a measurable function of $\{W_i\}_{i=1}^n$. 
		\end{enumerate}
		
		\item \cite{fang2019inference} assumption 4 is about the estimator of the derivative.
		
		Notice that $T_{-1, T_1(P)}' = T_{4,T_3(T_2(T_1(P)))}' \circ T_{3, T_2(T_1(P))}' \circ T_{2,T_1(P)}'$ is given by
		\begin{align*}
			&T_{-1, T_1(P)}': \mathbb{D}_{Tan} \rightarrow \mathbb{R}^2, &&T_{-1, T_1(P)}'(h) = D_4 D_3 T_{2, T_1(P)}'(h)
		\end{align*}
		Estimate this derivative with
		\begin{align*}
			&\widehat{T}_{-1,T_1(P)}' : \mathbb{D}_{Tan} \rightarrow \mathbb{R}^2, &&\hat{D}_4 \hat{D}_3 \widehat{T}_{2,T_1(P)}'(h)
		\end{align*}
		The estimator $\widehat{T}_{-1,T_1(P)}'$ satisfies the conditions of \cite{fang2019inference} lemma S.3.6, and therefore \cite{fang2019inference} assumption 4. These conditions are
		\begin{enumerate}
			\item Modulus of continuity: $\lVert \widehat{T}_{-1, T_1(P)}'(h_1) - \widehat{T}_{-1,T_1(P)}'(h_2) \rVert \leq C_n \lVert h_1 - h_2 \rVert$ for some $C_n = O_p(1)$.
			
			\item Pointwise consistency: for any $h$, $\lVert \widehat{T}_{-1, T_1(P)}(h) - T_{-1,T_1(P)}(h) \rVert = o_p(1)$.
		\end{enumerate}
		
		To see these claims in detail:
		\begin{enumerate}
			\item For any matrix $A$, let $\lVert A \rVert_o = \sup_{x ; \lVert x \rVert_2 = 1} \lVert A x \rVert_2$ be the operator norm. 
			\begin{align*}
				\lVert \widehat{T}_{-1, T_1(P)}'(h_1) - \widehat{T}_{-1,T_1(P)}'(h_2) \rVert &= \lVert \hat{D}_4 \hat{D}_3 \widehat{T}_{2, T_1(P)}'(h_1) - \hat{D}_4 \hat{D}_3 \widehat{T}_{2, T_1(P)}'(h_2)\rVert \\
				&\leq \lVert \hat{D}_4 \hat{D}_3 \rVert_o \lVert \widehat{T}_{2,T_1(P)}'(h_1) - \widehat{T}_{2,T_1(P)}'(h_2) \rVert \\
				&\leq \lVert \hat{D}_4 \hat{D}_3 \rVert \lVert \lVert h_1 - h_2 \rVert
			\end{align*}
			where the last claim follows because $\widehat{T}_{2,T_1(P)}'$ is $1$-Lipschitz (shown below). Next notice $\hat{D}_4 \overset{p}{\rightarrow} D_4$ and $\hat{D}_3 \overset{p}{\rightarrow} D_3$ by the CMT, which implies $\lVert \hat{D}_4 \hat{D}_3\rVert = O_p(1)$ as required. \\
			
			To see that $\widehat{T}_{2, T_1(P)}'$ is $1$-Lipschitz, recall 
			\begin{align*}
				&\widehat{T}_{2,T_1(P)}'\left(\{H_{1,x}, H_{0,x}, h_{\eta_1, x}, h_{\eta_0, x}, h_{s, x}\}_{x \in \mathcal{X}}\right) \\
				&\hspace{1 cm} = \left(\left\{\widehat{OT}_{c_L, x}'(H_{1,x}, H_{0,x}), -\widehat{OT}_{c_H, x}'(H_{1,x}, H_{0,x}), h_{\eta_1, x}, h_{\eta_0, x}, h_{s, x}\right\}_{x \in \mathcal{X}}\right)
			\end{align*}
			The maps $\widehat{OT}_{c_L, x}, -\widehat{OT}_{c_H, x}$ are $1$-Lipschitz. Specifically, note that 
			\begin{align*}
				&\lvert \widehat{OT}_{c_L, x}'(H_{1,x}, H_{0, x}) - \widehat{OT}_{c_L, x}'(G_{1,x}, G_{0, x}) \rvert \\
				&\hspace{1 cm} = \left\lvert \sup_{(\varphi,\psi) \in \widehat{\Psi}_{c,x}}\left\{ H_{1,x}(\varphi) + H_{0,x}(\psi)\right\} - \sup_{(\varphi,\psi) \in \widehat{\Psi}_{c,x}}\left\{G_{1,x}(\varphi) + G_{0,x}(\psi)\right\}\right\rvert \\
				&\hspace{1 cm} \leq \sup_{\varphi \in \mathcal{F}_{1,x}} \lvert H_{1,x}(\varphi) - G_{1,x}(\varphi)\rvert + \sup_{\psi \in \mathcal{F}_{0,x}} \lvert H_{0,x}(\psi) - G_{0,x}(\psi) \rvert \\
				&\hspace{1 cm} = \lVert H_{1,x} - G_{1,x} \rVert_{\mathcal{F}_{1,x}} + \lVert H_{0,x} - G_{0,x} \rVert_{\mathcal{F}_{0,x}}
			\end{align*}
			and similarly, $-\widehat{OT}_{c_H, x}$ is $1$-Lipschitz. The other maps in $\widehat{T}_{2, T_1(P)}$ are the identity map, which is also $1$-Lipschitz. It follows that $\widehat{T}_{2, T_1(P)}$ is $1$-Lipschitz.\footnote{For $k = 1,2$, let $\mathbb{D}_k$, $\mathbb{E}_k$ be metric spaces. If $f_k : \mathbb{D}_k \rightarrow \mathbb{E}_k$ be Lipschitz with constants $L_k$, then $f : \mathbb{D}_1 \times \mathbb{D}_2 \rightarrow \mathbb{E}_1 \times \mathbb{E}_2$ given by $f(x_1,x_2) = (f_1(x_1), f_2(x_2))$ is Lipschitz with constant $\max\{L_1, L_2\}$. To see this, recall $\mathbb{D}_1 \times \mathbb{D}_2$ and $\mathbb{E}_1 \times \mathbb{E}_2$ are metricized with the norms $\lVert (x_1,x_2) \rVert_{\mathbb{D}_1 \times \mathbb{D}_2} = \lVert x_1 \rVert_{\mathbb{D}_1} + \lVert x_2 \rVert_{\mathbb{D}_2}$ and $\lVert (y_1,y_2)\rVert_{\mathbb{E}_1 \times \mathbb{E}_2} = \lVert y_1 \rVert_{\mathbb{E}_1} + \lVert y_2 \rVert_{\mathbb{E}_2}$, and note that 
				\begin{align*}
					\lVert f(x_1, x_2) - f(x_1', x_2') \rVert_{\mathbb{E}_1 \times \mathbb{E}_2} &= \lVert (f_1(x_1), f_2(x_2)) - (f_1(x_1'), f_2(x_2'))\rVert_{\mathbb{E}_1 \times \mathbb{E}_2} = \lVert f_1(x_1) - f_1(x_1') \rVert_{\mathbb{E}_1} + \lVert f_2(x_2) - f_2(x_2') \rVert_{\mathbb{E}_2} \\
					&\leq L_1\lVert x - x_1' \rVert_{\mathbb{D}_1} + L_2 \lVert x_2 - x_2'\rVert_{\mathbb{D}_2} \leq \max\{L_1, L_2\}\lVert x - x_1' \rVert_{\mathbb{D}_1} + \max\{L_1, L_2\} \lVert x_2 - x_2'\rVert_{\mathbb{D}_2} \\
					&= \max\{L_1, L_2\} \times \lVert (x_1, x_2) - (x_1', x_2') \rVert_{\mathbb{D}_1 \times \mathbb{D}_2}
			\end{align*}}

			\item To show pointwise consistency, fix $h = \left(\{H_{1,x}, H_{0,x}, h_{\eta_1, x}, h_{\eta_0, x}, h_{s, x}\}_{x \in \mathcal{X}}\right)$ and note that 
			\begin{align*}
				&\lVert \widehat{T}_{-1, T_1(P)}'(h) - T_{-1, T_1(P)} \rVert = \lVert \hat{D}_4 \hat{D}_3 \widehat{T}_{2,T_1(P)}(h) - D_4 D_3 T_{2,T_1(P)}(h) \rVert \\
				&\hspace{1 cm} \leq \lVert (\hat{D}_4 \hat{D}_3 - D_4 D_3) T_{2,T_1(P)}'(h) \rVert + \lVert D_4 D_3(\widehat{T}_{2,T_1(P)}(h) - T_{2,T_1(P)}(h)) \rVert \\
				&\hspace{1 cm} \leq \lVert \hat{D}_4 \hat{D}_3 - D_4 D_3 \rVert_o \lVert T_{2,T_1(P)}'(h) \rVert + \lVert D_4 D_3 \rVert_o \lVert \widehat{T}_{2,T_1(P)}(h) - T_{2,T_1(P)}(h) \rVert
			\end{align*}
			Since $\hat{D}_4 \hat{D}_3 \overset{p}{\rightarrow} D_4 D_3$ by the CMT, it suffices to show 
			\begin{equation*}
				\lVert \widehat{T}_{2,T_1(P)}(h) - T_{2,T_1(P)}(h) \rVert = o_p(1)
			\end{equation*}
			The only nonzero coordinates correspond to $\widehat{OT}_{c_L, x}^{L\prime}(H_{1,x}, H_{0,x})$ and $-\widehat{OT}_{c_H, x}^{H\prime}(H_{1,x}, H_{0,x})$:
			\begin{align*}
				&\lVert \widehat{T}_{2,T_1(P)}(h) - T_{2,T_1(P)}(h) \rVert^2 \\
				&\hspace{1 cm} = \left(\widehat{OT}_{c_L, x}'(H_{1,x}, H_{0,x}) - OT_{c_L, (P_{1 \mid x}, P_{0 \mid x})}'(H_{1, x}, H_{0, x})\right)^2 \\
				&\hspace{4 cm} + \left(\widehat{OT}_{c_H, x}'(H_{1,x}, H_{0,x}) - OT_{c_H, (P_{1 \mid x}, P_{0 \mid x})}'(H_{1, x}, H_{0, x})\right)^2 \\
				&\hspace{1 cm} = o_p(1) + o_p(1)
			\end{align*}
			where the last $o_p(1)$ claim follows from lemma \ref{Lemma: inference, bootstrap, Fang and Santos alternative works for OT}.

		\end{enumerate}
		We conclude through \cite{fang2019inference} lemma S.3.6 that \cite{fang2019inference} assumption 4 is satisfied.
	\end{enumerate}
	
	Finally, apply \cite{fang2019inference} theorem 3.2 to find that 
	\begin{align*}
		\sup_{f \in \text{BL}_1} \left\lvert E\left[f(\hat{D}_4 \hat{D}_3 \widehat{T}_{2,T_1(P)}(\sqrt{n}(T_1(\mathbb{P}_n^*) - T_1(\mathbb{P}_n)))) \right] - E\left[f(T_P'(\mathbb{G}))\right]\right\rvert = o_p(1)
	\end{align*}
	as was to be shown.
\end{proof}

%% file: appendix/OTJointPO_appendix_duality.tex
\section{Appendix: duality in optimal transport}
\label{Appendix: duality in optimal transport}

This appendix contains terminology, notation, and results regarding optimal transport used in this paper. Most of these results can be found in the monographs \cite{villani2003topics}, \cite{villani2009optimal}, or \cite{santambrogio2015optimal}. 

\subsection{Primal and dual problems}
\label{Appendix: duality in optimal transport, primal and dual problems}

Let $\mathcal{Y}_1, \mathcal{Y}_0 $ be Polish subsets of $\mathbb{R}$, equipped with their Borel sigma algebras. Let $\mathcal{P}(\mathcal{Y}_d)$ be the set of probability distributions defined on $\mathcal{Y}_d$, and $P_d \in \mathcal{P}(\mathcal{Y}_d)$. Let $\mathcal{P}(\mathcal{Y}_1 \times \mathcal{Y}_0)$ be the set of probability distributions on the product space $\mathcal{Y}_1 \times \mathcal{Y}_0$.

A probability measure $\pi \in \mathcal{P}(\mathcal{Y}_1 \times \mathcal{Y}_0)$ has marginals $P_1$ and $P_0$ if 
\begin{align}
	&\text{For all } A \subset \mathcal{Y}_1 \text{ measurable, } \pi(A \times \mathcal{Y}_0) = P_1(A) = \int \mathbbm{1}_A(y_1) dP_1(y_1) \label{Defn: pi has marginal P_1} \\
	&\text{For all } B \subset  \mathcal{Y}_0 \text{ measurable, } \pi(\mathcal{Y}_1 \times B) = P_0(B) = \int \mathbbm{1}_B(y_0) dP_0(y_0) \label{Defn: pi has marginal P_0}
\end{align} 
The collection of such joint distributions with marginals $P_1$ and $P_0$ is denoted
\begin{equation}
	\Pi(P_1, P_0) = \left\{\pi \in \mathcal{P}(\mathcal{Y}_1 \times \mathcal{Y}_0) \; ; \; \pi \text{ satisfies } \eqref{Defn: pi has marginal P_1} \text{ and } \eqref{Defn: pi has marginal P_0}\right\} \label{Defn: Pi, distributions with given marginals}
\end{equation}
The \textbf{cost function} is a measurable function $c : \mathcal{Y}_1 \times \mathcal{Y}_0 \rightarrow \mathbb{R}$. The functional $I : \mathcal{P}(\mathcal{Y}_1 \times \mathcal{Y}_0) \rightarrow \mathbb{R} \cup \{+\infty\}$ is defined as
\begin{equation}
	I_c[\pi] = \int c(y_1,y_0) d\pi(y_1,y_0) \label{Defn: primal objective}
\end{equation}
The \textbf{optimal cost} $OT_c(P_1, P_0)$ is the infimum of $I_c[\pi]$ over $\Pi(P_1, P_0)$:
\begin{equation}
	OT_c(P_1, P_0) = \inf_{\pi \in \Pi(P_1, P_0)} I_c[\pi] = \inf_{\pi \in \Pi(P_1, P_0)} \int c(y_1,y_0) d\pi(y_1,y_0) \label{Defn: optimal transport primal problem, appendix}
\end{equation}
This minimization problem in \eqref{Defn: optimal transport primal problem, appendix} is known as \textbf{optimal transport}. When attained, a solution to \eqref{Defn: optimal transport primal problem, appendix} is called an \textbf{optimal transference plan} or \textbf{optimal coupling}. Attainment is common; \cite{villani2009optimal} theorem 4.1 implies:
\begin{restatable}[Optimal transport is attained]{lemma}{lemmaOptimalTransportAttainment}
	\label{Lemma: optimal transport is attained}
	\singlespacing
	
	Let $c : \mathcal{Y}_1 \times \mathcal{Y}_0 \rightarrow \mathbb{R}$ be lower semicontinuous and bounded from below. Then there exists $\pi^* \in \Pi(P_1, P_0)$ such that 
	\begin{equation*}
		E_{\pi^*}[c(Y_1, Y_0)] = \inf_{\pi \in \Pi(P_1, P_0)} \int c(y_1,y_0) d\pi(y_1,y_0)
	\end{equation*}
\end{restatable}

The dual problem will require some additional notation. For any probability measure $P$ let $L^1(P)$ denote the $P$-integrable functions. Define
\begin{equation}
	\Phi_c = \left\{(\varphi, \psi) \in L^1(P_1) \times L^1(P_0) \; ; \; \varphi(y_1) + \psi(y_0) \leq c(y_1,y_0)\right\}, \label{Defn: Phi_c, appendix} 
\end{equation}
and $J : L^1(P_1) \times L^1(P_0) \rightarrow \mathbb{R}$ by
\begin{equation}
	J(\varphi, \psi) = \int_{\mathcal{Y}_1} \varphi(y_1) dP_1(y_1) + \int_{\mathcal{Y}_0} \psi(y_0) dP_0(y_0) \label{Defn: dual problem objective}
\end{equation}
The \textbf{dual problem} of optimal transport is 
\begin{equation}
	\sup_{(\varphi, \psi) \in \Phi_c} J(\varphi, \psi) = \sup_{(\varphi,\psi) \in \Phi_c} \int \varphi(y_1) dP_1(y_1) + \int \psi(y_0) dP_0(y_0) \label{Defn: dual problem}
\end{equation}

\subsection{Duality}
\label{Appendix: duality in optimal transport, duality}

For any topological space $\mathcal{Z}$, let $\mathcal{C}_b(\mathcal{Z})$ denotes the set of functions $f : \mathcal{Z}\rightarrow \mathbb{R}$ that are continuous and bounded, and 
\begin{equation}
	\Phi_c \cap \mathcal{C}_b = \left\{(\varphi, \psi) \in \mathcal{C}_b(\mathcal{Y}_1) \times \mathcal{C}_b(\mathcal{Y}_0) \; ; \; \varphi(y_1) + \psi(y_0) \leq c(y_1, y_0)\right\} \label{Defn: Phi_c cap C_b}
\end{equation}

The following weak duality statement is \cite{villani2003topics} proposition 1.5. 

\begin{restatable}[Weak duality]{lemma}{lemmaWeakDuality} 
	\label{Lemma: weak duality}
	\singlespacing
	
	\begin{equation*}
		\sup_{(\varphi, \psi) \in \Phi_c \cap \mathcal{C}_b} J(\varphi, \psi) \leq \sup_{(\varphi, \psi) \in \Phi_c} J(\varphi, \psi) \leq \inf_{\pi \in \Pi(P_1, P_0)} I_c[\pi]
	\end{equation*}
\end{restatable}
The following strong duality statement can be directly inferred from \cite{villani2009optimal} theorem 5.10, or \cite{santambrogio2015optimal} theorem 1.42, and so is presented without proof.

\begin{restatable}[Strong duality]{theorem}{theoremStrongDuality}
	\label{Theorem: strong duality}
	\singlespacing
	
	Let $c : \mathcal{Y}_1 \times \mathcal{Y}_0 \rightarrow \mathbb{R}$ be lower semi-continuous and bounded from below. Then 
	\begin{equation}
		\inf_{\pi \in \Pi(P_1, P_0)} I_c[\pi] = \sup_{\varphi, \psi \in \Phi_c} J(\varphi, \psi) = \sup_{(\varphi, \psi)\in \Phi_c \cap \mathcal{C}_b} J(\varphi, \psi) \label{Display: theorem, strong duality result}
	\end{equation}
	Moreover, the infimum of the left-hand side of \eqref{Display: theorem, strong duality result} is attained. 
\end{restatable}

\subsection{$c$-concave functions}
\label{Appendix: duality in optimal transport, c-concave functions}

For any function $\varphi : \mathcal{Y}_1 \rightarrow \mathbb{R}$ and cost function $c(y_1,y_0)$, define the \textbf{c-transform} of $\varphi$ as the function $\varphi^c : \mathcal{Y}_0 \rightarrow \mathbb{R}$ given by
\begin{equation*}
	\varphi^c(y_0) = \inf_{y_1 \in \mathcal{Y}_1} \{c(y_1, y_0) - \varphi(y_1)\}.
\end{equation*}
Similarly, $\psi^c(y_1) = \inf_{y_0 \in \mathcal{Y}_0} \{c(y_1,y_0) - \psi(y_0)\}$ is the $c$-transform of $\psi$. $\varphi$ is called \textbf{$c$-concave} if $\varphi^{cc} = (\varphi^c)^c = \varphi$. If $\varphi$ is $c$-concave, then $(\varphi, \varphi^c)$ is called a \textbf{$c$-concave conjugate pair}. 

The following lemma \ref{Lemma: c-concave functions, generating c-concave functions} is exercise 2.35 found in \cite{villani2003topics} and presented without proof.

\begin{restatable}[\cite{villani2003topics} exercise 2.35]{lemma}{lemmaGeneratingCConcaveFunctions}
	\label{Lemma: c-concave functions, generating c-concave functions}
	\singlespacing
	
	Let $\mathcal{Y}_1$ and $\mathcal{Y}_0$ be nonempty sets and $c : \mathcal{Y}_1 \times \mathcal{Y}_0 \rightarrow \mathbb{R}$ be an arbitrary function. Let $\varphi : \mathcal{Y}_1 \rightarrow \mathbb{R}$. Then
	\begin{enumerate}[label=(\roman*)]
		\item $\varphi(y_1) + \varphi^c(y_0) \leq c(y_1,y_0)$ for all $(y_1, y_0) \in \mathcal{Y}_1 \times \mathcal{Y}_0$ 
		\item $\varphi^{cc}(y_1) \geq \varphi(y_1)$ for all $y_1 \in \mathcal{Y}_1$, and
		\item $\varphi^{ccc}(y_0) = \varphi^c(y_0)$ for all $y_0 \in \mathcal{Y}_0$
	\end{enumerate}
	It follows that $\varphi^{cc} = \varphi$ if and only if $\varphi$ is $c$-concave.
\end{restatable}

For $H \subseteq \left\{(f,g) \; ; \; f : \mathcal{Y}_1 \rightarrow \mathbb{R}, \text{ and } g : \mathcal{Y}_0 \rightarrow \mathbb{R}\right\}$, let 

\begin{align}
	\mathcal{F}_c^c(H) &= \left\{\varphi^c : \mathcal{Y}_0 \rightarrow \mathbb{R} \; ; \; \exists (f,g) \in H \text{ s.t. } \varphi^c(y_0) = \inf_{y_1 \in \mathcal{Y}_1}\{c(y_1,y_0) - f(y_1)\}\right\} \label{Defn: F_c^c for arbitrary set of functions} \\
	\mathcal{F}_c(H) &= \left\{\varphi : \mathcal{Y}_1 \rightarrow \mathbb{R} \; ; \; \exists \varphi^c \in F_c^c(H) \text{ s.t. } \varphi(y_1) = \inf_{y_0 \in \mathcal{Y}_0}\{c(y_1,y_0) - \varphi^c(y_0)\}\right\} \notag
\end{align}
$\mathcal{F}_c(H)$ is called the \textbf{$c$-concave functions generated by $H$}, and $\mathcal{F}_c^c(H)$ the \textbf{$c$-conjugates generated by $H$.}\footnote{$H$ is a typically a subset of $L^1(P_1) \times L^1(P_0)$. As defined the sets $\mathcal{F}_c(H)$ and $\mathcal{F}_c^c(H)$ only depend on the functions in $H$ that map $\mathcal{Y}_0$ to $\mathbb{R}$. This notational choice is more natural with the reasoning of lemma \ref{Lemma: c-concave functions, dual restricted to measurable and integrable c-concave functions satisfies strong duality} below.} Notice that not every $(\varphi, \psi) \in \mathcal{F}_c(H) \times \mathcal{F}_c^c(H)$ is a $c$-concave conjugate pair.

\begin{restatable}[Restricting the dual to $c$-concave functions]{lemma}{lemmaRestrictingTheDualToCConcaveFunctions}
	\label{Lemma: c-concave functions, dual restricted to measurable and integrable c-concave functions satisfies strong duality}
	\singlespacing
	
	Let $\Phi_{cs} \subseteq \Phi_c$ be such that
	\begin{enumerate}
		\item strong duality holds: $\inf_{\pi \in \Pi(P_1, P_0)} I_c[\pi] = \sup_{(\varphi, \psi) \in \Phi_{cs}} J(\varphi, \psi)$, and
		\item the $c$-concave functions generated by $\Phi_{cs}$ are integrable: $\mathcal{F}_c(\Phi_{cs}) \times \mathcal{F}_c^c(\Phi_{cs}) \subset L^1(P_1) \times L^1(P_0)$ 
	\end{enumerate}
	then 
	\begin{equation*}
		\inf_{\pi \in \Pi(P_1, P_0)} I_c[\pi] = \sup_{\varphi \in \mathcal{F}_c(\Phi_{cs})} J(\varphi, \varphi^c) = \sup_{(\varphi,\psi) \in \Phi_c \cap \big(\mathcal{F}_c(\Phi_{cs}) \times \mathcal{F}_c^c(\Phi_{cs})\big)} J(\varphi,\psi).
	\end{equation*}
\end{restatable}
\begin{proof}
	\singlespacing
	
	Let $(\varphi, \psi) \in \Phi_{cs}$. $\psi(y_0) \leq c(y_1,y_0) - \varphi(y_1)$ implies $\psi(y_0) \leq \varphi^c(y_0)$, and lemma \ref{Lemma: c-concave functions, generating c-concave functions} shows both that $\varphi(y_1) \leq \varphi^{cc}(y_1)$ and the pair $(\varphi^{cc}, \varphi^c)$ is a $c$-concave conjugate pair; thus $(\varphi^{cc}, \varphi^c) \in \Phi_c \cap \big(\mathcal{F}_c(\Phi_{cs}) \times \mathcal{F}_c^c(\Phi_{cs})\big)$.
	
	Since $\varphi^{cc}$ and $\varphi^c$ are integrable by assumption, $J(\varphi, \psi) \leq J(\varphi^{cc}, \varphi^c)$ and hence 
	\begin{equation*}
		\inf_{\pi \in \Pi(P_1, P_0)} I_c[\pi] = \sup_{(\varphi, \psi) \in \Phi_{cs}} J(\varphi, \psi) \leq \sup_{\varphi^{cc} \in \mathcal{F}_c(\Phi_{cs})} J(\varphi^{cc}, \varphi^c) \leq \sup_{(\varphi,\psi) \in \Phi_c \cap (\mathcal{F}_c(\Phi_{cs}) \times \mathcal{F}_c^c(\Phi_{cs}))} J(\varphi,\psi)
	\end{equation*}
	Finally, since $\Phi_c \cap (\mathcal{F}_c(\Phi_{cs}) \times \mathcal{F}_c^c(\Phi_{cs})) \subset \Phi_c$, it follows that
	\begin{equation*}
		\sup_{\varphi \in \mathcal{F}_c(\Phi_{cs})} J(\varphi, \varphi^c) \leq \sup_{(\varphi, \psi) \in \Phi_c} J(\varphi, \psi) = \inf_{\pi \in \Pi(P_1, P_0)} I_c[\pi]
	\end{equation*}
	with the final equality following from strong duality. 
\end{proof}

\begin{restatable}[Continuous cost function implies measurability of $c$-concave functions]{lemma}{lemmaContinuousCostFunctionMeasurableCConcaveFunctions}
	\label{Lemma: c-concave functions, if c is c is continuous then c-transforms are measurable} 
	\singlespacing
	
	If $c : \mathcal{Y}_1 \times \mathcal{Y}_0 \rightarrow \mathbb{R}$ is continuous, then for any $\psi : \mathcal{Y}_0 \rightarrow \mathbb{R}$, $\varphi(y_1) = \inf_{y_0 \in \mathcal{Y}_0} \{c(y_1,y_0) - \psi(y_0)\}$ and $\varphi^c(y_0) = \inf_{y_1 \in \mathcal{Y}_1} \{c(y_1,y_0) - \varphi(y_1)\}$ are upper semicontinuous and hence measurable.
\end{restatable}
\begin{proof}
	\singlespacing
	
	The pointwise infimum of a family of upper semicontinuous functions is upper semicontinuous (\cite{aliprantis2006infinite} Lemma 2.41). Since $c(y_1,y_0)$ is continuous, for any fixed $y_0 \in \mathcal{Y}_0$ the function $y_1 \mapsto c(y_1,y_0) - \psi(y_0)$ is continuous and hence
	\begin{equation*}
		\varphi(y_1) = \inf_{y_0 \in \mathcal{Y}_0} \{c(y_1,y_0) - \psi(y_0)\}
	\end{equation*}
	is upper semicontinuous. Similarly, $\varphi^c(y_0) = \inf_{y_1 \in \mathcal{Y}_1}\{c(y_1,y_0) - \varphi(y_1)\}$ is upper semicontinuous. Being upper semicontinuous, $\varphi$ and $\varphi^c$ are measurable.
\end{proof}
\begin{remark}
	\label{Remark: continuous cost function implies measurable c-concave functions} 
	\singlespacing
	
	Compare lemma \ref{Lemma: c-concave functions, if c is c is continuous then c-transforms are measurable} with \cite{villani2009optimal} Remark 5.5 discussing measurability of $c$-concave functions. Note that continuity of $c$ is sufficient but not necessary for measurability of $c$-concave functions; see section \ref{Appendix: duality in optimal transport, c-concave functions, indicator cost functions} for counterexamples.
\end{remark}

\begin{restatable}[Universal bound on the the dual problem feasible set]{lemma}{lemmaUniversalBoundDualProblemsFeasibleSet}
	\label{Lemma: c-concave functions, universal bound for c-concave functions}
	\singlespacing
	
	Suppose $c : \mathcal{Y}_1 \times \mathcal{Y}_0 \rightarrow \mathbb{R}$ is bounded, and let $c_L = \inf_{(y_1,y_0) \in \mathcal{Y}_1 \times \mathcal{Y}_0} c(y_1,y_0)$, $c_H = \sup_{(y_1,y_0) \in \mathcal{Y}_1 \times \mathcal{Y}_0} c(y_1,y_0)$.

	\begin{enumerate}
		\item For any bounded functions $\varphi : \mathcal{Y}_1 \rightarrow \mathbb{R}$ and $\psi : \mathcal{Y}_0 \rightarrow \mathbb{R}$, $\varphi^c$ and $\psi^c$ are bounded. \label{Lemma: c-concave functions, universal bound for c-concave functions, claim 1}
		\item For any bounded, measurable $c$-conjugate pair $(\varphi, \varphi^c)$ there exists $\bar{\varphi}$ such that \label{Lemma: c-concave functions, universal bound for c-concave functions, claim 2}
		\begin{enumerate}[label=(\roman*)]
			\item $\bar{\varphi}$ and $\bar{\varphi}^c$ satisfy the bounds:
			\begin{align*}
				&c_L \leq \bar{\varphi}(y_1) \leq c_H &&c_L - c_H \leq \bar{\varphi}^c(y_0) \leq 0
			\end{align*}
			for all $(y_1,y_0) \in \mathcal{Y}_1 \times \mathcal{Y}_0$.
			
			\item $J(\varphi, \varphi^c) = J(\bar{\varphi}, \bar{\varphi}^c)$.
		\end{enumerate}
	\end{enumerate}
\end{restatable}
\begin{proof}
	\singlespacing
	For claim \ref{Lemma: c-concave functions, universal bound for c-concave functions, claim 1}, let $\varphi$ be bounded and note that 
	\begin{equation}
		c_L - \sup \varphi \leq \underbrace{\inf_{y_1 \in \mathcal{Y}_1} \{c(y_1, y_0) - \varphi(y_1)\}}_{=\varphi^c(y_0)} \leq c_H - \sup \varphi \label{Display: lemma proof, c-concave bounds on varphi^c}
	\end{equation}
	are finite bounds on $\varphi^c$. The upper bound on $\varphi^c$ follows from the existence of a sequence $\{y_{1j}\}_{j=1}^\infty$ with $\varphi(y_{1j}) \rightarrow \sup_{y_1 \in \mathcal{Y}_1} \varphi(y_1)$, because $\varphi^c(y_0) = \inf_{y_1 \in \mathcal{Y}_1} \{c(y_1, y_0) - \varphi(y_1)\} \leq c(y_{1j}, y_0) - \varphi(y_{1j}) \leq c_H - \varphi(y_{1j})$ for all $j$. The same argument shows $\psi^c$ is bounded, specifically, 
	\begin{equation}
		c_L - \sup \psi \leq \underbrace{\inf_{y_0 \in \mathcal{Y}_0} \{c(y_1,y_0) - \psi(y_0)\}}_{= \psi^c(y_1)} \leq c_H - \sup \psi \label{Display: lemma proof, c-concave bounds on varphi^cc}
	\end{equation}
	
	For claim \ref{Lemma: c-concave functions, universal bound for c-concave functions, claim 2}, let $(\varphi, \varphi^c)$ be a $c$-conjugate pair, i.e. $\varphi(y_1) = \inf_{y_0 \in \mathcal{Y}_0} \{c(y_1,y_0) - \varphi^c(y_0)\}$. Notice that for any $s \in \mathbb{R}$, 
	\begin{align*}
		(\varphi + s)^c(y_0) = \inf_{y_1 \in \mathcal{Y}_1} \{c(y_1, y_0) - \varphi(y_1) - s\} = \varphi^c(y_0) - s \\
		(\varphi + s)^{cc}(y_0) = \inf_{y_0 \in \mathcal{Y}_0} \{c(y_1, y_0) - \varphi^c(y_1) + s\} = \varphi(y_1) + s
	\end{align*}
	Define $\bar{\varphi}(y_1) = \varphi(y_1) - \sup \varphi + c_H$, and notice that $\sup \bar{\varphi} = c_H$. Thus \eqref{Display: lemma proof, c-concave bounds on varphi^c} implies $c_L - c_H \leq \bar{\varphi}^c(y_0) \leq 0$ for all $y_0 \in \mathcal{Y}_0$, and so \eqref{Display: lemma proof, c-concave bounds on varphi^cc} implies $c_L \leq \bar{\varphi}^{cc}(y_1) = \bar{\varphi}(y_1) \leq c_H$. Finally,
	\begin{align*}
		J(\varphi, \varphi^c) &= \int \varphi(y_1) dP_1(y_1) + \int \varphi^c(y_0) dP_0(y_0) \\
		&= \int \varphi(y_1) - \sup \varphi + c_H dP_1(y_1) + \int \varphi^c(y_0) + \sup \varphi - c_H dP_0(y_0) \\
		&= J(\bar{\varphi}, \bar{\varphi}^c)
	\end{align*}
	which completes the proof.
\end{proof}
\begin{remark}
	\label{Remark: discussion of universal bounds on c-concave functions}
	\singlespacing
	
	Lemma \ref{Lemma: c-concave functions, universal bound for c-concave functions} shows that it is often without loss of generality to restrict the dual to classes of functions sharing universal bounds. For an example, see lemma \ref{Lemma: c-concave functions, smooth costs, strong duality} below.
	
	Note that when $c_L = 0$, the bounds simplify to 
	\begin{align*}
		&0 \leq \bar{\varphi}(y_1) \leq \lVert c \rVert_\infty, &&-\lVert c \rVert_\infty \leq \bar{\varphi}^c(y_0) \leq 0
	\end{align*}
	as in \cite{villani2003topics} Remark 1.13. Also note that, when any universal bound suffices, one can take
	\begin{align*}
		&-\lVert c\rVert_\infty \leq \bar{\varphi}(y_1) \leq \lVert c \rVert_\infty, &&-2\lVert c \rVert_\infty \leq \bar{\varphi}^c(y_0) \leq 0
	\end{align*}
	which depend only on $\lVert c \rVert_\infty = \sup_{(y_1,y_0) \in \mathcal{Y}_1\times \mathcal{Y}_0} \lvert c(y_1, y_0) \rvert$. 
\end{remark}

\subsubsection{$c$-concave functions of smooth cost functions}
\label{Appendix: duality in optimal transport, c-concave functions, smooth cost functions}

For $\alpha \in (0,1]$ and $L > 0$, $c : \mathcal{Y}_1 \times \mathcal{Y}_0 \rightarrow \mathbb{R}$ is called $(\alpha, L)$-\textbf{H\"older continuous} if 
\begin{equation*}
	\lvert c(y_1, y_0) - c(y_1', y_0') \rvert \leq L \lVert (y_1, y_0) - (y_1', y_0') \rVert^\alpha
\end{equation*}
for all $(y_1, y_0), (y_1', y_0') \in \mathcal{Y}_1 \times \mathcal{Y}_0$. 

\begin{restatable}[H\"older cost implies H\"older c-concave functions]{lemma}{lemmaHolderCostImpliesHolderCConcaveFunctions}
	\singlespacing
	\label{Lemma: c-concave functions, Holder cost implies Holder c-concave functions}
	
	Let $c : \mathcal{Y}_1 \times \mathcal{Y}_0 \rightarrow \mathbb{R}$ be $(\alpha,L)$-H\"older continuous. For any $g : \mathcal{Y}_0 \rightarrow \mathbb{R}$, 
	\begin{align*}
		&\varphi(y_1) = \inf_{y_0 \in \mathcal{Y}_0} \{c(y_1,y_0) - g(y_0)\}, &&\varphi^c(y_0) = \inf_{y_1 \in \mathcal{Y}_1} \{c(y_1, y_0) - \varphi(y_1)\}
	\end{align*}
	are $(\alpha, L)$-H\"older continuous.
\end{restatable}
\begin{proof}
	\singlespacing
	
	H\"older continuity implies $c(y_1,y_0) \leq c(y_1', y_0) + L \lvert y_1 - y_1' \rvert^\alpha$ holds for any $y_0 \in \mathcal{Y}_0$ and any $y_1, y_1' \in \mathcal{Y}_1$. It follows that
	\begin{equation*}
		\varphi(y_1) = \inf_{y_0' \in \mathcal{Y}_0} \{c(y_1, y_0') - g(y_0')\} \leq c(y_1,y_0) - g(y_0) \leq c(y_1', y_0) - g(y_0) + L\lvert y_1 - y_1' \rvert^\alpha
	\end{equation*}
	implying $\varphi(y_1) - (c(y_1', y_0) - g(y_0)) \leq L \lvert y_1 - y_1' \rvert^\alpha$. Therefore
	\begin{equation*}
		\varphi(y_1) - \varphi(y_1') = \varphi(y_1) - \inf_{y_0 \in \mathcal{Y}_0} \{c(y_1',y_0) - g(y_0)\} \leq L\lvert y_1 - y_1' \rvert^\alpha
	\end{equation*}
	holds for any $y_1, y_1' \in \mathcal{Y}_1$. This implies $\varphi(y_1') - \varphi(y_1) \leq L \lvert y_1' - y_1 \rvert^\alpha$, hence $\varphi$ is $(\alpha,L)$-H\"older. The same argument implies $\varphi^c$ is $(\alpha,L)$-H\"older.
\end{proof}

Lemmas \ref{Lemma: c-concave functions, smooth costs, strong duality}, \ref{Lemma: weak convergence, Donsker results, c-concave functions for smooth costs}, and \ref{Lemma: weak convergence, completeness, smooth costs c-concave functions}, are relevant for compact $\mathcal{Y}_1, \mathcal{Y}_0 \subset \mathbb{R}$, and $L$-Lipscthiz $c : \mathcal{Y}_1 \times \mathcal{Y}_0 \rightarrow \mathbb{R}$. Under these assumptions, define
\begin{align}
	\mathcal{F}_c &= \left\{\varphi : \mathcal{Y}_1 \rightarrow \mathbb{R} \; ; \; -\lVert c \rVert_\infty \leq \varphi(y_1) \leq \lVert c \rVert_\infty, \; \lvert \varphi(y_1) - \varphi(y_1') \rvert \leq L \lvert y_1 - y_1'\rvert \right\} \label{Defn: F_c for smooth costs, strong duality appendix} \\
	\mathcal{F}_c^c &= \left\{\psi : \mathcal{Y}_0 \rightarrow \mathbb{R} \; ; \; -2\lVert c \rVert_\infty \leq \psi(y_0) \leq 0, \; \lvert \psi(y_0) - \psi(y_0') \rvert \leq L \lvert y_0 - y_0'\rvert \right\} \label{Defn: F_c^c for smooth costs, strong duality appendix}
\end{align}

\begin{restatable}[Strong duality for smooth cost functions]{lemma}{lemmaCConcaveFunctionsLipschitzCostCompactSupportsDonsker}
	\singlespacing
	\label{Lemma: c-concave functions, smooth costs, strong duality}	
	
	Let $\mathcal{Y}_1, \mathcal{Y}_0 \subset \mathbb{R}$ be compact, $c : \mathcal{Y}_1 \times \mathcal{Y}_0 \rightarrow \mathbb{R}$ be $L$-Lipschitz, and $\mathcal{F}_c$, $\mathcal{F}_c^c$ be given by \eqref{Defn: F_c for smooth costs, strong duality appendix} and \eqref{Defn: F_c^c for smooth costs, strong duality appendix} respectively. Then strong duality holds:
	\begin{equation*}
		\inf_{\pi \in \Pi(P_1, P_0)} I_c[\pi] = \sup_{(\varphi, \psi) \in \Phi_c \cap (\mathcal{F}_c \times \mathcal{F}_c^c)} J(\varphi, \psi)
	\end{equation*}
\end{restatable}
\begin{proof}
	\singlespacing
	
	First notice lemma \ref{Lemma: c-concave functions, Holder cost implies Holder c-concave functions} implies $\mathcal{F}_c(\Phi_c \cap \mathcal{C}_b)$ and $\mathcal{F}_c^c(\Phi_c \cap \mathcal{C}_b)$ consist of $L$-Lipschitz functions.\footnote{Note that $\mathcal{F}_c(\Phi_c \cap \mathcal{C}_b)$ and $\mathcal{F}_c^c(\Phi_c \cap \mathcal{C}_b)$ are not necessarily $\mathcal{F}_c$ and $\mathcal{F}_c^c$ defined in the statement of the lemma.} Since $c$ is continuous and $\mathcal{Y}_1 \times \mathcal{Y}_0$ is compact, $\lVert c \rVert_\infty = \sup_{y_1, y_0 \in \mathcal{Y}_1 \times \mathcal{Y}_0} \lvert c(y_1,y_0) \rvert < \infty$. Continuity implies these $c$-concave functions are measurable, and lemma \ref{Lemma: c-concave functions, universal bound for c-concave functions} shows they are bounded. Thus $\mathcal{F}_c(\Phi_c \cap \mathcal{C}_b) \times \mathcal{F}_c^c(\Phi_c \cap \mathcal{C}_b) \subseteq L^1(P_1) \times L^1(P_0)$, and so lemma \ref{Lemma: c-concave functions, dual restricted to measurable and integrable c-concave functions satisfies strong duality} implies
	\begin{equation*}
		\inf_{\pi \in \Pi(P_1, P_0)} I_c[\pi] = \sup_{\varphi \in \mathcal{F}_c(\Phi_c \cap \mathcal{C}_b)} J(\varphi,\varphi^c)
	\end{equation*}
	Lemma \ref{Lemma: c-concave functions, universal bound for c-concave functions} and remark \ref{Remark: discussion of universal bounds on c-concave functions} further shows that for every $\varphi \in \mathcal{F}_c(\Phi_c \cap \mathcal{C}_b)$, a shifted function $\bar{\varphi}$ is such that $\sup_{y_1 \in \mathcal{Y}_1} \lvert \bar{\varphi}(y_1) \rvert \leq \lVert c \rVert_\infty$, $-2\lVert c \rVert \leq \bar{\varphi}^c(y_0) \leq 0$, $\bar{\varphi}$ and $\bar{\varphi}^c$ are $L$-lipschitz, and $J(\varphi, \varphi^c) = J(\bar{\varphi}, \bar{\varphi}^c)$. Thus 
	\begin{equation*}
		\sup_{\varphi \in \mathcal{F}_c(\Phi_c \cap \mathcal{C}_b)} J(\varphi,\varphi^c) = \sup_{\varphi \in \mathcal{F}_c} J(\varphi,\varphi^c) 
	\end{equation*}
	Furthermore, 
	\begin{equation*}
		\sup_{\varphi \in \mathcal{F}_c} J(\varphi,\varphi^c) \leq \sup_{(\varphi, \psi) \in \Phi_c \cap (\mathcal{F}_c \times \mathcal{F}_c^c)} J(\varphi,\psi) \leq \sup_{(\varphi,\psi) \in \Phi_c} J(\varphi,\psi) = \inf_{\pi \in \Pi(P_1,P_0)} I_c[\pi]
	\end{equation*}
	completes the proof.
\end{proof}

\begin{remark}
	\label{Remark: differentiability implies Lipschitz}
	\singlespacing
	
	Suppose $\mathcal{Y}_1$ and $\mathcal{Y}_0$ are compact and $c(y_1,y_0)$ is continuously differentiable on an open set containing $\mathcal{Y}_1 \times \mathcal{Y}_0$. Then $c$ restricted to $\mathcal{Y}_1 \times \mathcal{Y}_0$ is bounded and Lipschitz. 
	
	That $c : \mathcal{Y}_1 \times \mathcal{Y}_0 \rightarrow \mathbb{R}$ is bounded follows from $c$ being continuous, $\mathcal{Y}_1 \times \mathcal{Y}_0$ being compact, and the extreme value theorem. To see that $c$ restricted to $\mathcal{Y}_1 \times \mathcal{Y}_0$ is $L$-Lipschitz, let $(y_1, y_0), (y_1',y_0') \in \mathcal{Y}_1 \times \mathcal{Y}_0$ be arbitrary and note that the mean value theorem applied to $g(t) = c(t(y_1, y_0) + (1-t)(y_1',y_0'))$ implies there exists $s \in (0,1)$ such that 
	\begin{align*}
		(c(y_1, y_0) - c(y_1',y_0'))  &= g(1) - g(0) = g'(s) \\
		&= \left\langle \nabla c(s(y_1, y_0) + (1-s)(y_1',y_0')), (y_1, y_0) - (y_1', y_0') \right\rangle 
	\end{align*}
	Notice that Cauchy-Schwarz then implies
	\begin{align*}
		\lvert c(y_1, y_0) - c(y_1',y_0') \rvert &\leq \lVert \nabla c(s(y_1, y_0) + (1-s)(y_1',y_0')) \rVert \lVert (y_1, y_0) - (y_1', y_0') \rVert \\
		&\leq \sup_{(y_1'', y_0'') \in \mathcal{Y}_1 \times \mathcal{Y}_0} \lVert \nabla c(y_1'', y_0'') \rVert \lVert (y_1, y_0) - (y_1', y_0') \rVert
	\end{align*}
	Finally, notice $L = \sup_{(y_1'', y_0'') \in \mathcal{Y}_1 \times \mathcal{Y}_0} \lVert \nabla c(y_1'', y_0'') \rVert$ is finite because $\mathcal{Y}_1 \times \mathcal{Y}_0$ is compact and $(y_1, y_0) \mapsto  \lVert \nabla c(y_1, y_0) \rVert$ is continuous.
\end{remark}

\subsubsection{$c$-concave functions when $c(y_1,y_0) = \mathbbm{1}\{(y_1,y_0) \in C\}$}
\label{Appendix: duality in optimal transport, c-concave functions, indicator cost functions}

\begin{restatable}[Strong duality with indicator costs]{theorem}{theoremStrongDualityForIndicatorCosts}
	\label{Theorem: strong duality for indicator costs} 
	\singlespacing
	
	Let $C$ be a nonempty, open subset of $\mathcal{Y}_1 \times \mathcal{Y}_0$, and $c : \mathcal{Y}_1 \times \mathcal{Y}_0 \rightarrow \mathbb{R}$ given by $c(y_1,y_0) = \mathbbm{1}_C(y_1,y_0) = \mathbbm{1}\{(y_1,y_0) \in C\}$. Then
	\begin{equation*}
		\inf_{\pi \in \Pi(P_1, P_0)} \int \mathbbm{1}_C(y_1,y_0) d\pi(y_1,y_0) = \sup_{(A, B) \in \Phi_c^I} \int \mathbbm{1}_A(y_1) dP_1(y_1) - \int \mathbbm{1}_B(y_0) d\nu(y_0)
	\end{equation*}
	where 
	\begin{equation*}
		\Phi_c^I = \left\{(A, B)\; ; \; A \subset \mathcal{Y}_1 \text{ is closed and nonempty, } B \subset \mathcal{Y}_0 \text{ is measurable, and } \mathbbm{1}_A(y_1) - \mathbbm{1}_B(y_0) \leq \mathbbm{1}_C(y_1,y_0) \right\}
	\end{equation*}
\end{restatable}
\begin{proof}
	\singlespacing
	
	\cite{villani2003topics} Theorem 1.27 implies 
	\begin{align*}
		\inf_{\pi \in \Pi(P_1, P_0)} \int \mathbbm{1}_C(y_1,y_0) d\pi(y_1,y_0) = \sup_{A \text{ closed}} \int \mathbbm{1}_A(y_1) dP_1(y_1) - \int \mathbbm{1}_{A^C}(y_0) dP_0(y_0)
	\end{align*}
	where $A^C = \left\{y \in \mathcal{Y}_0 \; ; \; \exists y_1 \in A, \; (y_1, y_0) \not \in C\right\}$ is the \textbf{projection} of $(A \times \mathcal{Y}_0) \setminus C$ onto $\mathcal{Y}_0$. Measurability of $A^C$ is guaranteed by the measurable projection theorem; see \cite{crauel2002random} theorem 2.12. It is clear that 
	\begin{align*}
		\sup_{A \text{ closed}} \int \mathbbm{1}_A(y_1) dP_1(y_1) - \int \mathbbm{1}_{A^C}(y_0) dP_0(y_0) \leq \sup_{A \subseteq \mathcal{Y}_1, B \subseteq \mathcal{Y}_0} \int \mathbbm{1}_A(y_1) dP_1(y_1) - \int \mathbbm{1}_B(y_0) d\nu(y_0)
	\end{align*}
	with $A$, $B$ measurable. Notice it is without loss to exclude $A = \varnothing$, because $J(\mathbbm{1}_\varnothing, -\mathbbm{1}_B) \leq 0 = J(\mathbbm{1}_{\mathcal{Y}_1}, \mathbbm{1}_{\mathcal{Y}_0})$ and $\mathbbm{1}_{\mathcal{Y}_1}(y_1) - \mathbbm{1}_{\mathcal{Y}_0}(y_0) = 0 \leq \mathbbm{1}_C(y_1,y_0)$ for all $(y_1, y_0) \in \mathcal{Y}_1 \times \mathcal{Y}_0$. Thus 	
	\begin{equation*}
		\sup_{A \subseteq \mathcal{Y}_1, B \subseteq \mathcal{Y}_0} \int \mathbbm{1}_A(y_1) dP_1(y_1) - \int \mathbbm{1}_B(y_0) d\nu(y_0) = \sup_{(A, B) \in \Phi_c^I} \int \mathbbm{1}_A(y_1) dP_1(y_1) - \int \mathbbm{1}_B(y_0) d\nu(y_0)
	\end{equation*}
	Weak duality (lemma \ref{Lemma: weak duality}) implies 
	\begin{equation*}
		\sup_{(A, B) \in \Phi_c^I} \int \mathbbm{1}_A(y_1) dP_1(y_1) - \int \mathbbm{1}_B(y_0) dP_0(y_0) \leq \inf_{\pi \in \Pi(P_1, P_0)} \int \mathbbm{1}_C(y_1,y_0) d\pi(y_1,y_0) 
	\end{equation*}
	and the result follows.
\end{proof}

The strong duality result of theorem \ref{Theorem: strong duality for indicator costs} is especially useful when combined with a careful characterization of the corresponding $c$-concave functions. To describe these, let $A \subseteq \mathcal{Y}_1$ be nonempty, and define
\begin{align}
	&A^C = \left\{y_0 \in \mathcal{Y}_0 \; ; \; \exists y_1 \in A, \; (y_1, y_0) \not \in C \right\}, &&A^{CC} = \left\{y_1 \in \mathcal{Y}_1 \; ; \; \forall y_0 \in \mathcal{Y}_0 \setminus A^C, \; (y_1, y_0) \in C \right\}, \label{Defn: A^C, A^CC} \\
	&C_{0m} = \left\{y_0 \in \mathcal{Y}_0 \; ; \; \forall y_1 \in \mathcal{Y}_1, \; (y_1,y_0) \in C\right\}, &&C_{1m} = \left\{y_1 \in \mathcal{Y}_1 \; ; \; \forall y_0 \in \mathcal{Y}_0, \; (y_1, y_0) \in C\right\} \label{Defn: C_0m, C_1m} \\
	&C_{0m}^C = \begin{cases} C_{1m} &\text{ if } C_{0m} = \varnothing \\ \varnothing &\text{ if } C_{0m} \neq \varnothing \end{cases}, &&C_{1m}^C = \begin{cases} C_{0m} &\text{ if } C_{1m} = \varnothing \\ \varnothing &\text{ if } C_{1m} \neq \varnothing \end{cases} \label{Defn: C_0m^C, C_1m^C}
\end{align}
Note that $A^C$ is well defined whenever $A \neq \varnothing$, and to ensure $A^{CC}$ is well defined we require $A^C \neq \mathcal{Y}_0$. $C_{0m}$ is denoted as such because $\mathbbm{1}_{C_{0m}}(y_0) = \inf_{y_1 \in \mathcal{Y}_1} \mathbbm{1}_C(y_1,y_0)$ is the subset of $\mathcal{Y}_{\underline{0}}$ found by \underline{m}inimizing $\mathbbm{1}_C(y_1,y_0)$ over $y_1 \in \mathcal{Y}_1$.

\begin{restatable}[$c$-concave functions for indicator costs]{lemma}{lemmaCConcaveFunctionsIndicators}
	\label{Lemma: c-concave functions, indicator costs, c-conjugate pairs generated from indicators}
	\singlespacing
	
	Let $C$ be a nonempty, open subset of $\mathcal{Y}_1 \times \mathcal{Y}_0$, $c : \mathcal{Y}_1 \times \mathcal{Y}_0 \rightarrow \mathbb{R}$ given by $c(y_1,y_0) = \mathbbm{1}_C(y_1,y_0)$, $A \subseteq \mathcal{Y}_1$ be closed and nonempty, and $\varphi(y_1) = \mathbbm{1}_A(y_1) = \mathbbm{1}\{y_1 \in A\}$. Then
	\begin{enumerate}
		\item $\varphi^c(y_0) = -\mathbbm{1}_{A^C}(y_0)$,
		
		\item if $A^C \neq \mathcal{Y}_0$, then $\varphi^{cc}(y_1) = \mathbbm{1}_{A^{CC}}(y_1)$, and 
		
		\item If $A^C = \mathcal{Y}_0$, then $J(\varphi^{cc}, \varphi^c) = J(\mathbbm{1}_{C_{1m}}, 0)$
		
	\end{enumerate}
\end{restatable}
\begin{proof}
	\singlespacing
	
	\begin{enumerate}
		\item Notice $\mathbbm{1}_C(y_1, y_0) - \mathbbm{1}_A(y_1) \in \{-1, 0, 1\}$, and
		\begin{equation*}
			\varphi^c(y_0) = \inf_{y_1 \in \mathcal{Y}_1} \{\mathbbm{1}_C(y_1, y_0) - \mathbbm{1}_A(y_1)\}
		\end{equation*}
		will never take value $1$ because any $y_1 \in A$ implies the objective is at most $0$. Furthermore, if there exists $y_1 \in A$ such that $(y_1, y_0) \not \in C$, then the infimum attains $-1$. If there does not exist such $y_1$, then $\varphi^c(y_0) = 0$. Thus $\varphi^c(y_0) = -\mathbbm{1}_{A^C}(y_0)$. 
		
		\item Suppose $A^C \neq \mathcal{Y}_0$. Notice that $\mathbbm{1}_C(y_1, y_0) + \mathbbm{1}_{A^C}(y_0)$ takes values in $\{0, 1, 2\}$, and
		\begin{equation*}
			\varphi^{cc}(y_1) = \inf_{y_0 \in \mathcal{Y}_0} \{\mathbbm{1}_C(y_1, y_0) + \mathbbm{1}_{A^C}(y_0)\}
		\end{equation*}
		will never equal $2$ because $\mathcal{Y}_0 \setminus A^C \neq \varnothing$. Moreover, the infimum will equal $1$ if and only if $(y_1, y_0) \in C$ for all $y_0 \in \mathcal{Y}_0 \setminus A^C$; thus $\varphi^{cc}(y_1) = \mathbbm{1}_{A^{CC}}(y_1)$. 
		
		\item
		If $A^C = \mathcal{Y}_0$, then $\varphi^{cc}(y_1) = \inf_{y_0 \in \mathcal{Y}_0}\{\mathbbm{1}_C(y_1, y_0) + 1\} = \mathbbm{1}_{C_{1m}}(y_1) + 1$ and
		\begin{align*}
			\varphi^{ccc}(y_0) = \inf_{y_1 \in \mathcal{Y}_1} \{\mathbbm{1}_C(y_1, y_0) - \mathbbm{1}_{C_{1m}}(y_1) - 1\} = \mathbbm{1}_{C_{1m}^C}(y_0) -1
		\end{align*}
		To see that $(\mathbbm{1}_{C_{1m}})^c = 0$ if $C_{1m} \neq \varnothing$, notice the objective $\mathbbm{1}_C(y_1, y_0) - \mathbbm{1}_{C_{1m}}(y_0)$ takes values in $\{-1, 0, 1\}$, and because $C_{1m} \neq \varnothing$ will never take value $1$. For the objective to take value $-1$ at a given $y_1$, it must be the case that $\mathbbm{1}_{C_{1m}}(y_1) = 1$ and there exists $y_0$ such that $\mathbbm{1}_C(y_1, y_0) = 0$, but this contradicts the definition $C_{1m} = \left\{y_1 \in \mathcal{Y}_1 \; ; \; \forall y_0 \in \mathcal{Y}_0, \; (y_1, y_0) \in C\right\}$.
		
		However, recall that $\varphi^{ccc}(y_0) = \varphi^c(y_0)$ as shown in lemma \ref{Lemma: c-concave functions, generating c-concave functions}. Since $\varphi^c(y_0) = -\mathbbm{1}_{A^C}(y_0) = -\mathbbm{1}_{\mathcal{Y}_0}(y_0) = -1$, this implies $(\mathbbm{1}_{C_{0m}^C})(y_0) = 0$. Then notice that
		\begin{equation*}
			J(\varphi^{cc}, \varphi^c) = J(\mathbbm{1}_{C_{1m}} + 1, -1) = J(\mathbbm{1}_{C_{1m}}, 0) 
		\end{equation*}
	\end{enumerate}
\end{proof}
\begin{remark}
	\label{Remark: c-concave functions for indicator costs} 
	\singlespacing
	
	Compare theorem \ref{Theorem: strong duality for indicator costs} and lemma \ref{Lemma: c-concave functions, indicator costs, c-conjugate pairs generated from indicators} with \cite{villani2003topics} theorem 1.27. 
\end{remark}

\begin{restatable}[Convex $C$ implies $c$-concave functions defined with convex sets]{lemma}{lemmaCConcaveFunctionsIndicatorsConvexC}
	\label{Lemma: c-concave functions, indicator costs, c-conjugate pairs generated from indicators for convex sets}
	\singlespacing
	
	Let $C$ be a nonempty, open, convex subset of $\mathcal{Y}_1 \times \mathcal{Y}_0$, and $c : \mathcal{Y}_1 \times \mathcal{Y}_0 \rightarrow \mathbb{R}$ given by $c(y_1,y_0) = \mathbbm{1}_C(y_1,y_0)$. Let $A \subseteq \mathcal{Y}_1$ be nonempty.
	\begin{enumerate}
		\item $A^C$ equals $\mathcal{Y}_0 \setminus B$ for some convex set $B$. \label{Lemma: c-concave functions, indicator costs, c-conjugate pairs generated from indicators for convex sets, claim 1}
		\item If $A^C \neq \mathcal{Y}_0$, then $A^{CC}$ is convex. \label{Lemma: c-concave functions, indicator costs, c-conjugate pairs generated from indicators for convex sets, claim 2}
		\item $C_{1m}$ is convex. \label{Lemma: c-concave functions, indicator costs, c-conjugate pairs generated from indicators for convex sets, claim 3}
	\end{enumerate}
\end{restatable}
\begin{proof}
	\singlespacing
	
	For claim \ref{Lemma: c-concave functions, indicator costs, c-conjugate pairs generated from indicators for convex sets, claim 1}, notice that 
	\begin{align*}
		A^C &= \left\{y_0 \in \mathcal{Y}_0 \; ; \; \exists y_1 \in A, \; (y_1, y_0) \in \big(\mathcal{Y}_1 \times \mathcal{Y}_0\big) \setminus C \right\} = \bigcup_{y_1 \in A} \left\{y_0 \in \mathcal{Y}_0 \; ; \; (y_1, y_0) \in \big(\mathcal{Y}_1 \times \mathcal{Y}_0\big) \setminus C\right\} \\
		&= \bigcup_{y_1 \in A} \mathcal{Y}_0 \setminus \left\{y_0 \in \mathcal{Y}_0 \; ; \; (y_1, y_0) \in C\right\} = \mathcal{Y}_0 \setminus \bigcap_{y_1 \in A} \left\{y_0 \in \mathcal{Y}_0 \; ; \; (y_1, y_0) \in C\right\}
	\end{align*}
	Since $C$ is convex, $\{y \in \mathcal{Y}_0 \; ; \; (y_1, y_0) \in C\}$ is also convex for any $y_1$. The intersection of an arbitrary collection of convex sets is convex, so $A^C = \mathcal{Y}_0 \setminus B$ for some convex $B$. \\
		
	Consider claim \ref{Lemma: c-concave functions, indicator costs, c-conjugate pairs generated from indicators for convex sets, claim 2} next. Notice that 
	\begin{align*}
		A^{CC} = \left\{y_1 \in \mathcal{Y}_1 \; ; \; \forall y_0 \in \mathcal{Y}_0 \setminus A^C, \; (y_1, y_0) \in C \right\} = \bigcap_{y_0 \in \mathcal{Y}_0 \setminus A^C} \left\{y_1 \in \mathcal{Y}_1 \; ; \; (y_1, y_0) \in C\right\}
	\end{align*}
	Since $C$ is convex, $\left\{y_1 \in \mathcal{Y}_1 \; ; \; (y_1, y_0) \in C\right\}$ is convex as well, and thus $A^{CC}$ is convex. \\
		
	Finally, we show claim \ref{Lemma: c-concave functions, indicator costs, c-conjugate pairs generated from indicators for convex sets, claim 3}. Similar to $A^{CC}$, notice that 
	\begin{equation*}
		C_{1m} = \left\{y_1 \in \mathcal{Y}_1 \; ; \; \forall y_0 \in \mathcal{Y}_0, \; (y_1,y_0) \in C\right\} = \bigcap_{y_0 \in \mathcal{Y}_0} \left\{y_1 \in \mathcal{Y}_1 \; ; \; (y_1, y_0) \in C\right\}
	\end{equation*}
	is the intersection of convex sets and therefore convex.
\end{proof}

Refer to the convex subsets of $\mathbb{R}$ as \textbf{intervals}; specifically, $I \subset \mathbb{R}$ is called an interval if $I$ takes the form
\begin{align*}
	&(\ell, u) &&[\ell, u) && (\ell, u] && [\ell, u]
\end{align*}
where $\ell = -\infty$ is allowed for $(\ell, u)$ and $(\ell, u]$ and $u = \infty$ is allowed for $(\ell, u)$ and $[\ell, u)$. $I^c$ is the complement of the interval $I$.

Lemmas \ref{Lemma: c-concave functions, indicator of convex set costs, strong duality}, \ref{Lemma: weak convergence, Donsker results, c-concave functions for indicator of convex set costs}, and \ref{Lemma: weak convergence, completeness, indicator of convex set costs c-concave functions} are relevant when the cost function is $c(y_1, y_0) = \mathbbm{1}\{(y_1,y_0) \in C\}$ for some nonempty, open, convex $C \subseteq \mathcal{Y}_1 \times \mathcal{Y}_0$. When this is so, define 
\begin{align}
	\mathcal{F}_c &= \left\{\varphi : \mathcal{Y}_1 \rightarrow \mathbb{R} \; ; \; \varphi(y_1) = \mathbbm{1}_I(y_1) \text{ for some interval } I\right\} \label{Defn: F_c for indicator costs of convex C, strong duality appendix} \\
	\mathcal{F}_c^c &= \left\{\psi : \mathcal{Y}_0 \rightarrow \mathbb{R} \; ; \; \psi(y_0) = -\mathbbm{1}_{I^c}(y_0) \text{ for some  interval } I\right\} \label{Defn: F_c^c for indicator costs of convex C, strong duality appendix}
\end{align}

\begin{restatable}[Strong duality for indicator cost functions of a convex set]{lemma}{lemmaCConcaveFunctionsIndicatorsConvexCDonsker}
	\label{Lemma: c-concave functions, indicator of convex set costs, strong duality}
	\singlespacing
	Let $\mathcal{Y}_1, \mathcal{Y}_0 \subseteq \mathbb{R}$, $C \subseteq \mathcal{Y}_1 \times \mathcal{Y}_0$ be nonempty, open, and convex, and let $c : \mathcal{Y}_1 \times \mathcal{Y}_0 \rightarrow \mathbb{R}$ be given by $c(y_1,y_0) = \mathbbm{1}_C(y_1,y_0)$. Let $\mathcal{F}_c$ and $\mathcal{F}_c^c$ be given by \eqref{Defn: F_c for indicator costs of convex C, strong duality appendix} and \eqref{Defn: F_c^c for indicator costs of convex C, strong duality appendix} respectively. Then strong duality holds:
	\begin{equation}
		\inf_{\pi \in \Pi(P_1, P_0)} \int \mathbbm{1}_C(y_1,y_0) d\pi(y_1, y_0) = \sup_{(\varphi, \psi) \in \Phi_c \cap \big(\mathcal{F}_c \times \mathcal{F}_c^c\big)} \int \varphi(y_1) dP_1(y_1) + \int \psi(y_0) dP_0(y_0) \label{Display: lemma, strong duality for indicators of convex sets}
	\end{equation}
\end{restatable}
\begin{proof}
	\singlespacing
	
	Recall that theorem \ref{Theorem: strong duality for indicator costs} shows 
	\begin{equation*}
		\inf_{\pi \in \Pi(P_1, P_0)} \int \mathbbm{1}_C(y_1,y_0) d\pi(y_1,y_0) = \sup_{(A, B) \in \Phi_c^I} \int \mathbbm{1}_A(y_1) dP_1(y_1) - \int \mathbbm{1}_B(y_0) d\nu(y_0)
	\end{equation*}
	where 
	\begin{equation*}
		\Phi_c^I = \left\{(A, B)\; ; \; A \subset \mathcal{Y}_1 \text{ is closed and nonempty, } B \subset \mathcal{Y}_0 \text{ is measurable, and } \mathbbm{1}_A(y_1) - \mathbbm{1}_B(y_0) \leq \mathbbm{1}_C(y_1,y_0) \right\}
	\end{equation*}
	We will apply lemma \ref{Lemma: c-concave functions, dual restricted to measurable and integrable c-concave functions satisfies strong duality}. Let $\varphi(y_1) = \mathbbm{1}_A(y_1)$ for some closed and nonempty $A \subset \mathcal{Y}_1$. There are two possibilities: 
	\begin{enumerate}
		\item $A^C = \mathcal{Y}_0$, in which case $J(\varphi^{cc}, \varphi^c) = J(\mathbbm{1}_{C_{1m}}, 0)$, or
		\item $A^C \neq \mathcal{Y}_0$, in which case $J(\varphi^{cc}, \varphi^c) = J(\mathbbm{1}_{A^{CC}}, -\mathbbm{1}_{A^C})$. 
	\end{enumerate}
	Since $C$ is convex, $C_{1m}$, and $A^{CC}$ are convex subsets of $\mathbb{R}$ (i.e., intervals), as shown in lemma \ref{Lemma: c-concave functions, indicator costs, c-conjugate pairs generated from indicators for convex sets}. $A^C$ is the complement of an interval, and $0 = \mathbbm{1}_\varnothing(y_0)$ is the indicator of the complement of $\mathbb{R}$, which is the interval $(-\infty,\infty)$. Since all functions involved are bounded, they are all integrable, and lemma \ref{Lemma: c-concave functions, dual restricted to measurable and integrable c-concave functions satisfies strong duality} implies 
	\begin{align*}
		\inf_{\pi \in \Pi(P_1, P_0)} \int \mathbbm{1}_C(y_1,y_0) d\pi(y_1, y_0) = \sup_{(\varphi, \psi) \in \Phi_c \cap \big(\mathcal{F}_c(\Phi_c^I) \times \mathcal{F}_c^c(\Phi_c^I)\big)} \int \varphi(y_1) dP_1(y_1) + \int \psi(y_0) dP_0(y_0)
	\end{align*}
	Finally, note that $\mathcal{F}_c(\Phi_c^I) \subseteq \mathcal{F}_c$ and $\mathcal{F}_c^c(\Phi_c^I) \subseteq \mathcal{F}_c^c$, which implies the strong duality claim in display \eqref{Display: lemma, strong duality for indicators of convex sets} holds.
\end{proof}

\subsection{Special cases: $c_L(y_1, y_0, \delta) = \mathbbm{1}\{y_1 - y_0 < \delta\}$ and $c_H(y_1, y_0, \delta) = \mathbbm{1}\{y_1 - y_0 > \delta\}$}
\label{Appendix: duality in optimal transport, CDF special case}

\begin{restatable}[]{lemma}{lemmaSpecialCaseCDFWeakInequality}
	\label{Lemma: duality special case CDF}
	\singlespacing
	
	Let $F_1(y) = P_1(Y_1 \leq y) = \int \mathbbm{1}\{y_1 \leq y\} dP_1(y_1)$ denote the cumulative distribution function (CDF) of $P_1$, and let $F_0$ the CDF of $P_0$. Let $c_L(y_1, y_0, \delta) = \mathbbm{1}\{y_1 - y_0 < \delta\}$. Then
	\begin{align}
		OT_{c_L}(P_1, P_0) &= \inf_{\pi \in \Pi(P_1, P_0)} \int \mathbbm{1}\{y_1 - y_0 < \delta\} d\pi(y_1, y_0) \notag \\
		&= \max\left\{\sup_y \{F_1(y) - F_0(y - \delta)\}, P_1(Y_1 < \min\{\mathcal{Y}_0\} + \delta)\right\} \label{Display: lemma, strong duality for CDF lower bound}
	\end{align}
\end{restatable}
\begin{proof} 
	\singlespacing
	
	Let $C = \{y_1 - y_0 < \delta\}$. Apply theorem \ref{Theorem: strong duality for indicator costs} and lemma \ref{Lemma: c-concave functions, indicator costs, c-conjugate pairs generated from indicators} to find that 
	\begin{align*}
		OT_{c_L}(P_1, P_0) = \max\{\sup_{A \in \mathcal{A}} P_1(Y_1 \in A^{CC}) - P_0(Y_0 \in A^C), P_1(Y_1 \in C_{1m})\}
	\end{align*}
	where 
	\begin{gather*}
		A^C = \left\{y_0 \in \mathcal{Y}_0 \; ; \; \exists y_1 \in A, \; (y_1, y_0) \not \in C \right\}, \quad A^{CC} = \left\{y_1 \in \mathcal{Y}_1 \; ; \; \forall y_0 \in \mathcal{Y}_0 \setminus A^C, \; (y_1, y_0) \in C \right\}, \\
		C_{1m} = \left\{y_1 \in \mathcal{Y}_1 \; ; \; \forall y_0 \in \mathcal{Y}_0, \; (y_1, y_0) \in C\right\}.
	\end{gather*}
	and $\mathcal{A}$ is the collection of closed, nonempty subsets of $\mathcal{Y}_1$ such that $A^C \neq \mathcal{Y}_0$. 
	
	First consider $\sup_{A \in \mathcal{A}} P_1(Y_1 \in A^{CC}) - P_0(Y_0 \in A^C)$. Let $A \in \mathcal{A}$ and $\varphi(y_1) = \mathbbm{1}_A(y_1)$. Thus
	\begin{align*}
		A^C &= \{y \in \mathcal{Y}_0 \; ; \; \exists y_1 \in A, \; (y_1, y_0) \not \in C\} = \{y_0 \in \mathcal{Y}_0 \; ; \; y_0 \leq \max\{A\} - \delta\}, \\
		A^{CC} &= \left\{y_1 \in \mathcal{Y}_1 \; ; \; \forall y_0 \in \mathcal{Y}_0 \setminus A^C, \; y_1 - y_0 < \delta \right\} = \left\{y_1 \in \mathcal{Y}_1 \; ; \; y_1 \leq \max\{A\}\right\}
	\end{align*}
	where we've used the fact that $A^C \neq \mathcal{Y}_0$ implies $\sup\{A\} < \infty$ and so $\sup\{A\} = \max\{A\}$ because $A$ is closed. Therefore
	\begin{align*}
		J(\varphi^{cc}, \varphi^c) &= P_1(Y_1 \in A^{CC}) - P_0(Y_0 \in A^c) \\
		&= P_1(Y_1 \leq \max\{A\}) - P_0(Y_0 \leq \max\{A\} - \delta)
	\end{align*}
	which takes the form $F_1(y) - F_0(y - \delta)$ for $y = \max\{A\}$. 
	
	Now consider $P_1(Y_1 \in C_{1m})$, and notice that 
	\begin{align*}
		C_{1m} &= \left\{y_1 \in \mathcal{Y}_1 \; ; \; \forall y_0 \in \mathcal{Y}_0, \; (y_1, y_0) \in C\right\} = \left\{y_1 \in \mathcal{Y}_1 \; ; \; \forall y_0 \in \mathcal{Y}_0, \; y_1 - y_0 < \delta\right\} \\
		&= \left\{y_1 \in \mathcal{Y}_1 \; ; \; \forall y_0 \in \mathcal{Y}_0, \; y_1 < \min\{\mathcal{Y}_0\} + \delta\right\}
	\end{align*}
	Thus $P_1(Y_1 \in C_{1m}) = P_1(Y_1 < \min\{\mathcal{Y}_0\} + \delta)$. The result follows.
\end{proof}

\begin{remark}
	\singlespacing
	$C_{1m}$ may be closed; e.g., let $\mathcal{Y}_1 = [0,1] \cup [3, 10]$, let $\mathcal{Y}_0 = [2,10]$, and $\delta = 0$. Then $C_{1m} = \{y_1 \in \mathcal{Y}_1 \; ; \; y_1 < 2\} = [0,1]$.
\end{remark}

\begin{restatable}[]{corollary}{corollarySpecialCaseCDFWeakInequality}
	\label{Lemma: duality special case CDF, corollary}
	\singlespacing
	
	Let $c_L(y_1, y_0, \delta) = \mathbbm{1}\{y_1 - y_0 < \delta\}$ and $P_1$, $P_0$ have continuous cumulative distribution functions $F_1(y) = P_1(Y_1 \leq y)$ and $F_0(y) = P_0(Y_0 \leq y)$ respectively. Then
	\begin{align}
		OT_{c_L}(P_1, P_0) = \inf_{\pi \in \Pi(P_1, P_0)} \int \mathbbm{1}\{y_1 - y_0 < \delta\} d\pi(y_1, y_0) = \sup_y \{F_1(y) - F_0(y - \delta)\}\label{Display: lemma, strong duality for CDF lower bound, corollary}
	\end{align}
\end{restatable}
\begin{proof}
	\singlespacing
	
	Continuity of the cumulative distribution functions implies $P_1(Y_1 = \delta + \min\{\mathcal{Y}_0\}) = P_0(Y_0 = \min\{\mathcal{Y}_0\}) = 0$, and thus
	\begin{align*}
		P_1(Y_1 < \delta + \min\{\mathcal{Y}_0\}) = P_1(Y_1 \leq \delta + \min\{\mathcal{Y}_0\}) - P_0(Y_0 \leq \min\{\mathcal{Y}_0\})
	\end{align*}
	Which takes the form $F_1(y) - F_0(y - \delta)$ for $y = \delta + \min\{\mathcal{Y}_0\}$. It follows that 
	\begin{equation*}
		\max\left\{\sup_y \{F_1(y) - F_0(y - \delta)\}, P_1(Y_1 < \min\{\mathcal{Y}_0\} + \delta)\right\} = \sup_y \{F_1(y) - F_0(y - \delta)\}
	\end{equation*}
	and lemma \ref{Lemma: duality special case CDF} gives the result. 
\end{proof}

\begin{restatable}[]{lemma}{lemmaSpecialCaseCDFComplementWeakInequality}
	\label{Lemma: duality special case CDF complement}
	\singlespacing
	
	Let $c_H(y_1, y_0, \delta) = \mathbbm{1}\{y_1 - y_0 > \delta\}$. Then
	\begin{align}
		OT_{c_H}(P_1, P_0) &= \inf_{\pi \in \Pi(P_1, P_0)} \int \mathbbm{1}\{y_1 - y_0 > \delta\} d\pi(y_1, y_0) \notag \\
		&= \max\left\{\sup_y \{P_1([y, \infty)) - P_0([y - \delta, \infty))\}, P_1((\max\{\mathcal{Y}_0\} + \delta, \infty))\right\} \label{Display: lemma, strong duality for CDF complement lower bound}
	\end{align}
\end{restatable}
\begin{proof} 
	\singlespacing
	
	The proof is similar to that of lemma \ref{Lemma: duality special case CDF}. Let $C = \{y_1 - y_0 > \delta\}$. Apply theorem \ref{Theorem: strong duality for indicator costs} and lemma \ref{Lemma: c-concave functions, indicator costs, c-conjugate pairs generated from indicators} to find that 
	\begin{align*}
		OT_{c_L}(P_1, P_0) = \max\{\sup_{A \in \mathcal{A}} P_1(Y_1 \in A^{CC}) - P_0(Y_0 \in A^C), P_1(Y_1 \in C_{1m})\}
	\end{align*}
	where 
	\begin{gather*}
		A^C = \left\{y_0 \in \mathcal{Y}_0 \; ; \; \exists y_1 \in A, \; (y_1, y_0) \not \in C \right\}, \quad A^{CC} = \left\{y_1 \in \mathcal{Y}_1 \; ; \; \forall y_0 \in \mathcal{Y}_0 \setminus A^C, \; (y_1, y_0) \in C \right\}, \\
		C_{1m} = \left\{y_1 \in \mathcal{Y}_1 \; ; \; \forall y_0 \in \mathcal{Y}_0, \; (y_1, y_0) \in C\right\}.
	\end{gather*}
	and $\mathcal{A}$ is the collection of closed, nonempty subsets of $\mathcal{Y}_1$ such that $A^C \neq \mathcal{Y}_0$. 
	
	Consider $\sup_{A \in \mathcal{A}} P_1(Y_1 \in A^{CC}) - P_0(Y_0 \in A^C)$. Let $A \in \mathcal{A}$ and $\varphi(y_1) = \mathbbm{1}_A(y_1)$, and notice that 
	\begin{align*}
		A^C &= \{y \in \mathcal{Y}_0 \; ; \; \exists y_1 \in A, \; (y_1, y_0) \not \in C\} = \{y_0 \in \mathcal{Y}_0 \; ; \; y_0 \geq \min\{A\} - \delta\}, \\
		A^{CC} &= \left\{y_1 \in \mathcal{Y}_1 \; ; \; \forall y_0 \in \mathcal{Y}_0 \setminus A^C, \; y_1 - y_0 < \delta \right\} = \left\{y_1 \in \mathcal{Y}_1 \; ; \; y_1 \geq \min\{A\}\right\}
	\end{align*}
	Where as in the proof of lemma \ref{Lemma: duality special case CDF}, $A^C \neq \mathcal{Y}_0$ implies $\inf\{A\} > -\infty$ and so $\inf\{A\} = \min\{A\}$ because $A$ is closed. Thus
	\begin{align*}
		J(\varphi^{cc}, \varphi^c) &= P_1(Y_1 \in A^{CC}) - P_0(Y_0 \in A^c) \\
		&= P_1(Y_1 \geq \min\{A\}) - P_0(Y_0 \geq \min\{A\} - \delta)
	\end{align*}
	which takes the form $P_1([y, \infty)) - P_0([y - \delta, \infty))$ for $y = \min\{A\}$. 
	
	Now consider $P_1(Y_1 \in C_{1m})$, and notice that 
	\begin{align*}
		C_{1m} &= \left\{y_1 \in \mathcal{Y}_1 \; ; \; \forall y_0 \in \mathcal{Y}_0, \; (y_1, y_0) \in C\right\} = \left\{y_1 \in \mathcal{Y}_1 \; ; \; \forall y_0 \in \mathcal{Y}_0, \; y_1 - y_0 > \delta\right\} \\
		&= \left\{y_1 \in \mathcal{Y}_1 \; ; \; \forall y_0 \in \mathcal{Y}_0, \; y_1 > \max\{\mathcal{Y}_0\} + \delta\right\}
	\end{align*}
	Thus $P_1(Y_1 \in C_{1m}) = P_1(Y_1 > \max\{\mathcal{Y}_0\} + \delta)$. The result follows.
\end{proof}

\begin{restatable}[]{corollary}{corollarySpecialCaseCDFWeakInequality}
	\label{Lemma: duality special case CDF complement, corollary}
	\singlespacing
	
	Let $c_H(y_1, y_0, \delta) = \mathbbm{1}\{y_1 - y_0 > \delta\}$ and $P_1$, $P_0$ have continuous cumulative distribution functions $F_1(y) = P_1(Y_1 \leq y)$ and $F_0(y) = P_0(Y_0 \leq y)$ respectively. Then
	\begin{align}
		OT_{c_L}(P_1, P_0) = \inf_{\pi \in \Pi(P_1, P_0)} \int \mathbbm{1}\{y_1 - y_0 > \delta\} d\pi(y_1, y_0) = \sup_y \{F_0(y - \delta) - F_1(y)\}\label{Display: lemma, strong duality for CDF complement lower bound, corollary}
	\end{align}
\end{restatable}
\begin{proof}
	\singlespacing
	
	Continuity of the cumulative distribution functions implies that for any $y$,
	\begin{align*}
		P_1([y, \infty)) - P_0([y - \delta, \infty)) &= P_1((y, \infty)) - P_0((y - \delta, \infty)) \\
		&= (1 - F_1(y)) - (1 - F_0(y - \delta)) \\
		&= F_0(y - \delta) - F_1(y)
	\end{align*}
	and furthermore,
	\begin{align*}
		P_1(Y_1 > \delta + \max\{\mathcal{Y}_0\}) &= 1 - F_1(\delta + \min\{\mathcal{Y}_0\}) - (1 - F_0(\max\{\mathcal{Y}_0\}) \\
		&= F_0(\max\{\mathcal{Y}_0\}) - F_1(\delta + \max\{\mathcal{Y}_0\})
	\end{align*}
	equals $F_0(y - \delta) - F_1(y)$ for $y = \max\{\mathcal{Y}_0\} + \delta$. Finally, lemma \ref{Lemma: duality special case CDF complement} gives
	\begin{align*}
		OT_{c_H}(P_1, P_0) &= \max\left\{\sup_y \{P_1([y, \infty)) - P_0([y - \delta, \infty))\}, P_1((\max\{\mathcal{Y}_0\} + \delta, \infty))\right\} \\
		&= \sup_y \{F_0(y - \delta) - F_1(y)\}
	\end{align*}
\end{proof}

%% file: appendix/OTJointPO_appendix_differentiation.tex
\section{Appendix: miscellaneous lemmas}
\label{Appendix: misc lemmas}

\subsection{Continuity}
\label{Appendix: misc lemmas, subsection continuity}

\begin{restatable}[Continuity of maps between bounded function spaces]{lemma}{lemmaContinuityMapsBetweenBoundedFunctionSpaces}
	\label{Lemma: bounded function spaces, continuity with f: R^K -> R^M}
	\singlespacing
	
	Let $f : \mathbbm{D}_f \subseteq \mathbb{R}^K \rightarrow \mathbb{R}^M$ be uniformly continuous. Define the subset of bounded functions on $T$ taking values in $\mathbb{D}_f$:
	\begin{equation*}
		\ell^\infty(T, \mathbb{D}_f) = \left\{g : T \rightarrow \mathbb{R}^K \; ; \; g(t) \in \mathbb{D}_f, \; \sup_{t \in T} \lVert g(t) \rVert < \infty\right\} \subseteq \ell^\infty(T)^K
	\end{equation*}
	Let $F : \ell^\infty(T, \mathbb{D}_f) \rightarrow \ell^\infty(T)^M$ be defined pointwise as $F(g)(t) = f(g(t))$. Then $F$ is uniformly continuous.
\end{restatable}
\begin{proof}
	\singlespacing
	
	To see that $F : \ell^\infty(T, \mathbb{D}_f) \rightarrow \ell^\infty(T)^M$ is well defined, recall that uniform continuity of $f$ implies $f$ is bounded on bounded sets. Since $\{g(t) \; ; \; t \in T\}$ is bounded for any $g \in \ell^\infty(T, \mathbb{D}_f)$, this implies $\sup_t \lVert f(g(t)) \rVert < \infty$ and hence $F(g) \in \ell^\infty(T)^M$. 
		
	To see uniform continuity of $F$, let $\varepsilon > 0$ and use uniform continuity of $f$ to choose $\delta > 0$ such that for all $x, \tilde{x} \in \mathbb{D}_f$, 
	\begin{equation*}
		\lVert x - \tilde{x} \rVert < \delta \implies \lVert f(x) - f(\tilde{x}) \rVert < \varepsilon/2
	\end{equation*}
	Now let $g, \tilde{g} \in \ell^\infty(T, \mathbb{D}_f)$ satisfy $\lVert g - \tilde{g} \rVert_T = \sup_{t \in T} \lVert g(t) - \tilde{g}(t) \rVert < \delta$. Then $\lVert g(t) - \tilde{g}(t) \rVert < \delta$ for all $t \in T$, and hence $\lVert f(g(t)) - f(\tilde{g}(t)) \rVert < \varepsilon/2$ for all $t \in T$, and therefore
	\begin{equation*}
		\lVert F(g) - F(\tilde{g}) \rVert_T = \sup_{t \in T} \lVert f(g(t)) - f(\tilde{g}(t)) \rVert \leq \frac{\varepsilon}{2} < \varepsilon
	\end{equation*}
	which completes the proof.
\end{proof}

\begin{restatable}[]{corollary}{corollaryContinuityMapsBetweenBoundedFunctionSpaces}
	\label{Lemma: bounded function spaces, continuity with f: R^K -> R^M, corollary}
	\singlespacing
	
	Let $f : \mathbbm{D}_f \subseteq \mathbb{R}^K \rightarrow \mathbb{R}^M$ be continuous and bounded on bounded subsets of $\mathbb{D}_f$. Let $g_0 \in \ell^\infty(T, \mathbb{D}_f)$ where $\ell^\infty(T, \mathbb{D}_f)$ is as defined in lemma \ref{Lemma: bounded function spaces, continuity with f: R^K -> R^M}. Suppose that for some $\delta > 0$,
	\begin{equation*}
		g(T)^\delta \equiv \left\{x \in \mathbb{R}^K \; ; \; \inf_{t \in T} \lVert g_0(t) - x \rVert \leq \delta\right\}
	\end{equation*}
	is a subset of $\mathbb{D}_f$. Then $F : \ell^\infty(T, \mathbb{D}_f) \rightarrow \ell^\infty(T)^M$ defined pointwise by $F(g)(t) = f(g(t))$ is continuous at $g_0$.
\end{restatable}
\begin{proof}
	\singlespacing
	
	For any $g \in \ell^\infty(T, \mathbb{D}_f)$, we have $F(g) \in \ell^\infty(T)^M$ because $\{x \; ; \; x = g(t) \text{ for some } t \in T\}$ is bounded and $f$ is bounded on bounded subsets.
	
	Let $\{g_n\}_{n=1}^\infty \subseteq \ell^\infty(T, \mathbb{D}_f)$ be such that $g_n \rightarrow g_0$ in $\ell^\infty(T)^K$. It suffices to show that $F(g_n) \rightarrow F(g_0)$ in $\ell^\infty(T)^M$. Let $\tilde{f} : g(T)^\delta \rightarrow \mathbb{R}^M$ be the restriction of $f$ to $g_0(T)^\delta$; i.e., $\tilde{f}(x) = f(x)$. Note that because $g_0(T)^\delta$ is a closed and bounded subset of $\mathbb{R}^K$, it is compact, and hence $\tilde{f}$ is uniformly continuous by the Heine-Cantor theorem. Apply lemma \ref{Lemma: bounded function spaces, continuity with f: R^K -> R^M} to find that 
	\begin{align*}
		&\tilde{F} : \ell^\infty(T, g(T)^\delta) \rightarrow \ell^\infty(T)^M, &&\tilde{F}(g)(t) = \tilde{f}(g(t)) = f(g(t))
	\end{align*}
	is continuous. Since $g_n \rightarrow g_0$ in $\ell^\infty(T)^K$, there exists $N$ such that for all $n \geq N$, $\lVert g_n - g_0 \rVert_T= \sup_{t \in T} \lVert g_n(t) - g_0(t) \rVert < \delta$. Let $\tilde{g}_k = g_{k+N}$. Notice that $\tilde{g}_k(T) = \left\{x \in \mathbb{R}^K \; ; \; x = g_k(t) \text{ for some } t \in T\right\} \subseteq g_0(T)^\delta$, and hence $\tilde{g}_k \in \ell^\infty(T, g_0(T)^\delta)$. Continuity of $\tilde{F}$ and $\tilde{g}_k \rightarrow g_0$ implies $\tilde{F}(\tilde{g}_k) \rightarrow \tilde{F}(\tilde{g}_0)$. Thus
	\begin{align*}
		0 &= \lim_{k \rightarrow \infty} \lVert \tilde{F}(\tilde{g}_k) - \tilde{F}(g_0) \rVert_T = \lim_{k \rightarrow \infty} \lVert F(g_{k + N}) - F(g_0) \rVert_T = \lim_{n \rightarrow \infty} \lVert F(g_n) - F(g_0) \rVert_T 
	\end{align*}
	which completes the proof.
\end{proof}

\begin{restatable}[Uniform continuity of restricted sup]{lemma}{lemmaContinuityOfRestrictedSup}
	\label{Lemma: bounded function spaces, uniform continuity of restricted sup}
	\singlespacing
	
	For any set $X$, subset $A \subseteq X$, and bounded real-valued functions $f,g \in \ell^\infty(X)$, 
	\begin{equation}
		\left\lvert \sup_{x \in A} f(x) - \sup_{x \in A} g(x) \right\rvert \leq \sup_{x \in A} \lvert f(x) - g(x) \rvert \label{Display: lemma, bounded function spaces, continuity of restricted sup}
	\end{equation}
	and therefore $\sigma_A : \ell^\infty(X) \rightarrow \mathbb{R}$ given by $\sigma_A(f) = \sup_{x \in A} f(x)$ is uniformly continuous.
\end{restatable}
\begin{proof}
	\singlespacing
	Observe that
	\begin{equation*}
		\sup_{x \in A} f(x) - \sup_{x \in A} g(x) \leq \sup_{x \in A} \{f(x) - g(x)\} \leq \sup_{x \in A} \lvert f(x) - g(x) \rvert 
	\end{equation*}
	and 
	\begin{equation*}
		-\left[\sup_{x \in A} f(x) - \sup_{x \in A} g(x) \right] = \sup_{x \in A} g(x) - \sup_{x \in A} f(x) \leq \sup_{x \in A} \{g(x) - f(x)\} \leq \sup_{x \in A} \lvert f(x) - g(x) \rvert 
	\end{equation*}
	Together these inequalities imply
	\begin{equation*}
		-\sup_{x \in A} \lvert f(x) - g(x) \rvert \leq \sup_{x \in A} f(x) - \sup_{x \in A} g(x) \leq \sup_{x \in A} \lvert f(x) - g(x) \rvert
	\end{equation*}
	which is equivalent to \eqref{Display: lemma, bounded function spaces, continuity of restricted sup}. 
	
	To see uniform continuity, let $\varepsilon > 0$ and choose $\delta = \varepsilon$. Whenever $\lVert f - g \rVert_X = \sup_{x \in X} \lvert f(x) - g(x) \rvert < \delta$,
	\begin{equation*}
		\left\lvert \sigma_A(f) - \sigma_A(g) \right\rvert = \left\lvert \sup_{x \in A} f(x) - \sup_{x \in A} g(x) \right\rvert \leq \sup_{x \in A} \lvert f(x) - g(x) \rvert \leq \sup_{x \in X} \lvert f(x) - g(x) \rvert < \delta = \varepsilon
	\end{equation*}
	which completes the proof.
\end{proof}

\subsection{Differentiability}
\label{Appendix: misc lemmas, subsection differentiability}

This appendix reviews definitions and various facts related to Hadamard directional differentiability. The following definitions can be found in \cite{fang2019inference}.

Let $\mathbb{D}$, $\mathbb{E}$ be Banach spaces (complete, normed, vector spaces), and $\phi : \mathbb{D}_\phi \subseteq \mathbb{D} \rightarrow \mathbb{E}$. 
\begin{enumerate}[label=(\roman*)]
	\singlespacing
	
	\item $\phi$ is (fully) \textbf{Hadamard differentiable} at $x_0 \in \mathbb{D}_\phi$ tangentially to $\mathbb{D}_0 \subseteq \mathbb{D}$ if there exists a continuous linear map $\phi_{x_0}' : \mathbb{D}_0 \rightarrow\mathbb{E}$ such that
	\begin{equation*}
		\lim_{n \rightarrow \infty} \left\lVert \frac{\phi(x_0 + t_n h_n) - \phi(x_0)}{t_n} - \phi_{x_0}'(h) \right\rVert_{\mathbb{E}} = 0
	\end{equation*}
	for all sequences $\{h_n\}_{n=1}^\infty \subseteq \mathbb{D}$ and $\{t_n\}_{n=1}^\infty\subseteq \mathbb{R}$ such that $h_n \rightarrow h \in \mathbb{D}_0$ and $t_n \rightarrow 0$ as $n \rightarrow \infty$, and $x_0 + t_n h_n \in \mathbb{D}_\phi$ for all $n$. 
	
	\item $\phi$ is \textbf{Hadamard directionally differentiable} at $x_0 \in \mathbb{D}_\phi$ tangentially to $\mathbb{D}_0 \subseteq \mathbb{D}$ if there exists a continuous map $\phi_{x_0}' : \mathbb{D}_0 \rightarrow\mathbb{E}$ such that
	\begin{equation*}
		\lim_{n \rightarrow \infty} \left\lVert \frac{\phi(x_0 + t_n h_n) - \phi(x_0)}{t_n} - \phi_{x_0}'(h) \right\rVert_{\mathbb{E}} = 0
	\end{equation*}
	for all sequences $\{h_n\}_{n=1}^\infty \subseteq \mathbb{D}$ and $\{t_n\}_{n=1}^\infty\subseteq \mathbb{R}_+$ such that $h_n \rightarrow h \in \mathbb{D}_0$ and $t_n \downarrow 0$ as $n \rightarrow \infty$, and $x_0 + t_n h_n \in \mathbb{D}_\phi$ for all $n$.
\end{enumerate}

\cite{fang2019inference} proposition 2.1 shows that linearity is the key property distinguishing directional and full Hadamard differentiability. Specifically, if $\phi$ is Hadamard directionally differentiable at $x_0$ tangentially to a subspace $\mathbb{D}_0$, and $\phi_{x_0}'$ is linear, then $\phi$ is in fact fully Hadamard differentiable at $x_0$ tangentially to $\mathbb{D}_0$.

Hadamard directional differentiability obeys the chain rule. 
\begin{restatable}[Chain rule]{lemma}{lemmaChainRule}
	\label{Lemma: Hadamard differentiability, chain rule}
	\singlespacing
	
	Let $\mathbb{D}_1$, $\mathbb{D}_2$, and $\mathbb{E}$ be Banach spaces and $\phi_1 :  \mathbb{D}_{\phi_1} \subseteq \mathbb{D}_1 \rightarrow \mathbb{D}_2$, $\phi_2 : \mathbb{D}_{\phi_2} \subseteq \mathbb{D}_2 \rightarrow \mathbb{E}$ be functions. Suppose 
	
	\begin{enumerate}[label=(\roman*)]
		\item $\phi_1(\mathbb{D}_{\phi_1}) = \left\{y \in \mathbb{D}_2 \; ; \; y = \phi_1(x) \text{ for some } x \in \mathbb{D}_{\phi_1}\right\} \subseteq \mathbb{D}_{\phi_2}$, \label{Lemma: Hadamard differentiability, chain rule, assumption: chain rule composition well defined}
		\item $\phi_1$ is Hadamard directionally differentiable at $x_0 \in \mathbb{D}_{\phi_1}$ tangentially to $\mathbb{D}_1^T \subseteq \mathbb{D}_1$, with derivative $\phi_{1,x_0}'(h)$, and \label{Lemma: Hadamard differentiability, chain rule, assumption: chain rule inner function differentiable}
		\item $\phi_2$ is Hadamard directionally differentiable  at $\phi_1(x_0) \in \mathbb{D}_{\phi_2}$ tangentially to $\mathbb{D}_2^T \subseteq \mathbb{D}_2$, with derivative $\phi_{2,\phi_1(x_0)}'(h)$ \label{Lemma: Hadamard differentiability, chain rule, assumption: chain rule outer function differentiable}
	\end{enumerate}
	Let $\mathbb{D}^T = \left\{x \in \mathbb{D}_1^T \; ; \; \phi_{1,x_0}'(x) \in \mathbb{D}_2^T\right\}$. The composition function
	\begin{align*}
		&\phi : \mathbb{D}_{\phi_1} \rightarrow \mathbb{E}, && \phi(x) = \phi_2(\phi_1(\theta))
	\end{align*}
	is Hadamard directionally differentiable at $x_0$ tangentially to $\mathbb{D}^T$, with 
	\begin{align*}
		&\phi_{x_0}' : \mathbb{D}^T \rightarrow \mathbb{E}, &&\phi_{x_0}'(h) = \phi_{2,\phi_1(x_0)}'(\phi_{1,x_0}'(h))
	\end{align*}
\end{restatable}
\begin{proof}
	\singlespacing
	That $\phi$ is well defined is clear from assumption \ref{Lemma: Hadamard differentiability, chain rule, assumption: chain rule composition well defined}. To show its Hadamard directional differentiability, let $\{h_n\}_{n=1}^\infty \subseteq \mathbb{D}_{\phi_1}$ and $\{t_n\}_{n=1}^\infty \subseteq \mathbb{R}_+$ be such that $h_n \rightarrow h \in \mathbb{D}^T$, $t_n \downarrow 0$, and $x_0 + t_n h_n \in \mathbb{D}_{\phi_1}$ for all $n$. Assumption \ref{Lemma: Hadamard differentiability, chain rule, assumption: chain rule inner function differentiable} implies that 
	\begin{equation}
		\lim_{n\rightarrow\infty} \left\lVert \frac{\phi_1(x_0 + t_n h_n) - \phi_1(x_0)}{t_n} - \phi_{1,x_0}'(h)\right\rVert_{\mathbb{D}_2} = 0 \label{Display: lemma proof, chain rule phi1 differentiable}
	\end{equation}
	Let $g_n = \frac{1}{t_n}\left[\phi_1(x_0 + t_nh_n) - \phi_1(x_0)\right]$, $g = \phi_{1,x_0}'(h)$, and notice that \eqref{Display: lemma proof, chain rule phi1 differentiable} implies $g_n \rightarrow g$ in $\mathbb{D}_2$.
	
	Assumption \ref{Lemma: Hadamard differentiability, chain rule, assumption: chain rule composition well defined} implies $\phi_1(x_0) + t_n g_n = \phi_1(x_0 + t_n h_n) \in \mathbb{D}_{\phi_2}$ for each $n$, and the definition of $\mathbb{D}^T$ implies $g \in \mathbb{D}_2^T$. Assumption \ref{Lemma: Hadamard differentiability, chain rule, assumption: chain rule outer function differentiable} implies that
	\begin{equation}
		\lim_{n \rightarrow \infty} \left\lVert \frac{\phi_2(\phi_1(x_0) + t_n g_n) - \phi_2(\phi_1(x_0))}{t_n} - \phi_{2,\phi_1(x_0)}'(g)\right\rVert_{\mathbb{E}} = 0 \label{Display: lemma proof, chain rule phi2 differentiable}
	\end{equation}
	Substitute $\phi_2(\phi_1(x_0) + t_n g_n) = \phi_2(\phi_1(x_0 + t_n h_n))$, and $g = \phi_{1,x_0}'(h)$, into \eqref{Display: lemma proof, chain rule phi2 differentiable} to find 
	\begin{equation*}
		\lim_{n \rightarrow \infty} \left\lVert \frac{\phi_2(\phi_1(x_0 + t_n h_n)) - \phi_2(\phi_1(x_0))}{t_n} - \phi_{2,\phi_1(x_0)}'(\phi_{1,x_0}'(h))\right\rVert_{\mathbb{E}} = 0
	\end{equation*}
	which completes the proof.
\end{proof}
\begin{remark}
	\label{Remark: hadamard directional differentiability chain rule constraints}
	\singlespacing
	
	When defining and differentiating composition of functions, the outer function's properties determine restrictions that must be placed on the inner function to ensure the composition function is well defined and differentiable. 
	
	A familiar example of this is that the domain of the ``inner function'' $\phi_1$ may need to be restricted to ensure the composition map is well defined. For a simple example, $x^3$ is well defined and differentiable for any $x \in \mathbb{R}$, but $\log(x^3)$ is only well defined (and differentiable) for $x \in (0,\infty)$.
	
	A less familiar example shows up only when considering Hadamard differentiability tangentially to a set. The tangent spaces of each function jointly determine the tangent space of the derivative of the composition map. 
\end{remark}

The next lemma shows that Hadamard directionally differentiable functions can be ``stacked''.

\begin{restatable}[Stacking Hadamard differentiable functions]{lemma}{lemmaStackHadamardDifferentiability}
	\label{Lemma: Hadamard differentiability, stacking functions}
	\singlespacing
	
	Let $\mathbb{D}$, $\mathbb{E}_1$, and $\mathbb{E}_2$ be Banach spaces, and $\mathbb{D}_\phi \subseteq \mathbb{D}$. Suppose $\phi^{(1)} : \mathbb{D}_\phi \rightarrow \mathbb{E}_1$ and $\phi^{(2)} : \mathbb{D}_\phi \rightarrow \mathbb{E}_2$ are Hadamard directionally differentiable tangentially to $\mathbb{D}_0 \subseteq \mathbb{D}$ at $x_0 \in \mathbb{D}_\phi$ with derivatives $\phi_{x_0}^{(1)\prime}: \mathbb{D}_0 \rightarrow \mathbb{E}_1$ and $\phi_{x_0}^{(2)\prime} : \mathbb{D}_0 \rightarrow \mathbb{E}_2$. Define
	\begin{align*}
		&\phi : \mathbb{D}_\phi \rightarrow \mathbb{E}_1 \times \mathbb{E}_2, &&\phi(x) =
		\begin{pmatrix}
			\phi^{(1)}(x), &\phi^{(2)}(x) 
		\end{pmatrix}
	\end{align*}
	Then $\phi$ is Hadamard directionally differentiable tangentially to $\mathbb{D}_0$ at $x_0$, with derivative 
	\begin{align*}
		&\phi_{x_0}' : \mathbb{D}_0 \rightarrow \mathbb{E}_1 \times \mathbb{E}_2, &&\phi_{x_0}'(h) = 
		\begin{pmatrix}
			\phi_{x_0}^{(1)\prime}(h), &\phi_{x_0}^{(2)\prime}(h)
		\end{pmatrix}
	\end{align*}
\end{restatable}
\begin{proof}
	Hadamard directional differentiability of $\phi^{(1)}$ and $\phi^{(2)}$ tangentially to $\mathbb{D}_0$ at $x_0$ implies that for any sequences $\{h_n\}_{n=1}^\infty \subseteq \mathbb{D}$ and $\{t_n\} \subseteq \mathbb{R}_+$ such that $h_n \rightarrow h \in \mathbb{D}_0$, $t_n \downarrow 0$, and $x_0 + t_n h_n \in \mathbb{D}_\phi$ for all $n$,
	\begin{align*}
		&\lim_{n \rightarrow \infty} \left\lVert \frac{\phi^{(1)}(x_0 + t_n h_n) - \phi^{(1)}(x_0)}{t_n} - \phi_{x_0}^{(1)\prime}(h) \right\rVert_{\mathbb{E}_1} = 0 , \text{ and } \\
		&\lim_{n \rightarrow \infty} \left\lVert \frac{\phi^{(2)}(x_0 + t_n h_n) - \phi^{(2)}(x_0)}{t_n} - \phi_{x_0}^{(2)\prime}(h) \right\rVert_{\mathbb{E}_2} = 0
	\end{align*}
	Since $\lVert (e_1, e_2) - (\tilde{e}_1, \tilde{e}_2) \rVert_{\mathbb{E}_1 \times \mathbb{E}_2} = \lVert e_1 - \tilde{e}_1 \rVert_{\mathbb{E}_1} + \lVert e_2 - \tilde{e}_2 \rVert_{\mathbb{E}_2}$ metricizes $\mathbb{E}_1 \times \mathbb{E}_2$ (\cite{aliprantis2006infinite} lemma 3.3), we have 
	\begin{align*}
		&\left\lVert \frac{\phi(x_0 + t_n h_n) - \phi(x_0)}{t_n} - \phi_{x_0}'(h) \right\rVert_{\mathbb{E}_1 \times \mathbb{E}_2} \\
		&\hspace{1 cm} = \left\lVert \frac{\begin{pmatrix} \phi^{(1)}(x_0 + t_n h_n), &\phi^{(2)}(x_0 + t_n h_n) \end{pmatrix} - \begin{pmatrix}\phi^{(1)}(x_0), & \phi^{(2)}(x_0)\end{pmatrix}}{t_n} - \begin{pmatrix} \phi_{x_0}^{(1)}(h), &\phi_{x_0}^{(2)}(h) \end{pmatrix}\right\rVert_{\mathbb{E}_1 \times \mathbb{E}_2} \\
		&\hspace{1 cm} = \left\lVert \begin{pmatrix} \frac{\phi^{(1)}(x_0 + t_n h_n) - \phi^{(1)}(x_0)}{t_n} - \phi_{x_0}^{(1)\prime}(h), &\frac{\phi^{(2)}(x_0 + t_n h_n) - \phi^{(2)}(x_0 + t_n h_n)}{t_n} - \phi_{x_0}^{(2)\prime} \end{pmatrix}\right\rVert_{\mathbb{E}_1 \times \mathbb{E}_2} \\
		&\hspace{1 cm} = \left\lVert \frac{\phi^{(1)}(x_0 + t_n h_n) - \phi^{(1)}(x_0)}{t_n} - \phi_{x_0}^{(1)\prime}(h) \right\rVert_{\mathbb{E}_1} + \left\lVert \frac{\phi^{(2)}(x_0 + t_n h_n) - \phi^{(2)}(x_0)}{t_n} - \phi_{x_0}^{(2)\prime}(h) \right\rVert_{\mathbb{E}_2}
	\end{align*}
	Taking the limit as $n\rightarrow \infty$ gives the result.
\end{proof}

\subsubsection{Hadamard differentiability in bounded function spaces}

It is common to ``rearrange'' Donsker sets; i.e. view them not as scalar-valued but vector-valued with each coordinate occuring over a particular subset of functions (see \cite{van2000asymptotic} p. 270). The following lemma shows that one direction of the equivalence can be viewed as an application of the delta method.

\begin{restatable}[Rearranging Donsker sets]{lemma}{lemmaRearrangingDonskerSets}
	\label{Lemma: Hadamard differentiability, rearranging Donsker sets}
	\singlespacing
	
	Suppose $\mathcal{F} = \mathcal{F}_1 \cup \ldots \cup \mathcal{F}_K$ is $P$-Donsker, and $\sqrt{n}(\mathbb{P}_n - P) \overset{L}{\rightarrow} \mathbb{G}$ in $\ell^\infty(\mathcal{F})$. The map $\phi : \ell^\infty(\mathcal{F}) \rightarrow \ell^\infty(\mathcal{F}_1) \times \ldots \times \ell^\infty(\mathcal{F}_K)$ defined pointwise by 
	\begin{equation*}
		\phi(g)(f_1, \ldots, f_K) = (g(f_1), \ldots, g(f_K))
	\end{equation*}
	is fully Hadamard differentiable at any $P \in \ell^\infty(\mathcal{F})$ tangentially to $\ell^\infty(\mathcal{F})$, and is its own derivative:
	\begin{align*}
		&\phi_P' : \ell^\infty(\mathcal{F}) \rightarrow \ell^\infty(\mathcal{F}_1) \times \ldots \times \ell^\infty(\mathcal{F}_K), &&\phi_P'(h) = \phi(h)
	\end{align*}
	and hence
	\begin{align*}
		&\sqrt{n}(\phi(\mathbb{P}_n) - \phi(P)) \overset{L}{\rightarrow} \phi(\mathbb{G}) &&\text{ in } \ell^\infty(\mathcal{F}_1) \times \ldots \times \ell^\infty(\mathcal{F}_K)
	\end{align*}
\end{restatable}
\begin{proof}
	\singlespacing
	
	The map $\phi$ is linear; let $a, b \in \mathbb{R}$ and $g, h \in \ell^\infty(\mathcal{F})$ and notice that for any $(f_1, \ldots, f_K) \in \mathcal{F}_1 \times \ldots \times \mathcal{F}_K$, 
	\begin{align*}
		\phi(a g + b h)(f_1, \ldots, f_K) &= ((a g + b h)(f_1), \ldots, (a g + b h)(f_K)) \\
		&= (ag(f_1) + bh(f_1), \ldots, ag(f_K) + bh(f_K)) \\
		&= a (g(f_1), \ldots, g(f_K)) + b(h(f_1) \ldots, h(f_K)) \\
		&= a \phi(g)(f_1,\ldots, f_K) + b \phi(h)(f_1, \ldots, f_K) \\
		&= (a \phi(g) + b \phi(h))(f_1,\ldots, f_K)
	\end{align*}
	hence $\phi(a g + b h) = (a \phi(g) + b \phi(h))$, as these functions agree on all of $\mathcal{F}_1 \times \ldots \times \mathcal{F}_K$. 
	
	Next observe that $\phi$ is continuous. Recall that the product topology on $\ell^\infty(\mathcal{F}_1) \times \ldots \times \ell^\infty(\mathcal{F}_K)$ is generated by the norm
	\begin{align*}
		\lVert (g_1, \ldots, g_K) - (h_1, \dots, h_K) \rVert_{\mathcal{F}_1 \times \ldots \times \mathcal{F}_K} = \max\{\lVert g_1 - h_1 \rVert_{\mathcal{F}_1}, \ldots, \lVert g_K - h_K \rVert_{\mathcal{F}_K}\}
	\end{align*}
	see \cite{aliprantis2006infinite} lemma 3.3. Thus
	\begin{align*}
		\lVert \phi(g) - \phi(h) \rVert_{\mathcal{F}_1 \times \ldots \times \mathcal{F}_K} &= \max\left\{\sup_{f_1 \in \mathcal{F}_1} \lvert g(f_1) - h(f_1) \rvert, \ldots, \sup_{f_K \in \mathcal{F}_K} \lvert g(f_K) - h(f_K) \rvert\right\} \\
		&= \lVert g - h \rVert_\mathcal{F}
	\end{align*}
	and hence $\phi$ is continuous. 
	
	Since $\phi$ is linear and continuous, it is (fully) Hadamard differentiable at any point tangentially to $\ell^\infty(\mathcal{F})$ and is its own Hadamard derivative; indeed, for an: for all sequences $h_n \rightarrow h \in \ell^\infty(\mathcal{F})$ and $t_n \downarrow 0 \in \mathbb{R}$, one has $g + t_n h_n \in \ell^\infty(\mathcal{F})$ and
	\begin{equation*}
		\lim_{n \rightarrow \infty} \left\lVert \frac{\phi(g + t_n h_n) - \phi(g)}{t_n} - \phi(h) \right\rVert_{\mathcal{F}_1 \times \ldots \times \mathcal{F}_K} = \lim_{n \rightarrow \infty}  \left\lVert \phi(h_n) - \phi(h) \right\rVert_{\mathcal{F}_1 \times \ldots \times \mathcal{F}_K} = 0
	\end{equation*}
	
	Finally, since $\sqrt{n}(\mathbb{P}_n - P) \overset{L}{\rightarrow} \mathbb{G}$ in $\ell^\infty(\mathcal{F})$, the functional delta method (\cite{van2000asymptotic} theorem 20.8) implies $\sqrt{n}(\phi(\mathbb{P}_n) - \phi(P)) \overset{L}{\rightarrow} \phi(\mathbb{G})$ in $\ell^\infty(\mathcal{F}_1) \times \ldots \times \ell^\infty(\mathcal{F}_K)$.
\end{proof}

Although the following lemma and its corollary are stated for functions taking values in $\mathbb{R}$, by combining it with lemma \ref{Lemma: Hadamard differentiability, stacking functions} a similar result can be obtained for functions taking values in $\mathbb{R}^M$, similar to the setting of lemma \ref{Lemma: bounded function spaces, continuity with f: R^K -> R^M}. Compare \cite{vaart1997weak} lemma 3.9.25.

\begin{restatable}[Hadamard differentiability of maps between bounded function spaces]{lemma}{lemmaDifferentiability}
	\label{Lemma: Hadamard differentiability, maps between bounded function spaces}
	\singlespacing

	Let $f : \mathbb{D}_f \subseteq \mathbb{R}^K \rightarrow \mathbb{R}$. Suppose that
	\begin{enumerate}
		\item $f$ is continuously differentiable, and 
		\item the gradient of $f$,
		\begin{align*}
			&\nabla f : \mathbb{D}_f \rightarrow \mathbb{R}^K, &&\nabla f(x) =  \begin{pmatrix} \frac{\partial f}{\partial x_1}(x) & \ldots & \frac{\partial f}{\partial x_K}(x) \end{pmatrix}^\intercal,
		\end{align*}
		is uniformly continuous.
	\end{enumerate}
	Define the subset of $\ell^\infty(T)^K$ taking values in $\mathbb{D}_f$,
	\begin{equation*}
		\ell^\infty(T, \mathbb{D}_f) = \left\{g : T \rightarrow \mathbb{R}^K \; ; \; g(t) \in \mathbb{D}_f, \; \sup_{t \in T} \lVert g(t) \rVert < \infty\right\} \subseteq \ell^\infty(T)^K
	\end{equation*}
	and the subset of $\ell^\infty(T, \mathbb{D}_f)$ such that composition with $f$ defines a bounded function:
	\begin{equation*}
		\ell_f^\infty(T, \mathbb{D}_f) = \left\{g \in \ell^\infty(T, \mathbb{D}_f) \; ; \; \sup_{t \in T} \lvert f(g(t)) \rvert < \infty\right\}
	\end{equation*}

	Then $F : \ell_f^\infty(T, \mathbb{D}_f) \rightarrow \ell^\infty(T)$ defined pointwise with $F(g)(t) = f(g(t))$ is (fully) Hadamard differentiable tangentially to $\ell^\infty(T)^K$ at any $g_0 \in \ell_f^\infty(T, \mathbb{D}_f)$, with derivative $F_{g_0}' : \ell^\infty(T)^K \rightarrow \ell^\infty(T)$ given pointwise by
	\begin{align*}
		F_{g_0}'(h)(t) = \left[\nabla f(g_0(t))\right]^\intercal h(t) = \sum_{k=1}^K \frac{\partial f}{\partial x_k}(g_0(t)) h_k(t)
	\end{align*}
\end{restatable}
\begin{proof}
	\singlespacing
	
	The domain of $\ell_f^\infty(T, \mathbb{D}_f)$ ensures that $F : \ell_f^\infty(T, \mathbb{D}_f) \rightarrow \ell^\infty(T)$ is well defined.
	
	Let $\{h_n\}_{n=1}^\infty \subseteq \ell^\infty(T)^K$ and $\{r_n\}_{n=1}^\infty \subseteq \mathbb{R}$ such that $h_n \rightarrow h \in \ell^\infty(T)^K$, $r_n \rightarrow 0$, and $g_0 + r_n h_n \in \ell_f^\infty(T, \mathbb{D}_f)$ for each $n$. For each $n$ and each $t \in T$, apply the mean value theorem to find $\lambda_n(t) \in (0,1)$ such that $g_n(t) \coloneqq \lambda_n(t) (g_0(t) + r_n h_n(t)) + (1-\lambda_n(t)) g_0(t)$ satisfying\footnote{
		The mean value theorem being invoked here is the standard result: for any $x, \tilde{x} \in \mathbb{D}_f$, let $g_{x, \tilde{x}} : [0,1] \rightarrow \mathbb{R}$ be given by $g_{x, \tilde{x}}(\lambda) = f(\lambda \tilde{x} + (1-\lambda)x)$. Then $g_{x, \tilde{x}}(0) = f(x)$ and $g_{x, \tilde{x}}(1) = f(\tilde{x})$, and the mean value theorem tells us that there exists $\lambda \in (0,1)$ such that
		\begin{equation*}
			f(\tilde{x}) - f(x) = g_{x, \tilde{x}}(1) - g_{x, \tilde{x}}(0) = g_{x, \tilde{x}}'(\lambda) (1-0) = \left[\nabla f(\lambda \tilde{x} + (1-\lambda)x)\right]^\intercal (\tilde{x} - x) 
		\end{equation*}
	}
	\begin{align*}
		f(x_0(t) + r_n h_n(t)) - f(x_0(t)) &= \left[\nabla f(g_n(t))\right]^\intercal (x_0(t) + r_n h_n(t) - x_0(t)) \\
		&= r_n \left[\nabla f(g_n(t))\right]^\intercal h_n(t) 
	\end{align*}
	Use this to see that for all $n$ and all $t \in T$,
	\begin{align*}
		&\left\lvert \frac{f(g_0(t) + r_n h_n(t)) - f(g_0(t))}{r_n} - \nabla f(g_0(t))^\intercal h(t)\right\rvert = \left\lvert \nabla f(g_n(t))^\intercal h_n(t) - \nabla f(g_0(t))^\intercal h(t) \right\rvert \\
		&\hspace{1 cm} \leq \left\lvert \nabla f(g_n(t))^\intercal h_n(t) - \nabla f(g_0(t))^\intercal h_n(t)\right\rvert  + \left\lvert\nabla f(g_0(t))^\intercal h_n(t)- \nabla f(g_0(t))^\intercal h(t) \right\rvert \\
		&\hspace{1 cm} \leq \lVert \nabla f(g_n(t)) - \nabla f(g_0(t)) \rVert \times \lVert h_n(t) \rVert + \lVert \nabla f(g_0(t)) \rVert \times \lVert h_n(t) - h(t) \rVert
	\end{align*}
	where the first inequality is by the triangle inequality and the second by Cauchy-Schwarz in $\mathbb{R}^K$. It follows that
	\begin{align}
		&\sup_{t \in T} \left\lvert \frac{f(g_0(t) + r_n h_n(t)) - f(g_0(t))}{r_n} - \nabla f(g_0(t))^\intercal h(t)\right\rvert \notag \\
		&\hspace{2 cm} \leq \sup_{t \in T} \lVert \nabla f(g_n(t)) - \nabla f(g_0(t)) \rVert \times \sup_{t \in T} \lVert h_n(t) \rVert \label{Display: lemma proof, maps between bounded function spaces, inequality term 1}\\
		&\hspace{3 cm} + \sup_{t \in T} \lVert \nabla f(g_0(t)) \rVert \times \sup_{t \in T} \lVert h_n(t) - h(t) \rVert \label{Display: lemma proof, maps between bounded function spaces, inequality term 2}
	\end{align}

	Consider the term in \eqref{Display: lemma proof, maps between bounded function spaces, inequality term 1}. Recall that for some $\lambda_n(t) \in (0,1)$,
	\begin{align*}
		g_n(t) &= \lambda_n(t) (g_0(t) + r_n h_n(t)) + (1-\lambda_n(t)) g_0(t) \\
		&= \lambda_n(t)r_n h_n(t) + g_0(t) 
	\end{align*}
	and so
	\begin{equation*}
		\lVert g_n - g_0 \rVert_T = \sup_{t \in T} \lVert \lambda_n(t)r_n h_n(t) \rVert \leq \lvert r_n \rvert \times \sup_{t \in T} \lVert h_n(t) \rVert \rightarrow 0 
	\end{equation*}
	where the limit claim follows from $\sup_{t \in T} \lVert h_n(t) \rVert = \lVert h_n \rVert_T \rightarrow \lVert h \rVert_T < \infty$ (implying $\{\sup_{t \in T} \lVert h_n(t) \rVert\}_{n=1}^\infty$ is bounded) and $r_n \rightarrow 0$. Thus $g_n \rightarrow g_0$ in $\ell^\infty(T)^K$. Using this and uniform continuity of $\nabla f : \mathbb{D}_f \rightarrow \mathbb{R}^K$, lemma \ref{Lemma: bounded function spaces, continuity with f: R^K -> R^M} implies $\nabla f(g_n) \rightarrow \nabla f(g_0)$ in $\ell^\infty(T)^K$, i.e.
	\begin{equation*}
		\lVert \nabla f(g_n) - \nabla f(g_0) \rVert_T = \sup_{t \in T} \lVert \nabla f(g_n(t)) - \nabla f(g_0(t)) \rVert \rightarrow 0 
	\end{equation*}
	Using once again that $\{\sup_{t \in T} \lVert h_n(t) \rVert\}_{n=1}^\infty$ is bounded, this implies
	\begin{equation}
		\lim_{n\rightarrow \infty} \sup_{t \in T} \lVert \nabla f(g_n(t)) - \nabla f(g_0(t)) \rVert \times \sup_{t \in T} \lVert h_n(t) \rVert = 0 \label{Display: lemma proof, maps between bounded function spaces, inequality term 1 limit}
	\end{equation}
	
	Now consider the term in \eqref{Display: lemma proof, maps between bounded function spaces, inequality term 2}. $\sup_{t \in T} \lVert \nabla f(g_0(t)) \rVert < \infty$ because $\lVert \nabla f(\cdot)\rVert$ is uniformly continuous and $\sup_{t \in T} \lVert g_0(t) \rVert < \infty$, just as in the proof of lemma \ref{Lemma: bounded function spaces, continuity with f: R^K -> R^M}. Furthermore, $\lim_{n \rightarrow \infty} \sup_{t \in T} \lVert h_n(t) - h(t) \rVert = 0$, so 
	\begin{equation}
		\lim_{n \rightarrow \infty} \sup_{t \in T} \lVert \nabla f(g_0(t)) \rVert \times \sup_{t \in T} \lVert h_n(t) - h(t) \rVert = 0 \label{Display: lemma proof, maps between bounded function spaces, inequality term 2 limit}
	\end{equation}
	Combining \eqref{Display: lemma proof, maps between bounded function spaces, inequality term 1} through \eqref{Display: lemma proof, maps between bounded function spaces, inequality term 2 limit} we obtain
	\begin{align*}
		\lim_{n \rightarrow \infty} \sup_{t \in T} \left\lvert \frac{f(g_0(t) + r_n h_n(t)) - f(g_0(t))}{r_n} - \nabla f(g_0(t))^\intercal h(t)\right\rvert = 0
	\end{align*}
	which concludes the proof.
\end{proof}

\begin{remark}
	\label{Remark: why do a corollary to Hadamard differentiability of maps between bounded function spaces?}
	\singlespacing
	
	Lemma \ref{Lemma: Hadamard differentiability, maps between bounded function spaces} specifies the domain of $F$ as $\ell_f^\infty(T, \mathbb{D}_f) = \left\{g \in \ell^\infty(T, \mathbb{D}_f) \; ; \; \sup_{t \in T} \lvert f(g(t)) \rvert < \infty\right\}$. It is often straightforward to clarify the space $\ell_f^\infty(T, \mathbb{D}_f)$ in particular cases; for example, $\ell_f^\infty(T, \mathbb{D}_f) = \ell^\infty(T, \mathbb{D}_f)$ if $f$ satisfies any one of the following:
	\begin{enumerate*}[label=(\roman*)]
		\item $f$ is bounded,
		\item $f$ is Lipschitz, or 
		\item $f$ is bounded on bounded subsets (e.g., $f(x) = x$ is bounded on bounded subsets)
	\end{enumerate*}
	See also lemma \ref{Lemma: Hadamard differentiability, conditional distributions}.
	
	Lemma \ref{Lemma: Hadamard differentiability, maps between bounded function spaces} requires $\nabla f(\cdot)$ be uniformly continuous, but this often stronger than necessary. When hoping to argue $F : \ell_f^\infty(T, \mathbb{D}_f) \rightarrow \ell^\infty(T)$ defined pointwise with $F(g)(t) = f(g(t))$ is (fully) Hadamard differentiable at $g_0 \in \ell_f^\infty(T, \mathbb{D}_f)$, it suffices that $f$ is continuously differentiable on a closed set slightly larger than the (bounded) range of $g_0$. Compactness of this expanded range and the fact that continuous functions on compact sets are uniformly continuous allow us to apply the preceding lemma. This logic is formalized in the following corollary.
\end{remark}

\begin{restatable}[Hadamard differentiability of maps between bounded function spaces, corollary]{corollary}{corollaryHadamardDifferentiabilityMapsBetweenBoundedFunctionSpaces}
	\label{Lemma: Hadamard differentiability, maps between bounded function spaces, corollary}
	\singlespacing
	
	Let $f : \mathbb{D}_f \subseteq \mathbb{R}^K \rightarrow \mathbb{R}$ be continuously differentiable. 
	
	Define the subset of $\ell^\infty(T)^K$ taking values in $\mathbb{D}_f$,
	\begin{equation*}
		\ell^\infty(T, \mathbb{D}_f) = \left\{g : T \rightarrow \mathbb{R}^K \; ; \; g(t) \in \mathbb{D}_f, \; \sup_{t \in T} \lVert g(t) \rVert < \infty\right\} \subseteq \ell^\infty(T)^K
	\end{equation*}
	and the subset of $\ell^\infty(T, \mathbb{D}_f)$ such that composition with $f$ defines a bounded function:
	\begin{equation*}
		\ell_f^\infty(T, \mathbb{D}_f) = \left\{g \in \ell^\infty(T, \mathbb{D}_f) \; ; \; \sup_{t \in T} \lvert f(g(t)) \rvert < \infty\right\}
	\end{equation*}
	
	Let $g_0 \in \ell_f^\infty(T, \mathbb{D}_f)$, and suppose that for some $\delta > 0$,
	\begin{equation*}
		g_0(T)^\delta \equiv \left\{x \in \mathbb{R}^K \; ; \; \inf_{t \in T} \lVert x - g_0(t) \rVert \leq \delta\right\} \subseteq \mathbb{D}_f.
	\end{equation*}
	Then $F : \ell_f^\infty(T, \mathbb{D}_f) \rightarrow \ell^\infty(T)$ defined pointwise by $F(g)(t) = f(g(t))$ is (fully) Hadamard differentiable at $g_0$ tangentially to $\ell^\infty(T)^K$, with derivative $F_{g_0}' : \ell^\infty(T)^K \rightarrow \ell^\infty(T)$ given pointwise by
	\begin{align*}
		F_{g_0}'(h)(t) = \left[\nabla f(g_0(t))\right]^\intercal h(t) = \sum_{k=1}^K \frac{\partial f}{\partial x_k}(g_0(t)) h_k(t)
	\end{align*}
\end{restatable}
\begin{proof}
	\singlespacing
	
	Let $\tilde{f} : g_0(T)^\delta \rightarrow \mathbb{R}$ be the restriction of $f$ to $g_0(T)^\delta$. Note that $\tilde{f}$ is continuously differentiable on the compact $g_0(T)^\delta \subseteq \mathbb{R}^K$, hence $\nabla \tilde{f}$ is in fact uniformly continuous by the Heine-Cantor theorem. Apply lemma \ref{Lemma: Hadamard differentiability, maps between bounded function spaces} to find that 
	\begin{align*}
		&\tilde{F} : \ell_f^\infty(T, g_0(T)^\delta) \rightarrow \ell^\infty(T), &&\tilde{F}(g)(t) = \tilde{f}(g(t)) = f(g(t))
	\end{align*}
	is (fully) Hadamard differentiable at $g_0$, with derivative $\tilde{F}_{g_0}' : \ell^\infty(T)^K \rightarrow \ell^\infty(T)$ given pointwise by $\tilde{F}_{g_0}'(h)(t) = \left[\nabla f(g_0(t))\right]^\intercal h(t)$. By definition, this means that for any sequences $\{\tilde{h}_n\}_{n=1}^\infty \subseteq \ell^\infty(T)^K$ and $\{\tilde{r}_n\}_{n=1}^\infty \subseteq \mathbb{R}$ such that $\tilde{h}_n \rightarrow \tilde{h} \in \ell^\infty(T)^K$, $\tilde{r}_n \rightarrow 0$, and $g_0 + \tilde{r}_n \tilde{h}_n \in \ell^\infty(T, g_0(T)^\delta)$ for all $n$, 
	\begin{align}
		&\lim_{n \rightarrow \infty} \left\lVert \frac{\tilde{F}(g_0 + \tilde{r}_n \tilde{h}_n) - \tilde{F}(g_0)}{\tilde{r}_n} - F_{g_0}'(\tilde{h}) \right\rVert_T = 0 \label{Proof display: Cor: Hadamard differentiability of maps between bounded function spaces, restricted map is differentiable}
	\end{align}
	
	Let $\{h_n\}_{n=1}^\infty \subseteq \ell^\infty(T)^K$, $\{r_n\}_{n=1}^\infty \subseteq \mathbb{R}$ be such that $h_n \rightarrow h \in \ell^\infty(T)^K$, $r_n \rightarrow 0$, and $g_0 + r_n h_n \in \ell^\infty(T, \mathbb{D}_f)$ for all $n$. It suffices to show that 
	\begin{align*}
		&\left\lVert \frac{F(g_0 + r_nh_n) - F(g_0)}{r_n} - F_{g_0}'(h) \right\rVert_T \\
		&\hspace{1 cm} = \sup_{t \in T} \left\lvert \frac{f(g_0(t) + r_n h_n(t)) - f(g_0(t))}{r_n} - \left[\nabla f(g_0(t))\right]^\intercal h(t) \right\rvert
	\end{align*}
	has limit zero.
	
	Notice that $g_0 + r_n h_n \rightarrow g_0$ in $\ell^\infty(T)^K$, so for some $N$ we have that for all $n \geq N$, $\lVert g_0 + r_n h_n - g_0 \rVert_T = r_n \sup_{t \in T} \lVert h_n \rVert < \delta$. It follows that for $k \in \mathbb{N}$, $g_0 + r_{k + N} h_{k+N} \in \ell^\infty(T, g_0(T)^\delta)$ and hence $\tilde{r}_k = r_{k + N}$ and $\tilde{h}_k = h_{k + N}$ are sequences for which \eqref{Proof display: Cor: Hadamard differentiability of maps between bounded function spaces, restricted map is differentiable} applies. Therefore,
	\begin{align*}
		\lim_{n \rightarrow \infty} \left\lVert \frac{F(g_0 + r_nh_n) - F(g_0)}{r_n} - F_{g_0}'(h) \right\rVert_T &= \lim_{k \rightarrow \infty} \left\lVert \frac{F(g_0 + r_{k + N}h_{k + N}) - F(g_0)}{r_{k + N}} - F_{g_0}'(h) \right\rVert_T \\
		&= \lim_{k \rightarrow \infty} \left\lVert \frac{\tilde{F}(g_0 + \tilde{r}_k \tilde{h}_k) - \tilde{F}(g_0)}{\tilde{r}_k} - F_{g_0}'(h) \right\rVert_T \\
		&= 0
	\end{align*}
	Where the second equality follows from $\tilde{F}(g_0 + \tilde{r}_k \tilde{h}_k) = F(g_0 + r_{k+N} h_{k+N})$ and $\tilde{F}(g_0) = F(g_0)$. 
\end{proof}

The following lemma is lemma S.4.9 from \cite{fang2019inference}, but the authors state it for a metric space. The same proof works to show that statement holds in semimetric spaces as well.\footnote{
	Some useful facts about semimetrics:
	\begin{enumerate*}[label=(\roman*)]
		\item A semimetric defines a topology that is first countable (\cite{aliprantis2006infinite} pp. 70, 72), but this topology is not second countable or Hausdorff. The limits of sequences are not guaranteed to be unique.
		\item In a semimetric space, sequences still characterize the closures of sets, as well as continuity and semicontinuity of functions (\cite{aliprantis2006infinite}, theorems  2.40 and 2.42 on pp. 42-43).
		\item A subset of a semimetric space is compact if and only if it is complete and totally bounded (\cite{vaart1997weak}, footnote on p. 17).
	\end{enumerate*}
} The statement and proof are included here for completeness.

\begin{restatable}[Hadamard directional differentiability of supremum]{lemma}{lemmaHadamardDirectionalDifferentiabilityOfSupremum}
	\label{Lemma: Hadamard differentiability, supremum of bounded function}
	\singlespacing
	
	(\cite{fang2019inference} lemma S.4.9)
	
	Let $(\textbf{A}, d)$ be a compact semimetric space, $A$ a compact subset of $\textbf{A}$, and 
	\begin{align*}
		&\psi : \ell^\infty(\textbf{A}) \rightarrow \mathbb{R}, &&\psi(p) = \sup_{a \in A} \; p(a)
	\end{align*}
	Then $\psi$ is Hadamard directionally differentiable at any $p_0 \in \mathcal{C}(\textbf{A}, d)$ tangentially to $\mathcal{C}(\textbf{A}, d)$. $\Psi_A(p_0) = \argmax_{a \in A} p_0(a)$ is nonempty, and the directional derivative is given by
	\begin{align*}
		&\psi_{p_0}' : \mathcal{C}(\textbf{A}, d) \rightarrow \mathbb{R}, &&\psi_{p_0}'(p) = \sup_{a \in \Psi_A(p_0)} p(a)
	\end{align*}
\end{restatable}
\begin{proof}
	
	\singlespacing
	
	Let $p_0 \in \mathcal{C}(\textbf{A})$. Since $A$ is compact, $\Psi_A(p_0) = \argmax_{a \in A} p_0$ is nonempty (\cite{aliprantis2006infinite} theorem 2.43). Let $\{p_n\}_{n=1}^\infty \subseteq \ell^\infty(\textbf{A})$ and $\{t_n\}_{n=1}^\infty \subseteq \mathbb{R}_+$ such that $p_n \rightarrow p \in \mathcal{C}(\textbf{A})$ and $t_n \downarrow 0$. Notice that 
	\begin{align}
		&\left\lvert \frac{\psi(p_0 + t_n p_n) - \psi(p_0)}{t_n} -  \psi_{p_0}'(p)\right\rvert \notag \\
		&\hspace{1 cm} = \left\lvert \frac{\sup_{a \in A} \left\{p_0(a) + t_np_n(a)\right\} - \sup_{a \in A} p_0(a)}{t_n} - \sup_{a \in \Psi_A(p_0)} p(a)\right\rvert \notag \\
		&\hspace{1 cm} = \left\lvert \frac{\sup_{a \in A} \left\{p_0(a) + t_np_n(a)\right\} - \sup_{a \in \Psi_A(p_0)} p_0(a)}{t_n} - \sup_{a \in \Psi_A(p_0)} p(a)\right\rvert \notag\\
		&\hspace{1 cm} \leq \left\lvert \frac{\sup_{a \in \Psi_A(p_0)}\left\{p_0(a) + t_n p(a)\right\} - \sup_{a \in \Psi_A(p_0)}p_0(a)}{t_n} - \sup_{a \in \Psi_A(p_0)} p(a)\right\rvert \label{Display: lemma proof, Hadamard differentiability, supremum of bounded function, inequality term 1} \\
		&\hspace{2 cm} + \left\lvert \frac{\sup_{a \in A} \left\{p_0(a) + t_n p_n(a)\right\} - \sup_{a \in A}\left\{p_0(a) + t_n p(a)\right\}}{t_n} \right\rvert \label{Display: lemma proof, Hadamard differentiability, supremum of bounded function, inequality term 2} \\
		&\hspace{2 cm} + \left\lvert \frac{\sup_{a \in A} \left\{p_0(a) + t_n p (a)\right\} - \sup_{a \in \Psi_A(p_0)}\left\{p_0(a) + t_n p(a)\right\}}{t_n} \right\rvert \label{Display: lemma proof, Hadamard differentiability, supremum of bounded function, inequality term 3}
	\end{align}
	
	First, consider \eqref{Display: lemma proof, Hadamard differentiability, supremum of bounded function, inequality term 1}. Notice that $p_0$ is flat on $\Psi_A(p_0)$, so
	\begin{align}
		&\left\lvert \frac{\sup_{a \in \Psi_A(p_0)}\left\{p_0(a) + t_n p(a)\right\} - \sup_{a \in \Psi_A(p_0)}p_0(a)}{t_n} - \sup_{a\in \Psi_F(p_0)} p(a)\right\rvert \notag \\
		&\hspace{1 cm} = \left\lvert \sup_{a \in \Psi_A(p_0)} p(a) - \sup_{a\in \Psi_A(p_0)} p(a)\right\rvert = 0 \label{Display: lemma proof, Hadamard differentiability, supremum of bounded function, inequality term 1 limit is zero}
	\end{align}
	Next consider \eqref{Display: lemma proof, Hadamard differentiability, supremum of bounded function, inequality term 2}. Since $p_0 + t_n p_n$ and $p_0 + t_n p$ are elements of $\ell^\infty(\textbf{A})$, lemma \ref{Lemma: bounded function spaces, uniform continuity of restricted sup} implies
	\begin{align}
		&\left\lvert \frac{\sup_{a \in A} \left\{p_0(a) + t_n p_n(a)\right\} - \sup_{a \in A}\left\{p_0(a) + t_n p(a)\right\}}{t_n} \right\rvert \notag \\
		&\hspace{1 cm} \leq \sup_{a \in A} \lvert p_n(a) - p(a) \rvert \leq \lVert p_n - p \rVert_{\textbf{A}} \rightarrow 0 \label{Display: lemma proof, Hadamard differentiability, supremum of bounded function, inequality term 2 limit is zero}
	\end{align}
	Now consider \eqref{Display: lemma proof, Hadamard differentiability, supremum of bounded function, inequality term 3}. Notice that 
	\begin{align*}
		&\varphi : \mathcal{C}(\textbf{A}) \rightrightarrows \textbf{A}, &&\varphi(g) = A
	\end{align*}
	is a trivially continuous correspondence with nonempty, compact values. Furthermore,
	\begin{align*}
		&\Gamma_{p_0} : \mathcal{C}(\textbf{A}) \times \textbf{A} \rightarrow \mathbb{R}, &&\Gamma_{p_0}(g, a) = p_0(a) + g(a)
	\end{align*}
	is continuous on all of $\mathcal{C}(\textbf{A}) \times \textbf{A}$.\footnote{
		To see this, recall that the topology of $\mathcal{C}(\textbf{A}) \times \textbf{A}$ is generated by the semimetric $\rho((g, a), (\tilde{g},\tilde{a})) = \max\left\{\lVert g - \tilde{g} \rVert_{\textbf{A}}, d(a, \tilde{a})\right\}$ (\cite{aliprantis2006infinite} lemma 3.3). Let $\varepsilon > 0$ and $g \in \mathcal{C}(\textbf{A})$. Note that each element of $\mathcal{C}(\textbf{A})$ is a continuous function defined on a compact set, and is hence uniformly continuous by the Heine-Cantor theorem (lemma \ref{Lemma: heine cantor in semimetric spaces}). Use uniform continuity of $p_0$ and $g$ to choose $\delta_{p_0}, \delta_g > 0$ such that $d(a, \tilde{a}) < \delta_{p_0}$ implies $\lvert p_0(a) - p_0(\tilde{a}) \rvert < \varepsilon / 3$, and $d(a, \tilde{a}) < \delta_g$ implies $\lvert g(a) - g(\tilde{a}) \rvert < \varepsilon / 3$. Let $\delta = \min\{\delta_{p_0}, \delta_g, \varepsilon/3\}$, and notice that $\rho((g, a), (\tilde{g},\tilde{a})) < \delta$ implies $\lvert p_0(a) - p_0(\tilde{a}) \rvert < \varepsilon/3$, $\lvert g(a) - g(\tilde{a}) \rvert < \varepsilon/3$, and $\lVert g - \tilde{g} \rVert_{\textbf{A}} < \varepsilon/3$, and hence $\lvert \Gamma_{p_0}(g,a) - \Gamma_{p_0}(\tilde{g}, \tilde{a}) \rvert = \left\lvert p_0(a) + g(a) - p_0(\tilde{a}) - \tilde{g}(\tilde{a}) \right\rvert \leq \lvert p_0(a) - p_0(\tilde{a}) \rvert + \lvert g(a) - g(\tilde{a}) \rvert + \lvert g(\tilde{a}) - \tilde{g}(\tilde{a}) \rvert < \varepsilon$.
	}
	Thus $\sup_{a \in A} \left\{p_0(a) + g(a) \right\} = \max_{a \in \varphi(g)} \Gamma_{p_0}(g,a)$ satisfies the conditions of the Berge Maximum Theorem (\cite{aliprantis2006infinite} theorem 17.31), implying the argmax corresondence $\Phi : \mathcal{C}(\textbf{A}) \rightrightarrows \textbf{A}$ given by $\Phi(g) = \Psi_A(p_0 + g)$ is compact valued and upper hemicontinuous. 
	
	Let $\Psi_A(p_0)^{\epsilon} = \left\{a \in A \; ; \; \inf_{\tilde{a} \in \Psi_A(p_0)} d(a, \tilde{a}) \leq \epsilon\right\}$. Upper hemicontinuity and $\lVert t_n p \rVert_{\textbf{A}} \rightarrow 0$ implies that there exists $\delta_n \downarrow 0$ such that $\Psi_A(p_0 + t_n p) \subseteq \Psi_A(p_0)^{\delta_n}$.\footnote{
		To see this, recall the definition of $\Phi$ being upper hemicontinuous (uhc) given in \cite{aliprantis2006infinite}, definition 17.2: $\Phi$ is uhc at $g$ if for every neighborhood $U$ of $\Phi(g)$, the upper inverse image 
		\begin{equation*}
			\Phi^u(U) = \left\{h \in \mathcal{C}(\textbf{A}) \; ; \; \Phi(h) \subseteq U\right\}
		\end{equation*}
		is a neighborhood of $g$, i.e. $g$ is in the interior of $\Phi^u(U)$, so there exists $\eta > 0$ such that $\lVert g - \tilde{g} \rVert_{\textbf{A}} < \eta$ implies $\tilde{g} \in \Phi^u(U)$, and hence $\Phi(\tilde{g}) \subseteq U$. Since $\Psi_A$ is uhc and $\Psi_A(p_0)^\epsilon$ is a neighborhood of $\Psi_A(p_0)$, whenever $\lVert p_0 + t_n p - p_0 \rVert_{\textbf{A}} = t_n \lVert p \rVert_{\textbf{A}} < \epsilon$ we have that $\Psi(p_0 + t_n p) \subseteq \Psi_A(p_0)^\epsilon$. 
		
		Let 
		\begin{align*}
			\delta_n = \max_{a \in \Psi_A(p_0 + t_n p)} \min_{\tilde{a} \in \Psi_A(p_0)} d(a,\tilde{a})
		\end{align*}
		The inner $\min$ is attained because $d$ is continuous and the feasible set is compact. $a \mapsto \max_{\tilde{a} \in \Psi_A(p_0)} \{-d(a, \tilde{a})\}$ is continuous by the Maximum Theorem (\cite{aliprantis2006infinite} theorem 17.31), which implies $a \mapsto \min_{\tilde{a} \in \Psi_A(p_0)} d(a, \tilde{a})$ is continuous. The outer $\max$ is then attained because the feasible set is compact. Notice that $\delta_n$ that $\Psi_A(p_0 + t_n p) \subseteq \Psi_A(p_0)^{\delta_n}$. 
		
		Suppose for contradiction that $\delta_n \not\rightarrow 0$. Then there exists $\epsilon > 0$ and a subsequence $\{\delta_{n'}\}_{n'=1}^\infty$ such that $\delta_{n'} \geq \epsilon$ for all $n'$, which implies $\Psi_A(p_0 + t_{n'} p) \not\subseteq \Psi_A(p_0)^{\epsilon/2}$ for all $n'$. $\Psi_A(p_0)^{\epsilon/2}$ is a neighborhood of $\Psi_A(p_0) = \Phi(0)$, and $\Phi$ is uhc at $0$, hence $\Phi^u(\Psi_A(p_0)^{\epsilon/2})$ is a neighborhood of $0 \in \mathcal{C}(\textbf{A})$. So for some $\eta > 0$, $\lVert t_{n'} p \rVert < \eta$ implies $\Phi(t_{n'} p) = \Psi_A(p_0 + t_{n'} p) \subseteq \Psi_A(p_0)^{\epsilon/2}$. Since $t_{n'} p \rightarrow 0 \in \mathcal{C}(\textbf{A})$, there exist $n'$ with $\lVert t_{n'} p \rVert_{\textbf{A}} < \eta$, and for such $n'$ we have $\Psi_A(p_0 + t_{n'} p) \subseteq \Psi_A(p_0)^{\epsilon/2}$ by upper hemicontinuity. This is the desired contradiction; therefore $\delta_n \rightarrow 0$. 
		
		If $\delta_n$ does not converge monotonically to zero, set $\tilde{\delta}_n = \sup\{\delta_k \; ; \; k \geq n\}$. Note that $\tilde{\delta}_n \downarrow 0$ and $\tilde{\delta}_n \geq \delta_n$, the latter of which implies $\Psi_A(p_0 + t_n p) \subseteq \Psi_A(p_0)^{\tilde{\delta}_n}$.
	}

	It follows that 
	\begin{align*}
		&\left\lvert \frac{\sup_{a \in A} \left\{p_0(a) + t_n p (a)\right\} - \sup_{a \in \Psi_A(p_0)}\left\{p_0(a) + t_n p(a)\right\}}{t_n} \right\rvert \\
		&\hspace{1 cm} = \frac{1}{t_n} \left(\sup_{a \in \Psi_A(p_0)^{\delta_n}} \left\{p_0(a) + t_n p (a)\right\} - \sup_{a \in \Psi_A(p_0)}\left\{p_0(a) + t_n p(a)\right\}\right) 
	\end{align*}
	Let $a_{s,n} \in \argmax_{a \in \Psi_A(p_0)} \left\{p_0(a) + t_n p(a)\right\}$, which is nonempty because $\Psi_A(p_0)$ is compact and $p_0(a) + t_n p(a)$ is continuous. Let $a_{b,n} \in \Psi_A(p_0 + t_n p) \subseteq \Psi_A(p_0)^{\delta_n}$ satisfy $d(a_{b,n}, a_{s,n}) \leq \delta_n$, and notice that $\Psi_A(p_0 + t_n p) \subseteq \Psi_A(p_0)^{\delta_n}$ implies $\sup_{a \in \Psi_A(p_0)^{\delta_n}} \left\{p_0(a) + t_n p (a)\right\} = p_0(a_{b,n}) + t_n p (a_{b,n})$. So,
	\begin{align*}
		&\frac{1}{t_n} \left(\sup_{a \in \Psi_A(p_0)^{\delta_n}} \left\{p_0(a) + t_n p (a)\right\} - \sup_{a \in \Psi_A(p_0)}\left\{p_0(a) + t_n p(a)\right\}\right) \\
		&\hspace{1 cm} = p_0(a_{b,n}) + t_n p(a_{b,n}) - p_0(a_{s,n}) - t_n p(a_{s,n}) \\
		&\hspace{1 cm} \leq p(a_{b,n}) - p(a_{s,n})
	\end{align*}
	where the inequality follows because $a_{s,n}$ maximizes $p_0$ over $A$ while $a_{b,n}$ may not. Furthermore, $d(a_{b,n}, a_{s,n}) \leq \delta_n$ implies 
	\begin{equation*}
		p(a_{b,n}) - p(a_{s,n}) \leq \sup_{a, a' \in A, \; d(a,a') \leq \delta_n} \left\{p(a) - p(a')\right\}
	\end{equation*}
	and hence
	\begin{align}
		&\left\lvert \frac{\sup_{a \in A} \left\{p_0(a) + t_n p (a)\right\} - \sup_{a \in \Psi_A(p_0)}\left\{p_0(a) + t_n p(a)\right\}}{t_n} \right\rvert \notag \\
		&\hspace{1 cm} \leq \sup_{a, a' \in A, \; d(a,a') \leq \delta_n} \left\{p(a) - p(a')\right\} \notag \\
		&\hspace{1 cm} \rightarrow 0 \label{Display: lemma proof, Hadamard differentiability, supremum of bounded function, inequality term 3 limit is zero}
	\end{align}
	Where the limit claim follows from $p$ being a continuous function defined on a compact set, and so is in fact uniformly continuous by the Heine-Cantor theorem (lemma \ref{Lemma: heine cantor in semimetric spaces}). 
	
	To summarize, 
	\begin{align*}
		&\left\lvert \frac{\psi(p_0 + t_n p_n) - \psi(p_0)}{t_n} -  \psi_{p_0}'(p)\right\rvert \\
		&\hspace{1 cm} \leq \left\lvert \frac{\sup_{a \in \Psi_A(p_0)}\left\{p_0(a) + t_n p(a)\right\} - \sup_{a \in \Psi_A(p_0)}p_0(a)}{t_n} - \sup_{a\in \Psi_A(p_0)} p(a)\right\rvert \\
		&\hspace{2 cm} + \left\lvert \frac{\sup_{a \in A} \left\{p_0(a) + t_n p_n(a)\right\} - \sup_{a \in A}\left\{p_0(a) + t_n p(a)\right\}}{t_n} \right\rvert \\
		&\hspace{2 cm} + \left\lvert \frac{\sup_{a \in A} \left\{p_0(a) + t_n p (a)\right\} - \sup_{a \in \Psi_A(p_0)}\left\{p_0(a) + t_n p(a)\right\}}{t_n} \right\rvert
	\end{align*}
	along with \eqref{Display: lemma proof, Hadamard differentiability, supremum of bounded function, inequality term 1 limit is zero}, \eqref{Display: lemma proof, Hadamard differentiability, supremum of bounded function, inequality term 2 limit is zero}, and \eqref{Display: lemma proof, Hadamard differentiability, supremum of bounded function, inequality term 3 limit is zero} implies that $\psi$ is Hadamard directionally differentiable at any $p_0 \in \mathcal{C}(\textbf{A})$ tangentially to any $p \in \mathcal{C}(\textbf{A})$, with $\psi_{p_0}'(p) = \sup_{a \in \Psi_A(p_0)} p(a)$.
\end{proof}

\subsection{Other}
\label{Appendix: misc lemmas, subsection other}

The Heine-Cantor theorem is usually stated for metric spaces. As it is applied in the proof of lemma \ref{Lemma: Hadamard differentiability, supremum of bounded function} to a setting with semimetric spaces, the statement and standard proof are included here to make clear the result applies to semimetric spaces as well.

\begin{restatable}[Heine-Cantor theorem]{lemma}{lemmaHeineCantorSemimetric}
	\label{Lemma: heine cantor in semimetric spaces}
	\singlespacing
	
	Let $(X, d_X)$ and $(Y, d_Y)$ be semimetric spaces, $X$ compact, and $f : X \rightarrow Y$ continuous. Then $f$ is in fact uniformly continuous. 
\end{restatable}
\begin{proof}
	\singlespacing
	
	Let $\varepsilon > 0$. For each $x \in X$, use continuity of $f$ to choose $\delta_x$ such that 
	\begin{align*}
		d_X(x,x') < 2\delta_x \implies d_Y(f(x), f(x')) < \varepsilon/2
	\end{align*}
	Let $B_d(x) \subseteq X$ be the open ball of radius $d$ centered at $x$. Then $\bigcup_{x \in X} B_{\delta_x}(x)$ is an open cover of $X$. By compactness of $X$, there exists $x_1, \ldots, x_n$ such that $\bigcup_{i=1}^n B_{\delta_{x_i}}(x_i)$ covers $X$. Let $\delta = \min_{i \in \{1, \ldots, n\}} \delta_{x_i}$. As the minimum of a finite number of positive real numbers, we have $\delta > 0$. 
	
	Suppose $d_X(x,x') < \delta$. Since $\bigcup_{i=1}^n B_{\delta_{x_i}}(x_i)$ covers $X$, there exists $k \in \{1,\ldots, n\}$ such that $x \in B_{\delta_{x_k}}(x_k)$. Notice that 
	\begin{align*}
		d_X(x', x_k) &\leq d_X(x', x) + d_X(x, x_k) < \delta + \delta_{x_k} \leq 2\delta_{x_k}
	\end{align*}
	and thus $d_X(x,x') < \delta$ implies $d_X(x', x_k) < 2\delta_{x_k}$ and $d_X(x, x_k) < \delta_{x_k} < 2\delta_{x_k}$ for whichever $k$ is such that $x \in B_{\delta_{x_k}}(x_k)$. Then the definition of $\delta_{x_k}$ implies
	\begin{align*}
		d_Y(f(x), f(x')) &\leq d_Y(f(x), f(x_k)) + d_Y(f(x_k), f(x')) < \varepsilon/2 + \varepsilon/2 = \varepsilon
	\end{align*}
\end{proof}